\documentclass{article}
\usepackage{bm}
\usepackage[font=small]{caption}
\RequirePackage{amsthm,amsmath,amsfonts,amssymb}
\usepackage{omega}
\usepackage{parskip}
\usepackage[a4paper,left=2.5cm,right=2.5cm,top=2.5cm,bottom=2.5cm]{geometry}

\usepackage{bm}
\usepackage{natbib}
\usepackage[ruled]{algorithm2e}
\usepackage[dvipsnames]{xcolor}
\usepackage[colorlinks = true,
            linkcolor = NavyBlue,
            urlcolor  = NavyBlue,
            citecolor = NavyBlue,
            anchorcolor = NavyBlue]{hyperref} 

\usepackage[graphicx]{realboxes}
\usepackage{enumerate}
\usepackage{comment}
\RequirePackage{subfigure}
\RequirePackage{algorithmic}

\author{}
\numberwithin{equation}{section}
\theoremstyle{plain}
\newtheorem{rem}{Remark}[section]

\theoremstyle{definition}
\newtheorem{rema}{Remark}
\usepackage{chngcntr}

\newcommand{\ee}{\end{aligned} \end{equation}}
\newcommand{\eq}{\end{quote}}
\newcommand{\diag}{\mathrm{diag}}

\newcommand{\ep}{\end{parts}}

\newcommand{\bqp}{\begin{quote}\begin{parts}}

\newcommand{\epq}{\end{parts}\end{quote}}

\newcommand{\Rom}[1]{\text{\uppercase\expandafter{\romannumeral #1\relax}}}
\newcommand{\bee}{\begin{equation}\begin{aligned}}

\newcommand{\emm}{\end{bmatrix}}

\numberwithin{equation}{section}
\newcommand{\vertiii}[1]{{\vert\kern-0.25ex\vert\kern-0.25ex\vert #1 
    \vert\kern-0.25ex\vert\kern-0.25ex\vert}}
    
\newcommand\ma{\mathbf{A}}
\newcommand\mb{\mathbf{B}}
\newcommand\mc{\mathbf{C}}
\newcommand\md{\mathbf{D}}
\newcommand\me{\mathbf{E}}
\newcommand\mn{\mathbf{N}}
\newcommand\mh{\mathbf{H}}
\newcommand\mi{\mathbf{I}}
\newcommand\mq{\mathbf{Q}}

\newcommand\mr{\mathbf{R}}
\newcommand\ms{\mathbf{S}}

\newcommand\mv{\mathbf{V}}
\newcommand\mx{\mathbf{X}}
\newcommand\mz{\mathbf{Z}}
\newcommand\mw{\mathbf{W}}

\newcommand\mm{\mathbf{M}}

\newcommand\mt{\mathbf{T}}

\newcommand\mpp{\mathbf{P}}
\newcommand\mo{\mathbf{O}}
\newcommand\my{\mathbf{Y}}
\newcommand\muu{\mathbf{U}}
\newcommand\mSigma{\bm{\Sigma}}

\newcommand\mLambda{\bm{\Lambda}}

\usepackage{etoolbox}
\makeatletter
\csundef{listofalgorithms}
\makeatother

\usepackage{algorithm}

\title{\bf Chain-linked Multiple Matrix Integration via Embedding Alignment}

\author{Runbing Zheng %\thanks{
    \\
    Department of Applied Mathematics and Statistics, Johns Hopkins University\\
    and \\
    Minh Tang \\
    Department of Statistics, North Carolina State University}
    
\date{}

\begin{document}

\maketitle

{\fontsize{11pt}{20pt}
\selectfont

\begin{abstract}
Motivated by the increasing demand for multi-source data integration in various scientific fields, in this paper we study matrix completion in scenarios where the data exhibits certain block-wise missing structures -- specifically, where only a few noisy submatrices representing (overlapping) parts of the full matrix are available. We propose the Chain-linked Multiple Matrix Integration (CMMI) procedure to efficiently combine the information that can be extracted from these individual noisy submatrices. CMMI begins by deriving entity embeddings for each observed submatrix, then aligns these embeddings using overlapping entities between pairs of submatrices, and finally aggregates them to reconstruct the entire matrix of interest. We establish, under mild regularity conditions, entrywise error bounds and normal approximations for the CMMI estimates. Simulation studies and real data applications show that CMMI is computationally efficient and effective in recovering the full matrix, even when overlaps between the observed submatrices are minimal. 
%Suppose a complete low-rank matrix for some entities is not available, and there are multiple sources providing noisy submatrices of partial entities. 
%We propose an algorithm to integrate the multi-source data and recover the complete matrix. 
%We prove that the algorithm yields estimates whose entrywise fluctuations are normally distributed around the entries of the population matrix.
%Motivated by the increasing demand of multi-source data integration in various scientific fields, we study the problem of matrix completion based on multiple observed submatrices to deal with the block-wise missing structure.

\end{abstract}

\noindent%
{\it Keywords:} $2 \to \infty$ norm, normal approximations, matrix completion, data integration
\vfill

}

{\fontsize{11pt}{16.9pt}
\selectfont

\counterwithout{figure}{section}
\counterwithout{theorem}{section}
\counterwithout{assumption}{section}
\counterwithout{algorithm}{section}

\section{Introduction}

The development of large-scale data collection and sharing has sparked considerable research interests in integrating data from diverse sources to efficiently uncover underlying signals. This problem is especially pertinent in fields such as healthcare research \citep{zhou2021multi},
genomic data integration \citep{maneck2011genomic,tseng2015integrating,%zang2016high,
cai2016structured},
single-cell data integration \citep{stuart2019comprehensive,%argelaguet2021computational,
ma2023your},
and chemometrics \citep{mishra2021recent}.
In this paper we consider a formulation of the problem %\citep{stuart2019comprehensive,%argelaguet2021computational,
%ma2023your,zhou2021multi}, 
where each source $i$ corresponds to a partially observed submatrix $\mathbf{M}^{(i)}$ of some matrix $\mathbf{M}$, and the goal is to integrate these $\{ \mathbf{M}^{(i)}\}$ to recover $\mathbf{M}$ as accurately as possible. 

As a first motivating example, consider pointwise mutual information (PMI) constructed from different electronic healthcare records (EHR) datasets.
PMI quantifies the association between a pair of clinical concepts, and matrices representing these associations can be derived from co-occurrence summaries of various EHR datasets \citep{ahuja2020surelda,zhou2022multiview}. 
However, due to the lack of interoperability across healthcare systems \citep{rajkomar2018scalable}, different EHR data often involve non-identical concepts with limited overlap, resulting in substantial differences among their PMI matrices. 
The analysis of PMI matrices from different EHR datasets can thus be viewed as a 
multi-source matrix integration problem. Specifically, let $\mathcal{U}$ represent some concept set and 
suppose there is a symmetric PMI matrix $\mpp\in\mathbb{R}^{N\times N}$ associated with $\mathcal{U}$, where $N := |\mathcal{U}|$. For the $i$th EHR, we denote its clinical concept by $\mathcal{U}_i\subset \mathcal{U}$ and let $n_i:=|\mathcal{U}_i|$.
The PMI matrix derived from the $i$th EHR, $\ma^{(i)}\in\mathbb{R}^{n_i\times n_i}$, then corresponds to the 
principal submatrix of $\mpp$ associated with $\mathcal{U}_i$. As it is often the case that the union of all the entries in $\{\ma^{(i)}\}$ constitutes only a strict subset of those in $\mpp$, our aim is to integrate these $\{\ma^{(i)}\}$ to recover the unobserved entries in $\mpp$. %between all clincical concepts in $\mathcal{U}$.
{\color{black}
Another example involving symmetric matrices integration appears in neuroscience, where symmetric covariance matrices are computed from calcium imaging data to characterize functional connectivity among neurons. Due to experimental constraints, only a strict subset of neurons are observed in each recording session. Integrating these incomplete covariance matrices enables the reconstruction of global neuronal interaction networks and accurate identification of functional hubs \citep{chang2022low}.
}

An example of asymmetric matrix integration arises in single-cell matrix data, where rows represent genomic features, columns represent cells, and each entry records some specific information about a feature in the corresponding cell. 
A key challenge in the joint analysis for this type of data is to devise efficient computational strategies to integrate different data modalities \citep{ma2020integrative,lahnemann2020eleven}, as the experimental design 
may lead to a collection of single-cell data matrices for different, but potentially overlapping, sets of cells and features,
{\color{black} such as those generated across batches, tissues, or technologies. Completing such partially overlapping data is crucial for constructing unified representations of cell populations and improving downstream tasks like clustering or trajectory inference.
}
More specifically, 
let $\mpp\in\mathbb{R}^{N\times M}$ be the population matrix for all involved features and cells where $N:=|\mathcal{U}|, M:=|\mathcal{V}|$ (with $\mathcal{U}$ and $\mathcal{V}$ denoting the sets of genomic features and cells, respectively). Each single-cell data matrix $\ma^{(i)}\in\mathbb{R}^{n_i\times m_i}$ is then a submatrix of $\mpp$ corresponding to some $\mathcal{U}_i \subset \mathcal U$ and $\mathcal{V}_i \subset \mathcal V$; here we denote $n_i:=|\mathcal{U}_i|$ and $m_i:=|\mathcal{V}_i|$. Our aim is once again to integrate the collection of $\{\ma^{(i)}\}$ to reconstruct the original $\mpp$.

%The above examples involving EHR and single-cell data are special cases of the matrix completion with noise and block-wise missing structures.
%The traditional formulation of matrix completion is to recover a, possibly low-rank, matrix $\mathbf{M}$ based on a random subset of its entries which may be contaminated by noise \citep{keshavan2010matrix,candes2010matrix}. %,nguyen2019low}.
%If the observed entries are sampled uniformly at random then
%algorithms based on nuclear norm minimization (NNM) are shown to have near-optimal results \citep{candes2010power,candes2012exact}. More computationally efficient variants of NNM, such as those based on penalized or constrained NNM
%\citep{cai2010singular,candes2011tight,koltchinskii2011nuclear,tanner2013normalized,chen2019inference} or iteratively reweighted least squares \citep{fornasier2011low,mohan2012iterative} are also available. 
%When the rank of $\mathbf{M}$ is known or can be estimated, additional methods such as those based on greedy optimization \citep{lee2010admira}, alternating minimization \citep{jain2013low}, Riemannian optimization \citep{vandereycken2013low}, truncated NNM \citep{hu2012fast}, stochastic gradient descent \citep{jin2016provable}, and spectral decomposition \citep{sun2016guaranteed,cho2017asymptotic,chen2020nonconvex,yan2024inference} are applicable.
%For cases where entries are sampled independently but not uniformly, techniques like weighted trace-norm regularization \citep{srebro2010collaborative,foygel2011learning} or adding max-norm constraints \citep{cai2016matrix} can be effective.

The above examples involving EHR and single-cell data are special cases of the matrix completion with noise and block-wise missing structures. 
However, the existing literature on matrix completion mainly focuses on recovering a possibly low-rank matrix based on uniformly sampled observed entries 
or independently sampled observed entries which may be contaminated by noise; see, e.g., \citet{candes2012exact,cai2010singular,candes2011tight,koltchinskii2011nuclear,tanner2013normalized,chen2019inference,fornasier2011low,mohan2012iterative,lee2010admira,vandereycken2013low,hu2012fast,sun2016guaranteed,cho2017asymptotic,chen2020nonconvex,srebro2010collaborative,cai2016matrix,foygel2011learning} for an incomplete list of references. 

 These assumptions of uniform or independent sampling in standard matrix completion models are generally violated in applications of matrix integration, thus necessitating the development of efficient methods for tackling the block-wise missing structures.
Some examples of this development include the generalized integrative principal component analysis (GIPCA) of
\cite{zhu2020generalized}, structured matrix completion (SMC) of \cite{cai2016structured}, block-wise overlapping noisy matrix integration (BONMI) of \cite{zhou2021multi}, and symmetric positive semidefinite matrix completion (SPSMC) of \cite{bishop2014deterministic}. 
The GIPCA procedure operates under the setting where each data matrix have some common samples and completely different variables, and furthermore assumes that each entry in these matrices are from some exponential family of distribution, with entries in the same matrix having the same distributional form.
%\cite{du2023multinomial} propose a multinomial imputed-factor Logistic regression model with multi-source functional block-wise missing data as covariates, where each data source is functional data and a certain data sources could be missing for some subjects.
SMC is a spectral procedure for recovering the missing block of an approximately low-rank matrix when a subset of the rows and columns are observed; thus, SMC is designed to impute only a single missing block at a time.
%SMC is a spectral procedure for recovering the missing block of an approximately low rank matrix; however, SMC is designed to impute only a single missing block at a time. 
%BONMI is also a spectral procedure for recovering a missing block (or submatrix) in an approximately low rank matrix but, in contrast to SMC, assumes that the missing block share some (limited) overlap with a {\em pair} of observed sub-matrices.
BONMI is also a spectral procedure for recovering a missing block (or submatrix) in an approximately low-rank matrix but, in contrast to SMC, 
assumes that this missing block is associated with a given {\em pair} of observed submatrices that share some (limited) overlap. 
% motivated by the application of PMI matrices integration for EHR data, \cite{zhou2021multi} propose a spectral method called block-wise overlapping noisy matrix integration (BONMI)under a low rank assumption to integrate a pair of overlapping positive semi-definite matrices to recover their related missing block, and BONMI integrates only a pair of observed overlapping matrices at a time.
SPSMC has a similar spectral procedure with BONMI to recover a low-rank symmetric positive semidefinite matrix using some observed principal submatrices. While BONMI combines submatrices pair by pair, SPSMC sequentially integrates each new submatrix with the combined structure formed by all previously integrated submatrices.
The key idea behind BONMI and SPSMC is to align (via an orthogonal transformation) the spectral embeddings given by the leading (scaled) eigenvectors of the two overlapping submatrices and then impute the missing block by taking the outer product of these aligned embeddings. 
%Specifically, BONMI first obtains the low-dimensional embeddings of the corresponding entities for the two matrices, then aligns the embeddings for their overlapping entities by an orthogonal transformation.
%Embedding alignments %across different systems %, especially with  orthogonal transformation, 
%also appeared in other applications including bilingual dictionary induction \citep{%smith2017offline, 
%kementchedjhieva2018generalizing}, knowledge graphs integration \citep{lin2019guiding} %,fanourakis2023knowledge}, 
% and vertex nominations \citep{zheng2022vertex}.

In this paper, we extend the BONMI procedure, which integrates only two {\em overlapping} submatrices at a time, to simultaneously and jointly integrate $K \geq 2$ submatrices, and propose the Chain-linked Multiple Matrix Integration (CMMI) for more efficient and flexible matrix completion. 
{\color{black}
As a motivating example, suppose we have two overlapping pairs
of positive semidefinite submatrices $(\ma^{(1)}, \ma^{(2)})$ and
$(\ma^{(2)}, \ma^{(3)})$. For each submatrix $i$, let $\hat\mx^{(i)}$ be the $n_i \times d$ matrix whose columns are the $d$ leading eigenvectors of $\ma^{(i)}$ scaled by the square root of the corresponding eigenvalues. The rows of $\hat{\mx}^{(i)}$ represent the embeddings of the $n_i$
entities associated with $\ma^{(i)}$ into $\mathbb{R}^{d}$, and $\hat{\mx}^{(i)} \hat{\mx}^{(i)\top}$ correspond to the best rank-$d$ approximation to $\ma^{(i)}$. However, 
as the leading eigenvectors of $\ma^{(i)}$ are not necessarily unique,
we cannot directly use the inner product between different
$\hat\mx^{(i)}$ to estimate the unobserved entries. Rather, we first have to align $\{\hat{\mx}^{(1)}, \hat{\mx}^{(2)}, \hat{\mx}^{(3)}\}$ using their overlapping submatrices. 
More specifically, we align $\hat\mx^{(1)}$ and $\hat\mx^{(2)}$ by finding the orthogonal
transformation $\mw^{(1,2)}$ that maps the embeddings of
$\mathcal{U}_1 \cap \mathcal{U}_2$ in $\hat\mx^{(1)}$ to (approximate) their
counterparts in $\hat\mx^{(2)}$. The entries of 
$\hat\mx^{(1)}\mw^{(1,2)}\hat\mx^{(2)\top}$ then serve as estimates of the unobserved entries between $\mathcal{U}_1$ and $\mathcal{U}_2$.
Similarly, we compute $\mw^{(2,3)}$ to map the embeddings for
$\mathcal{U}_2 \cap \mathcal{U}_3$ in $\hat\mx^{(2)}$ to their
counterparts in $\hat\mx^{(3)}$, and entries of $\hat\mx^{(2)}\mw^{(2,3)}\hat\mx^{(3)\top}$ serve as estimates of the unobserved entries between $\mathcal{U}_2$ and $\mathcal{U}_3$. 
Finally, we can use $\hat\mx^{(1)}\mw^{(1,2)}\mw^{(2,3)}\hat\mx^{(3)\top}$ to estimate the unobserved entries between $\mathcal{U}_1$ and $\mathcal{U}_3$  
even when $\mathcal{U}_1 \cap \mathcal{U}_3 = \varnothing$ (so that
$\ma^{(1)}$ and $\ma^{(3)}$ are {\em non-overlapping}).
}
%As a motivating example, suppose we have two overlapping pairs of submatrices $(\ma^{(1)}, \ma^{(2)})$ and $(\ma^{(2)}, \ma^{(3)})$. For each submatrix $i$, $\hat\mx^{(i)}\in\mathbb{R}^{n_i\times d}$ such that $\ma^{(i)}\approx \hat\mx^{(i)}\hat\mx^{(i)\top}$ can be regarded estimated latent positions for the corresponding $n_i$ entity in $d$-dimension space. Using the overlapping entities between $\ma^{(1)}$ and $\ma^{(2)}$ (resp. $\ma^{(2)}$ and $\ma^{(3)}$) we can find an orthogonal transformation $\mw^{(1,2)}$ such that the overlapping entities between $\hat\mx^{(1)}$ and $\hat\mx^{(2)}$ can be aligned in the $d$-dimension space (resp. $\mw^{(2,3)}$ to align the overlapping entities between $\hat\mx^{(2)}$ and $\hat\mx^{(3)}$) to align the embeddings $\hat\mx^{(1)}$ and $\hat\mx^{(2)}$ (resp. $\hat\mx^{(2)}$ and $\hat\mx^{(3)}$). Then by combining $\mw^{(1,2)}$ and $\mw^{(2,3)}$, we can also align $\hat{\mx}^{(1)}$ to $\hat{\mx}^{(3)}$ and recover the missing block associated with $\ma^{(1)}$ and $\ma^{(3)}$ even when these submatrices are {\em non-overlapping}. 
Generalizing this observation we can show that as long as  $\{\ma^{(i)}\}$ are \textit{connected} then we can integrate them simultaneously to recover all the missing entries; here two submatrices $\ma^{(i)}$ and $\ma^{(j)}$ are said to be connected if there exists a sequence $i_0, i_1, \dots, i_L$ with $i_0 = i$, $i_L = j$ such that $\ma^{(i_{\ell-1})}$ and $\ma^{(i_{\ell})}$ are overlapping for all $\ell =1,\dots,L$.
The use of CMMI thus enables the recovery of many missing blocks that are unrecoverable by BONMI and furthermore allows for significantly smaller overlap between the observed submatrices. %; indeed our theoretical results show that CMMI is applicable whenever the number of overlapping entities between the submatrices is larger than or equal to their ranks. 
%thereby significantly improves the recoverable-to-known information ratio.
CMMI considers all possible overlapping pairs without relying on the integration order of submatrices, unlike SPSMC, enabling a more optimal recovery result.

%We also present theoretical results that are significantly more refined than those in \cite{zhou2021multi}. 

%Based on it, we obtain an expansion of the estimation error for the missing block with a max norm upper bound of the remainder. 
%It follows %a max norm upper bound for the estimation error of the missing block and it means 
%that our algorithm can recover the underlying truth for the missing block entrywisely, and furthermore, the entrywise fluctuations of our estimate are mean-zero normally distributed.

%In particular
%consider the setting where the entries of the submatrices are noisily observed while we allow for each observed submatrix to be noisy as well as have a proportion of uniformly missing entries.
%Furthermore, \cite{zhou2021multi} derive spectral norm upper bounds for the estimation errors of both 
%the missing block and the embedding matrix. 

%

%recovering the unknown entries and 
%In addition, the observation that the leading term dominating the error in estimating the missing block is only related to the two embedding matrices directly used for the estimation, and not to the other embedding matrices used to bridge them, inspires us to devise a strategy for optimally estimating each entry in complex multiple matrix integration problems.

The structure of our paper is as follows. In Section~\ref{sec:meth} we introduce the model for multiple observed principal submatrices of a whole symmetric positive semi-definite matrix, and propose CMMI to integrate a chain of connected overlapping submatrices.
Theoretical results for our CMMI procedures are presented in Section~\ref{sec:thm}.
In particular we derive error bounds in $2\to\infty$ norm for the spectral embedding of the submatrices and entrywise error bound for the recovered entries. Using these error bounds we show that our recovered entries are approximately normally distributed around their true values and that CMMI yields consistent estimates even with minimal overlaps between the observed submatrices. %, i.e., it can accurately integrate arbitrarily large submatrices with bounded overlaps.
We emphasize that the results in Section~\ref{sec:thm} also hold for BONMI (which is a special case of our results for $K=2$) and SPSMC, thereby providing significant refinements over those in \cite{zhou2021multi} and \cite{bishop2014deterministic}, which mainly focus on bounding the spectral or Frobenius norm errors of the missing block and embeddings.
%We emphasize that the results in Section~\ref{sec:thm} also hold for the BONMI model and, when restricted to this model, are significant refinements of those in \cite{zhou2021multi} which are based mainly on bounding the spectral norm errors of both the missing block and spectral embeddings. 
And our analysis handles both noisy and missing entries in the observed submatrices while \cite{zhou2021multi} and \cite{bishop2014deterministic} only consider the case of noisy entries. 
Numerical simulations and experiments on real data are presented in Sections~\ref{sec:simu} and \ref{sec:real}. In Section~\ref{sec:asy}, we extend our embedding alignment approach to the cases of symmetric indefinite matrices and asymmetric or rectangular matrices.
Detailed proofs of stated results and additional numerical experiments are provided in the supplementary material.
{\color{black}
Section~\ref{sec:integration} of the supplementary material extends the basic CMMI algorithm for chains to handle the integration of connected submatrices with arbitrarily complex structures, making it more applicable to real-world scenarios. We provide theoretical analysis of the generalized CMMI based on the results in Section~\ref{sec:thm}, and demonstrate its effectiveness through real data experiments.}

%Section~\ref{sec:integration} of the supplementary material explores more complex matrix integration challenges, such as scenarios where the connected submatrices do not form a single chain but rather complex. %Finally, in practical applications the observed sub-matrices might have a more complex structure than that of a simple chain.
%The theoretical results in Section~\ref{sec:thm} allow us to develop several effective strategies for addressing these issues. 
 %multiple feasible chains of overlapping submatrices exist for recovering a particular unobserved entry, with the challenge being how to select the optimal one. 
 %We propose strategies for handling these more intricate issues based on our foundational algorithm for a single chain, inspired by our theoretical analysis.
%Section~\ref{sec:discussion} provides further discussion on more complicated matrix integration challenges that can be encountered in real-world scenarios, such as when there are multiple feasible chains of overlapping submatrices to recover a particular unobserved entry and the challenge lies in selecting the best one, and we propose algorithms for handling these more intricate issues based on our foundational algorithm and theoretical analysis for a single chain.

\subsection{Notations}
We summarize some notations used in this paper. 
For any positive integer $n$, we denote by $[n]$ the set $\{1,2,\dots, n\}$. 
For two non-negative sequences
$\{a_n\}_{n \geq 1}$ and $\{b_n\}_{n \geq 1}$, we write $a_n \lesssim
b_n$ (resp. $a_n \gtrsim b_n$) if there exists some constant $C>0$
such that $a_n \leq C b_n$ (resp. $a_n \geq C b_n$) for all  $n \geq 1$, and we write $a_n \asymp b_n$ if $a_n\lesssim b_n$ and $a_n\gtrsim b_n$.
The notation $a_n \ll b_n$ (resp. $a_n \gg b_n$) means that there exists some sufficiently small (resp. large) constant $C>0$ such that $a_n \leq Cb_n$ (resp. $a_n \geq Cb_n$).
If $a_n/b_n$ stays bounded away from $+\infty$, we write $a_n=O(b_n)$ and $b_n=\Omega(a_n)$, and we use the notation $a_n=\Theta(b_n)$ to indicate that $a_n=O(b_n)$ and $a_n=\Omega(b_n)$.
If $a_n/b_n\to 0$, we write $a_n=o(b_n)$ and $b_n=\omega(a_n)$.
We say a sequence of events $\mathcal{A}_n$ holds with high probability if for any $c > 0$ there exists a finite constant $n_0$ depending only on $c$ such that $\mathbb{P}(\mathcal{A}_n)\geq 1-n^{-c}$ for all $n \geq n_0$.
We write $a_n = O_p(b_n)$ (resp. $a_n = o_p(b_n)$) to denote that $a_n = O(b_n)$ (resp. $a_n = o(b_n)$) holds with high probability.
We denote by $\mathcal{O}_d$ the set of $d \times d$ orthogonal
matrices.
% and let $\mathcal{O}(n,d):=\{\muu\in\mathbb{R}^{n\times d}\mid \muu^\top\muu=\mi_d\}$.
For any matrix $\mm\in \mathbb{R}^{A\times B}$ and index sets $\mathcal{A}\subseteq [A]$, $\mathcal{B}\subseteq [B]$, we denote by $\mm_{\mathcal{A},\mathcal{B}}\in \mathbb{R}^{|\mathcal{A}|\times |\mathcal{B}|}$ the submatrix of $\mm$ formed from rows $\mathcal{A}$ and columns $\mathcal{B}$, and we denote by $\mm_{\mathcal{A}}\in\mathbb{R}^{|\mathcal{A}|\times B}$ the submatrix of $\mm$ consisting of the rows indexed by $\mathcal{A}$. The Hadamard product between conformal matrices $\mathbf{M}$ and $\mathbf{N}$ is denoted by $\mathbf{M} \circ \mathbf{N}$.  
Given a matrix $\mathbf{M}$, we denote
its spectral, Frobenius, and infinity norms by $\|\mathbf{M}\|$, 
$\|\mathbf{M}\|_{F}$, and $\|\mathbf{M}\|_{\infty}$. 
We also denote the maximum entry (in modulus) of $\mathbf{M}$ by $\|\mathbf{M}\|_{\max}$ and the $2 \to \infty$ norm of
$\mathbf{M}$ by
$\|\mathbf{M}\|_{2 \to \infty} = \max_{\|\bm{x}\| = 1} \|\mathbf{M}
\bm{x}\|_{\infty} = \max_{i} \|\mathbf{m}_i\|,$
where $\mathbf{m}_i$ denotes the $i$th row of $\mathbf{M}$, i.e., $\|\mathbf{M}\|_{2 \to \infty}$ is the maximum of the $\ell_2$ norms of the rows of $\mathbf{M}$.
%We note that the $2 \to \infty$ norm is {\em not} sub-multiplicative. However, for any matrices $\mathbf{M}$ and $\mathbf{N}$ of conformal dimensions, we have \begin{equation*} \|\mathbf{M} \mathbf{N} \|_{2 \to \infty} \leq \min\{\|\mathbf{M}\|_{2 \to \infty} \times \|\mathbf{N}\|, \|\mathbf{M}\|_{\infty} \times \|\mathbf{N}\|_{2\to\infty}\};\end{equation*} see Proposition~6.5 in \cite{cape2019two}. 
%Perturbation bounds using the $2 \to \infty$ norm for the eigenvectors and/or singular vectors of a noisily observed matrix had recently attracted interests from the statistics community, see e.g., \cite{chen2021spectral,cape2019two,%lei2019unified,
%damle,abbe2020entrywise} and the references therein. 

\section{Methodology}

\label{sec:meth}

%\subsection{Model}

We are interested in an unobserved population matrix associated with $N$ entities denoted by $\mpp\in \mathbb{R}^{N\times N}$. We assume $\mpp$ is %symmetric, 
positive semi-definite with rank $d \ll N$; extensions to the case of symmetric but indefinite $\mpp$ as well as asymmetric or rectangular $\mpp$ are discussed in Section~\ref{sec:asy}. Denote the eigen-decomposition of $\mpp$ as
 $\muu\mLambda\muu^\top$, 
where $\mLambda\in\mathbb{R}^{d\times d}$ is a diagonal matrix whose diagonal entries are the non-zero eigenvalues of $\mpp$ in descending order, and the orthonormal columns of $\muu\in\mathbb{R}^{N\times d}$ constitute the corresponding eigenvectors.
The latent positions associated to the entities are given by $\mx=\muu\mLambda^{1/2}\in\mathbb{R}^{N\times d}$ and any entry in $\mpp$ can be written as the inner product of these latent positions, i.e., $\mpp = \mx \mx^{\top}$ so that $\mpp_{s,t}=\mathbf{x}_s^\top\mathbf{x}_t$ for any $s,t \in [N]$, where $\mathbf{x}_s$ and $\mathbf{x}_t$ denote the $s$th and $t$th row of $\mx$, respectively. %We can thus write $\mpp$ alternatively as 
% $$\mpp=\mx\mx^\top.$$

We assume that the entries of $\mpp$ are only partially observed, and furthermore, that the observed entries can be grouped into blocks. More specifically, suppose that we have $K$ sources and
for any $i\in[K]$ we denote the index set of the entities contained in the $i$th source by $\mathcal{U}_i\subseteq [N]$. 
For ease of exposition we also require $\mathcal{U}_i \cap (\cup_{j \not = i} \mathcal{U}_j) \not = \varnothing$ for all $i \in [K]$
as otherwise there exists some $i_*$ such that it is impossible to integrate observations from $\mathcal{U}_{i_*}$ with those from $\{\mathcal{U}_{j}\}_{j \not = i_*}$. %; note, however, that we do not require $\mathcal{U}_{i} \cap \mathcal{U}_j \not = \varnothing$ for all $j \not = i$. 
We denote $n_i := |\mathcal{U}_i|$ and the population matrix for the $i$th source by $\mpp^{(i)}\in\mathbb{R}^{n_i\times n_i}$. We then have
$$
\mpp^{(i)}=\mpp_{\mathcal{U}_i,\mathcal{U}_i}=\muu_{\mathcal{U}_i}\mLambda\muu_{\mathcal{U}_i}^\top
=\mx_{\mathcal{U}_i}\mx_{\mathcal{U}_i}^\top,
$$
where $\mpp_{\mathcal{U}_i,\mathcal{U}_i}$ is the submatrix of $\mpp$ formed from rows and columns in $\mathcal{U}_i$, $\muu_{\mathcal{U}_i}\in\mathbb{R}^{n_i\times d}$ contains the rows of $\muu$ in $\mathcal{U}_i$, and $\mx_{\mathcal{U}_i}\in\mathbb{R}^{n_i\times d}$ contains the latent positions of $\mathcal{U}_i$.

We also allow for missing and corrupted observations in each source, i.e., for the $i$th source %This missing data model proposed in \cite{yan2021inference}.
 we only get to observe
$\mpp^{(i)}_{s,t}+\mn^{(i)}_{s,t}$ for all $\mathbf{\Omega}^{(i)}_{s,t}=1.$
Here $\mathbf{\Omega}^{(i)}\in \{0,1\}^{n_i\times n_i}$ indicates the indices of the observed entries and $\mn^{(i)} \in \mathbb{R}^{n_i \times n_i}$ represents the random noise.
In particular $\mathbf{\Omega}^{(i)}$ and $\mn^{(i)}$ are both symmetric, and we assume the upper triangular entries of $\mathbf{\Omega}^{(i)}$ are i.i.d. Bernoulli random variables with success probability $q_i$ while 
the upper triangular entries of $\mn^{(i)}$ are independent, mean-zero sub-Gaussian random variables 
with Orlicz-2 norm bounded by $\sigma_i:=\max_{s,t\in [n_i]}\|\mn^{(i)}_{s,t}\|_{\psi_2}$.
 %We denote the proportion of non-missing entries in the $i$th source by $q_i$, so we have $\|\mathbf{\Omega}^{(i)}\|_F^2= q_i n_i^2$.
For this model, the matrix
\begin{equation}\label{eq:A(i)=...}
\begin{aligned}
\ma^{(i)}
%=\mb^{(i)}/q_i
=(\mpp^{(i)}+\mn^{(i)})\circ \mathbf{\Omega}^{(i)}/q_i
\end{aligned}
\end{equation}
is an unbiased estimate of $\mpp^{(i)}$, and thus a natural idea is to use the scaled leading eigenvectors $\hat\mx^{(i)}=\hat\muu^{(i)}(\hat\mLambda^{(i)})^{1/2}$ as an estimate for $\mx_{\mathcal{U}_i}$, where $\hat\mLambda^{(i)}$ and $\hat\muu^{(i)}$ contain the leading eigenvalues and the leading eigenvectors of $\ma^{(i)}$, respectively.
{\color{black} Note that in practice we use the empirical observed proportion $\hat{q}_i$ in place of $q_i$ when constructing $\mathbf{A}^{(i)}$. In particular, our evaluation of the algorithm's performance in the numerical experiments is based entirely on $\hat{q}_i$. 
    In contrast we assume that $q_i$ is known in our theoretical analysis. This is done, both for ease of exposition and without loss of generality, as $|\hat{q}_i - q_i| = O_p(n_i^{-1})$ for all $i$ and thus has no impact on the theoretical results stated in Section~\ref{sec:thm}.}
We now propose an algorithm to integrate and align $\{\hat\mx^{(i)}\}_{i\in[K]}$ for recovery of the unobserved entries in $\mpp$.

 %and denote the observed matrix by $\ma^{(i)}\in\mathbb{R}^{|\mathcal{U}_i|\times |\mathcal{U}_i|}$
% the observed matrix $\ma^{(i)}$
%$$\ma^{(i)}
%=\mathcal{P}_{\mathbf{\Omega}^{(i)}}(\mpp^{(i)}+\mn^{(i)})/q
%=(\mpp^{(i)}+\mn^{(i)})\circ \mathbf{\Omega}^{(i)}/q_i
%=(\muu_{\mathcal{U}_i}\mLambda\muu_{\mathcal{U}_i}^\top+\mn^{(i)})\circ \mathbf{\Omega}^{(i)}/q_i
%=(\mx_{\mathcal{U}_i}\mx_{\mathcal{U}_i}^\top+\mn^{(i)})\circ \mathbf{\Omega}^{(i)}/q_i,$$
%where $\mn^{(i)}$ is symmetric and the entries of $\mn^{(i)}$ are independent sub-Gaussian noise with variance $\sigma_i^2>0$. 
%Here
%$\mathcal{P}_{\mathbf{\Omega}^{(i)}}(\mm)
%=\mathbf{\Omega}^{(i)}\circ(\mm)$,
%where $\mathbf{\Omega}^{(i)}$ is symmetric and the entities of $\mathbf{\Omega}^{(i)}$ are i.i.d. Bernoulli($q_i$). 

%We assume $\|\mathbf{\Omega}^{(i)}\|_F^2= q_i n_i^2$, where $n_i=|\mathcal{U}_i|$.

%We assume $\bigcup_{i\in[K]}\mathcal{\mathcal{U}}_i=[N]$ and $\bigcup_{i\in[K]}\mathcal{\mathcal{V}}_i=[M]$. 

\subsection{Motivation of the algorithm}

\label{sec:mot}
We first summarize the BONMI algorithm of \cite{zhou2021multi}. We start with the \textit{noiseless} case for two overlapping submatrices $\mpp^{(1)}$ and $\mpp^{(2)}$ to illustrate the main ideas. %we have $2$ overlapping block-wise submatrices $\mpp^{(1)}$ and $\mpp^{(2)}$ without noise and missing entries. 
Our goal is to recover the unobserved entries for the white block in %the left panel of 
Figure~\ref{fig:K=2}; this is part of $\mpp_{\mathcal{U}_1,\mathcal{U}_2}$. %=\mx_{\mathcal{U}_1}\mx_{\mathcal{U}_2}^\top$). 
%and thus $\mpp^{(1)}$ and $\mpp^{(2)}$ also have rank $d$.
%$\mx^{(1)}$ and $\mx^{(2)}$ satisfy $\mpp^{(1)}=\mx^{(1)}\mx^{(1)\top}$ and $\mpp^{(2)}=\mx^{(2)}\mx^{(2)\top}$.
\begin{figure}[htbp!] 
\centering
\subfigure{\includegraphics[height=3.5cm]{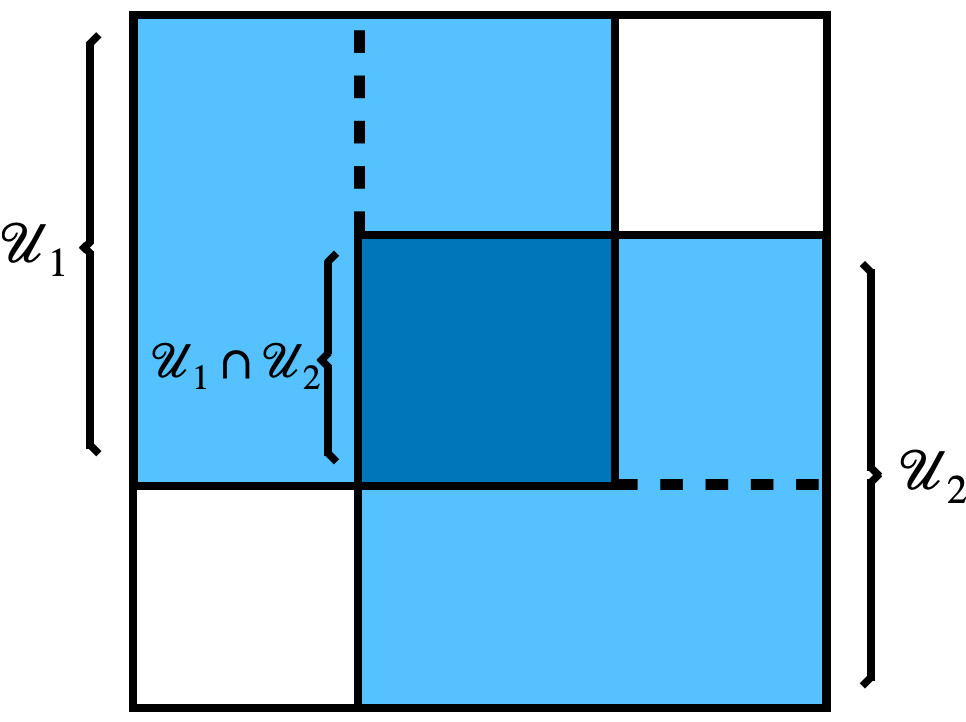}}% \qquad
%\subfigure{\includegraphics[height=4.5cm]{figure/paper/K=L.jpeg}}
\caption{%Left panel: 
a pair of overlapping observed submatrices. 
%Right panel: a chain of overlapping observed submatrices
}
\label{fig:K=2}
\end{figure}

%Here $\mpp_{\mathcal{U}_1,\mathcal{U}_2}$ denotes the submatrix of $\mpp$ formed from rows corresponding to $\mathcal{U}_1$ and columns corresponding to $\mathcal{U}_2$.
%Because $\mpp$ is not fully observed, $\mx_{\mathcal{U}_1}$ and $\mx_{\mathcal{U}_2}^\top$ cannot be obtained directly while $\mx^{(1)}$ and $\mx^{(2)}$ can be obtained from $\mpp^{(1)}$ and $\mpp^{(2)}$, respectively.
Based on $\mpp^{(1)}$ and $\mpp^{(2)}$ we can obtain latent position estimates for entities in $\mathcal{U}_1$ and $\mathcal{U}_2$, which we denote as $\mx^{(1)}$ and $\mx^{(2)}$. 
%Now we consider the relationships between $\mx^{(i)}$ and $\mx_{\mathcal{U}_i}$ for $i=1,2$.
Next note that
$$
\begin{aligned}
	\mx_{\mathcal{U}_1} \mx_{\mathcal{U}_1}^\top=\mpp_{\mathcal{U}_1,\mathcal{U}_1}=\mpp^{(1)}=\mx^{(1)}\mx^{(1)\top},
	\quad \mx_{\mathcal{U}_2} \mx_{\mathcal{U}_2}^\top=\mpp_{\mathcal{U}_2,\mathcal{U}_2}=\mpp^{(2)}=\mx^{(2)}\mx^{(2)\top},
\end{aligned}
$$
and hence there exist $\mw^{(1)},\mw^{(2)}\in\mathcal{O}_d$ such that 
\begin{equation}\label{eq:X}
\begin{aligned}
		\mx_{\mathcal{U}_1}=\mx^{(1)}\mw^{(1)},
		\quad \mx_{\mathcal{U}_2}=\mx^{(2)}\mw^{(2)}.
\end{aligned}
\end{equation}
Eq.~\eqref{eq:X} then implies
\begin{equation}\label{eq:P12}
	\mpp_{\mathcal{U}_1,\mathcal{U}_2}=\mx_{\mathcal{U}_1}\mx_{\mathcal{U}_2}^\top=\mx^{(1)}\mw^{(1)}\mw^{(2)\top}\mx^{(2)\top}=\mx^{(1)}\mw^{(1,2)}\mx^{(2)\top},
\end{equation}
where $\mw^{(1,2)}:=\mw^{(1)}\mw^{(2)\top}\in\mathcal{O}_d$, %, and notice that because $\mw^{(1)},\mw^{(2)}$ are unique, $\mw^{(1,2)}$ is also unique.
and thus we only need $\mw^{(1,2)}%:=$\mw^{(1)}\mw^{(2)\top}\in\mathcal{O}_d
$ to recover $\mpp_{\mathcal{U}_1,\mathcal{U}_2}$.

%Now we see how to obtain $\mw^{(1,2)}$.
Note that for entities in $\mathcal{U}_1\cap \mathcal{U}_2$, we have two equivalent representations of their latent positions. More specifically, let $\mx^{(1)}_{ \langle\mathcal{U}_1\cap \mathcal{U}_2\rangle}$ and $\mx^{(2)}_{ \langle\mathcal{U}_1\cap \mathcal{U}_2\rangle}$ be the rows of $\mx^{(1)}$ and $\mx^{(2)}$ corresponding to entities in $\mathcal{U}_1 \cap \mathcal{U}_2$. Then by
Eq.~\eqref{eq:X} we have %$$\mx^{(1)}_{ \langle\mathcal{U}_1\cap \mathcal{U}_2\rangle}\mw^{(1)}=\mx_{\mathcal{U}_1\cap \mathcal{U}_2}=\mx^{(2)}_{ \langle\mathcal{U}_1\cap \mathcal{U}_2\rangle}\mw^{(2)}.$$ It follows that
$\mx^{(2)}_{ \langle\mathcal{U}_1\cap \mathcal{U}_2\rangle}
=\mx^{(1)}_{ \langle\mathcal{U}_1\cap \mathcal{U}_2\rangle}\mw^{(1,2)}$
%\mw^{(1)}\mw^{(2)\top}$%=\mx^{(1)}_{ \langle\mathcal{U}_1\cap \mathcal{U}_2\rangle}\mw^{(1,2)}.$$
and thus $\mw^{(1,2)}$ can be obtained by aligning 
$\mx^{(1)}_{ \langle\mathcal{U}_1\cap \mathcal{U}_2\rangle}$ and $\mx^{(2)}_{ \langle\mathcal{U}_1\cap \mathcal{U}_2\rangle}$. 
% for the overlapping entities.
The resulting $\mw^{(1,2)}$ is unique
whenever $\mathrm{rk}(\mpp_{\mathcal{U}_1\cap \mathcal{U}_2,\mathcal{U}_1\cap \mathcal{U}_2}) = d$.
%To sum up, when integrating two observed submatrices $\mpp^{(1)}$ and $\mpp^{(2)}$ to recover the related unobserved block, it is necessary to find the proper orthogonal matrix $\mw^{(1,2)}$ to align the latent positions $\mx^{(1)}$ and $\mx^{(2)}$ so that they can be used together. This is because of the non-identifiability of the latent positions. And the overlapping entities make it possible to obtain $\mw^{(1,2)}$. %by aligning the latent positions of the overlapping entities $\mx^{(1)}_{\mathcal{U}_1\cap \mathcal{U}_2}$ and $\mx^{(2)}_{\mathcal{U}_1\cap \mathcal{U}_2}$.

The same approach also extends to the case where the $\mpp^{(1)}$ and $\mpp^{(2)}$
are partially and noisily observed. More specifically, suppose 
we observe $\ma^{(1)}$ and $\ma^{(2)}$ as defined in Eq.~\eqref{eq:A(i)=...}. We then obtain estimated latent positions $\hat\mx^{(1)}$ for $\mathcal{U}_1$ and $\hat\mx^{(2)}$ for $\mathcal{U}_2$ from $\ma^{(1)}$ and $\ma^{(2)}$, respectively. To align $\hat\mx^{(1)}$ and $\hat\mx^{(2)}$, we solve the orthogonal Procrustes problem
$$
\begin{aligned}
	\mw^{(1,2)}=\underset{\mo\in \mathcal{O}_{d}}{\operatorname{argmin}} \|\hat\mx^{(1)}_{ \langle\mathcal{U}_1\cap \mathcal{U}_2\rangle}\mo-\hat\mx^{(2)}_{ \langle\mathcal{U}_1\cap \mathcal{U}_2\rangle}\|_F
\end{aligned}
$$
%for the entities in $\mathcal{U}_1\cap \mathcal{U}_2$ to get $\mw^{(1,2)}$. 
and then estimate the unobserved block as part of
$\hat\mpp_{\mathcal{U}_1,\mathcal{U}_2}=\hat\mx^{(1)}\mw^{(1,2)}\hat\mx^{(2)\top}.$
%This method to %recover the block-wise unobserved data based on 
%integrate a pair of overlapping submatrices for a  positive semi-definite low-rank matrix is identical to that in \cite{zhou2021multi}. 
%Our description of this process from the perspective of aligning the latent positions now enables us to extend the same idea to a chain of overlapping submatrices and %also allow us to handle symmetric indefinite matrices as well as asymmetric or rectangular matrices (see Section~\ref{sec:asy}).

\subsection{Chain-linked Multiple Matrix Integration (CMMI)}

\label{sec:chain}

We now extend the ideas in Section~\ref{sec:mot} to a chain of overlapping submatrices.
Suppose our goal is to recover the entries in the yellow block %of the right panel 
in Figure~\ref{fig:K=L}.
Given a collection $\{\ma^{(i)}\}_{0 \leq i \leq L}$ such that $\mathcal{U}_{i-1} \cap \mathcal{U}_i \not = \varnothing $ for all $1\leq i \leq L$.
%$(\mathcal{U}_0,\mathcal{U}_1),(\mathcal{U}_1,\mathcal{U}_2),\dots,(\mathcal{U}_{L-1},\mathcal{U}_L)$ are overlapping. 
%the $\mpp_{\mathcal{U}_0,\mathcal{U}_L}$ contains the target entry and pairs $(\mathcal{U}_0,\mathcal{U}_1),(\mathcal{U}_1,\mathcal{U}_2),\dots,(\mathcal{U}_{L-1},\mathcal{U}_L)$ are overlapping. 
Then for each pair $(\mathcal{U}_{i-1},\mathcal{U}_i)$, we align the estimated latent position matrices $\hat{\mx}^{(i-1)}$ and $\hat{\mx}^{(i)}$ by solving the orthogonal Procrustes problem
$$
\mw^{(i-1,i)}=\underset{\mo\in \mathcal{O}_d}{\operatorname{argmin}} \|\hat\mx^{(i-1)}_{ \langle\mathcal{U}_{i-1}\cap \mathcal{U}_i\rangle}\mo-\hat\mx^{(i)}_{ \langle\mathcal{U}_{i-1}\cap \mathcal{U}_i\rangle}\|_F.
$$
Note that the solution of the orthogonal Procrustes problem between matrices $\hat\mx^{(i-1)}_{ \langle\mathcal{U}_{i-1}\cap \mathcal{U}_i\rangle}$ and $\hat\mx^{(i)}_{ \langle\mathcal{U}_{i-1}\cap \mathcal{U}_i\rangle}$ %of conformal dimensions 
is given by $\mm_1 \mm_2^{\top}$ where $\mm_1$ and $\mm_2$ contain the left and right singular vectors of $\hat\mx^{(i-1)\top}_{ \langle\mathcal{U}_{i-1}\cap \mathcal{U}_i\rangle} \hat\mx^{(i)}_{ \langle\mathcal{U}_{i-1}\cap \mathcal{U}_i\rangle}$, respectively \citep{procrustes}.
%for the overlapping entities $\mathcal{U}_{i-1}\cap \mathcal{U}_i$. 
$\hat{\mx}^{(0)}$ and $\hat{\mx}^{(L)}$ can be aligned by combining these $\{\mw^{(i-1,i)}\}_{1\leq i\leq L}$, which then yields
$$\hat{\mpp}_{\mathcal{U}_0, \mathcal{U}_{L}}= \hat\mx^{(0)}\mw^{(0,1)}\mw^{(1,2)}\cdots \mw^{(L-1,L)}\hat\mx^{(L)\top}$$
as an estimate for $\mpp_{\mathcal{U}_0, \mathcal{U}_{L}}$. See Algorithm~\ref{Alg_chain} for more details.

\begin{figure}[htbp!] 
\centering
%\subfigure{\includegraphics[height=4.5cm]{figure/paper/K=2.jpeg}}% \qquad
\subfigure{\includegraphics[height=4cm]{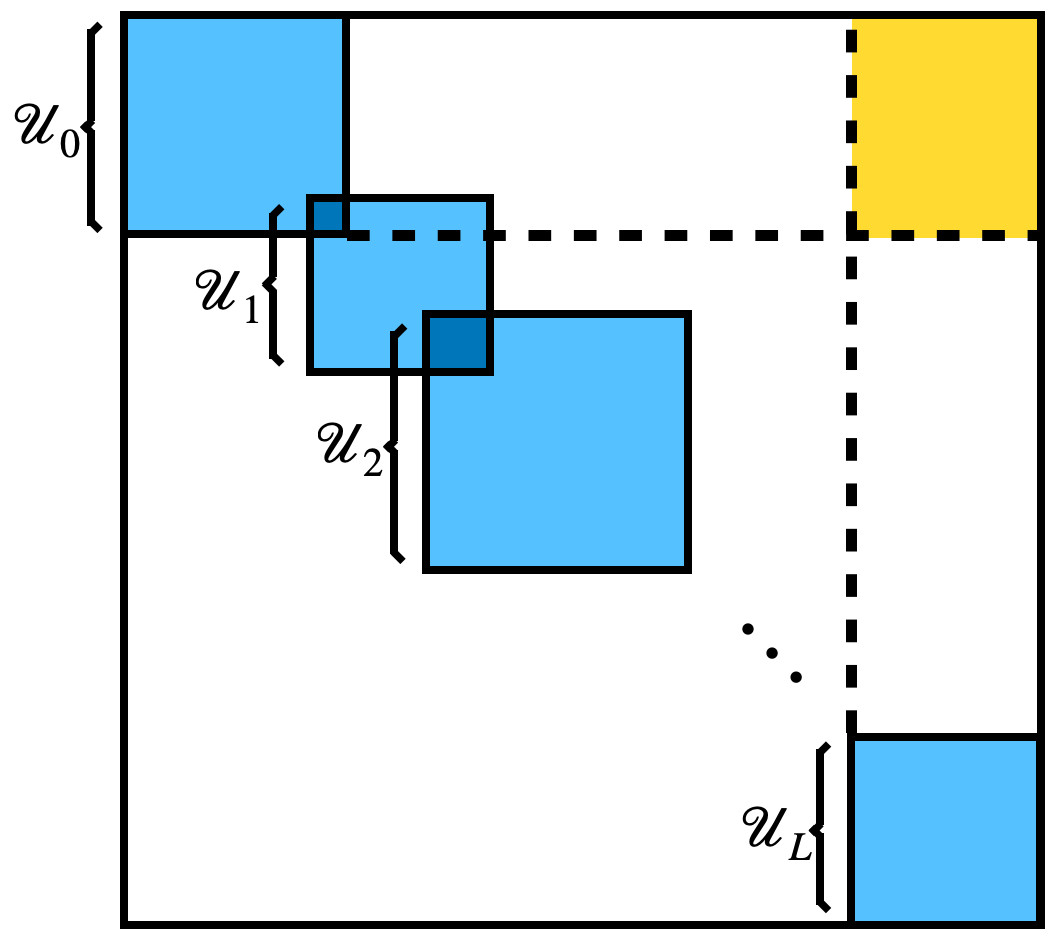}}
\caption{%Left panel: a pair of overlapping observed submatrices. Right panel: 
a chain of overlapping observed submatrices.
}
\label{fig:K=L}
\end{figure}

\begin{algorithm}[htp]
\caption{Chain-linked Multiple Matrix Integration (CMMI) algorithm}	
\label{Alg_chain}
\begin{algorithmic}
\REQUIRE Embedding dimension $d$, a chain of overlapping submatrices $\ma^{(i_0)},\ma^{(i_1)},\dots,\ma^{(i_L)}$ for $\mathcal{U}_{i_0},\mathcal{U}_{i_1},\dots,\mathcal{U}_{i_L}$ with $\min\{|\mathcal{U}_{i_0}\cap \mathcal{U}_{i_1}|,|\mathcal{U}_{i_1}\cap \mathcal{U}_{i_2}|,\dots,|\mathcal{U}_{i_{L-1}}\cap\mathcal{U}_{i_L}|\} \geq d$.% for some threshold $r\geq d$.
\begin{enumerate}
	\item For $0 \leq \ell \leq L$, obtain estimated latent position matrix for $\mathcal{U}_{i_\ell}$, denoted by $\hat\mx^{(i_\ell)}=\hat\muu^{(i_\ell)}(\hat\mLambda^{(i_\ell)})^{1/2}$, where $\hat\muu^{(i_\ell)}\in\mathbb{R}^{|\mathcal{U}_{i_\ell}|\times d}$ and the diagonal matrix $\hat\mLambda^{(i_\ell)}\in\mathbb{R}^{d\times d}$ contain the $d$ leading eigenvectors and eigenvalues of $\ma^{(i_\ell)}$, respectively.
	\item For $1 \leq \ell \leq L$, obtain $\mw^{(i_{\ell-1},i_\ell)}$ by solving the orthogonal Procrustes problem
\[
\mw^{(i_{\ell-1},i_\ell)}=\underset{\mo\in \mathcal{O}_d}{\operatorname{argmin}} \|\hat\mx^{(i_{\ell-1})}_{ \langle\mathcal{U}_{i_{\ell-1}}\cap \mathcal{U}_{i_\ell}\rangle}\mo-\hat\mx^{(i_\ell)}_{ \langle\mathcal{U}_{i_{\ell-1}}\cap \mathcal{U}_{i_\ell}\rangle}\|_F.
\]
	\item Compute $\hat\mpp_{\mathcal{U}_{i_0},\mathcal{U}_{i_L}}=\hat\mx^{(i_0)}\mw^{(i_0,i_1)}\mw^{(i_1,i_2)}\cdots \mw^{(i_{L-1},i_L)}\hat\mx^{(i_L)\top}$.
\end{enumerate} 
\ENSURE $\hat\mpp_{\mathcal{U}_{i_0},\mathcal{U}_{i_L}}$.
\end{algorithmic}
\end{algorithm}

{\color{black}For choosing $d$ in Algorithm~\ref{Alg_chain}, in practice we can first examine the eigenvalues of $\{\mathbf{A}^{(i)}\}$ to select the individual $\{d_i\}$, and then set $d = \max_i d_i$ to ensure that it captures all relevant signal components. One widely-used approach for selecting an individual $d_i$ is to inspect the scree plot for $\ma^{(i)}$ and identify a ``elbow" separating the signal eigenvalues from the noise eigenvalues. An example of this approach is the automated dimensionality selection procedure in \cite{zhu2006automatic} which maximizes a profile likelihood function. Other examples include residual subsampling \citep{han2019universal}, eigenvalue ratio tests \citep{ahn2013eigenvalue}, and inference based on empirical eigenvalue distributions \citep{onatski2010determining}.}

Compared to BONMI in \cite{zhou2021multi}, which handles only two {\em overlapping} submatrices at a time,
our proposed CMMI can actually combine all \textit{connected} submatrices, where two submatrices $\mpp^{(i)}$ and $\mpp^{(j)}$ are said to be connected if there exists a path of overlapping submatrices between them. Indeed, for the example in Figure~\ref{fig:K=L}, BONMI can only recover the entries associated with pairs of overlapping submatrices, namely \(\mpp_{\mathcal{U}_{0},\mathcal{U}_{1}}, \mpp_{\mathcal{U}_{1},\mathcal{U}_{2}}, \dots, \mpp_{\mathcal{U}_{L-1},\mathcal{U}_{L}}\), %and their symmetric counterparts 
while CMMI can recover the whole matrix
$\mpp$. In general, BONMI only recovers $O(1/L)$ fraction of the entries recoverable by CMMI.
%Therefore, the recoverable-to-known information ratio for BONMI is at most $2$, while for CMMI, it can approach $L$. 
Moreover, our theoretical results indicate that increasing $L$ has a minimal effect on the estimation error of CMMI (see Theorem~\ref{thm:R(i0,...,iL)}), and 
simulations and real data experiments in Sections~\ref{sec:simu} and \ref{sec:real} show that accurate recovery is possible even when $L = 20$. 
Our theoretical results also show that CMMI requires only minimal overlap between $\mathcal{U}_{i-1}$ and $\mathcal{U}_i$, e.g., $|\mathcal{U}_{i-1} \cap \mathcal{U}_i|$ can be as small as $d$, the embedding dimension of $\{\hat{\mx}^{(i)}\}$; see Remark~\ref{rem:example1_setting} for further discussion, and Section~\ref{sec:simu limit overlap} of the supplementary material for corresponding experimental results.

{\color{black}For more general cases encountered in practice, the structures of the observed submatrices can be more complex than simple chains. In Section~\ref{sec:integration} of the supplementary material, we extend the basic CMMI algorithm from Algorithm~\ref{Alg_chain} to integrate connected submatrices with arbitrarily complex structures.}
%For more general cases encountered in practice, there
%may exist multiple chains to recover a given unobserved entry. 
%We will present in Section~\ref{sec:integration} of the supplementary material several strategies to select and/or combine these chains. 
Although in certain cases, such as a simple chain, CMMI is identical to the sequential integration approach SPSMC in \cite{bishop2014deterministic}, 
CMMI offers a more refined strategy in more complex scenarios by considering all overlapping pairs, as the restriction to sequential integration imposes limitations on SPSMC, and how to determine an effective integration order is unresolved in \cite{bishop2014deterministic}. 

{\color{black}This idea of first obtaining individual estimates and then
sequentially aligning them to obtain a global estimate also appears in
the Spectral-Stitching algorithm in \cite{chen2016community} where the
goal is to determine the community membership of each vertex in a
graph with two communities. More specifically, the Spectral-Stiching
algorithm first partitions the vertex set into several overlapping
subsets of size $n$ such that any two adjacent subsets share $n/2$
common vertices.  It then applies spectral methods separately to each
subgraph to obtain community estimates. Since the community labels
obtained from different subgraphs may be inconsistent, the algorithm
sequentially stitches the individual community estimates together
using majority voting. }

\section{Theoretical Results}

\label{sec:thm}
We now present theoretical guarantees for the estimate $\hat\mpp_{\mathcal{U}_{i_0},\mathcal{U}_{i_L}}$ obtained by Algorithm~\ref{Alg_chain}.
%proposed for a chain of overlapping submatrices in Section~\ref{sec:chain}. To prepare for this, we first provide the theoretical result for any two overlapping observed submatrices $\ma^{(i)}$ and $\ma^{(j)}$.
We shall make the following assumptions on the underlying population matrices $\{\mpp_{\mathcal{U}_{i},\mathcal{U}_{i}}\}$ for the observed blocks. 
We emphasize that, because our results address either large-sample approximations or limiting distributions, these assumptions should be interpreted in the regime where $n_i$ is arbitrarily large and/or $n_i\to\infty$. 
\begin{assumption}
  \label{ass:main}
  For each $i$, the following conditions hold for sufficiently large $n_i$.
  \begin{itemize}
   %, then set $\tau=\lambda_{\max}/\lambda_{\min}$.
  %$$\lambda_{\max}\asymp\lambda_{\min}.$$
  %\item We set the sampling probability $\mathbb{P}(s\in\mathcal{U}_i)=p_i$, and set $p_0=\min_{i\in[K]}p_i$. We assume
  %$$p_0\geq C\sqrt{\mu_0 d \log N/N}$$
  %for some sufficiently large constant $C$, and $$\max_{i\in[K]} p_i/p_0\lesssim 1.$$
  \item We have $\mathrm{rk}(\mpp_{\mathcal{U}_i,\mathcal{U}_i}) = d$. Let $\lambda_{i,\max}$ and $\lambda_{i,\min}$ denote the largest and smallest non-zero eigenvalues of $\mpp_{\mathcal{U}_i,\mathcal{U}_i}$, and let $\muu^{(i)}\in\mathbb{R}^{n_i\times d}$ contain the eigenvectors corresponding to all non-zero eigenvalues.
        We then assume
  \begin{gather}
  \label{eq:rem_1}
  \|\muu^{(i)}\|\lesssim \frac{d^{1/2}}{n_i^{1/2}},\quad \text{and}\quad
  \frac{\lambda_{i,\max}}{\lambda_{i,\min}}\leq M\end{gather}  for some finite constant $M>0$.%See Remark~\ref{rm:lambda} for more discussion on this condition.
  \item $\ma^{(i)} = (\mpp_{\mathcal{U}_{i}, \mathcal{U}_i} + \mn^{(i)}) \circ \bm{\Omega}^{(i)}$ where $\mathbf{N}^{(i)}$ is a symmetric matrix whose (upper triangular) entries are independent mean-zero sub-Gaussian random variables with Orlicz-2 norm bounded by $\sigma_i$ and $\bm{\Omega}^{(i)}$ is a symmetric binary matrix whose (upper triangular) entries are i.i.d. Bernoulli random variables with success probability $q_i$. 
  \item Denote \[\mu_i := \lambda_{i,\min}/n_i, \quad  \gamma_i := (\|\mpp_{\mathcal{U}_i,\mathcal{U}_i}\|_{\max} + \sigma_i) \log^{1/2}{n_i}.\]
  We suppose $q_i n_i \gtrsim \log^{2}{n_i}$ and
  \begin{equation}\label{eq:assm1_con}
     \frac{\gamma_i}{(q_i n_i)^{1/2}\mu_i}\ll 1,
     \quad \frac{\gamma_i}{(q_i n_i\mu_i)^{1/2}}\lesssim \|\mx_{\mathcal{U}_i}\|_{2\to\infty}.
  \end{equation}

  \end{itemize}
\end{assumption}

{\color{black}We note that the conditions in Assumption~\ref{ass:main} are quite mild and typically seen in the matrix completion literature. For example Eq.~\eqref{eq:rem_1} is satisfied whenever $\mpp_{\mathcal{U}_i, \mathcal{U}_i}$ has bounded condition number and bounded coherence; see e.g., \cite{abbe2020entrywise,chen2021spectral,recht2011simpler}.
Next, $n_i q_i \gtrsim \log^{2}{n_i}$ is much less stringent compared to
$q_i \equiv 1$ as assumed in \cite{zhou2021multi} and \cite{bishop2014deterministic}. Finally Eq.~\eqref{eq:assm1_con} is satisfied whenever Eq.~\eqref{eq:rem_1} holds and 
$\|\mpp_{\mathcal{U}_i, \mathcal{U}_i}\|_{\max} + \sigma_i = O(1)$. See Remark~\ref{rem:example1_setting} for further discussion.
%Note that $\|\mpp_{\mathcal{U}_i,\mathcal{U}_i}\|_{\max}=\|\mx_{\mathcal{U}_i}\|_{2 \to \infty}^2  \lesssim d \frac{\lambda_{i,\max}}{n_i} \asymp \mu_i$, under Assumption~\ref{ass:main}.
Note that $\|\mpp_{\mathcal{U}_i,\mathcal{U}_i}\|_{\max} = \|\mx_{\mathcal{U}_i}\|_{2 \to \infty}^2 \lesssim d \frac{\lambda_{i,\max}}{n_i} \asymp \mu_i$ under Assumption~\ref{ass:main}, and see Section~\ref{sec:add discussion} of the supplementary
material for further discussion on $\mu_i$ and $\|\mx_{\mathcal{U}_i}\|_{2\to\infty}$.
}

\vspace{0.25cm}

\begin{rema}\label{rm:lambda}
Let $\mpp\in\mathbb{R}^{N\times N}$ have rank $d$. Denote by $\lambda_{\max}$ and $\lambda_{\min}$ the largest and smallest non-zero eigenvalues of $\mpp$, respectively, and let $\muu\in\mathbb{R}^{N\times d}$ contain the eigenvectors corresponding to the non-zero eigenvalues of $\mpp$. Suppose (1) $\mpp$ has bounded condition number, i.e., $\lambda_{\max}/\lambda_{\min}\leq M'$ for some constant $M'>0$, and $\muu$ has bounded coherence, i.e., $\|\muu\|_{2\to\infty}\lesssim d^{1/2}N^{-1/2}$; (2) for each $i$, $\mathcal{U}_i$ are drawn uniformly at random from $\mathcal{U}$. Then
\begin{equation}\label{eq:relation lambda_i}
	\frac{n_i}{N}\lambda_{\min} \lesssim \lambda_{i,\min} \leq \lambda_{i,\max}\lesssim \frac{dn_i}{N}\lambda_{\max}, 
  \quad \text{and} \quad \|\muu^{(i)}\|_{2\to\infty}\lesssim\frac{d^{1/2}}{n_i^{1/2}}
\end{equation}
 with high probability, and %the condition in 
 Eq.~\eqref{eq:rem_1} holds (see Lemma~\ref{lemma:i->global} for more details). %, where $n_i = |\mathcal{U}_i|$. 
  % $\lambda_{i,\max} \lambda_{i,\min}$, and $\muu^{(i)}$ %, and $\lambda_{\min}(\mpp_{\mathcal{U}_i\cap\mathcal{U}_j,\mathcal{U}_i\cap\mathcal{U}_j})$ 
  % satisfy the conditions in Assumption~\ref{ass:main} with high probability; see Lemma~\ref{lemma:i->global} for details. %This assumption is for simplicity and clarity in presenting the results.
%The condition in 
%
%
\end{rema}
%\begin{remark}
%	Eq.~\eqref{eq:assm1_con} ensures that $\hat\mx^{(i)}$ is a consistent estimator of $\mx_{\mathcal{U}_i}$, and is generally trivially satisfies; see Remark~\ref{rem:example1_setting}.
%\end{remark}

{\color{black} 
%We first present an informal, and simplified discussion of the theoretical results to provide a basic intuition about the properties of CMMI, as a formal statement of the general result can be somewhat complex.
We first present an informal, simplified version of the theoretical results to provide a basic intuition about the properties of CMMI, as the formal statement, while being  quite more general, is also more complex.
Suppose that $n_i \asymp n, q_i \equiv q, \sigma_i \equiv \sigma$, and $\lambda_{i,\min} \asymp \frac{n_i}{N} \lambda_{\min}$ for all $i$, with the overlap sizes satisfying $n_{i_{\ell-1},i_\ell} \equiv m \geq d$ for all $1 \leq \ell \leq L$. This setting corresponds to a scenario where the block sizes, overlap sizes, noise levels, and missing rates are uniform, and we further suppose that the entries are bounded (see Remark~\ref{rem:example1_setting} for more details on this case).
Then we have the entrywise error bound 
$$
        \|\hat\mpp_{\mathcal{U}_{i_0},\mathcal{U}_{i_L}}-\mpp_{\mathcal{U}_{i_0},\mathcal{U}_{i_L}}\|_{\max}
        \lesssim \frac{(1+\sigma)\log^{1/2}n}{(qn)^{1/2}}
        $$
        with high probability, provided that 
        $
        L\Big(\frac{(1+\sigma)\log^{1/2} n}{(qn)^{1/2}}
        +\frac{1}{m^{1/2}}\Big)
        \lesssim 1
        $.
Furthermore, under mild conditions we also establish an entrywise normal approximation. More specifically, for any $s \in [n_{i_0}], t \in [n_{i_L}]$, we have
        $$
	\tilde\sigma_{s,t}^{-1}
	\big( \hat{\mpp}_{\mathcal{U}_{i_0},\mathcal{U}_{i_L}} - \mpp_{\mathcal{U}_{i_0},\mathcal{U}_{i_L}}\big)_{s,t}
	    \rightsquigarrow \mathcal{N}(0,1)
$$
as $n\to\infty$, where the standard deviation satisfies
$
\tilde{\sigma}_{s,t} \lesssim \frac{\sigma^2 + (1 - q)}{qn}.
$

We now present the formal theoretical results, which are applicable to a much broader range of settings while also yielding more detailed analysis.
We first consider the case where we only have two overlapping submatrices $\ma^{(i)}$ and $\ma^{(j)}$.
%In this case, we can recover any missing 
%Then the entry $\mpp_{s,t}$ %$(u_s,u_t)$ 
%where $s\in \mathcal{U}_i, t\in \mathcal{U}_j$ % but $s,t \notin\mathcal{U}_i\cap \mathcal{U}_j$, 
%By Algorithm~\ref{Alg_chain}, we estimate it 
%is estimated by $\hat{\mathbf{x}}^{(i)\top}_{\langle s\rangle} \mw^{(i,j)} \hat{\mathbf{x}}^{(j)}_{\langle t\rangle}$, where $\hat{\mathbf{x}}^{(i)}_{\langle s\rangle}$ is the row corresponding to the entity $s$ in $\hat\mx^{(i)}$ and $\hat{\mathbf{x}}^{(j)}_{\langle t\rangle}$ is the row corresponding to the entity $t$ in  $\hat\mx^{(j)}$.
%So its estimation error is an entry in $
%\hat\mpp_{\mathcal{U}_{i_0},\mathcal{U}_{i_L}}
%-\mpp_{\mathcal{U}_i,\mathcal{U}_j} 
	%=\hat\mx^{(i)} \mw^{(i,j)} \hat\mx^{(j)\top}
	%-\mx_{\mathcal{U}_i}\mx_{\mathcal{U}_j}^\top$. 
	Theorem~\ref{thm:R(i,j)} presents an expansion for 
$ \hat\mpp_{\mathcal{U}_{i},\mathcal{U}_{j}}
-\mpp_{\mathcal{U}_i,\mathcal{U}_j} 
	=\hat\mx^{(i)} \mw^{(i,j)} \hat\mx^{(j)\top}
	-\mx_{\mathcal{U}_i}\mx_{\mathcal{U}_j}^\top$. 
 % $\hat\mx^{(i)} \mw^{(i,j)} \hat\mx^{(j)\top}
%	-\mx_{\mathcal{U}_i}\mx_{\mathcal{U}_j}^\top$.
}
\begin{theorem}
\label{thm:R(i,j)}
    Let $\ma^{(i)}$ and $\ma^{(j)}$ be overlapping submatrices satisfying Assumption~\ref{ass:main}. % with $L=1$. 
    For their overlap, suppose $\mathrm{rk}(\mpp_{\mathcal{U}_i \cap \mathcal{U}_{j}, \mathcal{U}_i \cap \mathcal{U}_j}) = d$, and define
    \begin{equation}\label{eq:notation ij}
    	\begin{aligned}
    	n_{i,j}:=&|\mathcal{U}_i \cap \mathcal{U}_{j}|,\quad \vartheta_{i,j} := \lambda_{\max}(\mx_{\mathcal{U}_i \cap \mathcal{U}_j}^{\top}\mx_{\mathcal{U}_i \cap \mathcal{U}_j}), \quad 
    \theta_{i,j} := \lambda_{\min}(\mx_{\mathcal{U}_i \cap \mathcal{U}_j}^{\top}\mx_{\mathcal{U}_i \cap \mathcal{U}_j}),
    \\\alpha_{i,j}:=&\frac{n_{i,j} \gamma_i  \gamma_j }{{\theta_{i,j}}(q_in_i  \mu_i)^{1/2} (q_jn_j  \mu_j)^{1/2} }
    + \frac{{(n_{i,j} \vartheta_{i,j})^{1/2}}}{{\theta_{i,j}}}\Big(\frac{\gamma_i^2 }{q_in_i  \mu_i^{3/2}}
    +\frac{\gamma_j^2 }{q_jn_j  \mu_j^{3/2}}\Big)\\
      & + \frac{n_{i,j}^{1/2}\|\mx_{\mathcal{U}_i\cap \mathcal{U}_j}\|_{2\to\infty}}{{\theta_{i,j}}}\Big(\frac{\gamma_i}{(q_i n_i\mu_i)^{1/2}} +\frac{\gamma_j}{(q_j n_j\mu_j)^{1/2}}\Big).
    \end{aligned}
    \end{equation}
%    Suppose $\alpha_{i,j}\ll \theta_{i,j}$.
    Let  $\me^{(i)}:=\ma^{(i)}-\mpp^{(i)}$ for any $i$. % $\me^{(j)}=\ma^{(j)}-\mpp^{(j)}$.
    We then have
    \begin{equation}
      \label{eq:expansion_theorem1}
	\begin{aligned}
 \hat\mpp_{\mathcal{U}_{i},\mathcal{U}_{j}}
-\mpp_{\mathcal{U}_i,\mathcal{U}_j} 
		&=\me^{(i)}\mx_{\mathcal{U}_i}(\mx_{\mathcal{U}_i}^\top\mx_{\mathcal{U}_i})^{-1} \mx_{\mathcal{U}_j}^\top
	    +\mx_{\mathcal{U}_i} (\mx_{\mathcal{U}_j}^\top\mx_{\mathcal{U}_j})^{-1} \mx_{\mathcal{U}_j}^\top \me^{(j)}
	    +\mr^{(i,j)} + \ms^{(i,j)},
	\end{aligned}
    \end{equation}
    where 
	$\mr^{(i,j)}$ and $\ms^{(i,j)}$ are random matrices satisfying
        \begin{equation}
          \label{eq:bound_theorem1}
	\begin{aligned}
	\|\mr^{(i,j)}\|_{\max}
		\lesssim & \frac{\gamma_i \gamma_j}{(q_i n_i \mu_i)^{1/2}(q_j n_j \mu_j)^{1/2}} + \Big(\frac{\gamma_i^2}{q_in_i  \mu_i^{3/2}}
	+\frac{\gamma_i }{q_i^{1/2}n_i \mu_i^{1/2}}\Big)
	\|\mx_{\mathcal{U}_j}\|_{2\to\infty} \\ &
	+\Big(\frac{\gamma_j^2}{q_jn_j  \mu_j^{3/2}}
	+\frac{\gamma_j }{q_j^{1/2}n_j \mu_j^{1/2}}\Big)
	\|\mx_{\mathcal{U}_i}\|_{2\to\infty},
	\end{aligned}
\end{equation}
\begin{equation}
          \label{eq:bound_theorem2}
	\begin{aligned}
\|\ms^{(i,j)}\|_{\max} &\lesssim 
      {\alpha_{i,j}}\|\mx_{\mathcal{U}_i}\|_{2 \to \infty}\|\mx_{\mathcal{U}_j}\|_{2 \to \infty}
	\end{aligned}
        \end{equation}
with high probability. Furthermore suppose
\begin{equation}\label{eq:thm1_add_condition}
\alpha_{i,j}
\min\Big\{\frac{(q_{i}n_{i}  \mu_{i})^{1/2}}{ \gamma_{i}}\|\mx_{\mathcal{U}_{i}}\|_{2 \to \infty},
\frac{(q_{j}n_{j}  \mu_{j})^{1/2}}{ \gamma_{j}}\|\mx_{\mathcal{U}_{j}}\|_{2 \to \infty}\Big\}
\lesssim 1.
\end{equation}
Then $\me^{({i})}\mx_{\mathcal{U}_{i}}(\mx_{\mathcal{U}_{i}}^\top\mx_{\mathcal{U}_{i}})^{-1} \mx_{\mathcal{U}_{j}}^\top
	    +\mx_{\mathcal{U}_{i}} (\mx_{\mathcal{U}_{j}}^\top\mx_{\mathcal{U}_{j}})^{-1} \mx_{\mathcal{U}_{j}}^\top \me^{({j})}$ is the dominant term and 
\begin{equation*}
	\begin{aligned}
 \| \hat{\mpp}_{\mathcal{U}_{i},\mathcal{U}_{j}} - \mpp_{\mathcal{U}_{i},\mathcal{U}_{j}}\|_{\max} 
	    \lesssim \frac{\gamma_{i}}{(q_{i}n_{i}  \mu_{i})^{1/2}}\|\mx_{\mathcal{U}_{j}}\|_{2 \to \infty} + 
\frac{\gamma_{j}}{(q_{j}n_{j}  \mu_{j})^{1/2}}\|\mx_{\mathcal{U}_{i}}\|_{2 \to \infty}
\end{aligned}
\end{equation*}
with high probability.

\end{theorem} 

\begin{rema}
  \label{rem:interpret}
  The expansion in Eq.~\eqref{eq:expansion_theorem1} consists of four terms, with the first two terms being linear transformations of the additive noise matrices $\me^{(i)}$ and $\me^{(j)}$. The third term $\mr^{(i,j)}$ corresponds to second-order estimation errors for $\hat{\mx}^{(i)}$ and
$\hat{\mx}^{(j)}$, and hence Eq.~\eqref{eq:bound_theorem1} only depends
on quantities associated with ${\mx}_{\mathcal{U}_i}$ and ${\mx}_{\mathcal{U}_j}$. The last term $\ms^{(i,j)}$ 
corresponds to the error when aligning the overlaps $\hat{\mx}_{ \langle\mathcal{U}_i \cap \mathcal{U}_j\rangle}^{(i)}$ and
$\hat{\mx}_{ \langle\mathcal{U}_i \cap \mathcal{U}_j\rangle}^{(j)}$ and hence Eq.~\eqref{eq:bound_theorem2} depends on $\mx_{\mathcal{U}_i \cap \mathcal{U}_j}$.
{\color{black} Finally, Eq.~\eqref{eq:thm1_add_condition} ensures that $\ms^{(i,j)}$ is bounded by the first two terms, and is a mild and natural assumption in many settings; see Remark~\ref{rem:example1_setting} for further discussion. }
\end{rema}

Next we consider a chain of overlapping submatrices as described in Algorithm~\ref{Alg_chain}. %, then we can recover $\mpp_{\mathcal{U}_{i_0},\mathcal{U}_{i_L}}=\mx_{\mathcal{U}_{i_0}}\mx_{\mathcal{U}_{i_L}}^\top$ as $\hat\mx^{(i_0)} 
	%	\mw^{(i_0,i_1)}\mw^{(i_1,i_2)}\cdots\mw^{(i_{L-1},i_L)}
	%	\hat\mx^{(i_L)\top}$.  
	Theorem~\ref{thm:R(i0,...,iL)} presents the expansion of the
estimation error for
$\hat{\mpp}_{\mathcal{U}_{i_0},\mathcal{U}_{i_L}}$ and
Theorem~\ref{thm:normal} leverages this expansion to derive an
entrywise normal approximation for
$\hat{\mpp}_{\mathcal{U}_{i_0},\mathcal{U}_{i_L}} -
\mpp_{\mathcal{U}_{i_0},\mathcal{U}_{i_L}}$. While the statements of Theorem~\ref{thm:R(i0,...,iL)} and Theorem~\ref{thm:normal}
appear somewhat complicated at first glance, this is intentional as they make the results more general and thus applicable to a wider range of settings. 
Indeed, we allow for $(n_i, \sigma_i, q_i)$ to have different magnitudes as well as the overlaps
to be of very different sizes $n_{i,j}$. For example we can have
$n_1 \gg n_2 \gg n_3$ but $q_1 \ll q_2 \ll q_3$ while $n_{1,2} \ll n_{2,3}$ but $\sigma_1 \asymp \sigma_2 \ll \sigma_3$. 
If $(n_i, q_i, \sigma_i, n_{i,j}) \equiv (n, q, \sigma, m)$ then 
these results can be simplified considerably; see Remark~\ref{rem:example1_setting}. 
\begin{theorem}
\label{thm:R(i0,...,iL)}
Consider a chain of overlapping submatrices $(\ma^{(i_0)}, \dots, \ma^{(i_L)})$ satisfying Assumption~\ref{ass:main}. 
For all overlaps $1\leq \ell \leq L$, suppose $\mathrm{rk}(\mpp_{\mathcal{U}_{i_{\ell-1}}  \cap \mathcal{U}_{i_{\ell}} , \mathcal{U}_{i_{\ell-1}} \cap \mathcal{U}_{i_\ell}}) = d$, 
and define $n_{{i_{\ell-1}} ,{i_{\ell}} },\vartheta_{{i_{\ell-1}} ,{i_{\ell}} },\theta_{{i_{\ell-1}} ,{i_{\ell}} },\alpha_{{i_{\ell-1}},{i_{\ell}}}$ as in Eq.~\eqref{eq:notation ij}.
%$$
%\begin{aligned}
%	n_{{i_{\ell-1}} ,{i_{\ell}} }:=&|\mathcal{U}_{\ell-1} \cap \mathcal{U}_{\ell}|,
%    \quad 
%    \vartheta_{{i_{\ell-1}} ,{i_{\ell}} } := \lambda_{\max}(\mx_{\mathcal{U}_{i_{\ell-1}}  \cap \mathcal{U}_{i_{\ell}} }^{\top}\mx_{\mathcal{U}_{i_{\ell-1}}  \cap \mathcal{U}_{i_{\ell}} }),
%    \quad 
%    \theta_{{i_{\ell-1}} ,{i_{\ell}} } := \lambda_{\min}(\mx_{\mathcal{U}_{i_{\ell-1}}  \cap \mathcal{U}_{i_{\ell}} }^{\top}\mx_{\mathcal{U}_{i_{\ell-1}}  \cap \mathcal{U}_{i_{\ell}} }),\\
%    \alpha_{{i_{\ell-1}},{i_{\ell}}}
 %       	:=&\frac{n_{{i_{\ell-1}},{i_{\ell}}} \gamma_{i_{\ell-1}}  \gamma_{i_{\ell}} }{{\theta_{{i_{\ell-1}},{i_{\ell}}}}(q_{i_{\ell-1}}n_{i_{\ell-1}}  \mu_{i_{\ell-1}})^{1/2} (q_{i_{\ell}}n_{i_{\ell}}  \mu_{i_{\ell}})^{1/2} }
 %       	+\frac{(n_{{i_{\ell-1}},{i_{\ell}}} \vartheta_{{i_{\ell-1}},{i_{\ell}}})^{1/2}}{{\theta_{{i_{\ell-1}},{i_{\ell}}}}}\Big(\frac{\gamma_{i_{\ell-1}}^2 }{q_{i_{\ell-1}}n_{i_{\ell-1}}  \mu_{i_{\ell-1}}^{3/2}}
%    +\frac{\gamma_{i_{\ell}}^2 }{q_{i_{\ell}}n_{i_{\ell}}  \mu_{i_{\ell}}^{3/2}}\Big)
%    \\&+\frac{n_{{i_{\ell-1}},{i_{\ell}}}^{1/2}\|\mx_{\mathcal{U}_{i_{\ell-1}}\cap \mathcal{U}_{i_{\ell}}}\|_{2\to\infty}}{{\theta_{{i_{\ell-1}},{i_{\ell}}}}}\Big(\frac{\gamma_{i_{\ell-1}}}{(q_{i_{\ell-1}} n_{i_{\ell-1}}\mu_{i_{\ell-1}})^{1/2}} +\frac{\gamma_{i_{\ell}}}{(q_{i_{\ell}} n_{i_{\ell}}\mu_j)^{1/2}}\Big).
%\end{aligned}
%$$
%Suppose $\alpha_{{i_{\ell-1}},{i_{\ell}}}\ll \theta_{{i_{\ell-1}},{i_{\ell}}}$ for all $1\leq\ell\leq L$.
    Let $\me^{(i)}:=\ma^{({i} )}-\mpp^{({i} )}$ for all $i$. % $\ell= \{0,L\}$.
     We then have 
	$$
	\begin{aligned}
 \hat{\mpp}_{\mathcal{U}_{i_0},\mathcal{U}_{i_L}} - \mpp_{\mathcal{U}_{i_0},\mathcal{U}_{i_L}}
	%	\hat\mx^{(i_0)} \tilde{\mw}^{(i_0, i_{L})}
	%	\mw^{(i_0,i_1)}%\mw^{(i_1,i_2)}
  %\cdots\mw^{(i_{L-1},i_L)}
	%	\hat\mx^{(i_L)\top}
	%    -\mx_{\mathcal{U}_{i_0}}\mx_{\mathcal{U}_{i_L}}^\top 
	    =\me^{({i_0})}\mx_{\mathcal{U}_{i_0}}(\mx_{\mathcal{U}_{i_0}}^\top\mx_{\mathcal{U}_{i_0}})^{-1} \mx_{\mathcal{U}_{i_L}}^\top
	    +\mx_{\mathcal{U}_{i_0}} (\mx_{\mathcal{U}_{i_L}}^\top\mx_{\mathcal{U}_{i_L}})^{-1} \mx_{\mathcal{U}_{i_L}}^\top \me^{({i_L})}
	 + \mr^{(i_0,i_L)} + \ms^{(i_0,i_1,\dots,i_L)},
	\end{aligned}
	$$
	where $\mr^{(i_0,i_L)}$ and $\ms^{(i_0,i_1,\dots,i_L)}$ are random matrices satisfying
        \begin{equation}
        \label{eq:bound_ri0iL}
 	\begin{split}
		\|\mr^{(i_0,i_L)}\|_{\max}  \lesssim &
\frac{\gamma_{i_{0}} \gamma_{i_{L}} }{(q_{i_{0}} n_{i_{0}} \mu_{i_{0}})^{1/2}(q_{i_{L}} n_{i_{L}} \mu_{i_{L}})^{1/2}} + \Big(\frac{\gamma_{i_{0}}^2}{q_{i_{0}}n_{i_{0}}  \mu_{i_{0}}^{3/2}}
	+\frac{\gamma_{i_{0}} }{q_{i_{0}}^{1/2}n_{i_{0}} \mu_{i_{0}}^{1/2}}\Big)
	\|\mx_{\mathcal{U}_{i_{L}}}\|_{2\to\infty} \\ &+
	\Big(\frac{\gamma_{i_{L}}^2}{q_{i_{L}}n_{i_{L}}  \mu_{i_{L}}^{3/2}}
	+\frac{\gamma_{i_{L}} }{q_{i_{L}}^{1/2}n_{i_{L}} \mu_{i_{L}}^{1/2}}\Big)
	\|\mx_{\mathcal{U}_{i_{0}}}\|_{2\to\infty}
    \end{split}
    \end{equation}
    \begin{equation}
    \label{eq:bound_si0iL}
  \|\ms^{(i_0,i_1,\dots,i_{L})}\|_{\max} \lesssim 
  \Big[\sum_{\ell=1}^{L}{\alpha_{{i_{\ell-1}},{i_{\ell}}}}
    \Big]\cdot
    \|\mx_{\mathcal{U}_{i_0}}\|_{2 \to \infty}\cdot	\|\mx_{\mathcal{U}_{i_L}}\|_{2 \to \infty}
\end{equation}
with high probability. Furthermore suppose
\begin{equation}\label{eq:thm2_add_condition}
\Big[\sum_{\ell=1}^{L}{\alpha_{{i_{\ell-1}},{i_{\ell}}}}\Big]
\min\Big\{\frac{(q_{i_0}n_{i_0}  \mu_{i_0})^{1/2}}{ \gamma_{i_0}}\|\mx_{\mathcal{U}_{i_0}}\|_{2 \to \infty},
\frac{(q_{i_L}n_{i_L}  \mu_{i_L})^{1/2}}{ \gamma_{i_L}}\|\mx_{\mathcal{U}_{i_L}}\|_{2 \to \infty}\Big\}
\lesssim 1.
\end{equation}
Then $\me^{({i_0})}\mx_{\mathcal{U}_{i_0}}(\mx_{\mathcal{U}_{i_0}}^\top\mx_{\mathcal{U}_{i_0}})^{-1} \mx_{\mathcal{U}_{i_L}}^\top
	    +\mx_{\mathcal{U}_{i_0}} (\mx_{\mathcal{U}_{i_L}}^\top\mx_{\mathcal{U}_{i_L}})^{-1} \mx_{\mathcal{U}_{i_L}}^\top \me^{({i_L})}$ is the dominant term and 
\begin{equation*}
	\begin{aligned}
 \| \hat{\mpp}_{\mathcal{U}_{i_0},\mathcal{U}_{i_L}} - \mpp_{\mathcal{U}_{i_0},\mathcal{U}_{i_L}}\|_{\max} 
	    \lesssim 
	    \frac{ \gamma_{i_0}}{(q_{i_0}n_{i_0}  \mu_{i_0})^{1/2}}\|\mx_{\mathcal{U}_{i_L}}\|_{2 \to \infty} + 
\frac{ \gamma_{i_L}}{(q_{i_L}n_{i_L}  \mu_{i_L})^{1/2}}\|\mx_{\mathcal{U}_{i_0}}\|_{2 \to \infty}
\end{aligned}
\end{equation*}
with high probability.
\end{theorem}
The only difference between Theorems~\ref{thm:R(i,j)} and \ref{thm:R(i0,...,iL)} is in the upper bound for $\ms^{(i_0,i_1, \dots, i_L)}$ compared to that for $\ms^{(i,j)}$. Indeed, $\me^{({i_0})}\mx_{\mathcal{U}_{i_0}}(\mx_{\mathcal{U}_{i_0}}^\top\mx_{\mathcal{U}_{i_0}})^{-1} \mx_{\mathcal{U}_{i_L}}^\top
	    +\mx_{\mathcal{U}_{i_0}} (\mx_{\mathcal{U}_{i_L}}^\top\mx_{\mathcal{U}_{i_L}})^{-1} \mx_{\mathcal{U}_{i_L}}^\top \me^{({i_L})}$ and $\mr^{(i_0,i_L)}$ in Theorem~\ref{thm:R(i0,...,iL)} only depend on $\mx_{\mathcal{U}_{i_0}}$ and $\mx_{\mathcal{U}_{i_L}}$, but not on the chain linking them, and thus their upper bounds are the same as that in Theorem~\ref{thm:R(i,j)} for $i = i_0$ and $j = i_L$. In contrast, from our discussion in Remark~\ref{rem:interpret}, the term $\ms^{(i_0,i_1,\dots,i_L)}$ corresponds to the alignment error between $\hat{\mx}^{(i_0)}$ and $\hat{\mx}^{(i_L)}$. As $\mathcal{U}_{i_0}$ and $\mathcal{U}_{i_L}$ need not share any overlap, this alignment is obtained via a sequence of orthogonal Procrustes transformations between $\hat{\mx}^{(i_{\ell-1})}_{ \langle\mathcal{U}_{i_{\ell-1}} \cap \mathcal{U}_{i_{\ell}}\rangle}$ and $\hat{\mx}^{(i_{\ell})}_{ \langle\mathcal{U}_{i_{\ell-1}} \cap \mathcal{U}_{i_{\ell}}\rangle}$ for $1\leq\ell \leq L$. The accumulated error for these transformations is reflected in the term $\sum_{\ell=1}^{L} {\alpha_{i_{\ell-1},i_{\ell}}}$ and depends on the whole chain. 
	   If $L$ is not too large relative to the overlaps $\{n_{i,j}\}$ then the error in $\ms^{(i_0,i_1,\dots,i_L)}$ is negligible compared to that of $\me^{({i_0})}\mx_{\mathcal{U}_{i_0}}(\mx_{\mathcal{U}_{i_0}}^\top\mx_{\mathcal{U}_{i_0}})^{-1} \mx_{\mathcal{U}_{i_L}}^\top
	    +\mx_{\mathcal{U}_{i_0}} (\mx_{\mathcal{U}_{i_L}}^\top\mx_{\mathcal{U}_{i_L}})^{-1} \mx_{\mathcal{U}_{i_L}}^\top \me^{({i_L})}$, and consequently our entrywise error rate depends only on $\mx_{\mathcal{U}_{i_0}}$ and $\mx_{\mathcal{U}_{i_L}}$ rather than on the chain linking them.
	    \begin{theorem}
\label{thm:normal}
Consider the setting of Theorem~\ref{thm:R(i0,...,iL)}. For $i\in\{i_0,i_L\}$, let $\md^{(i)}$ be a $n_{i}\times n_{i}$ matrix whose entries are of the form
$$
\begin{aligned}
	&\md^{(i)}_{k_1,k_2}:=
	[\mathrm{Var}(\mn^{(i)}_{k_1,k_2})
	+(1-q_{i})(\mpp_{k_1,k_2}^{(i)})^2]/q_{i}\text{ for any }k_1,k_2\in[n_{i}].
\end{aligned}
$$
Define $\mb^{(i_0,i_L)}, \mb^{(i_L,i_0)}$ as
\begin{gather*}
\mb^{(i_0,i_L)}:=\mx_{\mathcal{U}_{i_0}}(\mx_{\mathcal{U}_{i_0}}^\top\mx_{\mathcal{U}_{i_0}})^{-1} \mx_{\mathcal{U}_{i_L}}^\top,
%\text{ and }
\quad \mb^{(i_L,i_0)}:=\mx_{\mathcal{U}_{i_L}}(\mx_{\mathcal{U}_{i_L}}^\top\mx_{\mathcal{U}_{i_L}})^{-1} \mx_{\mathcal{U}_{i_0}}^\top.
\end{gather*}
For any fixed $(s,t)\in[n_{i_0}]\times [n_{i_L}]$, define $\tilde\sigma^2_{s,t}$ as
\begin{gather*}
		\tilde\sigma^2_{s,t}:=
\sum_{k_1=1}^{n_{i_0}}(\mb^{(i_0,i_L)}_{k_1,t})^2\md^{({i_0})}_{s,k_1}
    +\sum_{k_2=1}^{n_{i_L}}(\mb^{(i_L,i_0)}_{k_2,s})^2\md^{({i_L})}_{t,k_2}.%}_{\mSigma^{(i_0,i_L,2)}}.
\end{gather*}
% Note that $|\tilde\sigma^2_{s,t}|\lesssim (\sigma_{i_0}^2+(1-q_{i_0})\|\mpp^{(i_0)}\|_{\max}^2)/(n_{i_0} q_{i_0}) +
% (\sigma_{i_L}^2+(1-q_{i_L})\|\mpp^{(i_L)}\|_{\max}^2)/(n_{i_L} q_{i_L})$
% almost surely. 
Furthermore, denote
\[ \zeta_{s, t} := 
\frac{(\sigma_{i_0}^2+(1-q_{i_0})\|\mpp_{\mathcal{U}_{i_0},\mathcal{U}_{i_0}}\|_{\max}^2) \|\bm{x}_{i_L,t}\|^2}{q_{i_0} n_{i_0} \mu_{i_0}} +
\frac{(\sigma_{i_L}^2+(1-q_{i_L})\|\mpp_{\mathcal{U}_{i_L},\mathcal{U}_{i_L}}\|_{\max}^2)\|\bm{x}_{i_0,s}^2\|}{q_{i_L} n_{i_L} \mu_{i_L}},\]
where $\bm{x}_{i_0,s}$ and $\bm{x}_{i_L,t}$ denote the $s$-th row and $t$-th row of $\mx_{\mathcal{U}_{i_0}}$ and $\mx_{\mathcal{U}_{i_L}}$ respectively.
Note that $\tilde{\sigma}^2_{s,t}
                \lesssim \zeta_{s,t}$.
Suppose the following conditions 
\begin{gather}
	\label{eq:con_2}
	%\begin{aligned}
		\tilde{\sigma}^2_{s,t}
                \gtrsim \zeta_{s,t}, \\
% (\sigma_{i_0}^2+(1-q_{i_0})\|\mpp^{(i_0)}\|_{\max}^2)/(n_{i_0} q_{i_0} \mu_{i_0}) +
% (\sigma_{i_L}^2+(1-q_{i_L})\|\mpp^{(i_L)}\|_{\max}^2)/(n_{i_L} q_{i_L} \mu_{i_L})
	%\end{aligned}
%\end{equation}	
% and that
 %\begin{equation}
%  \label{eq:con_3}
% %\begin{aligned}
% 	\frac{\|\mpp\|_{\max}}{(\sigma^2+(1-q)\|\mpp\|_{\max}^2)^{1/2}}\lesssim 1, \\
%	\end{aligned}
%\end{equation}
% \begin{equation}
 \label{eq:con_3}	
\frac{\|\mx_{\mathcal{U}_{i}}\|_{2 \to \infty}\log^{1/2} (q_{i}n_{i})}{(q_{i}n_{i}  \mu_{i})^{1/2}}+\frac{\|\mpp_{\mathcal{U}_{i},\mathcal{U}_{i}}\|_{\max} \cdot \|\mx_{\mathcal{U}_{i}}\|_{2 \to \infty}}{(n_{i} \mu_{i})^{1/2} \sigma_{i}}=o(1)\text{ for }i=\{i_0,i_L\},\\
\label{eq:con_4}
\begin{aligned}
    (r_{\infty} + s_{\infty})/\zeta_{s,t}^{1/2} = o(1)
	% &\frac{(\|\mpp\|_{\max}+\sigma)^2N^{1/2}\log N}{(\sigma^2+(1-q)\|\mpp\|_{\max}^2)^{1/2}(pq)^{1/2}\lambda_{\min}}
	% +
	% \frac{(\|\mpp\|_{\max}+\sigma)\log^{1/2}N}
	% 	{(\sigma^2+(1-q)\|\mpp\|_{\max}^2)^{1/2}(pN)^{1/2}}
	% \\+&\frac{L(\|\mpp\|_{\max}+\sigma)^2\breve p(pN)^{1/2}\log N}{(\sigma^2+(1-q)\|\mpp\|_{\max}^2)^{1/2}q^{1/2}\theta}
	% +
	% \frac{L(\|\mpp\|_{\max}+\sigma)(\breve ppN)^{1/2}\lambda_{\min}\log^{1/2}N}
	% 	{(\sigma^2+(1-q)\|\mpp\|_{\max}^2)^{1/2}N\theta}
	% 	=o(1)
\end{aligned}
\end{gather}
are satisfied, where $r_{\infty}$ and $s_{\infty}$ are upper bounds for $\|\mr^{(i_0,i_L)}\|_{\max}$ and $\|\ms^{(i_0,\dots,i_L)}\|_{\max}$ given in Theorem~\ref{thm:R(i0,...,iL)}.
We then have
$
%\begin{aligned}
	\tilde\sigma_{s,t}^{-1}
	\big( \hat{\mpp}_{\mathcal{U}_{i_0},\mathcal{U}_{i_L}} - \mpp_{\mathcal{U}_{i_0},\mathcal{U}_{i_L}}\big)_{s,t}
	    \rightsquigarrow \mathcal{N}(0,1)
%\end{aligned}
$
as $\min\{n_{i_0},n_{i_L}\} \rightarrow \infty$. 
\end{theorem}

\iffalse
\begin{remark}
We note that, for ease of exposition, all conditions and the bounds in Theorem~\ref{thm:R(i0,...,iL)} and Theorem~\ref{thm:normal} are presented in terms of $\|\mpp\|_{\max},\lambda_{\min}, \sigma, p, \breve{p},q,\theta$ which are ``global quantities", i.e., they depend on all of the $\{\ma^{(i)}\}$ or $\{\mpp^{(i)}\}$. % even when some of these submatrices are not part of the chain linking $\ma^{(i_0)}$ and $\ma^{(i_{L})}$. 
Replacing these quantities with $n_{i_0}, \dots, n_{i_L}$ and other parameters specific to $\{\ma^{(i)}\}$ and $\{\mpp^{(i)}\}$ along the given chain, as in the statement of Theorem~\ref{thm:R(i,j)}, would lead to weaker conditions and sharper bounds but at the cost of more complicated notations.
\end{remark}
\fi 
\begin{rema}\label{rm:qleq1-c}
{\color{black} Eq.~\eqref{eq:con_2} provides a lower bound for the entrywise variance $\tilde{\sigma}^2_{s,t}$ and is necessary as our normal approximations are for 
    $\tilde{\sigma}_{s,t}^{-1}(\hat{\mpp}_{\mathcal{U}_{i_0},\mathcal{U}_{i_L}} - 
    \mpp_{\mathcal{U}_{i_0},\mathcal{U}_{i_L}})_{s,t}$. 
	Eq.~\eqref{eq:con_3} ensures that each independent component of the dominant terms for $\big( \hat{\mpp}_{\mathcal{U}_{i_0},\mathcal{U}_{i_L}} - \mpp_{\mathcal{U}_{i_0},\mathcal{U}_{i_L}}\big)_{s,t}$ is not too large compared to $\tilde{\sigma}_{s,t}$. 
	Eq.~\eqref{eq:con_4} guarantees that the remainder terms $\|\mr^{(i_0,i_L)}\|_{\max}$ and $\|\ms^{(i_0,\dots,i_L)}\|_{\max}$ are negligible (when scaled by $\tilde{\sigma}_{s,t}^{-1})$.
	%Eq.~\eqref{eq:con_2} guarantees a certain rate of the variance $\tilde{\sigma}^2_{s,t}$.
	%Eq.~\eqref{eq:con_3} ensures that the estimation error from random missing entries in each observed block is not too large compared to $\tilde{\sigma}^2_{s,t}$. 
	%Eq.~\eqref{eq:con_3} guarantees $r_{\infty}$ and $s_{\infty}$ are negligible compared to $\tilde{\sigma}_{s,t}$. 
	These conditions are very mild; see Remark~\ref{rem:example1_setting}.}
	%If $q\leq 1-c$ for some constant $c > 0$, then 
%	for Theorem~\ref{thm:normal} the condition in Eq.~\eqref{eq:con_3} always holds, and the condition in Eq.~\eqref{eq:con_4} simplifies to 
%		$\frac{(\|\mpp\|_{\max}+\sigma)N^{1/2}\log N}{(pq)^{1/2}\lambda_{\min}}
%	+
%	\frac{\log^{1/2}N}
%		{(pN)^{1/2}}
%	+\frac{L(\|\mpp\|_{\max}+\sigma)\breve p(pN)^{1/2}\log N}{q^{1/2}\theta}
%	+
%	\frac{L(\breve ppN)^{1/2}\lambda_{\min}\log^{1/2}N}
%		{N\theta}=o(1)$.
%	$\frac{L(\|\mpp\|_{\max}+\sigma)\log N}{pq^{1/2}\lambda_{\min}}=o(1)$ and $
%\frac{L \log^{1/2}N}
%		{p\breve p^{1/2}N}
%		=o(1)$. 
%Furthermore, i
	In addition, if $q_i = 1$ then %, similar to Remark~\ref{rm:no missing}, 
all terms depending on $\|\mpp_{\mathcal{U}_i,\mathcal{U}_i}\|_{\max}$ are dropped from these conditions. Specifically, $\gamma_i$ in Assumption~\ref{ass:main} simplifies to $\sigma_i\log^{1/2} n_i$, %$\zeta_{s, t}$ in Eq.~\eqref{eq:con_2} simplifies to $\frac{\sigma_{i_0}^2\|\bm{x}_{i_L,t}\|^2}{n_{i_0} \mu_{i_0}} +\frac{\sigma_{i_L}^2\|\bm{x}_{i_0,s}^2\|}{ n_{i_L} \mu_{i_L}}$,
%$\zeta_{s, t}$ in Theorem~\ref{thm:normal} simplifies to $\frac{\sigma_{i_0}^2\|\bm{x}_{i_L,t}\|^2}{ n_{i_0} \mu_{i_0}} +\frac{\sigma_{i_L}^2\|\bm{x}_{i_0,s}^2\|}{ n_{i_L} \mu_{i_L}}$; 
%the condition in 
and Eq.~\eqref{eq:con_3} simplifies to
$\frac{\|\mx_{\mathcal{U}_{i}}\|_{2 \to \infty}\log^{1/2} n_{i}}{(n_{i}  \mu_{i})^{1/2}}=o(1)$ for $i=\{i_0,i_L\}$.
If $\sigma_i = 0$ then all terms depending on $\sigma_i$ are dropped. Specifically, %the condition in 
Eq.~\eqref{eq:con_3} simplifies to $\frac{\|\mx_{\mathcal{U}_{i}}\|_{2 \to \infty}\log^{1/2} (q_{i}n_{i})}{(q_{i}n_{i}  \mu_{i})^{1/2}}=o(1)\text{ for }i=\{i_0,i_L\}$.
%the condition in Eq.~\eqref{eq:con_3} is eliminated, and the condition in Eq.~\eqref{eq:con_4} simplifies to
%$
%\frac{\sigma N^{1/2}\log N}{p^{1/2}\lambda_{\min}}
%	+
%	\frac{\log^{1/2}N}
%		{(pN)^{1/2}}
%	+\frac{L\sigma\breve p(pN)^{1/2}\log N}{q^{1/2}\theta}
%	+
%	\frac{L(\breve p pN)^{1/2}\lambda_{\min}\log^{1/2}N}
%		{N\theta}
%		=o(1)$ 
%$\frac{L\sigma\log N}{pq^{1/2}\lambda_{\min}}=o(1)$ and $\frac{L\log^{1/2}N}
%		{p\breve p^{1/2}N}
%		=o(1)$.
\end{rema}

\vspace{0.25cm}

\begin{rema}
%Suppose $\mathrm{Var}(\mn^{(i)}_{k_1,k_2})\equiv \sigma_i^2$
%Suppose the variances for each entry in $\mn^{(i_0)}$ and $\mn^{(i_L)}$ are identical to $\sigma_{i_0}^2$ and $
%\sigma_{i_L}^2$, respectively. Then f
Using Theorem~\ref{thm:normal} we can also construct a $(1 - \alpha)\times 100\%$ confidence interval for $\big(\mpp_{\mathcal{U}_{i_0},\mathcal{U}_{i_L}}\big)_{s,t}$
as 
$\big(\hat\mpp_{\mathcal{U}_{i_0},\mathcal{U}_{i_L}}\big)_{s,t}
\pm z_{\alpha/2} \hat{\tilde \sigma}_{s,t}$
     where $z_{\alpha/2}$ denotes the upper $\alpha/2$ quantile of the standard normal distribution and $\hat{\tilde \sigma}_{s,t}$ is a consistent estimate of 
     ${\tilde \sigma}_{s,t}$ based on
  \begin{gather*}
	    	\hat\mb^{(i_0,i_L)}=\hat\mx^{(i_0)}(\hat\mx^{(i_0)\top}\hat\mx^{(i_0)})^{-1} \tilde\mw^{(i_0,i_L)}\hat\mx^{(i_L)\top},\quad \hat\mb^{(i_L,i_0)}=\hat\mx^{(i_L)}(\hat\mx^{(i_L)\top}\hat\mx^{(i_L)})^{-1} \tilde\mw^{(i_0,i_L)\top}\hat\mx^{(i_0)\top}, \\
	     	\hat\md^{(i_0)}_{k_1,k_2}
	    	=(\ma^{(i_0)}_{k_1,k_2}-\hat\mpp^{(i_0)}_{k_1,k_2})^2,
   	\quad \hat\md^{(i_L)}_{k_1,k_2}
	    	=(\ma^{(i_L)}_{k_1,k_2}-\hat\mpp^{(i_L)}_{k_1,k_2})^2
      \end{gather*}
    with $\hat\mpp^{(i_0)}=\hat\mx^{(i_0)}\hat\mx^{(i_0)\top},\hat\mpp^{(i_L)}=\hat\mx^{(i_L)}\hat\mx^{(i_L)\top}$ and $\tilde\mw^{(i_0,i_L)}=\mw^{(i_0,i_1)}\mw^{(i_1,i_2)}\dots\mw^{(i_{L-1},i_L)}$;
    we leave the details to the interested readers. %$$
\end{rema}

\vspace{0.25cm}

\begin{rema}
\label{rem:example1_setting}

	We now provide an example to illustrate the above results.
	We first assume $\{n_i\}$ are of the same order, i.e., there exists an $n$ with $n_i \asymp n$ for all $i$. We also assume $q_i \equiv q$, $\sigma_i \equiv \sigma$ for all $i$. 
	We further suppose $\mpp \in \mathbb{R}^{N \times N}$ has $\Theta(N^2)$ entries that are lower bounded by some constant $c_0$ not depending on $N$, and that $\|\mpp\|_{\max} \leq c_1$ for some other constant $c_1$ not depending on $N$. Then $\|\mx\|_{2\to\infty}\lesssim 1$, and as $\mpp$ is low-rank with bounded condition number, we also have $\lambda_{\min} = \Theta(N)$. 
	By Eq.~\eqref{eq:relation lambda_i} we have $\lambda_{i,\min}\asymp \frac{n_i}{N}\lambda_{\min}$, 
    so $\mu_i\asymp 1$ for all $i$.
		For the overlaps we assume $n_{i_{\ell-1},i_{\ell}}\equiv m\geq d$ and 
        $\vartheta_{i_{\ell-1},i_{\ell}} \asymp \theta_{i_{\ell-1},i_{\ell}} = \Theta(m)$.
        Under this setting, the condition in Eq~\eqref{eq:assm1_con} simplifies to
        $
        \frac{(1+\sigma)\log^{1/2}n}{(qn)^{1/2}}\ll 1;
        $
        %$\{\alpha_{i_{\ell-1},i_\ell}\}$ in Theorem~\ref{thm:R(i0,...,iL)} have the bound 
        %$
        %\alpha_{i_{\ell-1},i_\ell}
        %\lesssim \frac{(1+\sigma)^2\log n}{qn}
        %+\frac{(\|\mpp\|_{\max}+\sigma)^2\log n}{qn}
        %+\frac{(1+\sigma)\log^{1/2} n}{m^{1/2}(qn)^{1/2}},
        %$
        %and t
        the error bounds in Eq.~\eqref{eq:bound_ri0iL} and Eq.~\eqref{eq:bound_si0iL} of Theorem~\ref{thm:R(i0,...,iL)} simplify to 
        $$\|\mr^{(i_0,i_L)}\|_{\max}
        	\lesssim \frac{(1+\sigma)^2\log n}{qn}+\frac{(1+\sigma)\log^{1/2} n}{q^{1/2}n},$$
        \begin{equation}\label{eq:bound_si0iL_2} 
        	\begin{aligned}
        	%&\|\mr^{(i_0,i_L)}\|_{\max}\lesssim \frac{(1+\sigma)^2\log n}{qn}+\frac{(1+\sigma)\log^{1/2} n}{q^{1/2}n},\\
        	\|\ms^{(i_0,i_1,\dots,i_L)}\|_{\max}
        	\lesssim L\Big(\frac{(1+\sigma)^2\log n}{qn}
        +\frac{(1+\sigma)\log^{1/2} n}{m^{1/2}(qn)^{1/2}}\Big)
        \end{aligned}
        \end{equation}
        with high probability. Furthermore provided the condition in Eq.~\eqref{eq:thm2_add_condition}, we also have
        $$
        \|\hat\mpp_{\mathcal{U}_{i_0},\mathcal{U}_{i_L}}-\mpp_{\mathcal{U}_{i_0},\mathcal{U}_{i_L}}\|_{\max}
        \lesssim \frac{(1+\sigma)\log^{1/2}n}{(qn)^{1/2}}
        $$
        with high probability,  and Eq.~\eqref{eq:thm2_add_condition} simplifies to
        $
        L\Big(\frac{(1+\sigma)\log^{1/2} n}{(qn)^{1/2}}
        +\frac{1}{m^{1/2}}\Big)
        \lesssim 1
        $ now.
        
        All conditions are then trivially satisfied and the estimate error converges to $0$ whenever $\sigma = O(1)$, $L=O(1)$, and $qn =\Omega(\log^2 n)$.
        Note that the overlap size $m$ can be as small as $d = \mathrm{rk}(\mpp)$; see Section~\ref{sec:simu limit overlap} for simulation results. 
     For Theorem~\ref{thm:normal}, %the $\zeta_{s, t}$ in Theorem~\ref{thm:normal} then has the bound
     %$
     %\zeta_{s,t}\asymp \frac{\sigma^2+(1-q)}{qn}
     %$, and 
     the condition in Eq.~\eqref{eq:con_3} is trivial, and the condition in Eq.~\eqref{eq:con_4} simplifies to
     \begin{equation}\label{eq:thm3_add}
     	\begin{aligned}
     	L\Big(\frac{(1+\sigma)^2\log n}{(\sigma^2+(1-q))^{1/2}(qn)^{1/2}}
        +\frac{(1+\sigma)\log^{1/2} n}{(\sigma^2+(1-q))^{1/2}m^{1/2}}\Big)
        =o(1).
     \end{aligned}
     \end{equation}
     Note that $\frac{(1+\sigma)}{(\sigma^2+(1-q))^{1/2}}$ is bounded when $q\leq 1-c$ for some constant $c>0$ (this condition can be omitted when $q=1$). 
     Then Eq.~\eqref{eq:thm3_add} is satisfied when $qn=\omega(\log^2 n)$ and $m=\omega (\log n)$, i.e. the overlap size is slightly larger than the minimal $d = \mathrm{rk}(\mpp)$.
     See Section~\ref{sec:add discussion} for further discussion on Assumption~\ref{ass:main} and the conditions in Theorems~\ref{thm:R(i,j)}, \ref{thm:R(i0,...,iL)}, and \ref{thm:normal}.
     \end{rema}

\subsection{Discussion and comparison with related work}
\label{sec:comparison}

%As mentioned in Section~\ref{sec:meth}, 
Our theoretical results are comparable to those of \cite{zhou2021multi} for two observed submatrices and to \cite{bishop2014deterministic} for a simple chain, %of observed submatrices, 
while being significantly stronger than both.
{\color{black}
%In particular, the error bounds in \cite{zhou2021multi} are in the spectral norm and those in \cite{bishop2014deterministic} are in the Frobenius norm, whereas based on our error $2\to\infty$ bound $\hat\mx^{(i)}$ we can have have error bound recovered matrix $\hat\mpp_{\mathcal{U}_{i_0},\mathcal{U}_{i_L}}$ are in the maximum entrywise norm and allow for heterogeneity among blocks.
In particular, the error bounds in \cite{zhou2021multi} and \cite{bishop2014deterministic} are given in terms of the spectral and Frobenius norms, which only provide coarse control over individual entries. In other words, the estimation error for each entry can only be bounded indirectly through the overall matrix norm and does not yield sharp {\em entrywise} guarantees.
In contrast, our analysis leverages a $2 \to \infty$ norm bound on $\hat{\mathbf{X}}^{(i)}$, which controls row-wise fluctuations and leads directly to a bound on the maximum entrywise error of the recovered matrix. 
Furthermore, our analysis also reveals the dominant error terms in our estimation from which we can obtain entrywise normal approximations. This type of inferential guarantees can not be achieved using either the spectral or Frobenius norm bounds.
In addition, our theoretical results allow for heterogeneity across blocks as we can handle cases where $(n_i, \sigma_i, q_i)$ have different magnitudes and the overlap sizes $n_{i,j}$ can 
vary across submatrices. %For example we can have $n_1 \gg n_2 \gg n_3$ but $q_1 \ll q_2 \ll q_3$ while $n_{12} \ll n_{23}$ but $\sigma_1 \asymp \sigma_2 \ll \sigma_3$. 
In contrast, \cite{zhou2021multi} assumes all $\{n_i\}$ are of the same order and all $\{n_{i,j}\}$ are of the same order, and their error bound for recovering the entries related to blocks $i$ and $j$ depends on $\sigma = \max_{k\in[m]} \sigma_k$ %(where the maximum is over all $k \in [m]$) 
rather than the specific $(\sigma_i,\sigma_j)$. % which are specific to blocks $i,j$
Our analysis is thus more general and also yields sharper results under heterogeneous settings.
Finally we relax the overlap requirement, as
\cite{zhou2021multi} requires %$(n_{i,j}/n_{i}) \asymp (n_i/N)$, meaning the 
overlaps to grow with submatrices, while %our derivations in Remark~\ref{rem:example1_setting} 
we show that the overlap size can be % as small as the minimal 
$d$.
}
Comparing our results to \cite{bishop2014deterministic}, we note that the error bound in Theorem~4 of
\cite{bishop2014deterministic} grows exponentially with the length of the chain (due to its dependency on $G^{k-2}$ where $k$ is the chain length), whereas our bound for $\hat\mpp_{\mathcal{U}_{i_0},\mathcal{U}_{i_L}}$ includes only a non-dominant term that grows linearly with the chain length; see Eq.~\eqref{eq:bound_si0iL} or Eq.~\eqref{eq:bound_si0iL_2}.
And the bound in Theorem~4 of \cite{bishop2014deterministic} is only applicable for small $\epsilon$ where $\|\ma^{(i)} - \mpp^{(i)}\|_F \leq \epsilon$, whereas our noise model allows $\epsilon$ to be of order $n_i$ with $\|\mx\|_F$ of order $\sqrt{N}$, rendering their result ineffective.

%\begin{remark}
We emphasize that the block-wise observation models in this paper, BONMI \citep{zhou2021multi} and SPSMC \citep{bishop2014deterministic} differ significantly from those in the standard matrix completion literature, which typically focuses on a single large matrix and assumes uniformly or independently sampled observed entries.
Nevertheless, the authors of BONMI have compared their results with other results in the standard matrix completion literature. For example, Remark~10 in \cite{zhou2021multi} shows that the upper bound for the spectral norm error of 
BONMI matches the minimax rate for the missing at random setting. 
As CMMI is an extension of BONMI to more than $2$ matrices, the above comparison is still valid.
Furthermore, our results for CMMI are in terms of the maximum entrywise norm and normal approximations, which are significant refinements of the spectral norm error in BONMI, and are also comparable to the best available results for standard matrix completion such as those in \cite{abbe2020entrywise} and \cite{chen2021spectral}. More specifically, consider the case of $n_i \equiv n \asymp N/L$ with $\lambda_{i,\min} \asymp \frac{n}{N} \lambda_{\min} \asymp \frac{1}{L} \lambda_{\min}$. Also suppose $q_i \equiv q$ and $\sigma_i \equiv \sigma$. Then CMMI has the maximum entrywise error bound of $\frac{(\|\mpp\|_{\max} + \sigma)\log^{1/2} N}{q^{1/2}(N/L)^{1/2}}$, which matches the rate in Theorem 3.4 of \cite{abbe2020entrywise} and Theorem~4.5 of \cite{chen2021spectral} up to a factor of $L^{-1/2}$, as the number 
of observed entries in our model is only $1/L$ times that for the standard matrix completion models. 
Finally the normal approximation result in Theorem~\ref{thm:normal} is analogous to Theorem~4.12 in \cite{chen2021spectral}, with the main difference being the expression for the normalizing variance as our model considers individual noise matrices $\me^{(i_0)}$ and $\me^{(i_L)}$ whereas \cite{chen2021spectral} consider a global noise matrix $\me$. % (which includes $\me^{(i_0)}$ and $\me^{(i_L)}$ as submatrices). 

Another related work is \cite{chang2022low} which considers matrix completion for sample covariance matrices with a spiked covariance structure. Sample covariance matrices differ somewhat from the data matrices considered in our paper as, while both our population data matrix $\mpp$ in Section~\ref{sec:meth} and their population covariance matrix $\mSigma$ are positive semidefinite, the entries of the sample covariance matrix $\hat{\mSigma}$ are {\em dependent}. Consequently, the settings in the two papers are related but not directly comparable.
Nevertheless, if we were to compare our results against Theorem~C.1 in \cite{chang2022low} (where we set $q_i \equiv 1$ in our model, as \cite{chang2022low} assume that the sample covariance submatrices are observed completely) then (1) we allow block sizes $\{p_k\}$ (theirs $\{p_k\}$ are our $\{n_i\}$) to differ significantly in magnitude; (2) more importantly, our error bounds depend at most linearly on the chain length, whereas their bounds grow exponentially with the chain length due to the dependency on $\xi^{K}$, (their $K$ is our $L$). As $\xi = \sqrt{\log \max_{i} n_i}$ (see Proposition~C.2 in \cite{chang2022low}), this results in a factor of $(\log \max_i n_i)^{K/2}$ that is highly undesirable as $K$ increases. 
 %\end{remark}

{\color{black}We finally note that for each observed submatrix with potentially independent random missing entries, estimating the latent positions $\hat{\mathbf{X}}^{(i)}$ using the scaled leading eigenvectors of $\mathbf{A}^{(i)}$ (as done in this paper) is rate-optimal. More specifically, from Lemma~\ref{lemma:hat X(i)W(i)-X} we have
$
\|\hat{\mathbf{P}}^{(i)} - \mathbf{P}^{(i)}\|_{\max} \lesssim \frac{(\|\mathbf{P}^{(i)}\|_{\max} + \sigma_i)\log^{1/2} n_i}{q_i^{1/2} n_i^{1/2}}
$
with high probability, 
where $\hat{\mathbf{P}}^{(i)} = \hat{\mathbf{X}}^{(i)} \hat{\mathbf{X}}^{(i)\top}$. This matches the best available entrywise error rates for matrix completion established in the literature, such as those in \cite{abbe2020entrywise} and \cite{chen2021spectral}. Lemma~\ref{lemma:hat X(i)W(i)-X} also implies
$\|\hat{\mathbf{P}}^{(i)} - \mathbf{P}^{(i)}\|_{F} \lesssim \frac{(\|\mathbf{P}^{(i)}\|_{\max} + \sigma_i)n_i^{1/2} \log^{1/2} n_i}{q_i^{1/2}}$ with high probability, and is the same (up to logarithmic factor) as the oracle bound from \cite[Eq.~(III.13)]{candes2010matrix}. The logarithmic factor is negligible and is due mainly to the fact that Lemma~\ref{lemma:hat X(i)W(i)-X} yields a concentration bound for $\|\hat{\mathbf{P}}^{(i)} - \mathbf{P}^{(i)}\|_{F}$ while \cite{candes2010matrix} is for
$\mathbb{E}[\|\hat{\mathbf{P}}^{(i)} - \mathbf{P}^{(i)}\|_{F}]$.
See Section~\ref{sec:simu initilization} of the supplementary material for empirical comparisons
between the SVD-based algorithm and other matrix completion methods for estimating $\{\mx^{(i)}\}$. 
Accurate initialization of $\{\hat{\mx}^{(i)}\}$ is crucial for the subsequent joint integration, as it leads to more precise estimates of the transformation matrices $\{\mathbf{W}^{(i,j)}\}$ and thereby contributing to improved overall recovery. 
% Indeed, our earlier analysis shows that the entry-wise error bounds for CMMI are rate-optimal even when compared to existing bounds for matrix completions in standard settings.
}

\section{Simulation Experiments}
\label{sec:simu}
 
We now present simulation experiments %for Algorithm~\ref{Alg_chain} 
to complement our theoretical results and compare the performance of CMMI against
existing state-of-the-art matrix completion algorithms.

\subsection{Estimation error of CMMI}

 \label{sec:simu1}

We simulate a chain of $(L+1)$ overlapping observed submatrices $\{\ma^{(i)}\}_{i=0}^L$ for the underlying population matrix $\mpp$ 
as described in Figure~\ref{fig:simulation_setting}, and then predict the unknown yellow block by  Algorithm~\ref{Alg_chain}.
Each $\mpp^{(i)}$ has the same dimension, i.e. $n_i\equiv n=pN$ for all $i=0,1,\dots,L$, and the overlap between $\mpp^{(i-1)}$ and $\mpp^{(i)}$ are set to $n_{i-1,i} \equiv m=\breve p n$ for all $i=1,\dots,L$.
We generate $\mpp=\muu\mLambda\muu^\top$ by sampling $\muu$ uniformly at random from the set of $N\times 3$ matrices %$\mo$ 
with orthonormal columns and set $\mLambda=\text{diag}(N, \tfrac{3}{4} N, \tfrac{1}{2} N)$. %, so that $\lambda_{\min} = \tfrac{1}{2} N$ in this setting.
We then generate symmetric noise matrices $\{\mn^{(i)}\}$ with $\mn^{(i)}_{st}\stackrel{i i d}{\sim}\mathcal{N}(0,\sigma^2)$ for all $i=0,1,\dots,L$ and all $s,t\in[n]$ with $s\leq t$. Finally, we set $\ma^{(i)} = (\mpp^{(i)} + \mn^{(i)}) \circ \bm{\Omega}^{(i)}$ where $\bm{\Omega}^{(i)}$ is a symmetric $n \times n$ matrix whose upper triangular entries are i.i.d. Bernoulli random variables with success probability $q_i \equiv q$.
 \begin{figure}[htbp!] 
\centering
\subfigure{\includegraphics[width=3.5cm]{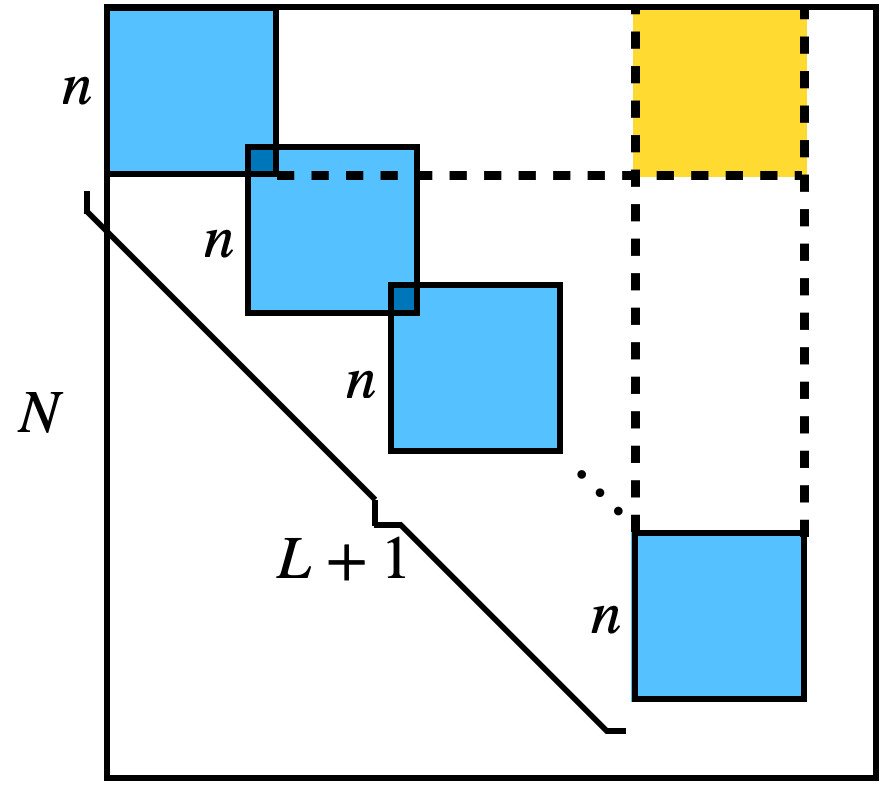}}
%\subfigure[]{\includegraphics[width=4cm]{image/simulation_disPCA/n_2.png}}
\caption{Simulation setting
}
\label{fig:simulation_setting}
\end{figure}

Recall that, by Theorem~\ref{thm:R(i0,...,iL)}, the estimation error for $\hat\mpp_{\mathcal{U}_{0},\mathcal{U}_{L}}-\mpp_{\mathcal{U}_{0},\mathcal{U}_{L}}$%$=\hat\mx^{(0)} 
%		\mw^{(0,1)}\mw^{(1,2)}\cdots\mw^{(L-1,L)}
%		\hat\mx^{(L)\top}
%	    -\mx_{\mathcal{U}_{0}}\mx_{\mathcal{U}_{L}}^\top$ 
can be decomposed into the first order approximation $\mm_\star:=\me^{(1)}\mx_{\mathcal{U}_{1}}(\mx_{\mathcal{U}_{1}}^\top\mx_{\mathcal{U}_{1}})^{-1} \mx_{\mathcal{U}_{L}}^\top
	    +\mx_{\mathcal{U}_{1}} (\mx_{\mathcal{U}_{L}}^\top\mx_{\mathcal{U}_{L}})^{-1} \mx_{\mathcal{U}_{L}}^\top \me^{({L})}$
	    and the remainder term $\mr^{(0,L)}+\ms^{(0,1,\dots,L)}$. Furthermore we also have
	    \begin{equation*}%\label{eq:simulation_error_rate}
	    	\begin{aligned}
	    	\|\mm_\star\|_{\max}\lesssim \frac{(1+\sigma)\log^{1/2}n}{(pqN)^{1/2}},
	    	\quad\|\mr^{(0,L)}+\ms^{(0,1,\dots,L)}\|_{\max}
	    		\lesssim L\Big(\frac{(1+\sigma)^2\log n}{pqN}+\frac{(1+\sigma)\log^{1/2}n}{q^{1/2}\breve p^{1/2} pN}\Big)
	    	\end{aligned}
	    \end{equation*}
        with high probability.
	   We compare the error rates for 
     $\|\hat\mpp_{\mathcal{U}_{0},\mathcal{U}_{L}}-\mpp_{\mathcal{U}_{0},\mathcal{U}_{L}}\|_{\max}$ against 
    $\|\mm_\star\|_{\max}$ and $\|\mr^{(0,L)}+\ms^{(0,1,\dots,L)}\|_{\max}$ by varying the value of one parameter among $N$, $p$, $\breve p$, $q$, $L$ and $\sigma$ while fixing the values of the remaining parameters.
	    %For ease of exposition, when we discuss the effect of a certain parameter on our convergence rate, we will treat all other parameters as fixed.
%Then according to Theorem~\ref{thm:R(i0,...,iL)}, 
%1) 
%2) when we focus on $p$, we have $\|\mm_\star\|_{\max}\lesssim p^{-1/2}$ and $\|\mr^{(0,\dots,L)}\|_{\max}\lesssim p^{-1}$ with high probability;
%3) when we focus on $\breve p$, we have $\|\mm_\star\|_{\max}\lesssim 1$ and $\|\mr^{(0,\dots,L)}\|_{\max}\lesssim 1+ \breve p^{-1/2}$ with high probability;
%4) when we focus on $q$, we have $\|\mm_\star\|_{\max}\lesssim q^{-1/2}$ and $\|\mr^{(0,\dots,L)}\|_{\max}\lesssim q^{-1}+  q^{-1/2}$ with high probability;
%5) when we focus on $L$, 
%6) when we focus on $\sigma$, we have $\|\mm_\star\|_{\max}\lesssim \sigma+1$ and $\|\mr^{(0,\dots,L)}\|_{\max}\lesssim \sigma^2+\sigma+1$ with high probability.
%Figure~\ref{fig:simulation_vary} reports %$\|\cdot\|_F/n$, %$\sqrt{\|\cdot\|_F^2/n^2}$, 
%which is bounded by $\|\cdot\|_{\max}$ but is more robust, 
%the maximum entrywise norm of
%of the total estimation error 
%$\|\hat\mpp_{\mathcal{U}_{0},\mathcal{U}_{L}}-\mpp_{\mathcal{U}_{0},\mathcal{U}_{L}}\|_{\max}$, 
%$\|\mm_\star\|_{\max}$ and $\|\mr^{(0,\dots,L)}\|_{\max}$ as we vary the values of one parameter among
%$N$, $p$, $\breve p$, $q$, $L$, and $\sigma$ while fixing the values of the remaining parameters. 
Empirical results for these error rates, averaged over $100$ Monte Carlo replicates, are summarized in Figure~\ref{fig:simulation_vary}. % for a plot of these empirical results, averaged over $100$ independent Monte Carlo replicates; these 
We note that the error rates in Figure~\ref{fig:simulation_vary} %The empirical results shown in Figure~\ref{fig:simulation_vary}
%are 
are consistent with the above bounds. % in Eq.~\eqref{eq:simulation_error_rate} %obtained by Theorem~\ref{thm:R(i0,...,iL)}.
\begin{figure}[htbp!] 
\centering
\subfigure[]{\includegraphics[height=4.25cm]{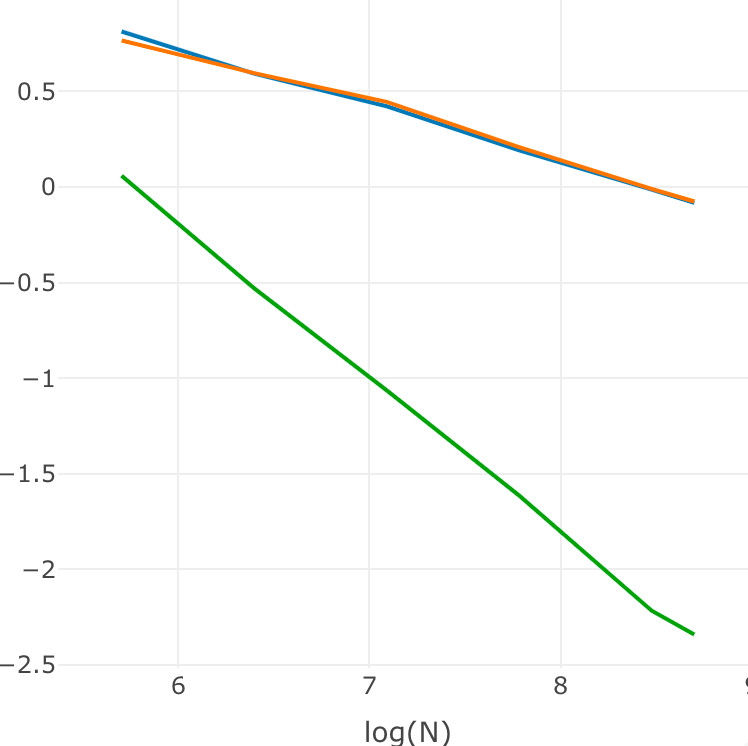}}
\subfigure[]{\includegraphics[height=4.25cm]{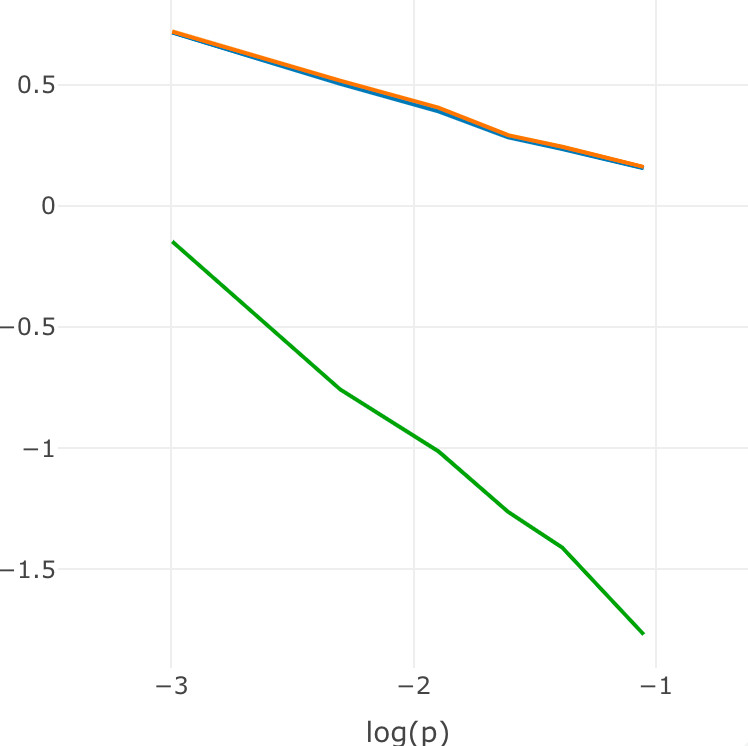}}
\subfigure[]{\includegraphics[height=4.25cm]{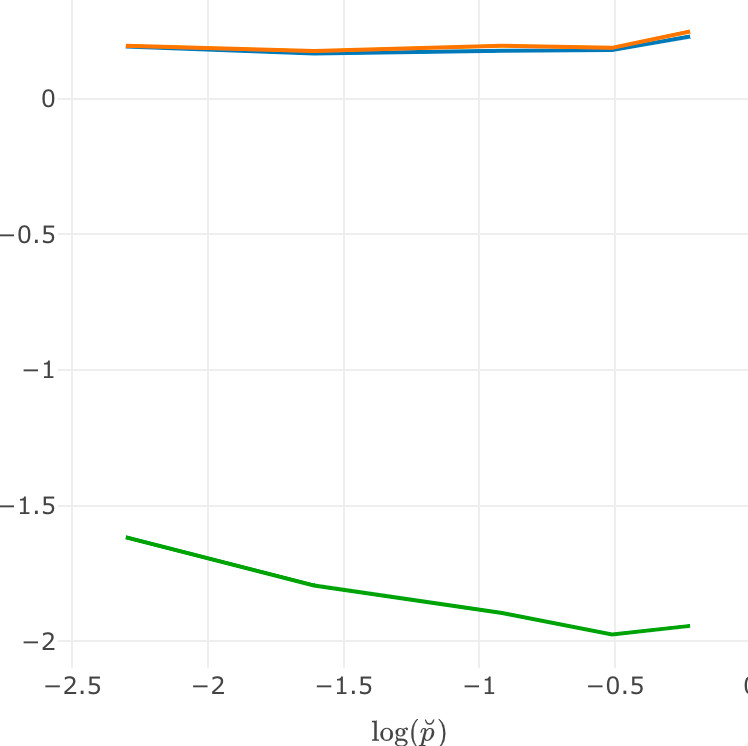}}\\
\subfigure[]{\includegraphics[height=4.25cm]{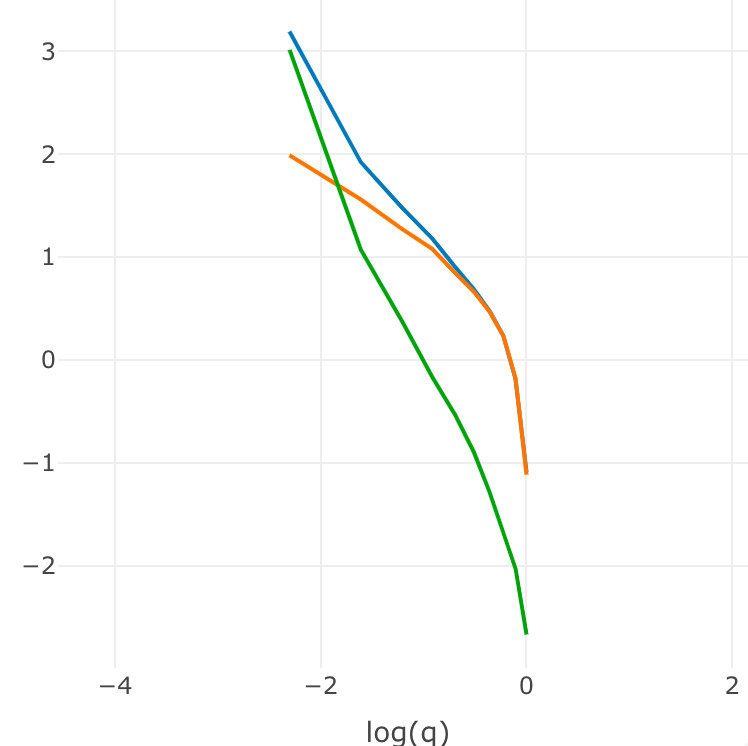}}
\subfigure[]{\includegraphics[height=4.25cm]{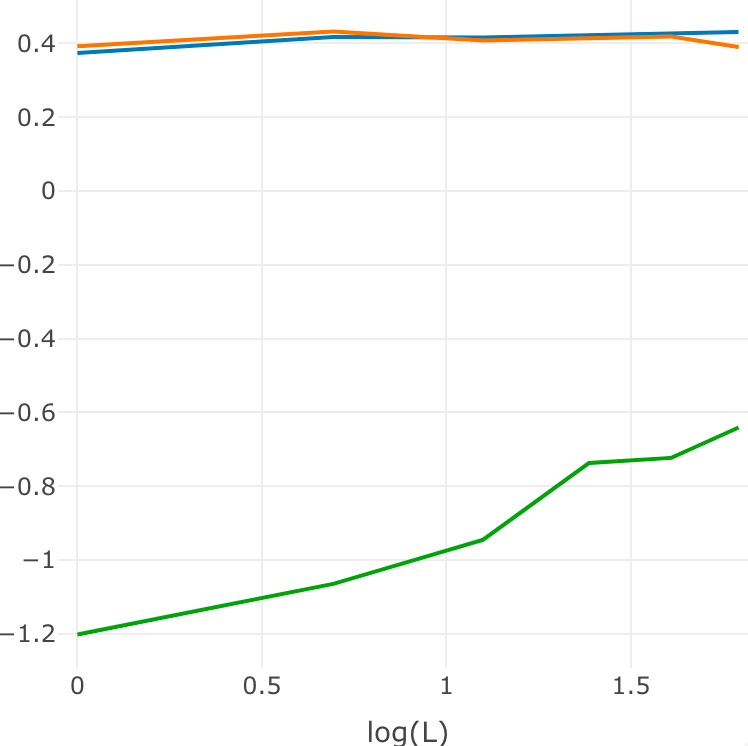}}
\subfigure[]{\includegraphics[height=4.25cm]{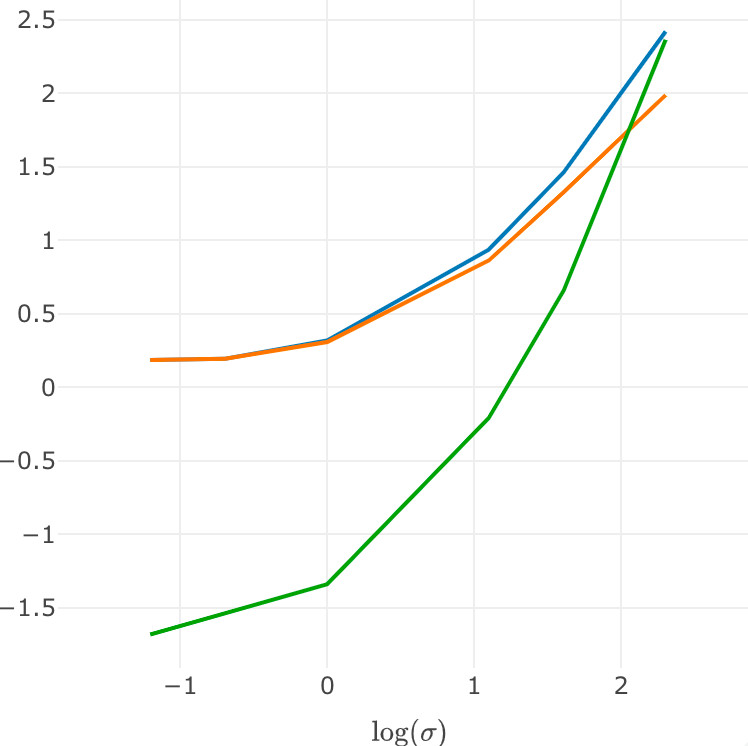}}
\caption{\footnotesize Log-log plot of the empirical error rates for %$\log (\|\cdot\|_F/n)$ %$\log \sqrt{\|\cdot\|_F^2/n^2}$ 
%of the estimation error 
$\|\hat\mpp_{\mathcal{U}_{0},\mathcal{U}_{L}}-\mpp_{\mathcal{U}_{0},\mathcal{U}_{L}}\|_{\max}$ (blue lines), 
its first order approximation $\|\mm_\star\|_{\max}$ (orange lines),
	    and the remainder term $\|\mr^{(0,L)}+\ms^{(0,1,\dots,L)}\|_{\max}$ (green lines) as we vary the value of a single parameter among $\{N, p, \breve p, q, L, \sigma\}$ while keeping the values of the remaining parameters fixed (note that we set $\lambda_{\min} = \tfrac{1}{2}N$). 
	    Panel (a): vary $N \in \{300,600,1200,2400,4800,6000\}$ for $p=0.3, \breve p=0.1, q=0.8, L=2, \sigma=0.5$.
	    Panel (b): vary $p \in \{0.05,0.1,0.15,0.2,0.25,0.35\}$ for $N=2400$, $\breve p=0.1$, $q=0.8$, $L=2$, $\sigma=0.5$.
	    Panel (c): vary $\breve p \in \{0.1,0.2,0.4,0.6,0.8\}$ for $N=2400$, $p=0.3$, $q=0.8$, $L=2$, $\sigma=0.5$.
	    Panel (d): vary $q \in \{0.1,0.2,\dots,0.9,1\}$ for $N=2400$, $p=0.3$, $\breve p=0.1$, $L=2$, $\sigma=0.5$.
	    Panel (e): vary $L \in \{1,2,3,4,5,6\}$ for $N=2400$, $p=0.15$, $\breve p=0.1$, $q=0.8$, $\sigma=0.5$.
	    Panel (f): vary $\sigma \in \{0.3,0.5,1,3,5,10\}$ for $N=2400$, $p=0.15$, $\breve p=0.1$, $q=0.8$, $L=2$, $\sigma=0.5$. 
	    %Other details can be referred in Section~\ref{sec:simu1}.
	    Error rates in each panel are averages based on $100$ independent Monte Carlo replicates. 
}
\label{fig:simulation_vary}
\end{figure}

   \subsection{Comparison with other matrix completion algorithms}
 
 \label{sec:comp}
 We use the same setting as in Section~\ref{sec:simu1}, %letting the observed submatrices be as depicted in Figure~\ref{fig:simulation_setting}, 
 but with $N = n + L \times (1 - \breve{p}) \times n \approx 2200$, so that the observed submatrices fully span the diagonal of the matrix.
 
 %We use the same setting as that in Section~\ref{sec:simu1}, but with $N=n+L\times(1-\breve p)\times n\approx 2200$ and let the observed submatrices be as depicted in Figure~\ref{fig:simulation_setting}. 
% \begin{figure}[h] 
%\centering
%\subfigure{\includegraphics[width=3.7cm]{figure/simulation/simulation_setting2.jpeg}}
%%\subfigure[]{\includegraphics[width=4cm]{image/simulation_disPCA/n_2.png}}
%\caption{Simulation setting for comparison with other algorithms
%}
%\label{fig:simulation_setting2}
%\end{figure}

 We vary $L$ and compare the performance of Algorithm~\ref{Alg_chain}
 %, which is called block-wise overlapping noisy matrix integration 
 (CMMI) with some existing state-of-art low-rank matrix completion algorithms, % which are for the uniform sampling setting. 
 including generalized spectral regularization (GSR) \citep{mazumder2010spectral}, fast alternating least squares (FALS) \citep{hastie2015matrix}, singular value thresholding (SVT) \citep{cai2010singular}, universal singular value thresholding (USVT) \citep{chatterjee2015matrix}, iterative regression against right singular vectors (IRRSV) \citep{troyanskaya2001missing}. Note that increasing $L$ leads to more 
 observed submatrices but, as each submatrix is of smaller dimensions, the total number of observed entries decreases with $L$ at rate of $N^2 q/L$. 
 Our performance metric for recovering the yellow unknown block  is in terms of the relative Frobenius norm error 
 $\|\hat\mpp_{\mathcal{U}_{0},\mathcal{U}_{L}}-\mpp_{\mathcal{U}_{0},\mathcal{U}_{L}}\|_F/\|\mpp_{\mathcal{U}_{0},\mathcal{U}_{L}}\|_F$. 
The error rates (averaged over $100$ independent Monte Carlo replicates) for %averaged over $100$ independent Monte Carlo replicates 
different algorithms are presented in Figure~\ref{fig:simulation_comparison}, and it
%Plots of the error rates (averaged over $100$ independent Monte Carlo replicates) for %averaged over $100$ independent Monte Carlo replicates 
%different algorithms and their running times are presented in the left and right panels of Figure~\ref{fig:simulation_comparison}, respectively.
 shows that CMMI outperforms all competing methods in terms of recovery accuracy. CMMI is also computationally efficient; see Section~\ref{sec:time comparison} of the supplementary material for details.  

\begin{figure}[htbp!] 
\centering
\subfigure{\includegraphics[height=5cm]{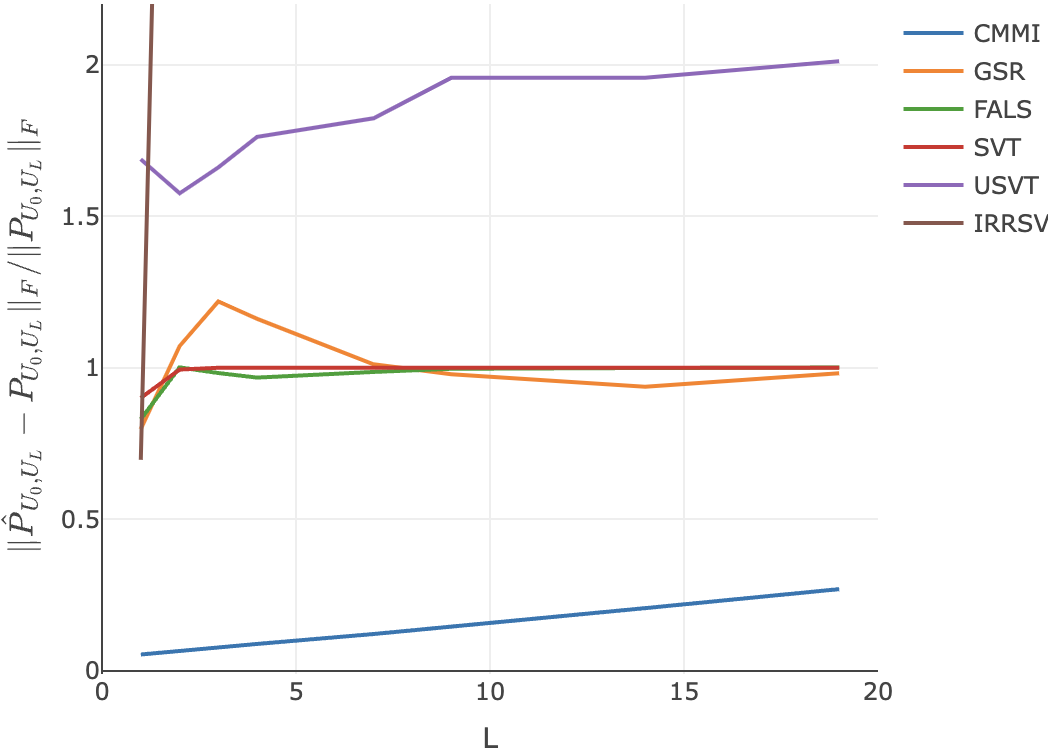}}
%\subfigure{\includegraphics[height=4.5cm]{figure/simulatoin_comparison/positive_block_time.png}}
%\subfigure[]{\includegraphics[width=4cm]{image/simulation_disPCA/n_2.png}}
\caption{%Comparison with other algorithms when we tune $L$.
\footnotesize Empirical errors $\|\hat\mpp_{\mathcal{U}_{0},\mathcal{U}_{L}}-\mpp_{\mathcal{U}_{0},\mathcal{U}_{L}}\|_F/\|\mpp_{\mathcal{U}_{0},\mathcal{U}_{L}}\|_F$ for CMMI %(blue)
 and other matrix completion algorithms %(SVT is red, IRRSV is brown, USVT is purple, SVT is red, FALS is green, and GSR is orange) 
as we vary $L \in \{1,2,3,4,7,9,14,19\}$ while fixing $N\approx 2200$, $\breve p=0.1$, $q=0.8$, $\sigma=0.5$.
The results are averaged over $100$ independent Monte Carlo replicates.
Note that the averaged relative $F$-norm errors of IRRSV are $\{0.7,    5.9,   21.0,   27.0,   96.3,   55.4,  206.2, 2299.0\}$ with some values being too large to be displayed.  
%The average running time (in log scale) over 100 replicates for algorithms, using 20-core parallel computing and 256 GB memory, is shown in the right panel. % The machine is equipped with 256 GB of memory and two 64-core, 2.25 GHz, 225-watt AMD EPYC 7742 processors.
}
\label{fig:simulation_comparison}
\end{figure}

 \section{Real Data Experiment: MEDLINE Co-occurrences}
\label{sec:real}
We compare the performance of CMMI against other matrix completion algorithms on %the MNIST database of grayscale images and 
MEDLINE database of co-occurrences citations. 
%\subsection{MEDLINE co-occurrences}
 The MEDLINE co-occurrences database \citep{medline1a} summarizes the MeSH Descriptors that occur together in MEDLINE citations from the MEDLINE/PubMed Baseline over a duration of several years. A standard approach for handling this type of data is to first transform the (normalized) co-occurrence counts into pointwise mutual information (PMI), an association measure widely used in natural language processing %\citep{church1990word,lu2023knowledge}. 
 More specifically, the PMI between two concepts $x$ and $y$ is
% This database contains the number of co-occurrences between clinical concepts in different years. 
%Pointwise mutual information (PMI) is an information-theoretic association measure widely used in natural language processing \citep{church1990word,lu2023knowledge}. 
%For two concepts $x$ and $y$, PMI$(x,y)$ is
defined as
$\mathrm{PMI}(x,y)=\log\tfrac{\mathbb{P}(x,y)}{\mathbb{P}(x)\mathbb{P}(y)},$
where $\mathbb{P}(x)$ and $\mathbb{P}(y)$ are the (marginal) occurrence probability of $x$ and $y$, and $\mathbb{P}(x,y)$ is the (joint) co-occurrence probability of $x$ and $y$.
%All of them can be calculated from the co-occurrence matrix.
%Therefore the PMI matrix for particular clinical concepts can be obtained from the co-occurrence dataset.

For our analysis of the MEDLINE data, we first select $7486$ clinical concepts which are most frequently cited during the twelve years period from $2011$ to $2022$, and construct the total PMI matrix $\tilde\mpp \in \mathbb{R}^{7486 \times 7486}$ between these concepts.  Next we split the $12$ years into $L+1$ time intervals of equal length, and for each time interval $0\leq i\leq L$  we construct the individual PMI matrix $\tilde\mpp^{(i)} \in \mathbb{R}^{7486 \times 7486}$. We randomly sample, for each interval, a subset $\mathcal{U}_i$ of $n = 1000$ concepts from those $7486$ cited concepts such that $|\mathcal{U}_{i-1} \cap \mathcal{U}_i| = 100$ for all $1 \leq i \leq L$. Finally we set $\ma^{(i)} = \tilde\mpp^{(i)}_{\mathcal{U}_i, \mathcal{U}_i}$ as the principal submatrix of $\tilde\mpp^{(i)}$ induced by $\mathcal{U}_i$. The collection $\{\ma^{(i)}\}_{0 \leq i \leq L}$ forms a chain of perturbed overlapping submatrices of $\tilde\mpp$.
(See Section~\ref{sec:hostic_realdata} of the supplementary material for a related analysis of this data
that might be more practically relevant.)
%and we only use the data from 2011 to 2022.
%We use the PMI matrix for these $7486$ clinical concepts obtained from the data for the total $12$ years as the population matrix $\mpp$. 

%By the (automatic) dimensionality selection procedure in \cite{zhu2006automatic}, the embedding dimension $d=23$ is selected.
%For $L=1$, we construct two observed submatriced as described in Figure~\ref{fig:simulation_setting2} for different years, which can be regarded as different data sources. Specifically, the total $12$ years are divided equally into the first $6$ years and the last $6$ years. And from the $7486$ clinical concepts, we randomly sample $n=1000$ clinical concepts $\mathcal{U}_0$ and $\mathcal{U}_1$ for the two $6$-year phases, respectively, where the overlap is set to $|\mathcal{U}_0 \cap \mathcal{U}_1| = 100$. Then the two observed submatrices are the two PMI matrices obtained from the data for the two $6$-year phases and their corresponding $n=1000$ clinical concepts, respectively.
%For $L=2,3$ or $5$, we can also conduct the similar operation to get three, four or six observed submatriced for different year phases.

Given $\{\ma^{(i)}\}$, we apply CMMI and other low-rank matrix completion algorithms to construct $\hat{\mpp}_{\mathcal{U}_{i_0}, \mathcal{U}_{i_L}}$ for the PMIs between clinical concepts in $\mathcal{U}_0$ and those in $\mathcal{U}_L$ in the total PMI matrix $\tilde\mpp$. Note that we specify $d=23$ for both CMMI and FALS, where this choice 
is based on applying the dimensionality selection procedure of
\cite{zhu2006automatic} to $\tilde\mpp$. In contrast we set $d = 3$ for GSR as its running time increase substantially for larger values of $d$. The values of $d$ for SVT and USVT are not specified, as both algorithms automatically determine $d$ using their respective eigenvalue thresholding procedures.
We then measure the similarities between the estimated PMIs in $\hat\mpp_{\mathcal{U}_0,\mathcal{U}_L}$ and the true total PMIs in $\tilde\mpp_{\langle\mathcal{U}_0\rangle,\langle\mathcal{U}_L\rangle}$ in terms of the Spearman's rank correlation $\rho$ (note that we only compare PMIs for pairs of concepts with positive co-occurrence). The Spearman's $\rho$ between two set of vectors takes value in $[-1,1]$ with $1$ (resp. $-1$) denoting perfect monotone increasing (resp. decreasing) relationship and $0$ suggesting no relationship. 
%a perfect Spearman's rank correlation of $+1$ or $-1$ occurs when each of the variables is a perfect monotone increasing or decreasing function of the other. A Spearman's rank correlation of $0$ means that there is no association between ranks. 
The results, averaged over $100$ independent Monte Carlo replicates, 
are summarized in Figure~\ref{fig:realdata_MRCOC}, where CMMI outperforms competing algorithms in accuracy. CMMI is also computationally efficient; see Section~\ref{sec:time comparison} of the supplementary material for details. %
%and show that the estimate $\hat{\mpp}_{\mathcal{U}_{0}, \mathcal{U}_{L}}$ obtained using CMMI has much higher rank-correlation compared to other algorithms especially for large $L$, which means that the matrix being recovered is larger and sparser. The computational advantage of CMMI is also obvious for this large data experiment.

 \begin{figure}[htbp!] 
\centering
\subfigure{\includegraphics[height=5cm]{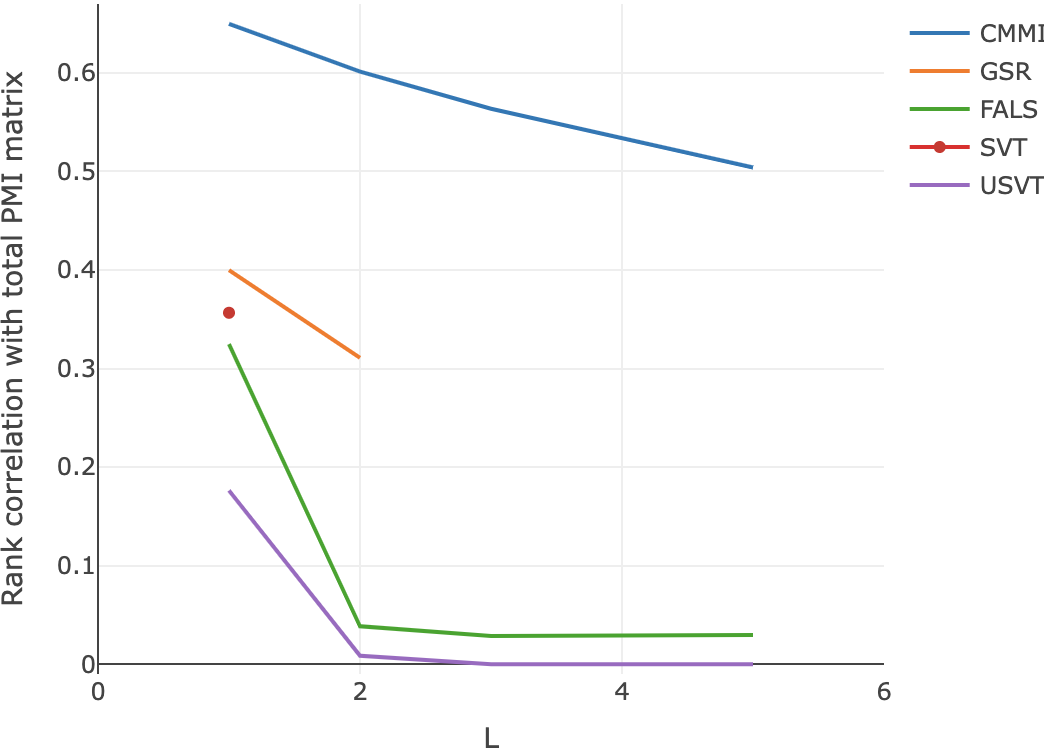}}
%\subfigure{\includegraphics[height=4.5cm]{figure/realdata/MEDLINE_block_time.png}}
\caption{%Comparison with other methods $L=\{1,2,3,4,7,9\}$.
\footnotesize Empirical estimates of Spearman's rank correlations for CMMI and other matrix completion algorithms as %$L$ changes. % (CMMI is blue, GSR is orange, FALS is green, SVT is red, and USVT is purple). 
%In particular, 
we vary $L=\{1,2,3,5\}$ while fixing $n=1000$ and $\breve p=0.1.$
The results are averaged over $100$ independent Monte Carlo replicates.
We are not able to evaluate IRRSV in this experiment due to the sparse nature of the co-occurence matrix (about 60\% zero entries in $\tilde\mpp$).
We only evaluate the performance of GSR for $L\leq 2$ and SVT for $L=1$, as these algorithms are computationally prohibitive with slight increases in $L$.
%The average running time (in log scale) over 100 replicates for algorithms, using 25-core parallel computing and 256 GB memory, is shown in the right panel. % The machine is equipped with 256 GB of memory and two 64-core, 2.25 GHz, 225-watt AMD EPYC 7742 processors.
}
\label{fig:realdata_MRCOC}
\end{figure}

 \section{Extensions to Indefinite or Asymmetric Matrices}
 \label{sec:asy}
 % We now describe how the methodologies presented in this paper can be extended to block-wise data integration of symmetric indefinite matrices and asymmetric/rectangular matrices.
  %low-rank matrices and asymmetric low-rank matrices.
 \subsection{CMMI for symmetric indefinite matrices}

Suppose $\mpp\in\mathbb{R}^{N\times N}$ is a symmetric
 indefinite low-rank matrix. 
 Let $d_{+}$ and $d_{-}$ be the number of positive and negative eigenvalues of $\mpp$ and set $d = d_+ + d_{-}\ll N$.
We denote the non-zero eigenvalues of $\mpp$ by $\lambda_1(\mpp)\geq\dots\geq\lambda_{d_+}(\mpp)>0>\lambda_{n-d_-+1}(\mpp)\geq\dots\geq\lambda_n(\mpp).$
Let $\mLambda_+:=\diag(\lambda_1(\mpp),\dots,\lambda_{d_+}(\mpp))$, $\mLambda_-:=\diag(\lambda_{n-d_-+1}(\mpp),\dots,\lambda_n(\mpp))$, and the orthonormal columns of 
$\muu_+\in\mathbb{R}^{N\times d_+}$ and $\muu_-\in\mathbb{R}^{N\times d_-}$ constitute the corresponding eigenvectors.
Then the eigen-decomposition of $\mpp$ is $\muu\mLambda\muu^\top$, where
$\mLambda:=\diag(\mLambda_+,\mLambda_-)$ and $\muu:=[\muu_+,\muu_-]$.
%Also denote by $\muu_+\in\mathbb{R}^{N\times d_+}$ and $\muu_-\in\mathbb{R}^{N\times d_-}$ the matrices 
%whose columns are the eigenvectors of $\mpp$ corresponding to 
%$\lambda_1(\mpp),\dots,\lambda_{d_+}(\mpp)$ and $\lambda_{n-d_-+1}(\mpp),\dots,\lambda_n(\mpp)$, respectively. Set $\muu:=[\muu_+,\muu_-]$ so that $\mpp=\muu\mLambda\muu^\top$.
 %For convenience, w
 
 Then $\mpp$ can be written as $
 %=\muu|\mLambda|^{1/2}\mi_{d_+,d_-}|\mLambda|^{1/2}\muu^\top
 \mpp =\mx\mi_{d_+,d_-}\mx^\top$ with $\mi_{d_+,d_-}=\diag(\mi_{d_+},-\mi_{d_-})$ and $\mx=\muu|\mLambda|^{1/2}$, and thus the rows of $\mx$ represent the latent positions for the entities. % i.e., $\mpp_{s,t}=\mathbf{x}_s^\top\mi_{d_+,d_-}\mathbf{x}_t$ for all $(s,t) \in [N] \times [N]$. 
%This allows us to write $\mpp$ alternatively as 
 %$$\mpp
 %=\muu|\mLambda|^{1/2}\mi_{d_+,d_-}|\mLambda|^{1/2}\muu^\top
 %=\mx\mi_{d_+,d_-}\mx^\top.$$
For any $i\in[K]$, let $\mathcal{U}_i\subseteq [N]$ denote the set of entities contained in the $i$th source, and let 
$\mpp^{(i)}$ %\in\mathbb{R}^{n_i\times n_i}$, where $n_i:=|\mathcal{U}_i|$, 
be the corresponding population matrix. We then have
$
\mpp^{(i)}=\mpp_{\mathcal{U}_i,\mathcal{U}_i}=\muu_{\mathcal{U}_i}\mLambda\muu_{\mathcal{U}_i}^\top
=\mx_{\mathcal{U}_i}\mi_{d_+,d_-}\mx_{\mathcal{U}_i}^\top.$
For each observed submatrix $\ma^{(i)}$ on $\mathcal{U}_i$, we compute the 
estimated latent position matrix $\hat\mx^{(i)}=\hat\muu^{(i)}|\hat\mLambda^{(i)}|^{1/2}$. Here 
$\hat\mLambda^{(i)}:=\diag(\hat\mLambda_+^{(i)},\hat\mLambda_-^{(i)})\in\mathbb{R}^{d\times d}$ and $\hat\mLambda_+^{(i)}$, 
$\hat\mLambda_-^{(i)}$ contain the $d_+$ largest positive and $d_{-}$ largest (in-magnitude) negative eigenvalues of $\ma^{(i)}$, respectively.
$\hat\muu^{(i)}:=[\hat\muu_+^{(i)},\hat\muu_-^{(i)}]$ contains the corresponding eigenvectors. 
%More specifically, denote all eigenvalues of $\ma^{(i)}$ by $\lambda_1(\ma^{(i)})\geq\dots\geq\lambda_{n_i}(\ma^{(i)})$.
%Let $\hat\mLambda_+^{(i)}:=\diag(\lambda_1(\ma^{(i)}),\dots,\lambda_{d_+}(\ma^{(i)}))$, $\hat\mLambda_-^{(i)}:=\diag(\lambda_{n_i-d_-+1}(\ma^{(i)}),\dots,\lambda_{n_i}(\ma^{(i)}))$, and $\hat\mLambda^{(i)}:=\diag(\hat\mLambda_+^{(i)},\hat\mLambda_-^{(i)})\in\mathbb{R}^{d\times d}$.
%Denote the eigenvector matrix of $\ma^{(i)}$ corresponding to $(\lambda_1(\ma^{(i)}),\dots,\lambda_{d_+}(\ma^{(i)}))$ and $(\lambda_{n_i-d_-+1}(\ma^{(i)}),\dots,\lambda_{n_i}(\ma^{(i)}))$ by $\hat\muu_+^{(i)}\in\mathbb{R}^{n_i\times d_+}$ and $\hat\muu_-^{(i)}\in\mathbb{R}^{n_i\times d_-}$, respectively, and let $\hat\muu^{(i)}:=[\hat\muu_+^{(i)},\hat\muu_-^{(i)}]$.

We start with the \textit{noiseless} case to illustrate the main idea. Consider $2$ overlapping block-wise submatrices $\mpp^{(1)}$ and $\mpp^{(2)}$ as shown in Figure~\ref{fig:K=2}. % and suppose $\mpp_{\mathcal{U}_1\cap \mathcal{U}_2,\mathcal{U}_1\cap \mathcal{U}_2}$ is of rank $d$, and thus $\mpp^{(1)}$ and $\mpp^{(2)}$ is also of rank $d$.
Now
$$
\begin{aligned}
	\mx_{\mathcal{U}_1} \mi_{d_+,d_-}\mx_{\mathcal{U}_1}^\top= \mpp^{(1)}=\mx^{(1)}\mi_{d_+,d_-}\mx^{(1)\top},
	\quad \mx_{\mathcal{U}_2} \mi_{d_+,d_-}\mx_{\mathcal{U}_2}^\top=\mpp^{(2)}=\mx^{(2)}\mi_{d_+,d_-}\mx^{(2)\top},
\end{aligned}
$$
%and hence $\operatorname{rank}(\mpp_{\mathcal{U}_1,\mathcal{U}_1})$ (resp. $\operatorname{rank}(\mpp_{\mathcal{U}_2,\mathcal{U}_2})$) achieves the number of columns of $\mx_{\mathcal{U}_1}$ and $\mx^{(1)}$ (resp. $\mx_{\mathcal{U}_2}$ and $\mx^{(2)}$). Then by simple algebraic analysis, 
and hence there exist matrices $\mw^{(1)} \in \mathcal{O}_{d_{+}, d_{-}}$ and $\mw^{(2)} \in \mathcal{O}_{d_+, d_{-}}$ such that 
\begin{equation*}%\label{eq:X2}
\begin{aligned}
		\mx_{\mathcal{U}_1}=\mx^{(1)}\mw^{(1)},
		\quad \mx_{\mathcal{U}_2}=\mx^{(2)}\mw^{(2)}.
\end{aligned}
\end{equation*}
Here $\mathcal{O}_{d_+,d_-}:=\{\mo\in\mathbb{R}^{d \times d}\mid \mo\mi_{d_+,d_-}\mo^\top=\mi_{d_+,d_-}\}$ is the indefinite orthogonal group. Then
%Eq.~\eqref{eq:X2} implies
\begin{equation*}%\label{eq:P122}
\begin{aligned}
	\mpp_{\mathcal{U}_1,\mathcal{U}_2}
	&=
	\mx_{\mathcal{U}_1}\mi_{d_+,d_-}\mx_{\mathcal{U}_2}^\top =\mx^{(1)}\mw^{(1)}\mi_{d_+,d_-}\mw^{(2)\top}\mx^{(2)\top}%=\mx^{(1)}\mw^{(1)}(\mw^{(2)})^{-1}\mi_{d_+,d_-}\mx^{(2)\top}\\&
	=\mx^{(1)}\mw^{(1,2)}\mi_{d_+,d_-}\mx^{(2)\top},
\end{aligned}
\end{equation*}
%and thus we only need $\mw^{(1,2)} :=\mw^{(1)}(\mw^{(2)})^{-1}\in\mathcal{O}_{d_+,d_-}$ to recover $\mpp_{\mathcal{U}_1,\mathcal{U}_2}$. %%, and notice that because $\mw^{(1)},\mw^{(2)}$ are unique, $\mw^{(1,2)}$ isalso unique.
%Thus $\mw^$
%To obtain $\mw^{(1,2)}$ from $\mx^{(1)}$ and $\mx^{(2)}$, we have by Eq.~\eqref{eq:X2} that 
%Note that from Eq.~\eqref{eq:X2} we also have
%$\mx^{(1)}_{ \langle\mathcal{U}_1\cap \mathcal{U}_2\rangle}\mw^{(1)}=\mx_{\mathcal{U}_1\cap \mathcal{U}_2}=\mx^{(2)}_{ \langle\mathcal{U}_1\cap \mathcal{U}_2\rangle}\mw^{(2)},$
% and hence
%$\mx^{(2)}_{ \langle\mathcal{U}_1\cap \mathcal{U}_2\rangle}
%=\mx^{(1)}_{ \langle\mathcal{U}_1\cap \mathcal{U}_2\rangle}\mw^{(1)}(\mw^{(2)})^{-1}$. %=\mx^{(1)}_{ \langle\mathcal{U}_1\cap \mathcal{U}_2\rangle}\mw^{(1,2)}.$$
where $\mw^{(1,2)} :=\mw^{(1)}(\mw^{(2)})^{-1}\in \mathcal{O}_{d_+,d_-}$. We can recover $\mw^{(1,2)}$ by aligning the latent positions for overlapping entities by
\begin{equation}
\label{eq:align_noiseless_indefinite}
\mw^{(1,2)}=\underset{\mo\in \mathcal{O}_{d_+,d_-}}{\operatorname{argmin}} \|\mx^{(1)}_{ \langle\mathcal{U}_1\cap \mathcal{U}_2\rangle}\mo-\mx^{(2)}_{ \langle\mathcal{U}_1\cap \mathcal{U}_2\rangle}\|_F.
\end{equation}
If $\mathrm{rk}(\mpp_{\mathcal{U}_1\cap \mathcal{U}_2,\mathcal{U}_1\cap \mathcal{U}_2}) = d$ then  
$\mw^{(1,2)} 
=% (\mx^{(1)\top}_{ \langle\mathcal{U}_1\cap \mathcal{U}_2\rangle} \mx^{(1)}_{ \langle\mathcal{U}_1\cap \mathcal{U}_2\rangle})^{-1}
%\mx^{(1)\top}_{ \langle\mathcal{U}_1\cap \mathcal{U}_2\rangle} \mx^{(2)}_{ \langle\mathcal{U}_1\cap \mathcal{U}_2\rangle}
%=
(\mx^{(1)}_{ \langle\mathcal{U}_1\cap \mathcal{U}_2\rangle})^{\dagger}\mx^{(2)}_{ \langle\mathcal{U}_1\cap \mathcal{U}_2\rangle}$ is the {\em unique} minimizer of Eq.~\eqref{eq:align_noiseless_indefinite}. 
Here $(\cdot)^{\dagger}$ denotes the Moore-Penrose pseudoinverse.
%\mw^{(1,2)}=\underset{\mo\in \mathbb{R}^{d\times d}}{\operatorname{argmin}} \|\mx^{(1)}_{ \langle\mathcal{U}_1\cap \mathcal{U}_2\rangle}\mo-\mx^{(2)}_{ \langle\mathcal{U}_1\cap \mathcal{U}_2\rangle}\|_F.

Now suppose $\ma^{(1)}$ and $\ma^{(2)}$ are noisy observations of $\mpp^{(1)}$ and $\mpp^{(2)}$. Let $\hat\mx^{(1)}$ and $\hat\mx^{(2)}$ be estimates of $\mx^{(1)}$ and $\mx^{(2)}$ as described above. 
Then to align $\hat\mx^{(1)}$ and $\hat\mx^{(2)}$, we can consider solving the indefinite orthogonal Procrustes problem
\begin{equation}
\label{eq:align_noisy_indefinite1}
\begin{aligned}
	\mw^{(1,2)}=\underset{\mo\in \mathcal{O}_{d_+,d_-}}{\operatorname{argmin}} \|\hat\mx^{(1)}_{ \langle\mathcal{U}_1\cap \mathcal{U}_2\rangle}\mo-\hat\mx^{(2)}_{ \langle\mathcal{U}_1\cap \mathcal{U}_2\rangle}\|_F.
\end{aligned}
\end{equation}
%The constraint $\mo\in \mathbb{O}_{d_+,d_-}$ leads to simpler 
%of the matrix.
However, in contrast to the noiseless case, there is no longer an 
analytical solution to Eq.~\eqref{eq:align_noisy_indefinite1}. We thus replace Eq.~\eqref{eq:align_noisy_indefinite1} with  the unconstrained least squares problem 
$%$
%\begin{aligned}
	\mw^{(1,2)}
	=\underset{\mo\in \mathbb{R}^{d\times d}}{\operatorname{argmin}} \|\hat\mx^{(1)}_{ \langle\mathcal{U}_1\cap \mathcal{U}_2\rangle}\mo-\hat\mx^{(2)}_{ \langle\mathcal{U}_1\cap \mathcal{U}_2\rangle}\|_F,
%\end{aligned}
$%$
whose solution is once again $\mw^{(1,2)} = (\hat\mx^{(1)}_{ \langle\mathcal{U}_1\cap \mathcal{U}_2\rangle})^{\dagger}\hat\mx^{(2)}_{ \langle\mathcal{U}_1\cap \mathcal{U}_2\rangle}$.
Given $\mw^{(1,2)}$, we estimate $\mpp_{\mathcal{U}_1, \mathcal{U}_2}$ by
$\hat\mpp_{\mathcal{U}_1,\mathcal{U}_2}=\hat\mx^{(1)}\mw^{(1,2)}\mi_{d_+,d_-}\hat\mx^{(2)\top}.$
Extending the above idea to a chain of overlapping submatrices is also straightforward. See Section~\ref{supp:indef} of the supplementary material for the detailed algorithm, associated theoretical result (Theorem~\ref{thm:R(il)_ind}) and numerical simulations. %Section~\ref{supp:indef} also include an extension of Theorem~\ref{thm:R(i0,...,iL)} to the indefinite setting.
\subsection{CMMI for asymmetric matrices}

Data integration for asymmetric matrices has many applications including genomic data integration \citep{maneck2011genomic,%zang2016high,
cai2016structured},
single-cell data integration \citep{stuart2019comprehensive,%argelaguet2021computational,
ma2023your}.
Suppose $\mpp\in\mathbb{R}^{N\times M}$ is a low-rank matrix. 
 Let $d$ be the rank of $\mpp$, and write the singular decomposition of $\mpp$ as
 $\mpp=\muu\mSigma\mv^\top,$
 where $\mSigma\in\mathbb{R}^{d\times d}$ is a diagonal matrix whose diagonal entries are composed of the singular values of $\mpp$ in a descending order, and orthonormal columns of $\muu\in\mathbb{R}^{N\times d}$ and $\mv\in\mathbb{R}^{M\times d}$ constitute the corresponding left and right singular vectors, respectively.
The left and right latent position matrices associated to the entities can be represented by $\mx=\muu\mSigma^{1/2}\in\mathbb{R}^{N\times d}$ and $\my=\mv\mSigma^{1/2}\in\mathbb{R}^{M\times d}$, respectively. %and we can write $\mpp$ alternatively as 
 %$$\mpp=\mx\my^\top.$$
 For the $i$th source we denote the index set of the entities for rows and columns by $\mathcal{U}_i\subseteq [N]$ and $\mathcal{V}_i\subseteq [M]$, and let
% We denote the population matrix for the $i$th source by $\mpp^{(i)}\in\mathbb{R}^{n_i\times m_i}$, where $n_i:=|\mathcal{U}_i|,m_i:=|\mathcal{V}_i|$. We then have
$
\mpp^{(i)}=\mpp_{\mathcal{U}_i,\mathcal{V}_i}=\muu_{\mathcal{U}_i}\mSigma\mv_{\mathcal{V}_i}^\top
=\mx_{\mathcal{U}_i}\my_{\mathcal{V}_i}^\top.$ For each noisily observed realization $\ma^{(i)}$ of $\mpp^{(i)}$, we obtain the estimated left latent positions $\hat\mx^{(i)}=\hat\muu^{(i)}(\hat\mSigma^{(i)})^{1/2}$ for entities in $\mathcal{U}_i$ and right latent positions $\hat\my^{(i)}=\hat\mv^{(i)}(\hat\mSigma^{(i)})^{1/2}$ for entities in $\mathcal{V}_i$.

\iffalse
Let $\mpp^{(1)}$ and $\mpp^{(2)}$ be two overlapping submatrices shown in Figure~\ref{fig:K=2a} without noise or missing entries. Suppose $\mathrm{rank}(\mpp^{(1)}) = \mathrm{rank}(\mpp^{(2)}) = 
\mathrm{rank}(\mx_{\mathcal{U}_1\cap \mathcal{U}_2}) = \mathrm{rank}(\my_{\mathcal{V}_1\cap \mathcal{V}_2}) = d$.
Now 
$$
\begin{aligned}
	\mx_{\mathcal{U}_1}\my_{\mathcal{V}_1}^\top=\mpp_{\mathcal{U}_1,\mathcal{V}_1}=\mpp^{(1)}=\mx^{(1)}\my^{(1)\top},
	\quad \mx_{\mathcal{U}_2} \my_{\mathcal{V}_2}^\top=\mpp_{\mathcal{U}_2,\mathcal{V}_2}=\mpp^{(2)}=\mx^{(2)}\my^{(2)\top}.
\end{aligned}
$$
%and %by our assumption we know $\operatorname{rank}(\mpp_{\mathcal{U}_1,\mathcal{V}_1})$ (resp. $\operatorname{rank}(\mpp_{\mathcal{U}_2,\mathcal{V}_2})$) achieves the number of columns of $\mx_{\mathcal{U}_1}$ and $\mx^{(1)}$ (resp. $\my_{\mathcal{V}_2}$ and $\my^{(2)}$). 
%hence, by simple algebraic analysis, there exist $d \times d$ matrices $\mw^{(1)},\mw^{(2)}$ such that 
%\begin{equation}\label{eq:X3}
%\begin{aligned}
	%	&\mx_{\mathcal{U}_1}=\mx^{(1)}\mw^{(1)},
	%	\quad \my_{\mathcal{V}_1}=\my^{(1)}(\mw^{(1)\top})^{-1},\\
		%&\mx_{\mathcal{U}_2}=\mx^{(2)}\mw^{(2)},
		%\quad \my_{\mathcal{V}_2}=\my^{(2)}(\mw^{(2)\top})^{-1}.
%\end{aligned}
%\end{equation}
Then there exist $d \times d$ matrices $\mw^{(1)}$ and $\mw^{(2)}$ such that
$$\begin{aligned}
		\mx_{\mathcal{U}_1}=\mx^{(1)}\mw^{(1)},
		\quad \my_{\mathcal{V}_1}=\my^{(1)}(\mw^{(1)\top})^{-1},
		\quad\mx_{\mathcal{U}_2}=\mx^{(2)}\mw^{(2)},
		\quad \my_{\mathcal{V}_2}=\my^{(2)}(\mw^{(2)\top})^{-1}.
\end{aligned}
$$
\fi
Suppose we want to recover the unobserved yellow submatrix in the left panel of Figure~\ref{fig:K=2a} as part of $\mpp_{\mathcal{U}_1,\mathcal{V}_2}$,
and only observe $\ma^{(1)}$ and $\ma^{(2)}$ as noisy versions of $\mpp^{(1)}$ and $\mpp^{(2)}$.
%Let $(\hat\mx^{(1)},\hat\my^{(1)})$ and $(\hat\mx^{(2)},\hat\my^{(2)})$ be the estimated latent positions matrices obtained from $\ma^{(1)}$ and $\ma^{(2)}$. 
Suppose $\mathrm{rk}(\mx_{\mathcal{U}_1\cap \mathcal{U}_2}) = \mathrm{rk}(\my_{\mathcal{V}_1\cap \mathcal{V}_2}) = d$. We first align the estimated latent positions $(\hat\mx^{(1)},\hat\my^{(1)})$ and $(\hat\mx^{(2)},\hat\my^{(2)})$ 
by solving the least squares problems
$$
\begin{aligned}
	\mw_\mx^{(1,2)}
	=\underset{\mo\in \mathbb{R}^{d\times d}}{\operatorname{argmin}} \|\hat\mx^{(1)}_{ \langle\mathcal{U}_1\cap \mathcal{U}_2\rangle}\mo-\hat\mx^{(2)}_{ \langle\mathcal{U}_1\cap \mathcal{U}_2\rangle}\|_F, \quad
\mw_\my^{(1,2)\top}=\underset{\mo\in \mathbb{R}^{d\times d}}{\operatorname{argmin}} \|\hat\my^{(2)}_{ \langle\mathcal{V}_1\cap \mathcal{V}_2\rangle}\mo-\hat\my^{(1)}_{\langle\mathcal{V}_1\cap \mathcal{V}_2\rangle}\|_F,
\end{aligned}
$$
and setting
$
\mw^{(1,2)}=\frac{1}{2}(\mw_\mx^{(1,2)}+\mw_\my^{(1,2)})
$.
We then estimate $\mpp_{\mathcal{U}_1,\mathcal{V}_2}$ by
$\hat\mpp_{\mathcal{U}_1,\mathcal{V}_2}=\hat\mx^{(1)}\mw^{(1,2)}\hat\my^{(2)\top}$ (see detailed derivations in Section~\ref{supp:asy} of the supplementary material).
Note that the unobserved white submatrix in the left panel of Figure~\ref{fig:K=2a} is part of $\mpp_{\mathcal{U}_2,\mathcal{V}_1}$ and can
be recovered using the similar procedure. % described above. 

\begin{figure}[htp!] 
\centering
\subfigure{\includegraphics[width=4.8cm]{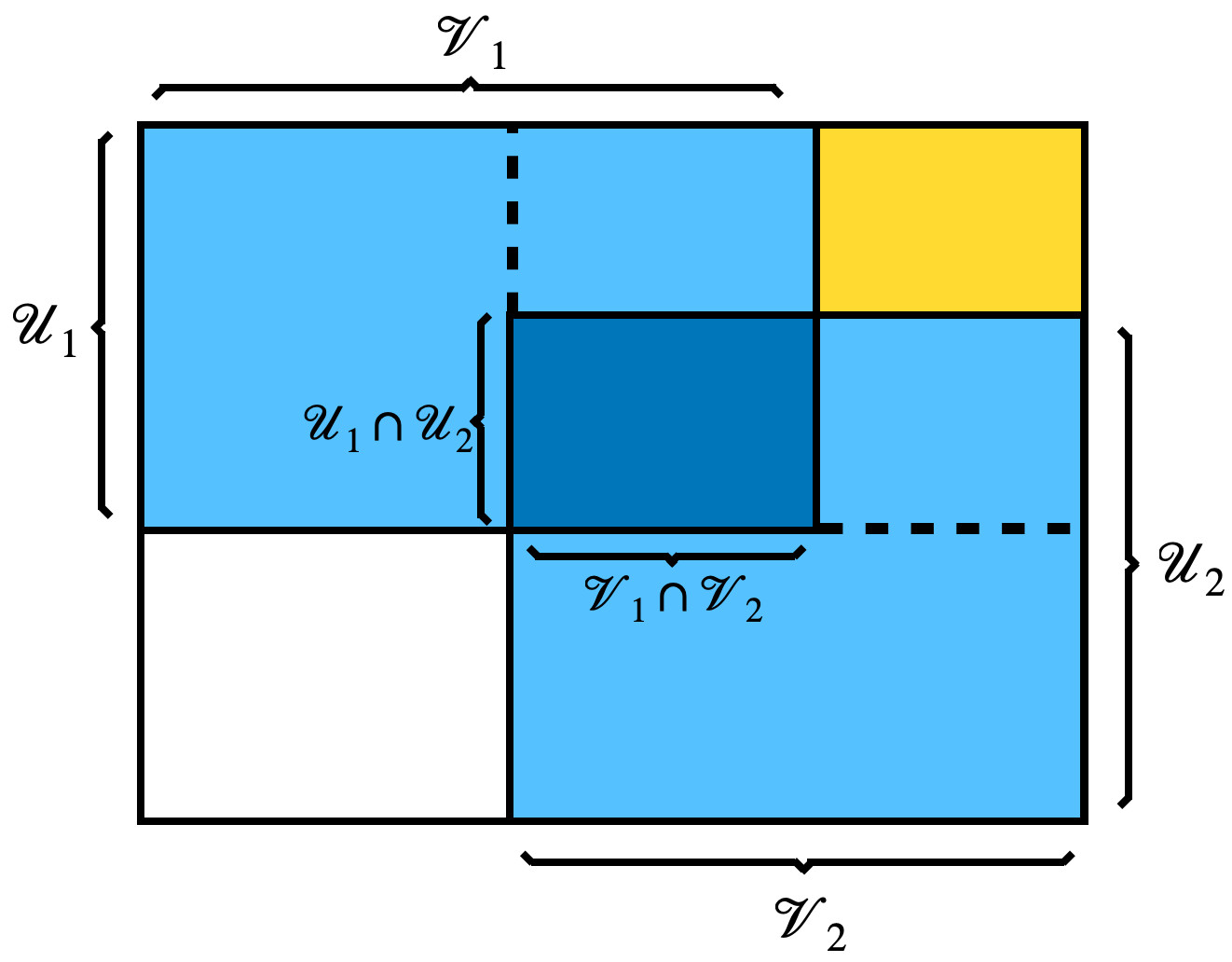}} \quad \quad
\subfigure{\includegraphics[width=4.8cm]{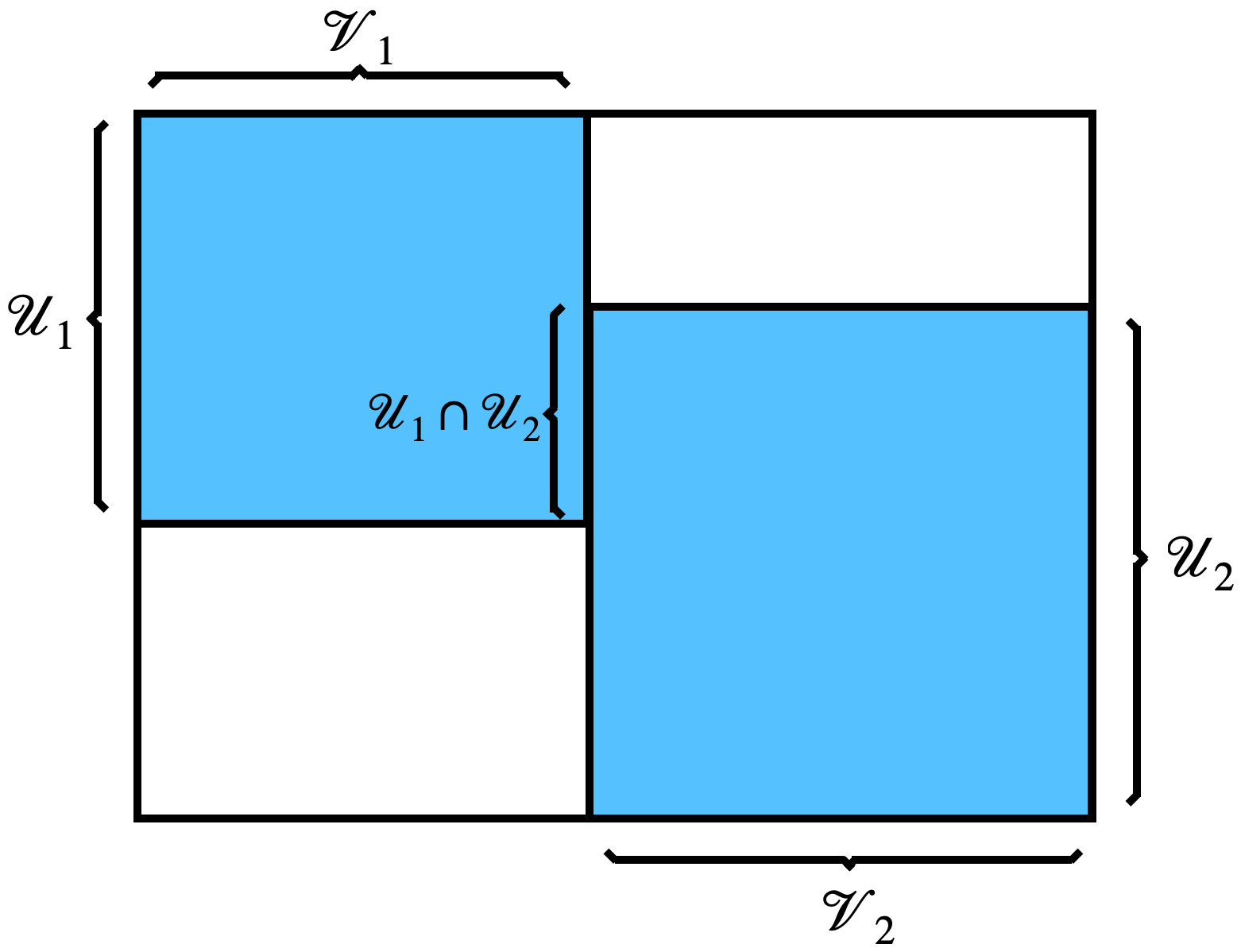}}
\caption{Left panel: a pair of overlapping observed submatrices of an asymmetric matrix. Right panel: overlapping row indices but no overlapping entries
}
\label{fig:K=2a}
\end{figure}

We emphasize that to integrate any two submatrices $\ma^{(1)}$
and $\ma^{(2)}$ of an asymmetric matrix, it is not necessary for them
to have any overlapping entries, i.e., it is not necessary that both
$\mathcal{U}_1\cap\mathcal{U}_2\neq \varnothing$ \textit{and}
$\mathcal{V}_1\cap\mathcal{V}_2\neq \varnothing$.  Indeed, if $|\mathcal{U}_1\cap\mathcal{U}_2| \geq d$,
\textit{or} (inclusive or) $|\mathcal{V}_1\cap\mathcal{V}_2| \geq d$
then we can recover $\mw^{(1,2)}$. Consider, for example, the situation in the right panel of Figure~\ref{fig:K=2a} and
suppose $\operatorname{rk}(\mx_{\mathcal{U}_1\cap\mathcal{U}_2})=d$. We can
then set $
\begin{aligned}
	\mw^{(1,2)}=\mw_\mx^{(1,2)}
	=\underset{\mo\in \mathbb{R}^{d\times d}}{\operatorname{argmin}} \|\hat\mx^{(1)}_{ \langle\mathcal{U}_1\cap \mathcal{U}_2\rangle}\mo-\hat\mx^{(2)}_{ \langle\mathcal{U}_1\cap \mathcal{U}_2\rangle}\|_F.
\end{aligned}
$
Extending the idea to a chain of overlapping submatrices is
straightforward; see Section~\ref{supp:asy} of the supplementary
material for the detailed algorithm and simulation results.

Finally,
we note that extending Theorem~\ref{thm:R(il)_ind} to the asymmetric setting is also straightforward if we assume the entries of
$\mathbf{N}^{(i)}$ are independent and that $|\mathcal{V}_i| \asymp
|\mathcal{U}_i|$ for all $i$. Indeed, we can simply apply
Theorem~\ref{thm:R(il)_ind} to the Hermitean dilations
of $\ma^{(i)}$. However, the asymmetric case also allows for 
richer noise models such as the rows of $\ma^{(i)}$ being independent
but the entries in each row are dependent, or imbalanced dimensions
where $|\mathcal{U}_i| \ll |\mathcal{V}_i|$ or vice versa. % for some indices
%$i$. 
We leave theoretical results for these more general settings to future work. 

}

\newpage
 
\bibliographystyle{chicago}
\bibliography{ref}

\newpage

\begin{center}%
    {\LARGE Supplementary Material for ``Chain-lined Multiple Matrix Integration via Embedding Alignment"\par}%
  \end{center}

\counterwithin{figure}{section}
\counterwithin{theorem}{section}
\counterwithin{assumption}{section}
\counterwithin{algorithm}{section}

\appendix

\section{Additional Numerical Results}

\setcounter{figure}{0}

\subsection{Entrywise normal approximations}
We now compare the entrywise behavior of $\hat{\mpp}_{\mathcal{U}_{i_0}, \mathcal{U}_{i_L}} - \mpp_{\mathcal{U}_{i_0}, \mathcal{U}_{i_L}}$ against the limiting distributions in Theorem~\ref{thm:normal}. % by Figure~\ref{fig:simulation_normal} and Figure~\ref{fig:simulation_cov}. 
In particular, we plot in Figure~\ref{fig:simulation_normal} histograms (based on $1000$ independent Monte Carlo replicates) of the $(i,j)$th entries where $(i,j) \in \{(1,1),(1,2),(1,3)\}$, and it is clear that the empirical distributions in Figure~\ref{fig:simulation_normal} %and %Figure~\ref{fig:simulation_cov} 
are well approximated by the normal distributions with parameters given in Theorem~\ref{thm:normal}.
%This numerical observation is in agreement with the asymptotic normality established in Theorem~\ref{thm:normal}.
\begin{figure}[htbp!] 
\centering
\subfigure{\includegraphics[width=4cm]{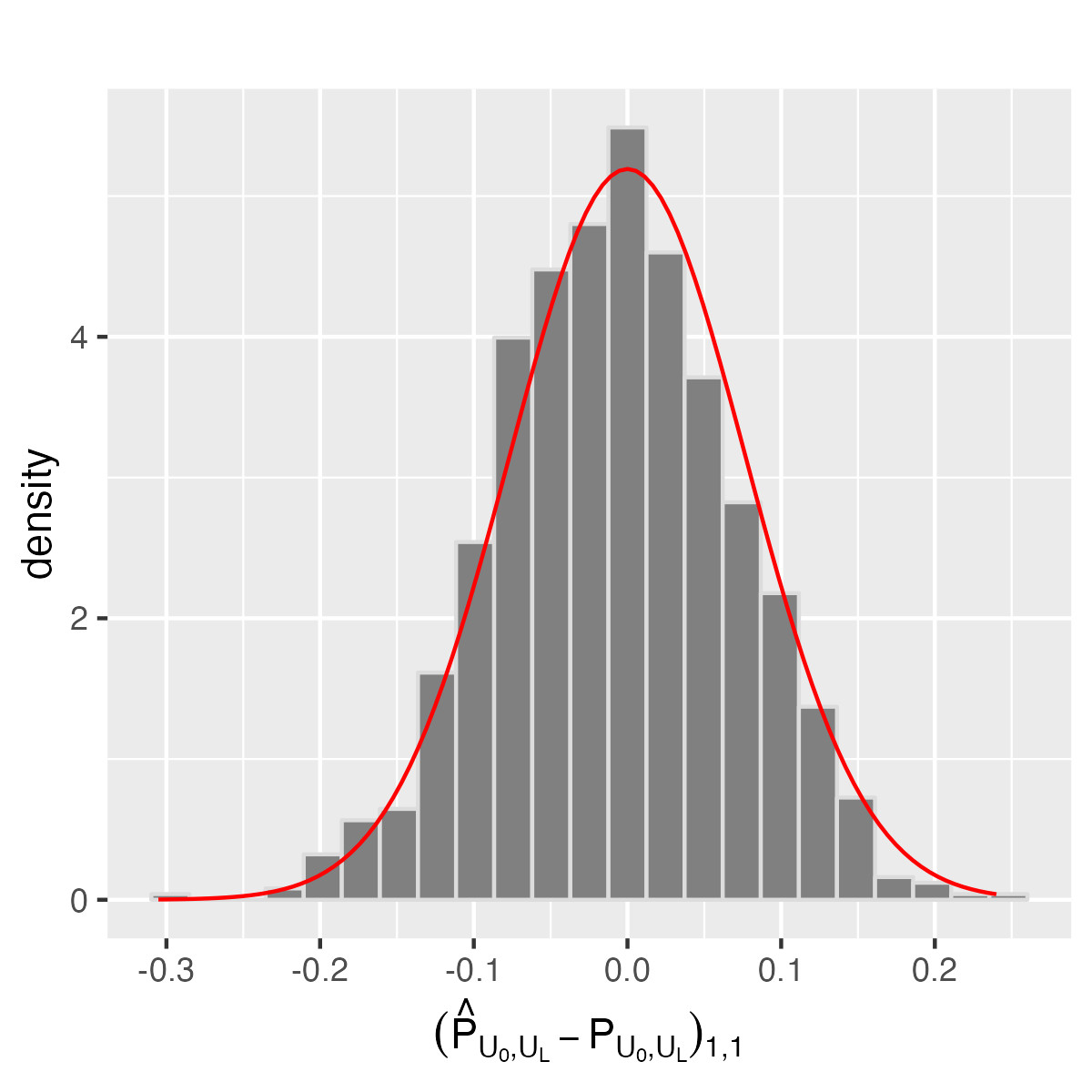}}
\subfigure{\includegraphics[width=4cm]{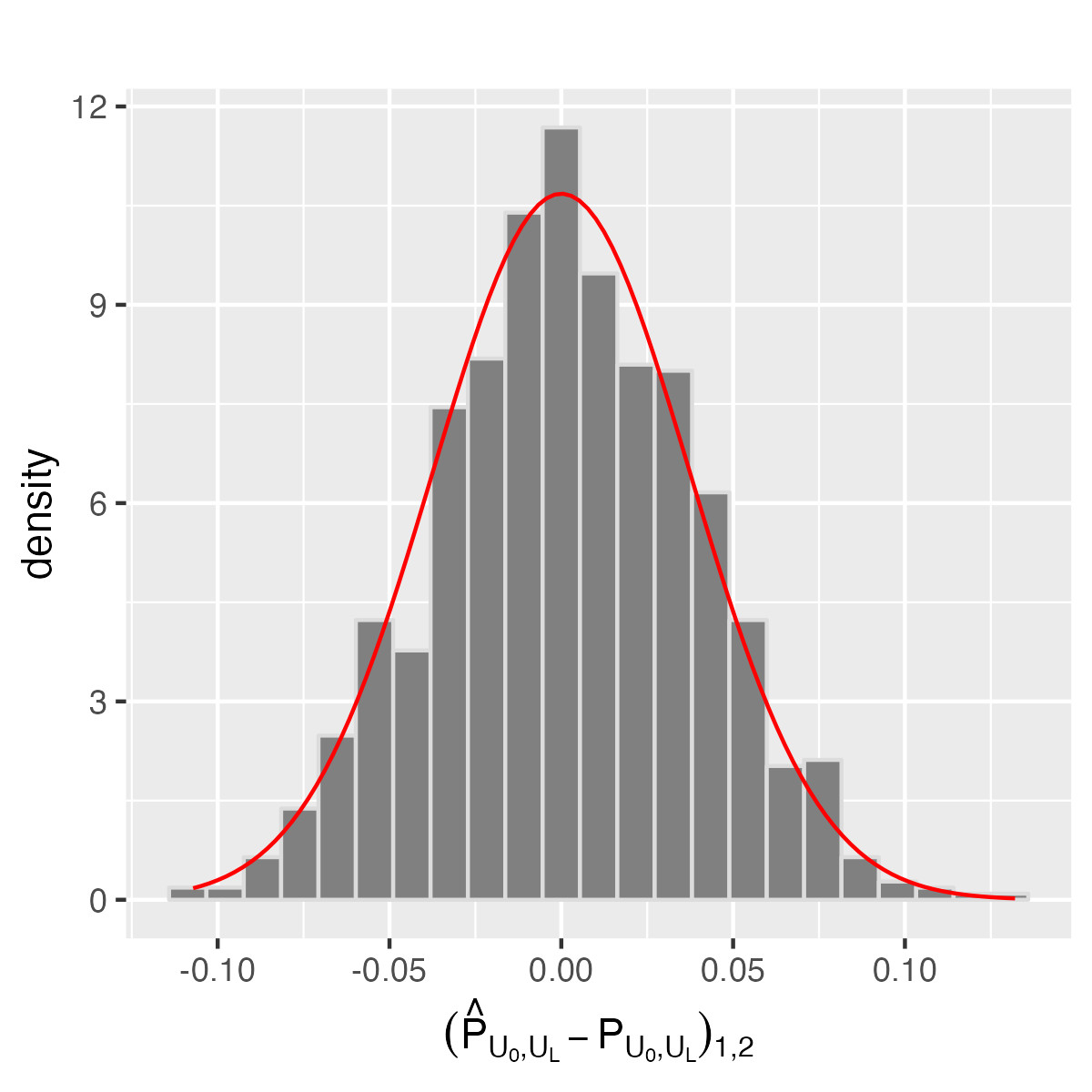}}
\subfigure{\includegraphics[width=4cm]{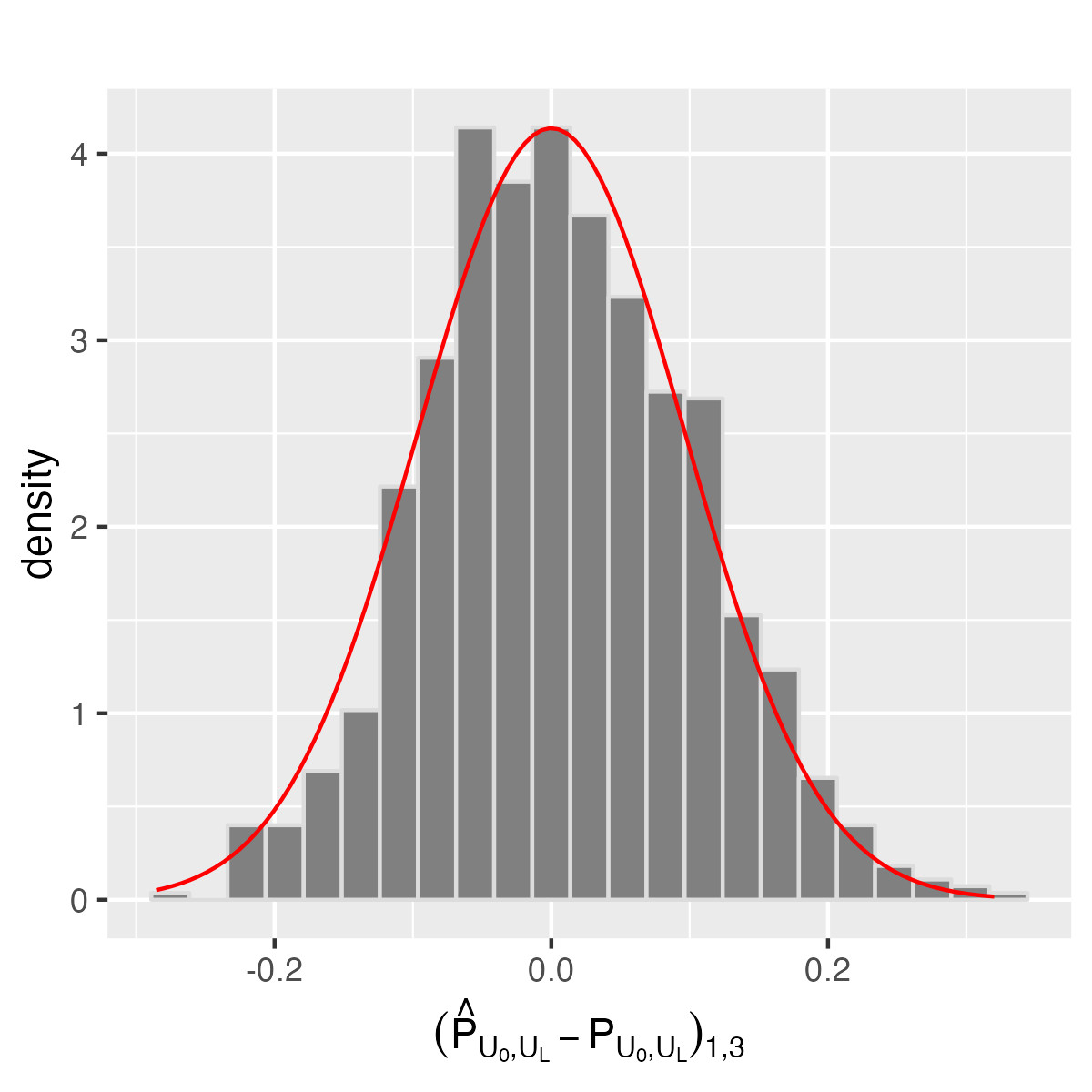}}
\caption{\footnotesize Histograms of the empirical distributions for the $(i,j)$th entry of $\hat\mpp_{\mathcal{U}_0,\mathcal{U}_L}-\mpp_{\mathcal{U}_0,\mathcal{U}_L}$; here $(i,j) = (1,1)$ (left panel), $(i,j) = (1,2)$ (middle panel), and $(i,j) = (1,3)$ (right panel).
These histograms are based on $1000$ independent Monte Carlo replicates where 
$N=%\lambda=
6000$, $p=0.3$, $\breve p=0.1$, $q=0.8$, $L=2$, and $\sigma=0.5$. 
The red lines are pdfs of the normal distributions with parameters given in Theorem~\ref{thm:normal}. }
\label{fig:simulation_normal}
\end{figure}

\subsection{Performance of CMMI with minimal overlaps}
\label{sec:simu limit overlap}
%According to Theorem~\ref{thm:R(i0,...,iL)}, the dominating term $\mm_\star$ of the estimation error is always not related to the overlap parameter and has the upper bound $\|\mm_\star\|_{\max}\lesssim \frac{(\|\mpp\|_{\max}+\sigma)\log ^{1/2}N}{(pN)^{1/2}q^{1/2}}$ with high probability.
%If the size of the overlaps are very small, for example $|\mathcal{U}_{i-1} \cap \mathcal{U}_i| = d$ for all $i$ (where $d := \mathrm{rk}(\mpp)$ is the smallest possible value for which the embeddings can still be aligned), then $\breve p \asymp \frac{1}{n} \asymp \frac{1}{pN}$
%and Theorem~\ref{thm:R(i0,...,iL)} implies
%, in this scenario, the estimation error %remainder term 
%has the upper bound 
%\[ 
%\hat{\mpp}
%\|\mr^{(1,\dots,L)}\|_{\max}\lesssim 
%		L\Big(\frac{(\|\mpp\|_{\max}+\sigma)^2\log N}{pq\lambda_{\min}}
	%+\frac{(\|\mpp\|_{\max}+\sigma)\log^{1/2}N}
	%	{ (pN)^{1/2}q^{1/2}} \Big)$ with high probability.
%Therefore, when $\lambda$ grows with $N$, for example, when $\lambda \asymp N$, where all entries in the population matrix are of the same order with some constant, if we focus only on the sizes of the observed submatrices $n \asymp pN$, we obtain an error rate of $n^{-1/2}$.
%This means that CMMI achieves accurate estimation even when the sizes of the observed submatrices $n$ tend to infinity but the overlaps are bounded, and thus it can accurately integrate large submatrices even with very limited overlaps.
		
We now examine the performance of CMMI when the overlap between the submatrices are very small. More specifically, we use the setting from Section~\ref{sec:comp} with $L = 2$ and $|\mathcal{U}_0 \cap \mathcal{U}_1| = |\mathcal{U}_1 \cap \mathcal{U}_2| = 3$; as $\mathrm{rk}(\mpp) = 3$, this is the smallest overlap for which the latent positions for the $\{\hat\mx^{(i)}\}_{i=0}^{L}$ can still be aligned. 
%Now, we examine the scenario where the overlaps are extremely limited, and the sizes of the observed submatrices $n = pN$ increase as discussed above. More specifically, we fix the number of overlapping entities between any two consecutive observed submatrices in the chain to be the rank $d=3$. 
We fix $q = 0.8, \sigma = 0.5$ and compute the estimation error $\|\hat\mpp_{\mathcal{U}_{0},\mathcal{U}_{L}} - \mpp_{\mathcal{U}_{0},\mathcal{U}_{L}}\|_{\max}$ for several values of $n$. The results are summarized in Figure~\ref{fig:simulation_overlap d}. Note that the slope of the line in the left panel of Figure~\ref{fig:simulation_overlap d} is approximately the same as the theoretical error rate of $\|\hat\mpp_{\mathcal{U}_{0},\mathcal{U}_{L}} - \mpp_{\mathcal{U}_{0},\mathcal{U}_{L}}\|_{\max} \lesssim n^{-1/2} \log^{1/2}{n}$ in Remark~\ref{rem:example1_setting}. In summary, CMMI can integrate arbitrarily large submatrices even with very limited overlap.
\begin{figure}[htbp!] 
\centering
\subfigure{\includegraphics[height=4.1cm]{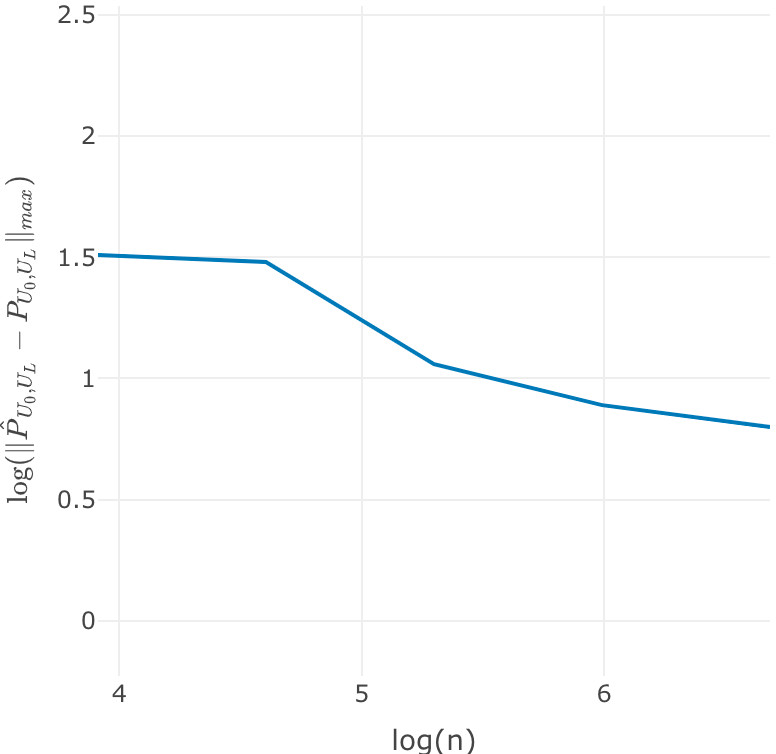}}
\subfigure{\includegraphics[height=4.1cm]{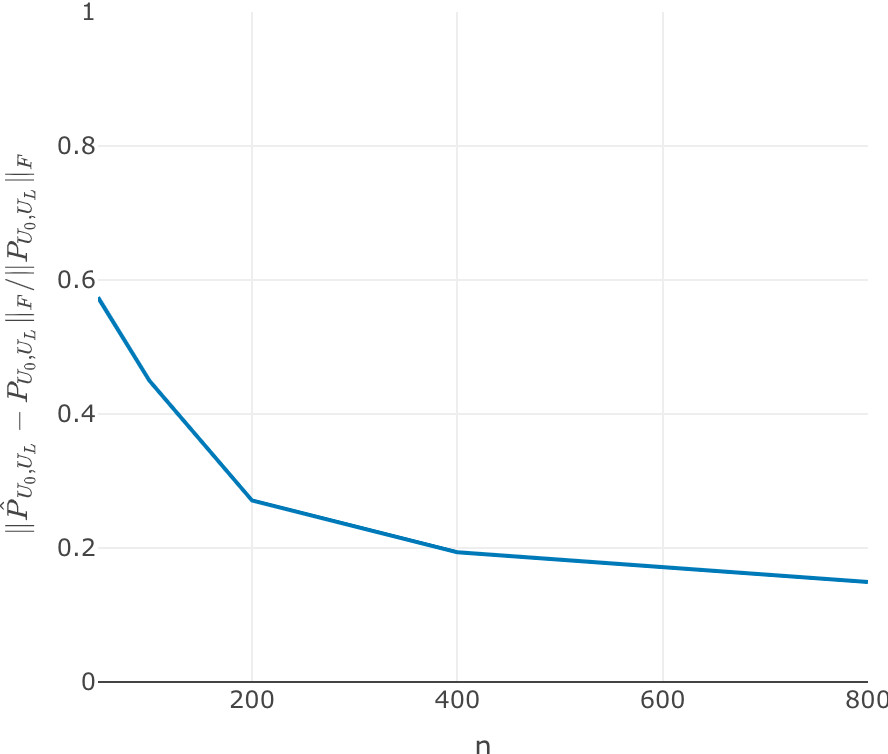}}
\caption{\footnotesize The left panel is a log-log plot of the empirical error rate for $\|\hat\mpp_{\mathcal{U}_{0},\mathcal{U}_{L}} - \mpp_{\mathcal{U}_{0},\mathcal{U}_{L}}\|_{\max}$ as we vary the values of $n \in \{50, 100, 200, 400, 800\}$ while fixing the overlap size as $d=3$, with $L=2$, $q=0.8$, and $\sigma=0.5$.
Right panel reports the empirical error rate for $\|\hat\mpp_{\mathcal{U}_{0},\mathcal{U}_{L}}-\mpp_{\mathcal{U}_{0},\mathcal{U}_{L}}\|_F/\|\mpp_{\mathcal{U}_{0},\mathcal{U}_{L}}\|_F$.
Results are based on $100$ independent Monte Carlo replicates.
%The right panel reports empirical estimates for the relative $F$-norm errors $\|\hat\mpp_{\mathcal{U}_{0},\mathcal{U}_{L}}-\mpp_{\mathcal{U}_{0},\mathcal{U}_{L}}\|_F/\|\mpp_{\mathcal{U}_{0},\mathcal{U}_{L}}\|_F$ as $n$ changes.
}
\label{fig:simulation_overlap d}
\end{figure}

\subsection{Comparison of the recovery of each $\mx^{(i)}$ with other algorithms}
\label{sec:simu initilization}

We use the same setting as in Section~\ref{sec:simu1}, but consider only a single observed block of size $n$. We evaluate how different algorithms recover the corresponding latent position matrix $\mx$ of the block.
The SVD-based algorithm computes $\hat\mx$ as the scaled leading eigenvectors of $\ma$ in Eq.~\eqref{eq:A(i)=...}, and for other matrix completion methods, $\hat\mx$ is obtained as the scaled leading eigenvectors of the recovered matrix.
Our performance metric for recovering the latent position matrix is in terms of the relative Frobenius norm error 
$\min_{\mw\in\mathcal{O}_d}\|\hat\mx\mw-\mx\|_F/\|\mx\|_F$. 
Plots of the error rates (averaged over $100$ independent Monte Carlo replicates) for different algorithms and their running times are presented in the left and right panels of Figure~\ref{fig:simulation_initialization}, respectively.
Figure~\ref{fig:simulation_initialization} shows that the SVD-based algorithm achieves comparable recovery accuracy relative to other matrix completion methods while being computationally efficient.
\begin{figure}[h!] 
\centering
\subfigure{\includegraphics[height=4.5cm]{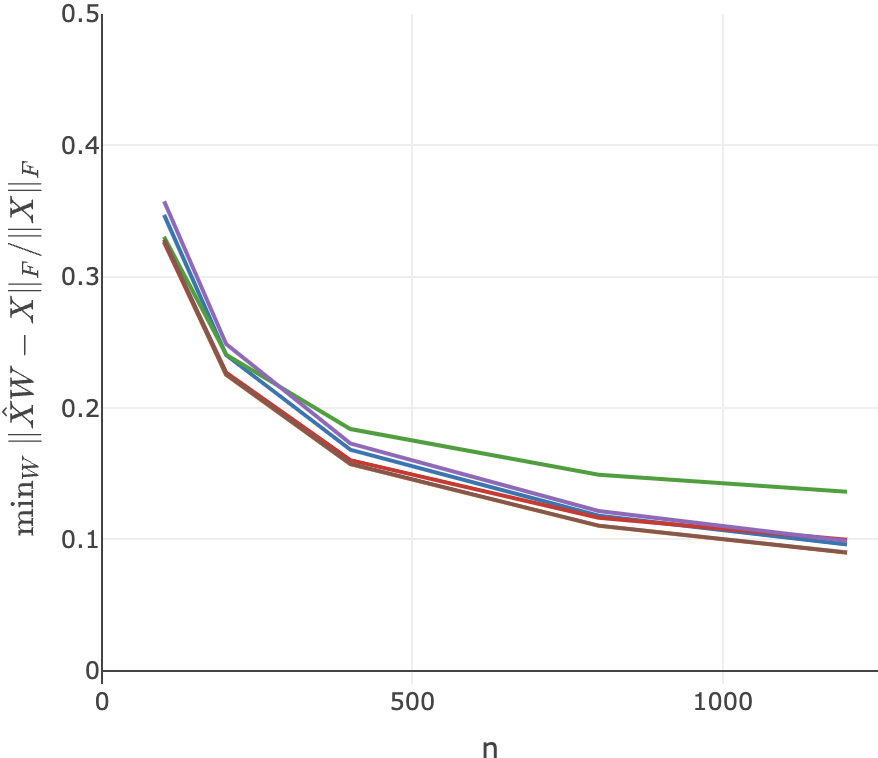}}
\subfigure{\includegraphics[height=4.5cm]{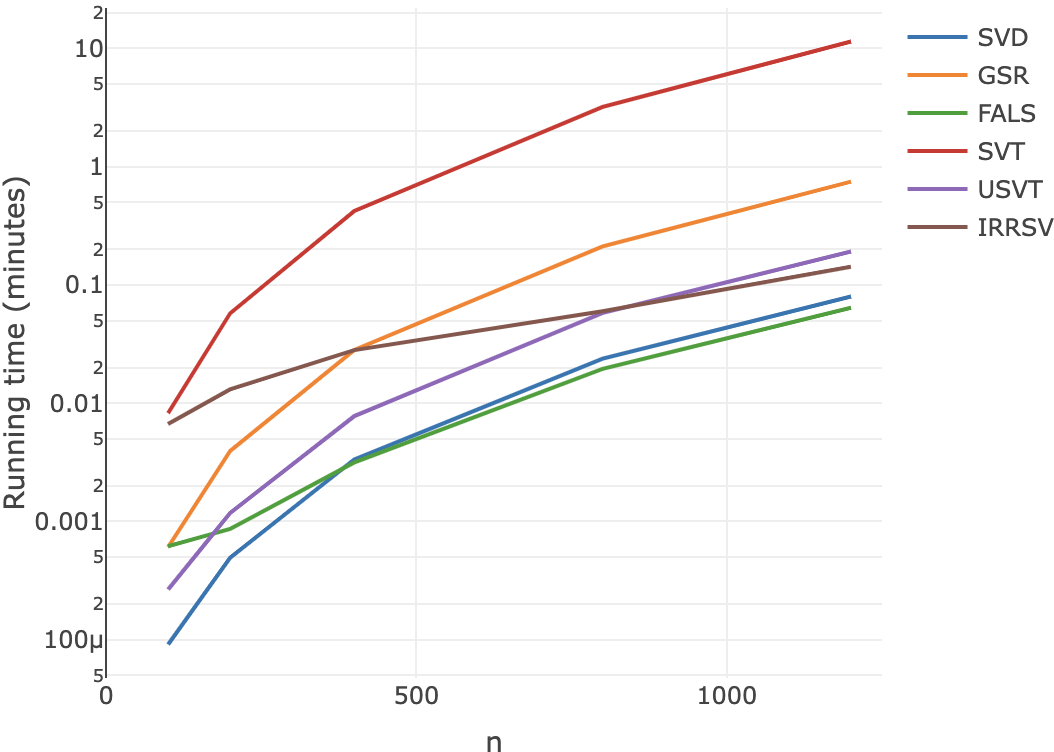}}
%\subfigure[]{\includegraphics[width=4cm]{image/simulation_disPCA/n_2.png}}
\caption{%Comparison with other algorithms when we tune $L$.
\footnotesize The left panel reports empirical errors $\min_{\mw\in\mathcal{O}_d}\|\hat\mx\mw-\mx\|_F/\|\mx\|_F$ for the SVD-based algorithm and other matrix completion algorithms as we vary $n \in \{100,200,400,800,1200\}$ while fixing $q=0.8$, $\sigma=2$.
The results are averaged over $100$ independent Monte Carlo replicates.
The average running time (in log scale) over 100 replicates for algorithms, using 25-core parallel computing and 256 GB memory, is shown in the right panel. % The machine is equipped with 256 GB of memory and two 64-core, 2.25 GHz, 225-watt AMD EPYC 7742 processors.
}
\label{fig:simulation_initialization}
\end{figure}

In some cases one may initialize
$\{\hat{\mx}^{(i)}\}$ using other matrix
completion algorithms to obtain slight improvements in
the joint integration. These gains are, however, limited (as 
Section~\ref{sec:comparison} shows that initialization using SVD-based algorithm is rate-optimal) while also being more computationally costly.

\subsection{Real data experiment: MNIST}
\label{sec:MNIST}

We compare the performance of CMMI against other matrix completion algorithms on the MNIST database of grayscale images.
%\subsection{MNIST}
 The MNIST database consists of $60000$ grayscale images of handwritten digits for the numbers $0$ through $9$. Each image is of size $28 \times 28$ pixels and can be viewed as a vector in $\{0,1,\dots,255\}^{784}$. Let $\mathbf{Y}$ denote the $60000 \times 784$ matrix whose rows represent these images, where each row is normalized to be of unit norm. 
We consider a chain of $L+1$ overlapping blocks, each block corresponding to a partially observed (cosine) kernel matrix for some subset of $n = 1000$ images. More specifically, 
%. We fix the dimension for each observed block $n=1000$. Therefore when we increase $L$ we have a longer chain of observed submatrices and have a larger population matrix.
% observed blocks based on the MNIST data by the following steps: 
\begin{enumerate}
\item for each $0 \leq i \leq L$ we generate a $n \times 784$ matrix $\my^{(i)}$ whose rows are sampled independently and {\em uniformly} from rows of $\my$ corresponding to one of the digits $\{0,1,2\}$, with the last $n\breve p$ rows of $\my^{(i-1)}$ and the first $n\breve p$ rows of $\my^{(i)}$ having the same labels;
\item we set $\mpp^{(i)} = \my^{(i)} \my^{(i)\top}$ for all $0 \leq i \leq L$;
\item finally, $\ma^{(i)} = \mpp^{(i)} \circ \bm{\Omega}^{(i)}$ where $\bm{\Omega}^{(i)}$ is a $n \times n$ symmetric matrix whose upper triangular entries are i.i.d. Bernoulli random variables with success probability $q$. 
%\item We first sample labels from $\{0,1,2\}$ uniformly at random for each row of $\mpp$.
%\item Next, for each $0 \leq i \leq L$, we generate a $n \times 784$ matrix $\my^{(i)}$ as embeddings. The rows of $\my^{(i)}$ have the same labels with the corresponding rows of the whole matrix. Then each row of $\my^{(i)}$ is randomly sampled from the rows of $\my$ which have the same label as this row of $\my^{(i)}$.
%\item We get each observed block $i=0,\dots,L$ as $\my^{(i)}\my^{(i)\top}$ with non-missing probability $q$.
\end{enumerate}
%And notice in this experiment, we only use the data for numbers $0$, $1$ and $2$.

Given above collection of $\{\ma^{(i)}\}_{0 \leq i \leq L}$, we compare the accuracy for jointly clustering the images in the first and last blocks. In particular, for CMMI we first construct an embedding $\hat{\mx}^{(i)} \in \mathbb{R}^{n \times d}$ using the $d$ leading scaled eigenvectors of $\ma^{(i)}$ for each $0 \leq i \leq L$.
We specify $d=36$ for CMMI, where this choice 
is based on applying the dimensionality selection procedure of
\cite{zhu2006automatic} to $\my$.
We then align $\hat{\mathbf{X}}^{(0)}$ to $\hat{\mathbf{X}}^{(L)}$ via
$\breve{\mathbf{X}}^{(0)} := \hat{\mathbf{X}}^{(0)} \mathbf{W}^{(0,1)} \cdots \mathbf{W}^{(L-1,L)}$,
and concatenate the rows of $\breve{\mathbf{X}}^{(0)}$ and $\hat{\mathbf{X}}^{(L)}$ into a $2n \times d$ matrix $\mathbf{Z}^{(0,L)}$. We cluster the rows of $\mathbf{Z}^{(0,L)}$ into three groups using $K$-means, and evaluate clustering accuracy against the true labels $\ell \in \{0,1,2\}$ using the Adjusted Rand Index (ARI).
		Note that
 ARI values range from $-1$ to $1$, with higher values indicating closer alignment between two sets of labels. 
 For the other low-rank matrix completion algorithms, we reconstruct $\mathbf{P}$ from $\{\mathbf{A}^{(i)}\}_{i=0}^{L}$ using $d = 36$ for FALS, and $d = 3$ for GSR, as the running time of GSR increases substantially with larger values of $d$.
 Letting $\hat{\mathbf{P}}$ denote the resulting estimate, we then compute $\hat{\mathbf{X}}$ such that $\hat{\mathbf{X}} \hat{\mathbf{X}}^{\top}$ is the best rank-$36$ approximation to $\hat{\mathbf{P}}$ in Frobenius norm (among all positive semidefinite matrices). We extract the $2n$ rows of $\hat{\mathbf{X}}$ corresponding to the images in $\mathbf{A}^{(0)}$ and $\mathbf{A}^{(L)}$, cluster these rows into 3 groups using $K$-means, and compute the Adjusted Rand Index (ARI) of the resulting cluster assignments.
 % Adjusted rank index for the 
%
%
%  
%=\hat\mx^{(0)} 
%		\mw^{(0,1)}\cdots\mw^{(L-1,L)}$ and $\hat\mx^{(L)}$. 
		%Notice that the initial embeddings $\hat\mx^{(0)}$ and $\hat\mx^{(L)}$ can not be compared because they are obtained from different submatrices and the embeddings have no identifiability, while 
%		Then $\breve\mx^{(0)}$ and $\hat\mx^{(L)}$ are comparable. We cluster the rows of $\breve\mx^{(0)}$ and $\hat\mx^{(L)}$ together into $3$ clusters, and compute the adjusted rank index (ARI) with the true labels of the two blocks to measure the classification accuracy.
%For other low-rank matrix completion algorithms, we compute the estimated latent positions
%adjacency spectral embedding (ASE) \citep{sussman2012consistent} 

%$\hat\mx$ of the whole completed matrix $\hat\mpp$ such that $\hat\mx\hat\mx^\top=\hat\mpp$. We cluster the corresponding rows of the ASE for the first block $\hat\mx_{\mathcal{U}_0}$ and the last block $\hat\mx_{\mathcal{U}_L}$ into $3$ clusters, and also compute ARI with the true labels to measure the classification accuracy.

Comparisons between the ARIs of CMMI and other matrix completion algorithms, for different numbers of submatrices $L$, are summarized in Figure~\ref{fig:realdata_MNIST}. %Note that the black dotted line in the left panel of Figure~\ref{fig:realdata_MNIST} are ARIs when applying $K$-means directly on $\my^{(0)}$ and $\my^{(L)}$, and thus represent the best possible clustering accuracy. 
We observe that CMMI outperforms all competing methods on this dataset.
CMMI also has strong advantages in computational efficiency; see Section~\ref{sec:time comparison} for details.
%, and its ARIs remain close to optimal even as $L$ increases. Finally, the right panel of Figure~\ref{fig:realdata_MNIST}  shows that the running time for CMMI is orders of magnitude smaller than that of other algorithms. 

\begin{figure}[htbp!] 
\centering
\subfigure{\includegraphics[height=5cm]{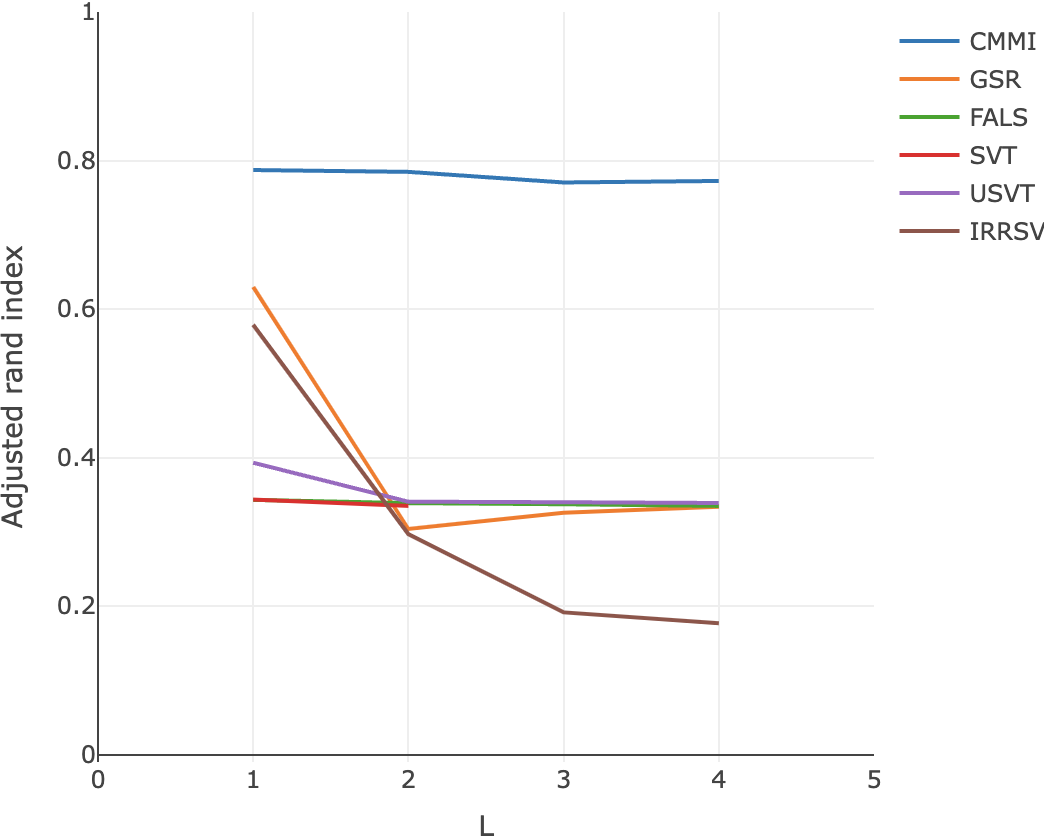}}
%\subfigure{\includegraphics[height=4.5cm]{figure/realdata/MNIST_block_time.png}}
\caption{%Comparison with other methods $L=\{1,2,3,4,7,9\}$.
\footnotesize{%The left panel reports 
ARIs for joint clustering of $(\hat{\mx}^{(0)}, \hat{\mx}^{(L)})$ for subsets of the MNIST dataset using CMMI and other matrix completion algorithms as we vary $L \in \{1,2,3,4\}$ while fixing $n=1000$, $\breve p=0.1$, $q=0.8$.
%The black dotted line is the ARIs obtained by clustering directly on the original data (i.e, observed without noise nor missing entries). 
The results are averaged over $100$ independent Monte Carlo replicates.
%The average running time (in log scale) over 100 replicates for algorithms, using 25-core parallel computing and 256 GB memory, is shown in the right panel. % The machine is equipped with 256 GB of memory and two 64-core, 2.25 GHz, 225-watt AMD EPYC 7742 processors.
We only evaluate the performance of SVT for $L\leq 2$, as these algorithms are computationally prohibitive with even slight increases in $L$.}
}
\label{fig:realdata_MNIST}
\end{figure}

\setcounter{figure}{0}
\setcounter{algorithm}{0}

\newpage 

{\color{black}

 \section{Integration of Multiple Matrices with Complex Connectivity}
 \label{sec:integration}
 
If we are given a chain of overlapping submatrices of some larger matrix $\mathbf{P}$, then Algorithm~\ref{Alg_chain} provides a simple and computationally efficient procedure for recovering the unobserved regions of $\mathbf{P}$. 
In practice, the structure of the observed submatrices can be more complex than a simple chain. Building on the idea of Algorithm~\ref{Alg_chain}, which can be used to integrate any pair of connected observed submatrices, we now introduce a procedure for integrating submatrices with arbitrary overlap structures. 
The procedure is illustrated for positive semidefinite matrices, and it can be easily extended to the cases of symmetric indefinite matrices and asymmetric or rectangular matrices. 
%In addition, in Section~\ref{sec:aggregate} we discuss the use of a preprocessing step proposed in \cite{zhou2021multi} to aggregate multiple observed values by introducing weights to each source. 

% \subsection{Algorithm for holistic recovery}

Suppose we have observed submatrices $\ma^{(1)},\ma^{(2)},\dots,\ma^{(K)}$ for $\mathcal{U}_{1},\mathcal{U}_{2},\dots,\mathcal{U}_{K}\subset [N]$ and want to integrate them.
Given Algorithm~\ref{Alg_chain}, a straightforward idea is to sequentially integrate each pair of connected submatrices along a chain connecting them. However, this strategy can lead to a significant amount of redundant computation. We now describe a more efficient approach that simplifies the integration process and allows all observed submatrices to be integrated simultaneously.

We consider the following undirected graph $\mathcal{G}$ to facilitate the integration process. Specifically, $\mathcal{G}$ has $K$ vertices $\{v_1, \dots, v_K\}$, where each vertex $v_i$ corresponds to the observed submatrix $\ma^{(i)}$ with estimated latent position matrix $\hat\mx^{(i)}$, and $v_i$ and $v_j$ are adjacent if and only if $\mathcal{U}_i$ and $\mathcal{U}_j$ are overlapping, i.e., $|\mathcal{U}_i \cap \mathcal{U}_j| \geq d$.
For each pair of adjacent vertices $v_i$ and $v_j$ in $\mathcal{G}$, we can compute an orthogonal matrix $\mw^{(i,j)}$ to align $\hat\mx^{(i)}$ and $\hat\mx^{(j)}$.

For any pair of vertices $v_i$ and $v_j$ in $\mathcal{G}$, if they are connected, meaning there exists a path between them, then the corresponding submatrices along this path form a chain that can be used to integrate $\hat\mx^{(i)}$ and $\hat\mx^{(j)}$.
In the following, we assume that $\mathcal{G}$ is connected, so that all latent position estimates $\{\hat\mx^{(i)}\}_{i \in [K]}$ can be integrated. If $\mathcal{G}$ is not connected, the integration procedure can be applied separately to each connected component.

Suppose we have a graph $\mathcal{G}$ as visualized in panel~(a) of Figure~\ref{fig:holistic1}. Note that $\mathcal{G}$ in panel~(a) of Figure~\ref{fig:holistic1} contains cycles, which means there exists at least one pair of vertices $v_i$ and $v_j$ with multiple paths connecting them. If, instead, there is a unique path between every pair of vertices, then all the latent position matrices $\{\hat\mx^{(i)}\}_{i \in [K]}$ can be consistently aligned, allowing all unobserved entries to be recovered simultaneously.
To resolve this issue, we consider a spanning tree of $\mathcal{G}$, as illustrated in panel~(b) of Figure~\ref{fig:holistic1}.
While Theorem~\ref{thm:R(i0,...,iL)} shows that the choice of spanning tree has a negligible effect on the estimation error, we may still prefer to construct a tree such that the paths pass through vertices with smaller estimation errors.
This can be achieved by setting the weight of any edge $e_{i,j}$ in $\mathcal{G}$ to $c_i+c_j$, where $c_i$ reflects the magnitude of the error for $\hat\mx^{(i)}$ as an estimate of $\mx_{\mathcal{U}_i}$. Specifically, motivated by Lemma~\ref{lemma:hat X(i)W(i)-X}, we define
$$c_i:=\|(\ma^{(i)}-\hat\mpp^{(i)})\hat\mx^{(i)}(\hat\mx^{(i)\top}\hat\mx^{(i)})^{-1}\|_F/n_i^{1/2},$$
where $\hat\mpp^{(i)}=\hat\mx^{(i)}\hat\mx^{(i)\top}$.
A minimum spanning tree (MST) of $\mathcal{G}$ is then computed based on these edge weights; see panel~(c) of Figure~\ref{fig:holistic1}.
\begin{figure}[htbp!] 
\centering
\subfigure[]{\includegraphics[height=4.3cm]{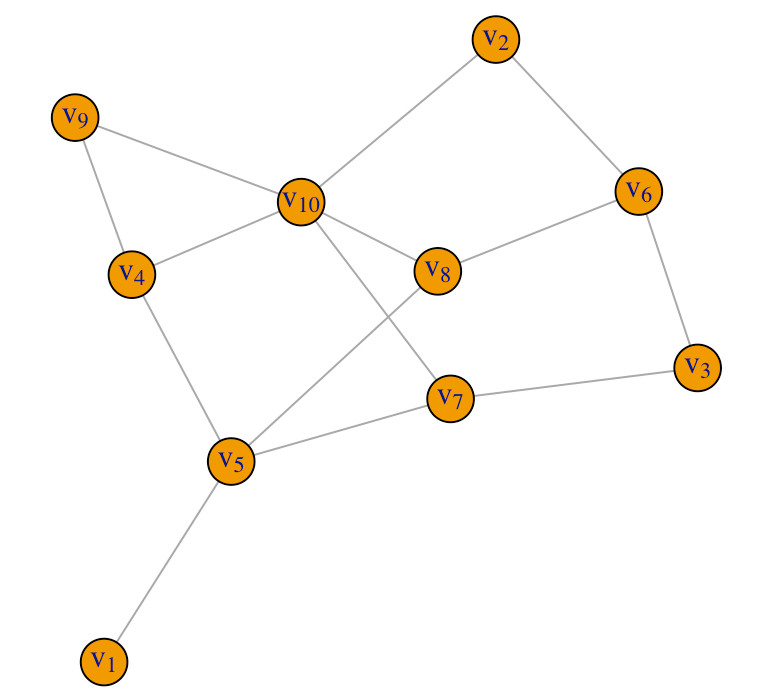}}
\subfigure[]{\includegraphics[height=4.3cm]{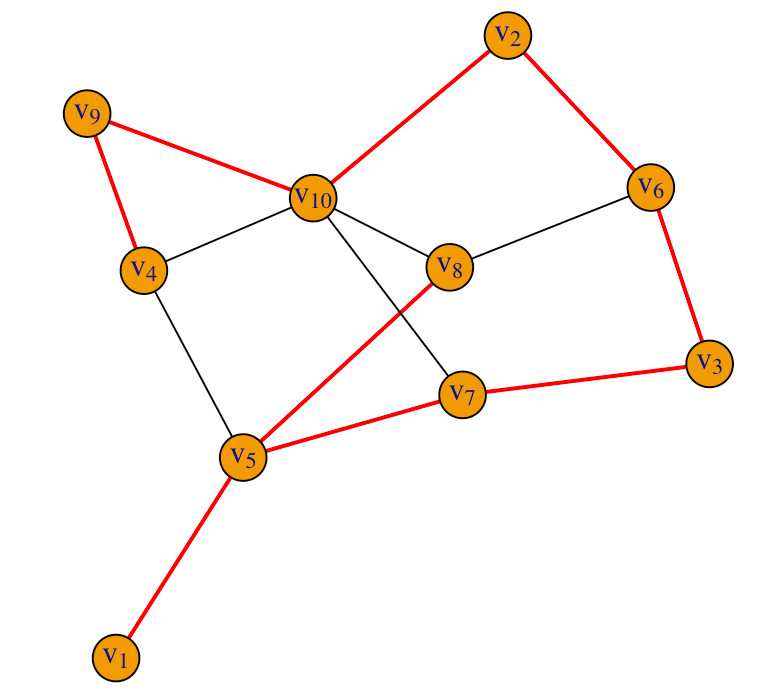}}
\subfigure[]{\includegraphics[height=4.3cm]{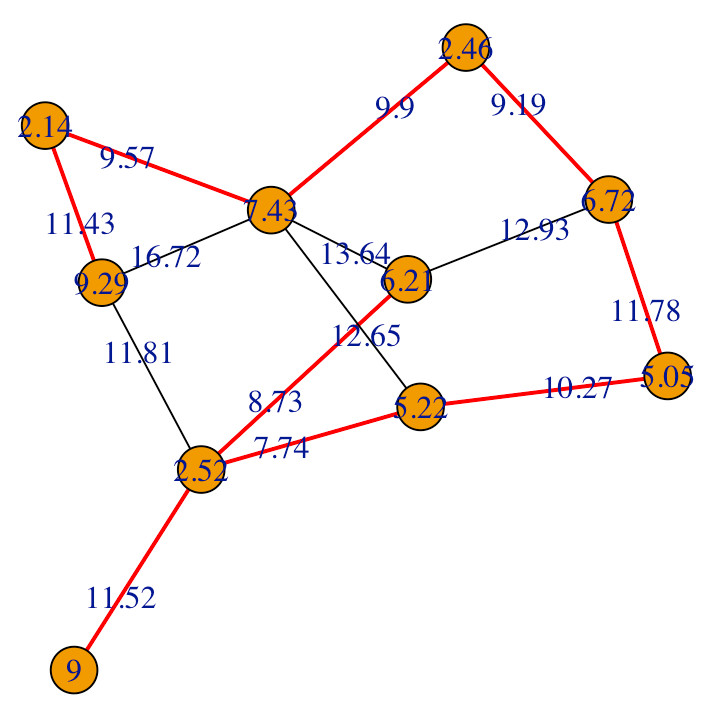}}
\caption{\footnotesize An example of a graph $\mathcal{G}$ described in Section~\ref{sec:integration} (see panel (a)). Panel (b) gives a spanning tree of $\mathcal{G}$. In Panel (c), we suppose the vertices have the estimation error magnitudes $\{c_i\}$ as shown in the labels of vertices. The edges weights between vertices $i$ and $j$ are given by $c_i+c_j$, and we highlight the minimum spanning tree of $\mathcal{G}$ using red colored lines.}
\label{fig:holistic1}
\end{figure}

Given a minimum spanning tree, we randomly select a reference index $i_\star$ and align the remaining latent position estimates $\{\hat\mx^{(i)}\}_{i \ne i_\star}$ to $\hat\mx^{(i_\star)}$ using the unique paths in the tree. Note that for some entities $\ell \in [N]$, we may obtain multiple estimated latent positions from different submatrices, denoted by $\{\hat\mx^{(i)}_{\langle \ell \rangle}\}_{\ell \in \mathcal{U}i}$. In such cases, we compute a weighted average of these estimates to obtain the integrated estimated latent position for entity $\ell$, where the weight assigned to each $\hat\mx^{(i)}_{\langle \ell \rangle}$ is given by $1 / c_i^2$.
See Algorithm~\ref{Alg_holistic} for  details.

Overall, the generalized CMMI procedure in Algorithm~\ref{Alg_holistic} provides a principled approach for aligning entity embeddings across arbitrarily connected submatrices and for aggregating these aligned embeddings into a unified representation.

\begin{algorithm}[htbp!]
\caption{Algorithm for holistic recovery}	
\label{Alg_holistic}
\begin{algorithmic}
\REQUIRE 
\small%\footnotesize
Embedding dimension $d$, observed submatrices $\ma^{(1)},\ma^{(2)},\dots,\ma^{(K)}$ for $\mathcal{U}_{1},\mathcal{U}_{2},\dots,\mathcal{U}_{K}\subset [N]$. %, an overlapping threshold $r\geq d$.

\STATE 
\vspace{0.1cm}
\textbf{Step 1 Constructing the weighted graph $\mathcal{G}$: }

\begin{enumerate}
	\item $\mathcal{G}$ have $K$ vertices $v_1,\dots,v_K$, and $v_i,v_j$ are adjacent if and only if
 $|\mathcal{U}_i\cap\mathcal{U}_j|\geq d$. 

\item For each $i\in [K]$, obtain estimated latent positions for $\mathcal{U}_i$, denoted by $\hat\mx^{(i)}$, and compute $c_i=\|(\ma^{(i)}-\hat\mpp^{(i)})\hat\mx^{(i)}(\hat\mx^{(i)\top}\hat\mx^{(i)})^{-1}\|_F/n_i^{1/2}$.% where $\hat\muu^{(i)}\in\mathbb{R}^{|\mathcal{U}_{i}|\times d}$ and the diagonal matrix $\hat\mLambda^{(i)}\in\mathbb{R}^{d\times d}$ contain the $d$ leading eigenvectors and eigenvalues of $\ma^{(i)}$, respectively.

%For each $i\in [K]$, compute $c_i= \|\ma^{(i)}-\hat\muu^{(i)}\hat\mLambda^{(i)}\hat\muu^{(i)\top}\|/|\mathcal{U}_{i}|,    $
%where $\hat\muu^{(i)}\hat\mLambda^{(i)}\hat\muu^{(i)\top}$ is the rank-$d$ eigendecomposition of $\ma^{(i)}$.

\item Set the weight of each edge $e_{i,j}$ as $c_i+c_j$.
\end{enumerate}

%\vspace{0.1cm}

\STATE \textbf{Step 2 Obtaining the aligned latent position estimates $\{\tilde\mx^{(i)}\}_{i\in[K]}$:}

\begin{enumerate}

\item Find the minimum spanning tree of $\mathcal{G}$ by Prim's algorithm or Kruskal's algorithm, and denote its edge set by $E_{MST}$.

\item For each edge $e_{i,j}\in E_{MST}$, obtain $\mw^{(i,j)}$ via the orthogonal Procrustes problem \[\mw^{(i,j)}=\underset{\mo\in\mathcal{O}_d}{\operatorname{argmin}}\|\hat\mx^{(i)}_{\langle\mathcal{U}_i\cap \mathcal{U}_j\rangle}\mo-\hat\mx^{(j)}_{\langle\mathcal{U}_i\cap \mathcal{U}_j\rangle}\|_F.\]
%and obtain $\mw^{(j,i)}=\mw^{(i,j)\top}$ simultaneously.

\item Choose one of the vertex denoted by ${i_\star}$ (for example, $i_\star=1$), and let $\tilde\mx^{(i_\star)}=\hat\mx^{(i_\star)}$.

\item For each $i\in [K]\backslash \{i_\star\}$, apply Breadth-First Search (BFS) to find a path from $i$ to $i_\star$, denoted by 
$(i_0=i,i_1,\dots, i_L=i_\star)$, and let $\tilde\mx^{(i)}=\hat\mx^{(i)}\mw^{(i_0,i_1)}\cdots\mw^{(i_{L-1},i_L)}$.

\end{enumerate}

\STATE \textbf{Step 3 Obtaining the holistic latent position estimate $\tilde\mx\in\mathbb{R}^{N\times d}$:}

For each $\ell\in[N]$, compute $S=\sum_{\ell\in\mathcal{U}_i}c_i^{-2}$, and compute the holistic estimated latent position as
$$
\tilde\mx_\ell
= \sum_{\ell\in\mathcal{U}_i} (c_i^{-2} /S)
    \hat\mx^{(i)}_{\langle \ell \rangle}.
$$

\ENSURE $\hat\mpp = \tilde{\mx} \tilde{\mx}^{\top}$.
\end{algorithmic}
\end{algorithm}

Our theoretical results and the analysis of Algorithm~\ref{Alg_chain} can be naturally extended to Algorithm~\ref{Alg_holistic}, albeit at the cost of more involved expressions and notations. In particular, as the holistic latent position estimate for some entities is computed as a weighted average of individual latent position estimates,
    the (entrywise) estimation error of the final $\hat\mpp$ is a weighted average of the errors from the individual estimates. 
More specifically, recall that the dominant term in the error for each individual estimate is
$\me^{({i_0})}\mx_{\mathcal{U}_{i_0}}(\mx_{\mathcal{U}_{i_0}}^\top\mx_{\mathcal{U}_{i_0}})^{-1} \mx_{\mathcal{U}_{i_L}}^\top
	    +\mx_{\mathcal{U}_{i_0}} (\mx_{\mathcal{U}_{i_L}}^\top\mx_{\mathcal{U}_{i_L}})^{-1} \mx_{\mathcal{U}_{i_L}}^\top \me^{({i_L})}$. 
            The dominant term in the error for the holistic latent position estimate is then a weighted average of multiple independent noise terms corresponding to different $(\me^{(i_0)},\me^{(i_L)})$ pairs. As a result, the holistic estimate, which aggregates information from multiple blocks, can yield significantly reduced error compared to individual
            estimates and lead to more accurate recovery.
            Furthermore, as the dominant error term is a weighted sum of independent terms,
            one can derive an entrywise normal approximation for the estimation error of $\hat\mpp$ using a similar analysis to that in Theorem~\ref{thm:normal} (but with a more complicated expression for the variance $\tilde{\sigma}_{s,t}^2$)

%The dominant term in the error for the holistic latent position estimate is then a weighted average of multiple independent noise terms corresponding to different $\me^{(i_0)}$ and different $\me^{(i_L)}$. As a result, the holistic estimate, which aggregates information from multiple blocks, can yield significantly reduced error compared to individual estimates, leading to more accurate recovery.
%Furthermore, since the dominant error term consists of independent terms as described above, a normal approximation for the estimation error of the $\hat\mpp$ entries can still be established using a similar analysis to that in Theorem~\ref{thm:normal}, although the resulting variance expression may be more intricate.

Compared with BONMI in \cite{zhou2021multi}, which integrates overlapping submatrices pairwise and selects the estimate for each unobserved entity based on the pair with the lowest sum of estimated noise levels, the generalized CMMI algorithm in Algorithm~\ref{Alg_holistic} provides a more refined strategy for multiple matrix integration. First, CMMI aligns all connected submatrices (not just directly overlapping ones), which can significantly expand the set of recoverable entries. Moreover, CMMI simultaneously leverages all connected submatrices by jointly aligning them and computing a weighted average of the estimated latent positions. In contrast, BONMI utilizes only a single selected pair of submatrices for each entity and therefore does not fully exploit the information available from all observed submatrices.

Compared with the sequential integration approach SPSMC in \cite{bishop2014deterministic}, CMMI also provides a more effective integration strategy. CMMI aligns the estimated latent positions from all submatrices more optimally by jointly considering all overlapping pairs, whereas SPSMC is constrained by its sequential structure, and the problem of determining an effective integration order is not addressed in \cite{bishop2014deterministic}. Furthermore, CMMI aggregates the aligned embeddings by incorporating information from all available submatrices, leading to more accurate estimates. In contrast, SPSMC assigns each entity a latent position based solely on the first submatrix in which it appears in the sequence, and thus also does not to fully leverage all available information.

 \subsection{Real data experiment: MEDLINE co-occurrences}
 \label{sec:hostic_realdata}
 
We note that the analysis of the MEDLINE co-occurrence dataset in
Section~\ref{sec:real} is based on a synthetic scenario where some of the observed entries are held out. While this might lead to a somewhat artificial use case, it is nevertheless intentional as we can then evaluate the performance of CMMI for simple chains (as implemented in
Algorithm~\ref{Alg_chain}). 
We now consider a more
realistic data integration problem for the MEDLINE data. More specifically, if we partition the citations
by year then each year tends to feature a
different set of frequently occurring clinical concepts.  The
PMIs computed between the high-frequency concepts in a given year
are likely to be less noisy compared to those involving rarely occurring
concepts. We now demonstrate that, by extracting the PMIs for top-frequency concepts in
each year and then integrating them using CMMI (or other matrix
integration algorithms), we recover more accurate co-occurrence
relationships between the clinical concepts than the PMIs computed directly from data aggregated across all the years. % (as demonstrated by the numerical
%results presented below).
 
 We consider MEDLINE co-occurrence data from the years 1993 to 2022. For each year, we extract a PMI submatrix based on the co-occurrence of the top 1000 most frequent clinical concepts. 
 Given a number of observed years $K$ (it is also the number of observed PMI submatrices), we aim to integrate these submatrices. 
In this experiment, we always select the most temporally distant years for the PMI submatrix integration task. For example, for $K = 2$, we integrate the PMI submatrices corresponding to the years 1993 and 2022.
These two $1000 \times 1000$ PMI submatrices together involve a total of $N=1540$ unique clinical concepts. Our goal is to recover the unobserved entries (approximately 25\%) in the resulting $1540 \times 1540$ PMI matrix.

When applying CMMI in Algorithm~\ref{Alg_holistic}, we determine the embedding dimension $d$ by first applying the automatic dimensionality selection procedure from \cite{zhu2006automatic} to each observed submatrix, and then selecting the largest resulting value as $d$ to retain sufficient information. 
For example, when $K = 2$, the procedure yields dimensions $12$ and $16$ for the two submatrices, and we set $d = 16$ for CMMI. For FALS, we use the same dimension $d$ as in CMMI, while for $GSR$, we always fix $d = 3$ to avoid excessive computational cost.
In addition, for low-rank matrix completion algorithms other than CMMI, we first construct a global matrix by merging all observed submatrices, where entries appearing in multiple submatrices are averaged, and then apply the matrix completion algorithm to this aggregated matrix.
 
We refer to the pre-trained clinical concept embeddings from \cite{beam2020clinical}, learned from massive sources of multimodal medical data, to evaluate how well the algorithms recover the unobserved entries.
Given pre-trained embedding vectors $\bm{v}_1, \dots, \bm{v}_N$ for $N$ clinical concepts, we construct a similarity matrix $\mathbf{P} \in \mathbb{R}^{N \times N}$ where each entry $\mathbf{P}_{ij}$ is the cosine similarity between $\bm{v}_i$ and $\bm{v}_j$. We then measure the similarities between the estimated PMIs for the unobserved entries with corresponding entries in $\mpp$ in terms of the Spearman's rank correlation.

We vary $K$ from $2$ to $15$ to compare the performance of different algorithms, and the results are summarized in Figure~\ref{fig:realdata_MEDLINE_dense}. 
Note that the baseline, shown as a dashed black line in the left panel of Figure~\ref{fig:realdata_MEDLINE_dense}, represents the performance obtained by directly computing PMIs from the co-occurrence data aggregated across the selected years.
Figure~\ref{fig:realdata_MEDLINE_dense} shows that integrating per-year PMIs can yield more faithful co-occurrence relationships than directly computing PMIs from aggregated data, especially when only a few years are observed. For example, when $K = 2$, the baseline achieves only $0.048$, while CMMI reaches $0.273$.
And Figure~\ref{fig:realdata_MEDLINE_dense} also shows that, compared to other matrix integration algorithms, CMMI has the highest accuracy while maintaining significant advantages in computational efficiency.
 
  \begin{figure}[htbp!] 
\centering
\subfigure{\includegraphics[height=6cm]{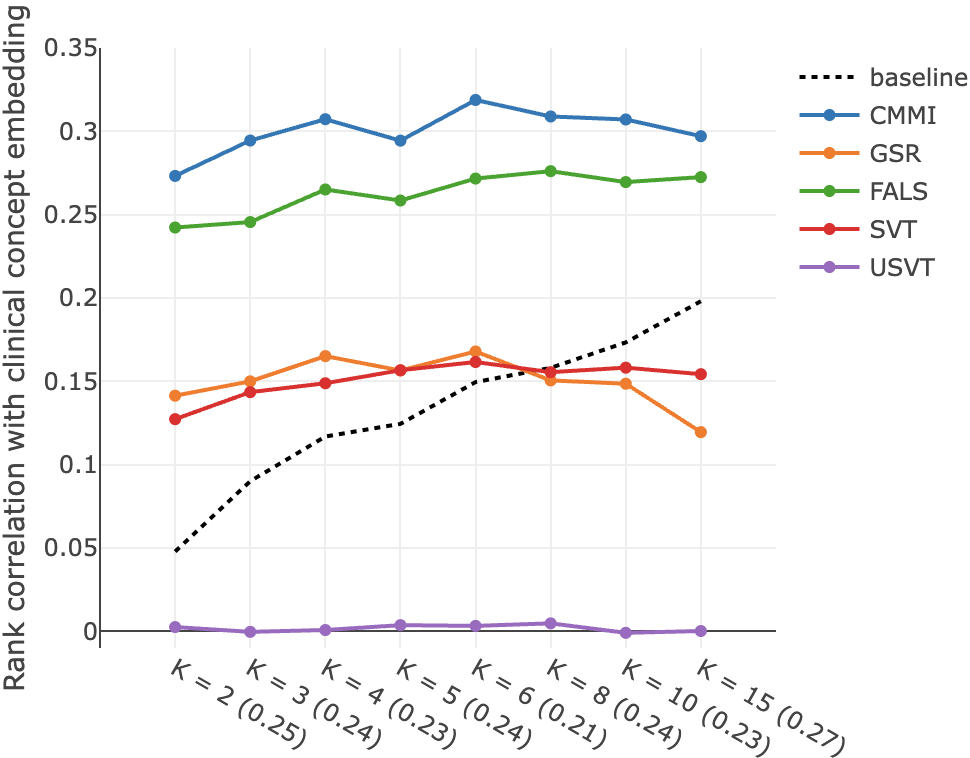}}
\subfigure{\includegraphics[height=6cm]{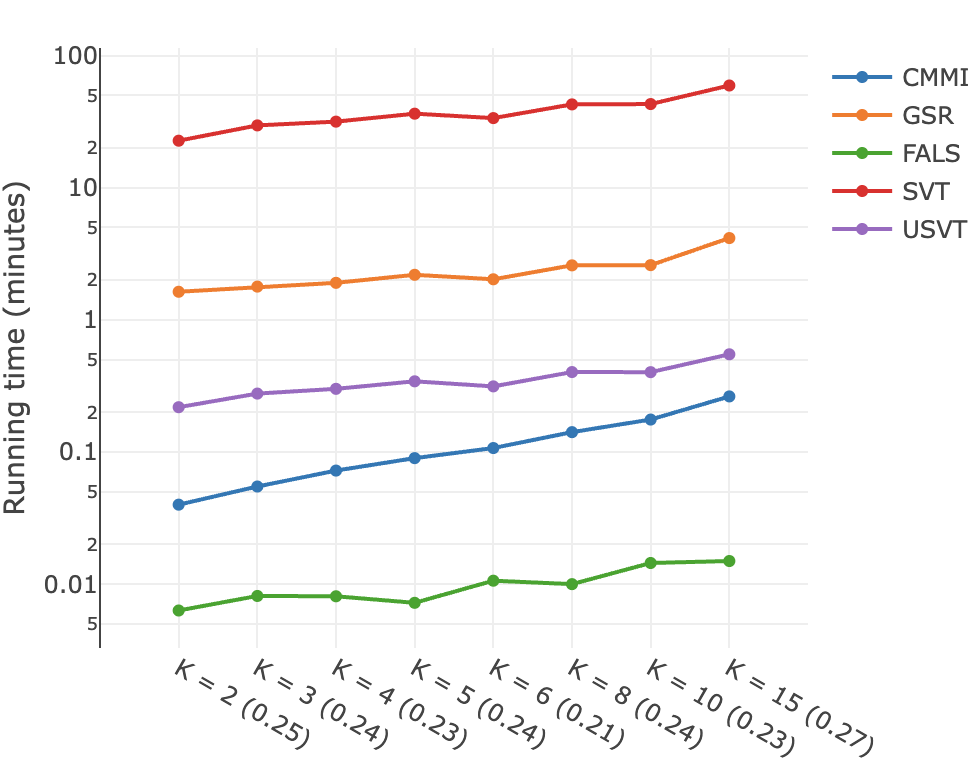}}
\caption{
\footnotesize The left panel reports empirical estimates of Spearman's rank correlations for CMMI and other low-rank matrix completion algorithms as we vary $K$ from $2$ to $15$. 
Note that we always use the most temporally distant $K$ years for the PMI submatrix integration task. For example, when $K = 2$, we use the years 1993 and 2022; when $K = 3$, we use the years 2022, 2008, and 1994; and when $K = 15$, we use the years 2022, 2020, \dots, 1994.
On the x-axis, each value of $K$ is followed (in parentheses) by the proportion of unobserved entries to be recovered in the entire matrix.
The running time (in log scale) for algorithms is shown in the right panel.
}
\label{fig:realdata_MEDLINE_dense}
\end{figure}

 \newpage 
 
\subsection{Computational running time 
for Sections~\ref{sec:comp}, \ref{sec:real}, and \ref{sec:MNIST}}

\label{sec:time comparison}

In the simulation setting in Section~\ref{sec:comp} and the synthetic
settings in Sections~\ref{sec:real} and~\ref{sec:MNIST}, only a single
missing submatrix was of interest, and the basic CMMI was used to
impute this submatrix.  In contrast, other existing low-rank matrix
completion algorithms are designed to impute all missing entries
across the entire matrix, which makes direct comparisons of
computational time somewhat unfair.  With the generalized CMMI, we are
now able to impute all unobserved entries, thereby allowing for a fairer
comparison of computational performance.

The results presented below are obtained following the same setup as
in Sections~\ref{sec:comp},\ref{sec:real}, and \ref{sec:MNIST}.
Instead of using the basic CMMI in Algorithm~\ref{Alg_chain} to impute
just a single missing block, we apply the generalized CMMI in
Algorithm~\ref{Alg_holistic} to recover all unobserved entries in the
entire matrix, and in evaluating algorithm performance, we consider
all unobserved entries rather than focusing on just a single
block of interest.

\textbf{Additional experimental results for Section~\ref{sec:comp}:}
\begin{figure}[h!] 
\centering
\subfigure{\includegraphics[height=5cm]{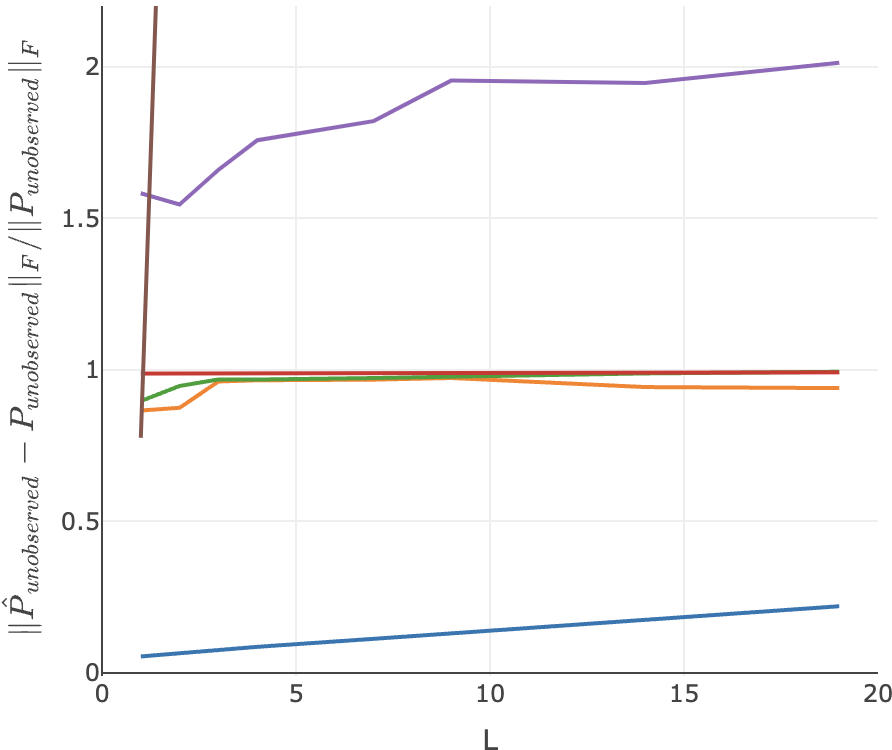}}
\subfigure{\includegraphics[height=5cm]{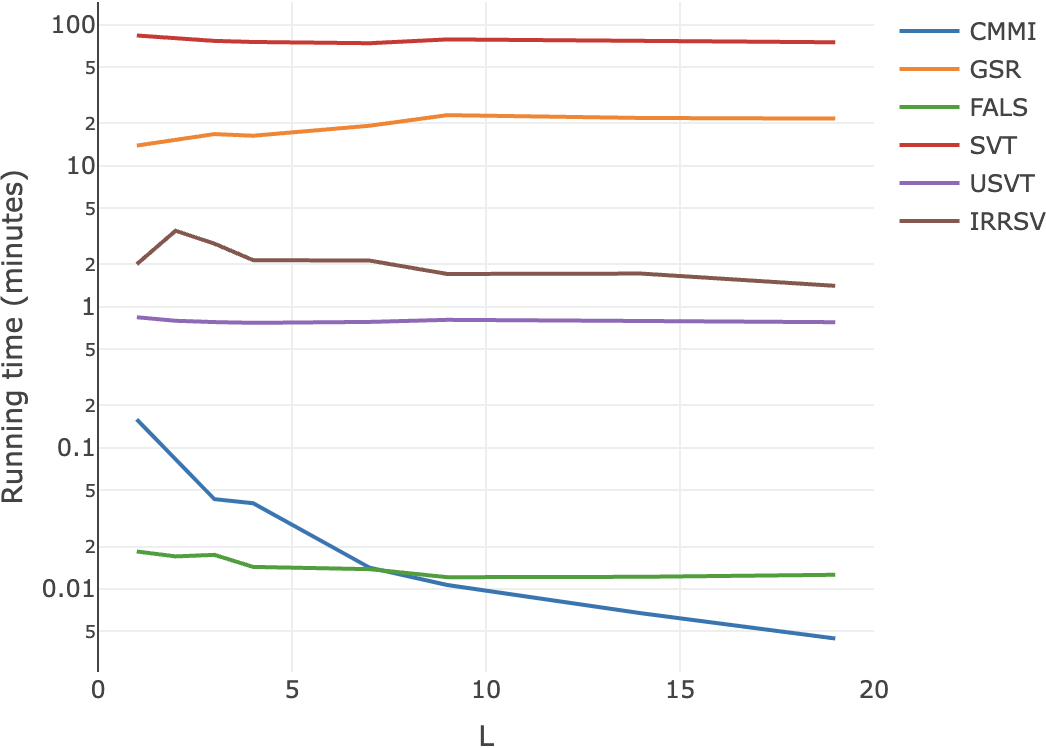}}
%\subfigure[]{\includegraphics[width=4cm]{image/simulation_disPCA/n_2.png}}
\caption{%Comparison with other algorithms when we tune $L$.
\footnotesize The left panel reports empirical errors $\|\hat\mpp_{\text{unobserved}}-\mpp_{\text{unobserved}}\|_F/\|\mpp_{\text{unobserved}}\|_F$, where $\mpp_{\text{unobserved}}$ and $\hat\mpp_{\text{unobserved}}$ denote the submatrices of unobserved entries in $\mathbf{P}$ and $\hat{\mathbf{P}}$, respectively, for CMMI and other matrix completion algorithms as we vary $L \in \{1,2,3,4,7,9,14,19\}$ while fixing $N\approx 2200$, $\breve p=0.1$, $q=0.8$, $\sigma=0.5$.
The results are averaged over $100$ independent Monte Carlo replicates.
Note that the averaged relative $F$-norm errors of IRRSV are $\{0.8,    4.4,   11.5,   15.2,   32.7,   37.5,  202.7, 2299.0\}$ and some of these values are too large to be displayed in this panel. 
The average running time (in log scale) over 100 replicates for algorithms, using 20-core parallel computing and 256 GB memory, is shown in the right panel. % The machine is equipped with 256 GB of memory and two 64-core, 2.25 GHz, 225-watt AMD EPYC 7742 processors.
}
\label{fig:simulation_comparison_all}
\end{figure}

\newpage
\textbf{Additional experimental results for Section~\ref{sec:real}:}
 \begin{figure}[h!] 
\centering
\subfigure{\includegraphics[height=5cm]{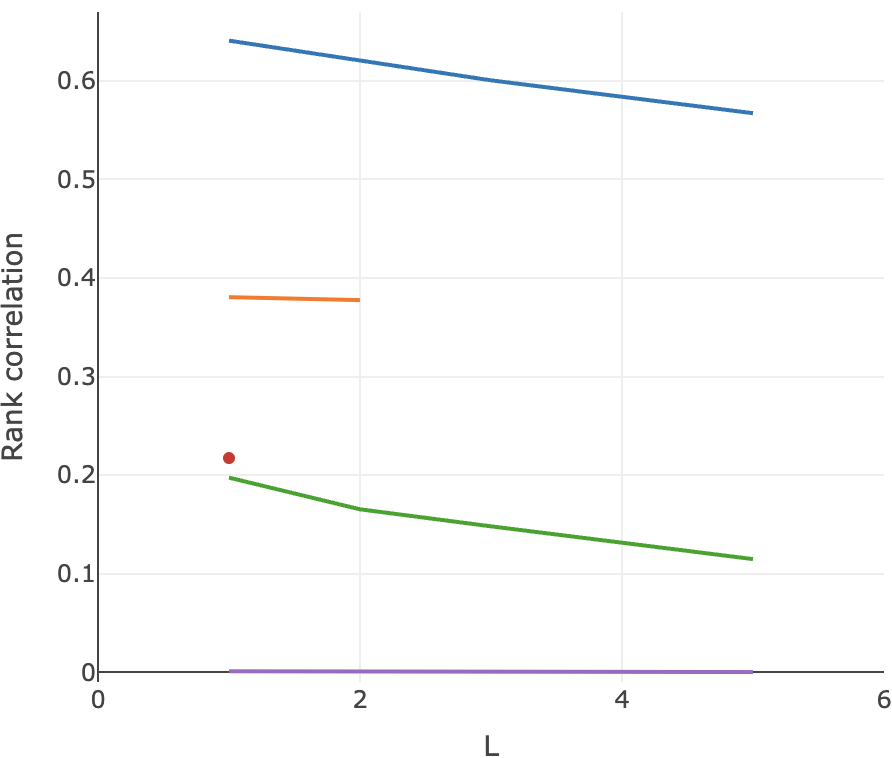}}
\subfigure{\includegraphics[height=5cm]{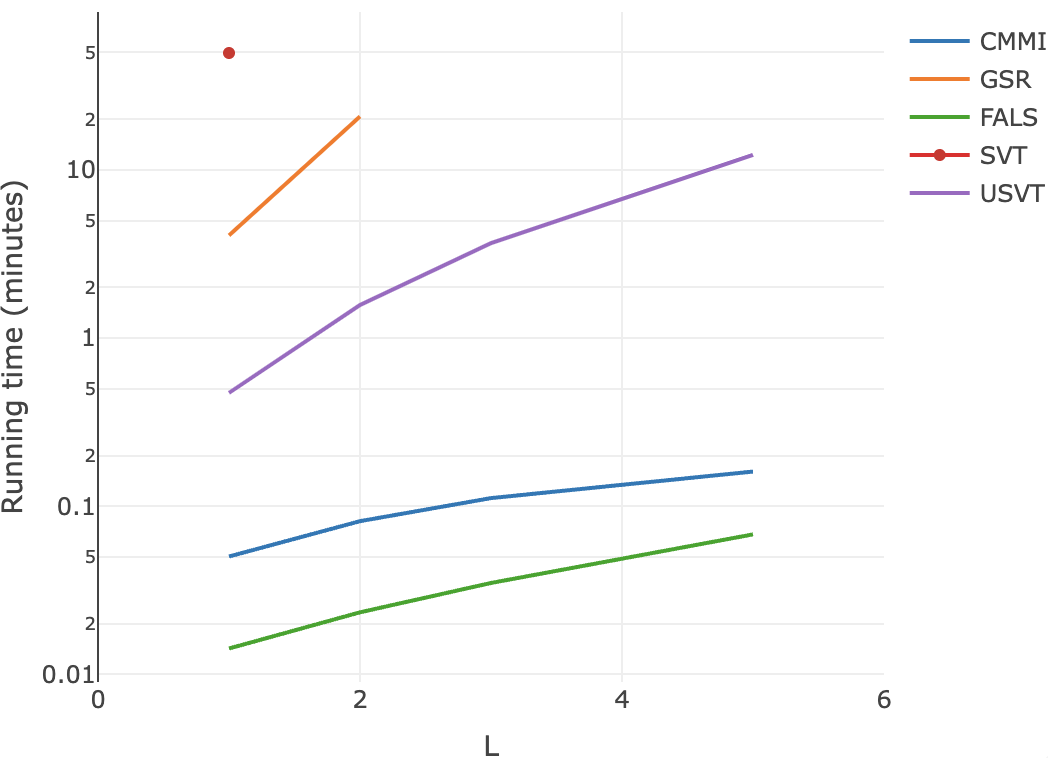}}
\caption{%Comparison with other methods $L=\{1,2,3,4,7,9\}$.
\footnotesize The left panel reports empirical estimates of Spearman's rank correlations for CMMI and other low-rank matrix completion algorithms as $L$ changes.
In particular, we vary $L=\{1,2,3,5\}$ while fixing $n=1000$ and $\breve p=0.1.$
The results are averaged over $100$ independent Monte Carlo replicates.
The average running time (in log scale) over 100 replicates for algorithms, using 25-core parallel computing and 256 GB memory, is shown in the right panel. % The machine is equipped with 256 GB of memory and two 64-core, 2.25 GHz, 225-watt AMD EPYC 7742 processors.
}
\label{fig:realdata_MRCOC_all}
\end{figure}

\textbf{Additional experimental results for Section~\ref{sec:MNIST}:}

\begin{figure}[h!] 
\centering
\subfigure{\includegraphics[height=5cm]{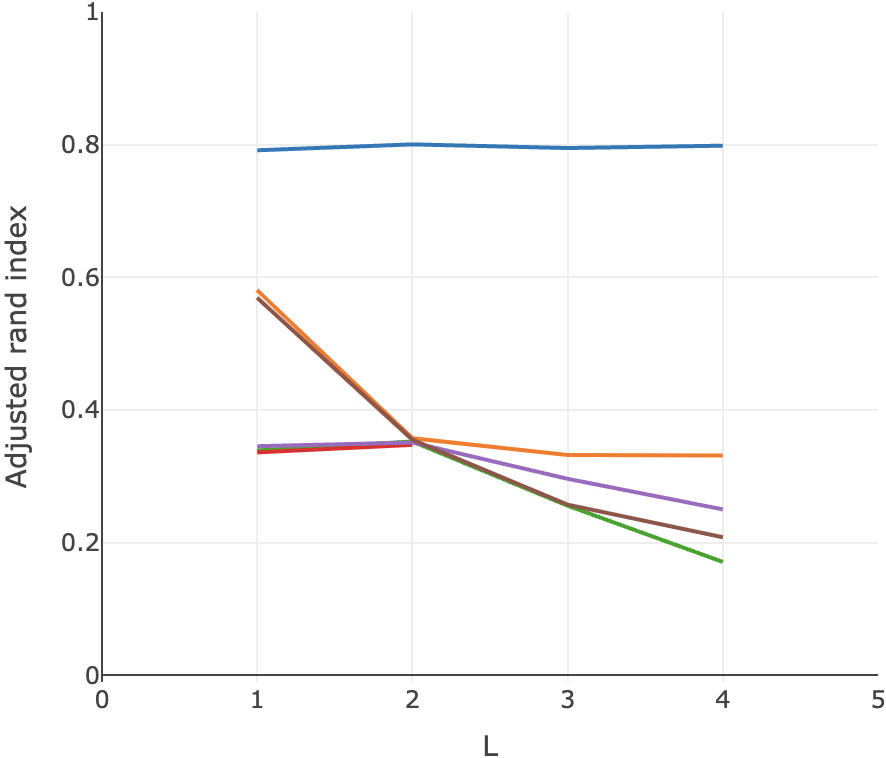}}
\subfigure{\includegraphics[height=5cm]{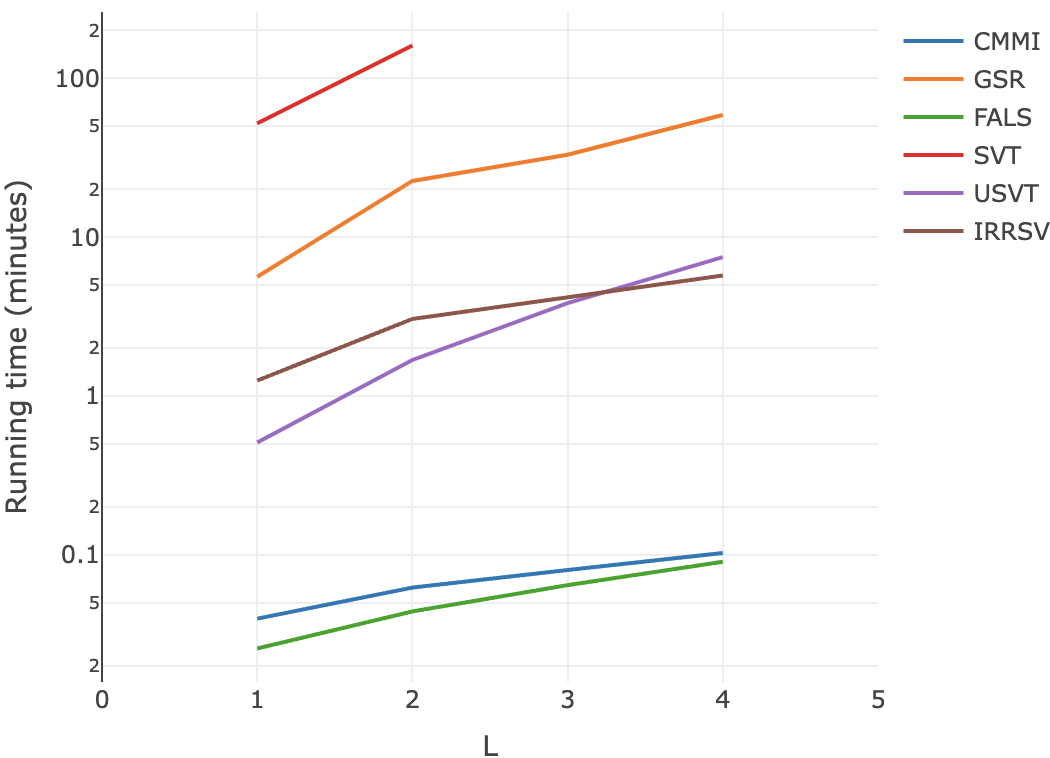}}
\caption{%Comparison with other methods $L=\{1,2,3,4,7,9\}$.
\footnotesize{The left panel reports ARIs for clustering $\hat{\mathbf{X}} \in \mathbb{R}^{N \times d}$, which consists of the top $d = 36$ scaled eigenvectors of the recovered full matrix $\hat{\mathbf{P}}$, obtained using CMMI and other low-rank matrix completion algorithms.
In particular, we vary $L \in \{1,2,3,4\}$ while fixing $n=1000$, $\breve p=0.1$, $q=0.8$.
%The black dotted line is the ARIs obtained by clustering directly on the original data (i.e, observed without noise nor missing entries). 
The results are averaged over $100$ independent Monte Carlo replicates.
The average running time (in log scale) over 100 replicates for algorithms, using 25-core parallel computing and 256 GB memory, is shown in the right panel. % The machine is equipped with 256 GB of memory and two 64-core, 2.25 GHz, 225-watt AMD EPYC 7742 processors.
}
}
\label{fig:realdata_MNIST_all}
\end{figure}

Figures~\ref{fig:simulation_comparison_all},
\ref{fig:realdata_MRCOC_all}, and \ref{fig:realdata_MNIST_all} show
that CMMI consistently outperforms other low-rank matrix completion
algorithms in recovering the entire matrix, and its computational time
remains highly competitive.
}

\setcounter{figure}{0}
\setcounter{algorithm}{0}
\setcounter{theorem}{0}
\newpage

\section{Algorithms and Simulation Results for Section~\ref{sec:asy}}

\subsection{Symmetric indefinite matrices integration}
\label{supp:indef}
\begin{algorithm}[htp]
\caption{CMMI for overlapping submatrices of a symmetric indefinite matrix}	
\label{Alg_chain2}
\begin{algorithmic}
\REQUIRE Embedding dimensions $d_+$ and $d_-$ for positive and negative eigenvalues, respectively; a chain of overlapping submatrices $\ma^{(0)},\ma^{(i_1)},\dots,\ma^{(i_L)}$ for $\mathcal{U}_{i_0},\mathcal{U}_{i_1},\dots,\mathcal{U}_{i_L}$ with $\min\{|\mathcal{U}_{i_0}\cap \mathcal{U}_{i_1}|,|\mathcal{U}_{i_1}\cap \mathcal{U}_{i_2}|,\dots,|\mathcal{U}_{i_{L-1}}\cap\mathcal{U}_{i_L}|\}\geq d$.
\begin{enumerate}
	\item For each $0\leq \ell \leq L$, obtain the estimated latent positions $\hat\mx^{(i_\ell)}=\hat\muu^{(i_\ell)}|\hat\mLambda^{(i_\ell)}|^{1/2}$. %for $\mathcal{U}_{i_\ell}$: 
%\begin{enumerate}	
%\item Compute eigen-decomposition of $\ma^{(i_\ell)}\in\mathbb{R}^{n_{i_\ell}\times n_{i_\ell}}$ where $n_{i_\ell}=|\mathcal{U}_{i_\ell}|$. 
%\item Order the eigenvalues $\lambda_1(\ma^{(i_\ell)})\geq\dots\geq\lambda_{n_{i_\ell}}(\ma^{(i_\ell)})$.
%\item Let $\hat\mLambda^{(i_\ell)}_+=\diag(\lambda_k(\ma^{(i_\ell)}) \colon k \leq d_+)$, $\hat\mLambda^{(i_\ell)}_-=\diag(\lambda_{k}(\ma^{(i_\ell)}) \colon k \geq n_{i_{\ell}} - d_- + 1)$.
%\item Let $\hat\muu_+^{(i_\ell)}$ and $\hat\muu_-^{(i_\ell)}$ be eigenvectors 
%corresponding to $\hat\mLambda^{(i_\ell)}_+$ and $\hat\mLambda^{(i_\ell)}_-$.% respectively.
%%\item Set $\hat\muu^{(i_\ell)}=[\hat\muu_+^{(i_\ell)}|\hat\muu_-^{(i_\ell)}]$  
%and $\hat\mLambda^{(i_\ell)}:=\diag(\hat\mLambda^{(i_\ell)}_+,\hat\mLambda^{(i_\ell)}_-)$.
%\end{enumerate}
	\item For each $1\leq \ell \leq L$, obtain $\mw^{(i_{\ell-1},i_\ell)}$ by solving the least square optimization problem
$$
\begin{aligned}
\mw^{(i_{\ell-1},i_\ell)}
&=\underset{\mo\in \mathbb{R}^{d\times d}}{\operatorname{argmin}} \|\hat\mx^{(i_{\ell-1})}_{\langle\mathcal{U}_{i_{\ell-1}}\cap \mathcal{U}_{i_\ell}\rangle}\mo-\hat\mx^{(i_\ell)}_{\langle\mathcal{U}_{i_{\ell-1}}\cap \mathcal{U}_{i_\ell}\rangle}\|_F.%\\
%&=((\hat\mx^{(i_{\ell-1})}_{\langle\mathcal{U}_{i_{\ell-1}}\cap \mathcal{U}_{i_\ell}\rangle})^\top
%\hat\mx^{(i_{\ell-1})}_{\langle\mathcal{U}_{i_{\ell-1}}\cap \mathcal{U}_{i_\ell}\rangle})^{-1}	
%(\hat\mx^{(i_{\ell-1})}_{\langle\mathcal{U}_{i_{\ell-1}}\cap \mathcal{U}_{i_\ell}\rangle})^\top
%\hat\mx^{(i_\ell)}_{\langle\mathcal{U}_{i_{\ell-1}}\cap \mathcal{U}_{i_\ell}\rangle}.
\end{aligned}
$$
	\item Compute $\hat\mpp_{\mathcal{U}_{i_0},\mathcal{U}_{i_L}}=\hat\mx^{(i_0)}\mw^{(i_0,i_1)}\mw^{(i_1,i_2)}\cdots \mw^{(i_L-1,i_L)}\mi_{d_+,d_-}\hat\mx^{(i_L)\top}$.
\end{enumerate} 
\ENSURE $\hat\mpp_{\mathcal{U}_{i_0},\mathcal{U}_{i_L}}$.
\end{algorithmic}
\end{algorithm}

Algorithm~\ref{Alg_chain2} presents a procedure for integration of
symmetric but possibly indefinite matrices.
We now state an extension of Theorem~\ref{thm:R(i0,...,iL)} to this setting. The main difference in this extension is the upper bound for
$\ms^{(i_0,i_1,\dots,i_L)}$ and this is due to the fact that the least square transformations $(\hat\mx^{(i_{\ell-1})}_{ \langle\mathcal{U}_{i_{\ell-1}}\cap \mathcal{U}_{i_{\ell}}\rangle})^{\dagger}\hat{\mx}^{(i_{\ell})}_{ \langle\mathcal{U}_{i_{\ell-1}}\cap \mathcal{U}_{i_{\ell}}\rangle}$ for $1 \leq \ell \leq L$
have spectral norms that can be smaller or larger than $1$,
and the accumulated error induced by these transformations need not grow linearly with $L$. If $L = 1$ then the bounds in Theorem~\ref{thm:R(il)_ind} are almost identical to those in 
Theorem~\ref{thm:R(i,j)}, but with a slightly different definition for $\alpha_{i,j}$. 
\begin{theorem}
{\em 
\label{thm:R(il)_ind}
    Consider a chain of overlapping submatrices $(\ma^{(i_0)}, \dots, \ma^{(i_L)})$ where, for each $0 \leq \ell \leq L$, 
    $\mpp^{(i_{\ell})}$ has $d_+$ positive eigenvalues and $d_{-}$ negative eigenvalues, satisfying Assumption~\ref{ass:main}.
    Set $d = d_{+} + d_{-}$. Here we define $\mu_i:=\lambda_{\min}(\mx^\top_{\mathcal{U}_i,\mathcal{U}_i}\mx_{\mathcal{U}_i,\mathcal{U}_i})/n_i$, $\gamma_i:=(\|\mpp_{\mathcal{U}_i,\mathcal{U}_i}\|_{\max}+\sigma_i)\log^{1/2}n_i$ for any $i$, and suppose $\frac{\gamma_i}{(q_in_i\mu_i)^{1/2}}\lesssim \|\mx^{(i)}\|_{2\to\infty}$ 
for $i \in \{i_0,i_L\}$.
    For all overlaps $1\leq \ell \leq L$, suppose $\mathrm{rk}(\mpp^{(i_{\ell})}_{\mathcal{U}_{i_{\ell-1}} \cap \mathcal{U}_{i_{\ell}},\mathcal{U}_{i_{\ell-1}} \cap \mathcal{U}_{i_{\ell}}}) = d$, and define
    $$
    \begin{aligned}
    \vartheta_{i_{\ell-1},i_\ell} &:= \lambda_{\max}(\mx^{(i_\ell)\top}_{\langle\mathcal{U}_{i_{\ell-1}} \cap \mathcal{U}_{i_{\ell}}\rangle}\mx^{(i_\ell)}_{\langle\mathcal{U}_{i_{\ell-1}} \cap \mathcal{U}_{i_{\ell}}\rangle}), \quad \theta_{i_{\ell-1},i_{\ell}} := \lambda_{\min}(\mx^{(i_{\ell-1})\top}_{\langle\mathcal{U}_{i_{\ell-1}} \cap \mathcal{U}_{i_{\ell}}\rangle} \mx^{(i_{\ell-1})}_{\langle\mathcal{U}_{i_{\ell-1}} \cap \mathcal{U}_{i_{\ell}}\rangle}),\\
    	\alpha_{i_{\ell-1},i_{\ell}}:=&\frac{n_{i_{\ell-1},i_{\ell}} \gamma_{i_{\ell-1}}  \gamma_{i_{\ell}} }{\theta_{{i_{\ell-1}},{i_{\ell}}}(q_{i_{\ell-1}}n_{i_{\ell-1}}  \mu_{i_{\ell-1}})^{1/2} (q_{i_{\ell}}n_{i_{\ell}}  \mu_{i_{\ell}})^{1/2} }
    +\frac{n_{{i_{\ell-1}},{i_{\ell}}}^{1/2}}{\theta_{{i_{\ell-1}},{i_{\ell}}}^{1/2}}\Bigl(\frac{\gamma_{i_{\ell-1}}^2}{q_{i_{\ell-1}} n_{i_{\ell-1}} \mu_{i_{\ell-1}}^{3/2}}+\frac{\vartheta_{{i_{\ell-1}},{i_{\ell}}}^{1/2}\gamma_{i_{\ell}}^2}{\theta_{{i_{\ell-1}},{i_{\ell}}}^{1/2}q_{i_{\ell}}n_{i_{\ell}}  \mu_{i_{\ell}}^{3/2}}\Bigr)
   \\&+\frac{n_{{i_{\ell-1}},{i_{\ell}}}\vartheta_{{i_{\ell-1}},{i_{\ell}}}^{1/2}\gamma_{i_{\ell-1}}^2}{\theta_{{i_{\ell-1}},{i_{\ell}}}^{3/2}q_{i_{\ell-1}}n_{i_{\ell-1}}  \mu_{i_{\ell-1}}}+\frac{n_{{i_{\ell-1}},{i_{\ell}}}^{1/2}\|\mx^{({i_{\ell-1}})}_{ \langle\mathcal{U}_{i_{\ell-1}}\cap\mathcal{U}_{i_{\ell}}\rangle}\|_{2\to\infty}}{\theta_{{i_{\ell-1}},{i_{\ell}}}}\Big(\frac{\gamma_{i_{\ell-1}}}{(q_{i_{\ell-1}}n_{i_{\ell-1}}\mu_{i_{\ell-1}})^{1/2}}+\frac{\vartheta_{{i_{\ell-1}},{i_{\ell}}}^{1/2}\gamma_{i_{\ell}}}{\theta_{{i_{\ell-1}},{i_{\ell}}}^{1/2}(q_{i_{\ell}}n_{i_{\ell}}\mu_{i_{\ell}})^{1/2}}\Big).
    \end{aligned}
    $$
    Suppose  $\frac{n_{{i_{\ell-1}},{i_{\ell}}}^{1/2}\gamma_{i_{\ell-1}}}{(q_{i_{\ell-1}}n_{i_{\ell-1}}  \mu_{i_{\ell-1}})^{1/2}} \ll \theta_{{i_{\ell-1}},{i_{\ell}}}^{1/2}$ for all $1\leq \ell\leq L$.
%    Recall that $\me^{(i)}=\ma^{(i)}-\mpp^{(i)}$ for any $i$.
    We then have 
    \[\hat{\mpp}_{\mathcal{U}_{i_0}, \mathcal{U}_{i_{L}}} - \mpp_{\mathcal{U}_{i_0}, \mathcal{U}_{i_{L}}} = 
      \me^{(i_0)} \mx_{\mathcal{U}_{i_0}} (\mx_{\mathcal{U}_{i_0}}^{\top} \mx_{\mathcal{U}_{i_0}})^{-1} \mx_{\mathcal{U}_{i_L}}^{\top} + \mx_{\mathcal{U}_{i_0}} (\mx_{\mathcal{U}_{i_L}}^{\top}\mx_{\mathcal{U}_{i_L}})^{-1} \mx_{\mathcal{U}_{i_L}}^{\top}
      \me^{(i_L)} + \mr^{(i_0,i_L)} + \ms^{(i_0,i_1, \dots, i_L)}, \]
   where $\mr^{(i_0,i_L)}$ and $\ms^{(i_0,i_1,\dots,i_L)}$ are random matrices satisfying
   \[ \begin{split} \|\mr^{(i_0,i_L)}\|_{\max}  \lesssim  &\Big(\frac{\gamma_{i_0}^2}{q_{i_0} n_{i_0}  \mu_{i_0}^{3/2}}
	+\frac{\gamma_{i_0} }{q_{i_0}^{1/2}n_{i_0} \mu_{i_0}^{1/2}}\Big)
	\|\mx^{(i_L)}\|_{2\to\infty} +
	\Big(\frac{\gamma_{i_L}^2}{q_{i_L} n_{i_L}  \mu_{i_L}^{3/2}}
	+\frac{\gamma_{i_L} }{q_{i_L}^{1/2}n_{i_L} \mu_{i_L}^{1/2}}\Big) 
	\|\mx^{(i_0)}\|_{2\to\infty} \\ &+ \frac{\gamma_{i_0} \gamma_{i_L}}{(q_{i_0} n_{i_0} \mu_{i_0})^{1/2}(q_{i_L} n_{i_L} \mu_{i_L})^{1/2}},
 \end{split}
 \]
 \[ \|\ms^{(i_0,i_1,\dots,i_L)}\|_{\max} \lesssim a_{L} \|\mx^{(i_0)}\|_{2 \to \infty} \cdot \|\mx^{(i_L)}\|_{2 \to \infty} \]
 with high probability. Here $a_{L}$ is a quantity defined recursively by $a_{1} = \alpha_{i_0,i_1}$ and
 \[  \quad a_{\ell} = a_{\ell-1} \cdot \Bigl(\alpha_{i_{\ell-1},i_{\ell}} + \Bigl[\frac{\vartheta_{i_{\ell-1},i_{\ell}}}{\theta_{i_{\ell-1},i_{\ell}}}\Bigr]^{1/2} \Bigr) + \Big\|\prod_{k=1}^{\ell-1} \bigl(\mx_{\langle\mathcal{U}_{i_{k-1}}\cap \mathcal{U}_{i_k}\rangle}^{(i_{k-1})}\bigr)^{\dagger} \mx_{\langle\mathcal{U}_{i_{k-1}}\cap \mathcal{U}_{i_k}\rangle}^{(i_{k})} \Bigr\| \cdot \alpha_{i_{\ell-1},i_{\ell}} \]
 for $2 \leq \ell \leq L$. }
\end{theorem}
We compare the performance of Algorithm~\ref{Alg_chain2} with other matrix completion methods. Consider the setting of Section~\ref{sec:comp}, but
with $\mLambda=\diag(N,\tfrac{1}{2} N, -\tfrac{1}{2} N,-N)$, and thus $d_+=d_-=2$.
Figure~\ref{fig:simulation_comparison2} shows the relative $F$-norm estimation error results for CMMI against other matrix completion
algorithms. 
  \begin{figure}[htbp] 
\centering
\subfigure{\includegraphics[height=5cm]{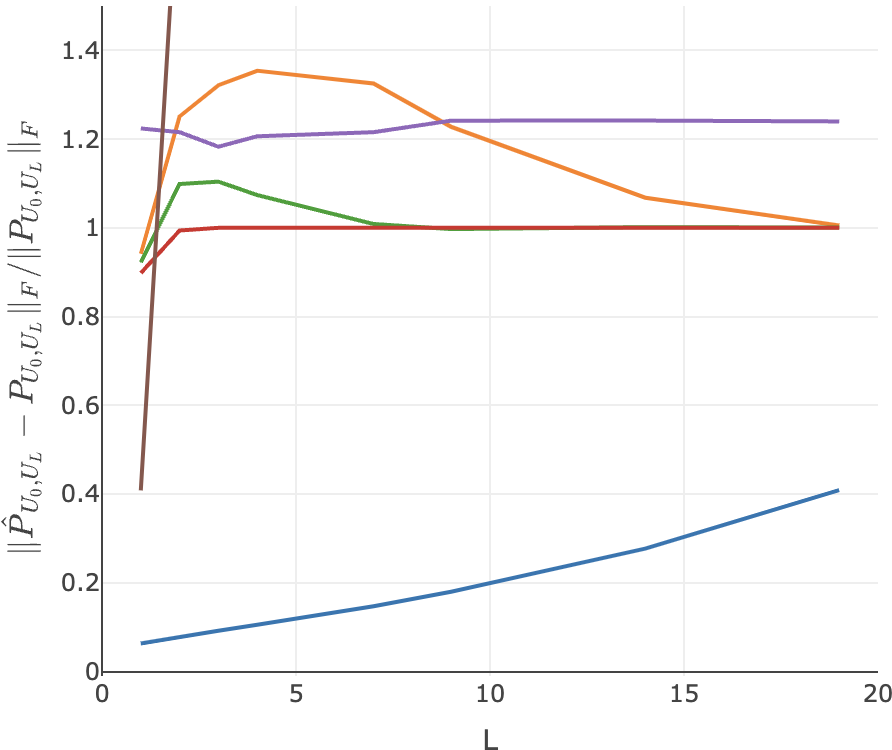}}
\subfigure{\includegraphics[height=5cm]{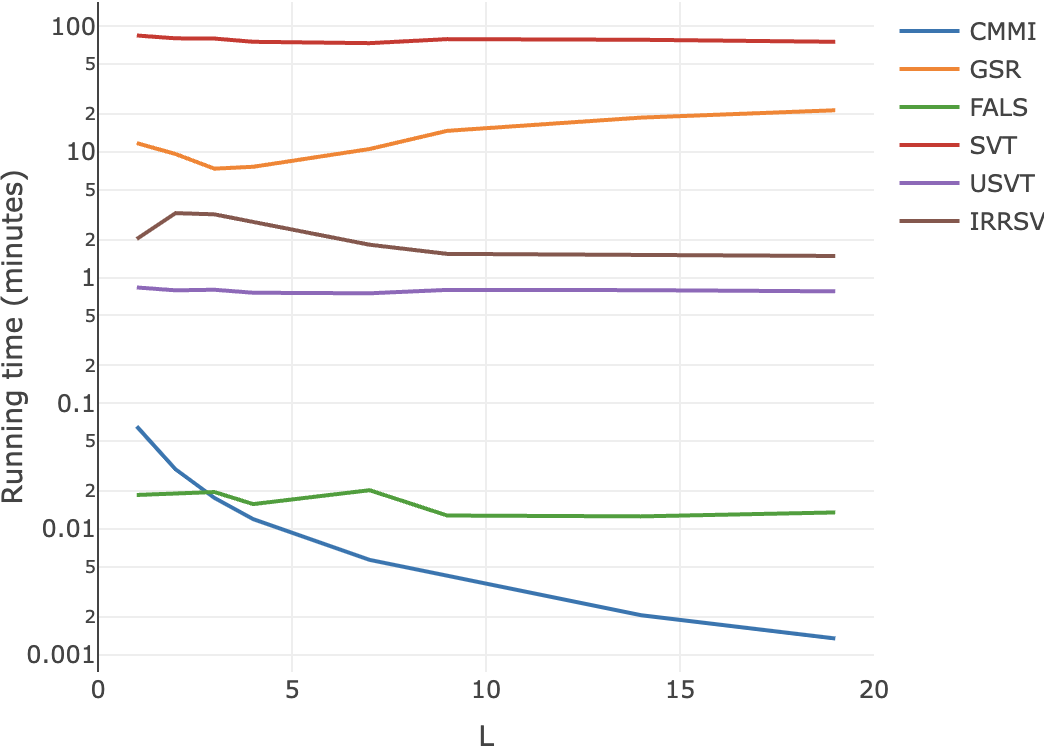}}
%\subfigure[]{\includegraphics[width=4cm]{image/simulation_disPCA/n_2.png}}
\caption{\footnotesize %Comparison with other algorithms when we tune $L$.
The left panel reports %empirical estimates of 
relative $F$-norm errors $\|\hat\mpp_{\mathcal{U}_{0},\mathcal{U}_{L}}-\mpp_{\mathcal{U}_{0},\mathcal{U}_{L}}\|_F/\|\mpp_{\mathcal{U}_{0},\mathcal{U}_{L}}\|_F$ for CMMI and other low-rank matrix completion algorithms as $L$ changes.
In particular, we vary $L=\{1,2,3,4,7,9,14,19\}$ while fixing $N\approx 2200$, $\breve p=0.1$, $q=0.8$, $\sigma=0.5$.
The results are averaged over $100$ independent Monte Carlo replicates.
Note that the averaged relative $F$-norm errors of IRRSV are $\{0.4,   1.9,   2.6,   3.8,  11.0,  70.4, 469.6, 305.0\}$ and some of these values are too large to be displayed in this panel.
The average running time (in log scale) over 100 replicates for algorithms, using 20-core parallel computing and 256 GB memory, is shown in the right panel. % The machine is equipped with 256 GB of memory and two 64-core, 2.25 GHz, 225-watt AMD EPYC 7742 processors.
}
\label{fig:simulation_comparison2}
\end{figure}

\subsection{Asymmetric matrices integration}
\label{supp:asy}

We provide detailed derivations of our idea for asymmetric matrices here.
We first consider the noiseless case to illustrate the idea.
 Let $\mpp^{(1)}$ and $\mpp^{(2)}$ be two overlapping submatrices shown in the left panel of Figure~\ref{fig:K=2a} without noise or missing entries. Suppose $\mathrm{rk}(\mpp^{(1)}) = \mathrm{rk}(\mpp^{(2)}) = 
\mathrm{rk}(\mx_{\mathcal{U}_1\cap \mathcal{U}_2}) = \mathrm{rk}(\my_{\mathcal{V}_1\cap \mathcal{V}_2}) = d$.
Now 
$$
\begin{aligned}
	\mx_{\mathcal{U}_1}\my_{\mathcal{V}_1}^\top=\mpp_{\mathcal{U}_1,\mathcal{V}_1}=\mpp^{(1)}=\mx^{(1)}\my^{(1)\top},
	\quad \mx_{\mathcal{U}_2} \my_{\mathcal{V}_2}^\top=\mpp_{\mathcal{U}_2,\mathcal{V}_2}=\mpp^{(2)}=\mx^{(2)}\my^{(2)\top}.
\end{aligned}
$$
%and %by our assumption we know $\operatorname{rank}(\mpp_{\mathcal{U}_1,\mathcal{V}_1})$ (resp. $\operatorname{rank}(\mpp_{\mathcal{U}_2,\mathcal{V}_2})$) achieves the number of columns of $\mx_{\mathcal{U}_1}$ and $\mx^{(1)}$ (resp. $\my_{\mathcal{V}_2}$ and $\my^{(2)}$). 
%hence, by simple algebraic analysis, there exist $d \times d$ matrices $\mw^{(1)},\mw^{(2)}$ such that 
%\begin{equation}\label{eq:X3}
%\begin{aligned}
	%	&\mx_{\mathcal{U}_1}=\mx^{(1)}\mw^{(1)},
	%	\quad \my_{\mathcal{V}_1}=\my^{(1)}(\mw^{(1)\top})^{-1},\\
		%&\mx_{\mathcal{U}_2}=\mx^{(2)}\mw^{(2)},
		%\quad \my_{\mathcal{V}_2}=\my^{(2)}(\mw^{(2)\top})^{-1}.
%\end{aligned}
%\end{equation}
Then there exist $d \times d$ matrices $\mw^{(1)}$ and $\mw^{(2)}$ such that
$$\begin{aligned}
		\mx_{\mathcal{U}_1}=\mx^{(1)}\mw^{(1)},
		\quad \my_{\mathcal{V}_1}=\my^{(1)}(\mw^{(1)\top})^{-1},
		\quad\mx_{\mathcal{U}_2}=\mx^{(2)}\mw^{(2)},
		\quad \my_{\mathcal{V}_2}=\my^{(2)}(\mw^{(2)\top})^{-1}.
\end{aligned}
$$
Suppose we want to recover the unobserved yellow submatrix in the left panel of Figure~\ref{fig:K=2a} as part of $\mpp_{\mathcal{U}_1,\mathcal{V}_2}=
 \mx_{\mathcal{U}_1}\my_{\mathcal{V}_2}^\top
	=\mx^{(1)}\mw^{(1)}(\mw^{(2)})^{-1}\my^{(2)\top}=\mx^{(1)}\mw^{(1,2)}\my^{(2)\top}
	%=\mx^{(1)}\mw^{(1,2)}\my^{(2)\top},
$ where $\mw^{(1,2)}:= \mw^{(1)}(\mw^{(2)})^{-1}$, %\end{aligned}
%\end{equation*}
%where $\mw^{(1,2)}:=\mw^{(1)}(\mw^{(2)})^{-1}$. %, and notice that because $\mw^{(1)},\mw^{(2)}$ are unique, $\mw^{(1,2)}$ is also unique.
 and thus our problem reduces to that of recovering $\mw^{(1,2)}$. %, first note that %$(\mx^{(1)},\my^{(1)})$ and $(\mx^{(2)},\my^{(2)})$ we also have
% $$
%\begin{aligned}
%	\mx^{(1)}_{ \langle\mathcal{U}_1\cap \mathcal{U}_2\rangle}\mw^{(1)}=\mx_{\mathcal{U}_1\cap \mathcal{U}_2}=\mx^{(2)}_{ \langle\mathcal{U}_1\cap \mathcal{U}_2\rangle}\mw^{(2)},
%	\quad\my^{(1)}_{\langle\mathcal{V}_1\cap \mathcal{V}_2\rangle}(\mw^{(1)\top})^{-1}=\my_{\mathcal{V}_1\cap \mathcal{V}_2}=\my^{(2)}_{\langle\mathcal{V}_1\cap \mathcal{V}_2\rangle}(\mw^{(2)\top})^{-1},
%\end{aligned}
%$$
%Once again, b
By straightforward algebra, we have
%and hence %we can recover $\mw^{(1,2)}$ via
$$
\begin{aligned}
	&\mx^{(2)}_{ \langle\mathcal{U}_1\cap \mathcal{U}_2\rangle}
%=\mx^{(1)}_{ \langle\mathcal{U}_1\cap \mathcal{U}_2\rangle}\mw^{(1)}(\mw^{(2)})^{-1}
=\mx^{(1)}_{ \langle\mathcal{U}_1\cap \mathcal{U}_2\rangle}\mw^{(1,2)}, \quad %=\mx^{(1)}_{ \langle\mathcal{U}_1\cap \mathcal{U}_2\rangle}\mw^{(1,2)},\\
&\my^{(1)}_{\langle\mathcal{V}_1\cap \mathcal{V}_2\rangle}
%=\my^{(2)}_{\langle\mathcal{V}_1\cap \mathcal{V}_2\rangle}(\mw^{(2)\top})^{-1}\mw^{(1)\top}
=\my^{(2)}_{\langle\mathcal{V}_1\cap \mathcal{V}_2\rangle}\mw^{(1,2)\top}, 
%=\my^{(2)}_{ \langle\mathcal{U}_1\cap \mathcal{U}_2\rangle}\mw^{(1,2)\top}.
\end{aligned}
$$
and $\mw^{(1,2)}$ can be obtained by aligning the latent positions for the 
overlapping entities, i.e., 
$$
\mw^{(1,2)}=\underset{\mo\in \mathbb{R}^{d\times d}}{\operatorname{argmin}} \|\mx^{(1)}_{ \langle\mathcal{U}_1\cap \mathcal{U}_2\rangle}\mo-\mx^{(2)}_{ \langle\mathcal{U}_1\cap \mathcal{U}_2\rangle}\|_F
\text{ or }
\mw^{(1,2)\top}=\underset{\mo\in \mathbb{R}^{d\times d}}{\operatorname{argmin}} \|\my^{(2)}_{ \langle\mathcal{V}_1\cap \mathcal{V}_2\rangle}\mo-\my^{(1)}_{\langle\mathcal{V}_1\cap \mathcal{V}_2\rangle}\|_F.
$$

Now suppose $\ma^{(1)}$ and $\ma^{(2)}$ are noisy observations of $\mpp^{(1)}$ and $\mpp^{(2)}$ with possible missing entries. Let $(\hat\mx^{(1)},\hat\my^{(1)})$ and $(\hat\mx^{(2)},\hat\my^{(2)})$ be the estimated latent positions matrices obtained from $\ma^{(1)}$ and $\ma^{(2)}$. 
We can align these estimates by solving the least squares problems
$$
\begin{aligned}
	\mw_\mx^{(1,2)}
	=\underset{\mo\in \mathbb{R}^{d\times d}}{\operatorname{argmin}} \|\hat\mx^{(1)}_{ \langle\mathcal{U}_1\cap \mathcal{U}_2\rangle}\mo-\hat\mx^{(2)}_{ \langle\mathcal{U}_1\cap \mathcal{U}_2\rangle}\|_F, \quad
\mw_\my^{(1,2)\top}=\underset{\mo\in \mathbb{R}^{d\times d}}{\operatorname{argmin}} \|\hat\my^{(2)}_{ \langle\mathcal{V}_1\cap \mathcal{V}_2\rangle}\mo-\hat\my^{(1)}_{\langle\mathcal{V}_1\cap \mathcal{V}_2\rangle}\|_F,
\end{aligned}
$$
and setting
$
\mw^{(1,2)}=\frac{1}{2}(\mw_\mx^{(1,2)}+\mw_\my^{(1,2)})
$.
We then extend this idea to a chain of overlapping submatrices and have Algorithm~\ref{Alg_chain3}.

\begin{algorithm}[htbp]
\caption{CMMI for overlapping submatrices of an asymmetric matrix}	
\label{Alg_chain3}
\begin{algorithmic}
\small %\footnotesize
\REQUIRE Embedding dimension $d$, a chain of overlapping submatrices $\ma^{(i_0)},\ma^{(i_1)},\dots,\ma^{(i_L)}$ for $(\mathcal{U}_{i_0},\mathcal{V}_{i_0}),(\mathcal{U}_{i_1},\mathcal{V}_{i_1}),\dots,(\mathcal{U}_{i_L},\mathcal{V}_{i_L})$; here for each $\ell\in[L]$, we have $|\mathcal{U}_{i_{\ell-1}}\cap \mathcal{U}_{i_{\ell}}|\geq d$ or $|\mathcal{V}_{i_{\ell-1}}\cap \mathcal{V}_{i_{\ell}}|\geq d$.
\begin{enumerate}
\item For each $0\leq\ell\leq L$, obtain estimated left latent positions for $\mathcal{U}_{i_\ell}$ as $\hat\mx^{(i_\ell)}=\hat\muu^{(i_\ell)}(\hat\mSigma^{(i_\ell)})^{1/2}$ and right latent positions for $\mathcal{V}_{i_\ell}$ as $\hat\my^{(i_\ell)}=\hat\mv^{(i_\ell)}(\hat\mSigma^{(i_\ell)})^{1/2}$.
%where the diagonal matrix $\hat\mSigma^{(i_\ell)}\in\mathbb{R}^{d\times d}$ contains the $d$ largest singular values of $\ma^{(i_\ell)}$, and $\hat\muu^{(i_\ell)}\in\mathbb{R}^{|\mathcal{U}_{i_\ell}|\times d}$ and $\hat\mv^{(i_\ell)}\in\mathbb{R}^{|\mathcal{V}_{i_\ell}|\times d}$ be the left and right singular vector matrix corresponding to $\hat\mSigma^{(i_\ell)}$, respectively.
    
	\item For each $1\leq \ell \leq L$, obtain $\mw^{(i_{\ell-1},i_\ell)}$:

\textbf{if} $|\mathcal{U}_{i_{\ell-1}}\cap \mathcal{U}_{i_{\ell}}|\geq d$ and $|\mathcal{V}_{i_{\ell-1}}\cap \mathcal{V}_{i_{\ell}}|\geq d$ \textbf{then}
    
    \quad Compute $\mw^{(i_{\ell-1},i_\ell)}=\frac{1}{2}(\mw_\mx^{(i_{\ell-1},i_\ell)}+\mw_\my^{(i_{\ell-1},i_\ell)})$ where
    $$
\begin{aligned}
\mw_\mx^{(i_{\ell-1},i_\ell)}
&=\underset{\mo\in \mathbb{R}^{d\times d}}{\operatorname{argmin}} \|\hat\mx^{(i_{\ell-1})}_{\langle\mathcal{U}_{i_{\ell-1}}\cap \mathcal{U}_{i_\ell}\rangle}\mo-\hat\mx^{(i_\ell)}_{\langle\mathcal{U}_{i_{\ell-1}}\cap \mathcal{U}_{i_\ell}\rangle}\|_F,\\
\mw_\my^{(i_{\ell-1},i_\ell)\top }
&=\underset{\mo\in \mathbb{R}^{d\times d}}{\operatorname{argmin}} \|\hat\my^{(i_\ell)}_{\langle\mathcal{V}_{i_{\ell-1}}\cap \mathcal{V}_{i_\ell}\rangle}\mo-\hat\my^{(i_{\ell-1})}_{\langle\mathcal{V}_{i_{\ell-1}}\cap \mathcal{V}_{i_\ell}\rangle}\|_F.
\end{aligned}
$$
\textbf{else if} $|\mathcal{U}_{i_{\ell-1}}\cap \mathcal{U}_{i_{\ell}}|\geq d$
 \textbf{then}
 
     \quad Compute
   $
\begin{aligned}
\mw^{(i_{\ell-1},i_\ell)}
&=\underset{\mo\in \mathbb{R}^{d\times d}}{\operatorname{argmin}} \|\hat\mx^{(i_{\ell-1})}_{\langle\mathcal{U}_{i_{\ell-1}}\cap \mathcal{U}_{i_\ell}\rangle}\mo-\hat\mx^{(i_\ell)}_{\langle\mathcal{U}_{i_{\ell-1}}\cap \mathcal{U}_{i_\ell}\rangle}\|_F.\end{aligned}
$

   \textbf{else} 
    
   \quad Compute $\mw^{(i_{\ell-1},i_\ell)}$ by
   $
\begin{aligned}
\mw^{(i_{\ell-1},i_\ell)\top}
&=\underset{\mo\in \mathbb{R}^{d\times d}}{\operatorname{argmin}} \|\hat\my^{(i_\ell)}_{\langle\mathcal{V}_{i_{\ell-1}}\cap \mathcal{V}_{i_\ell}\rangle}\mo-\hat\my^{(i_{\ell-1})}_{\langle\mathcal{V}_{i_{\ell-1}}\cap \mathcal{V}_{i_\ell}\rangle}\|_F.\end{aligned}
$
    
    \textbf{end if}

	\item Compute $\hat\mpp_{\mathcal{U}_{i_0},\mathcal{V}_{i_L}}=\hat\mx^{(i_0)}\mw^{(i_0,i_1)}\mw^{(i_1,i_2)}\cdots \mw^{(i_L-1,i_L)}\hat\my^{(i_L)\top}$.
\end{enumerate} 
\ENSURE $\hat\mpp_{\mathcal{U}_{i_0},\mathcal{V}_{i_L}}$.
\end{algorithmic}
\end{algorithm}

\newpage

 We compare the performance of Algorithm~\ref{Alg_chain3} with other matrix completion methods. %Consider a similar setting of Section~\ref{sec:comp} but for an asymmetric matrix as follows. 
 We simulate a chain of $(L+1)$ overlapping observed submatrices $\{\ma^{(i)}\}_{i=0}^L$ for the underlying population matrix $\mpp$ 
as described in Figure~\ref{fig:simulation_setting3}, and then predict the yellow unknown block by Algorithm~\ref{Alg_chain3}.
We let all observed submatrices have the same dimension $n\times m$, and let all overlapping parts have the same dimension $(\breve p_n n)\times (\breve p_m m)$.
For the observed submatrices, we generate the noise matrices $\{\mn^{(i)}\}$ by $\mn^{(i)}_{st}\stackrel{i i d}{\sim}\mathcal{N}(0,\sigma^2)$ for all $i=0,1,\dots,L$ and all $s\in[n],t\in[m]$, and we let all observed submatrices have the same non-missing probability $q$.
For the low-rank underlying population matrix $\mpp=\muu\mSigma\mv^\top$, we randomly generate $\muu$ and $\mv$ from $\{\mo\in\mathbb{R}^{N\times d}\mid \mo^\top\mo=\mi\}$ and $\{\mo\in\mathbb{R}^{M\times d}\mid \mo^\top\mo=\mi\}$, respectively. We fix the rank as $d=3$, and set $\mSigma=\text{diag}(N,\tfrac{3}{4}N,\tfrac{1}{2}N)$.
We fix the dimensions of the entire matrix at $N \approx 2200$ and $M \approx 2800$, and we vary $L$, the length of the chain, while ensuring that the observed submatrices fully span the diagonal of the matrix Recall that as $L$ increases, we have more observed submarices but each observed submatrix is of smaller dimensions, which then increases the difficulty of recovering the original matrix $\mpp$. 
Figure~\ref{fig:simulation_comparison3} shows the relative $F$-norm estimation error results of recovering the yellow region.

\begin{figure}[htp!] 
\centering
\subfigure{\includegraphics[width=4.5cm]{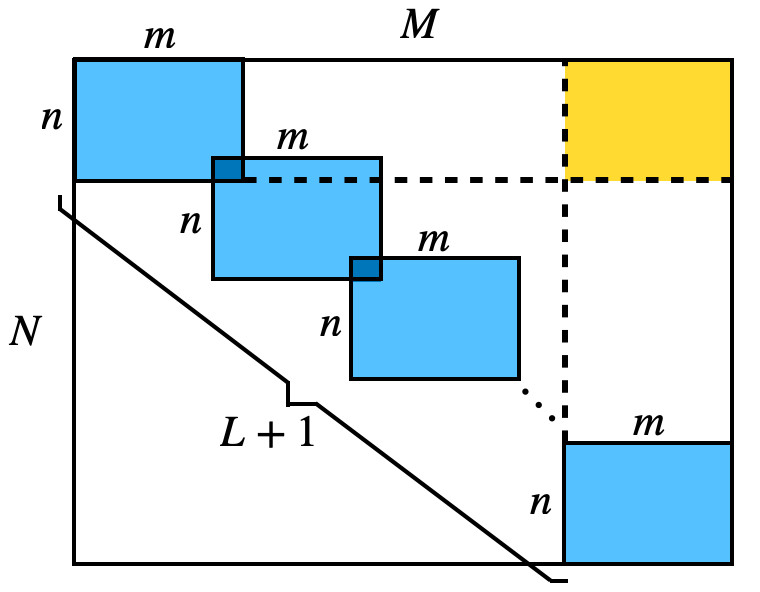}}
%\subfigure[]{\includegraphics[width=4cm]{image/simulation_disPCA/n_2.png}}
\caption{Simulation setting for an asymmetric matrix
}
\label{fig:simulation_setting3}
\end{figure}

\newpage

   \begin{figure}[htp!] 
\centering
\subfigure{\includegraphics[height=5cm]{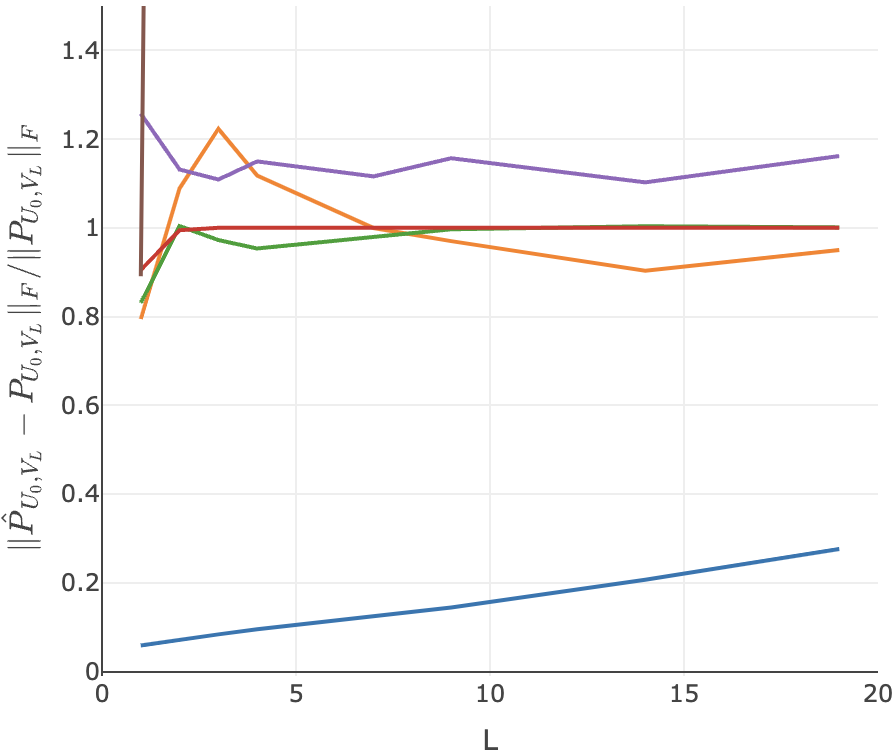}}
\subfigure{\includegraphics[height=5cm]{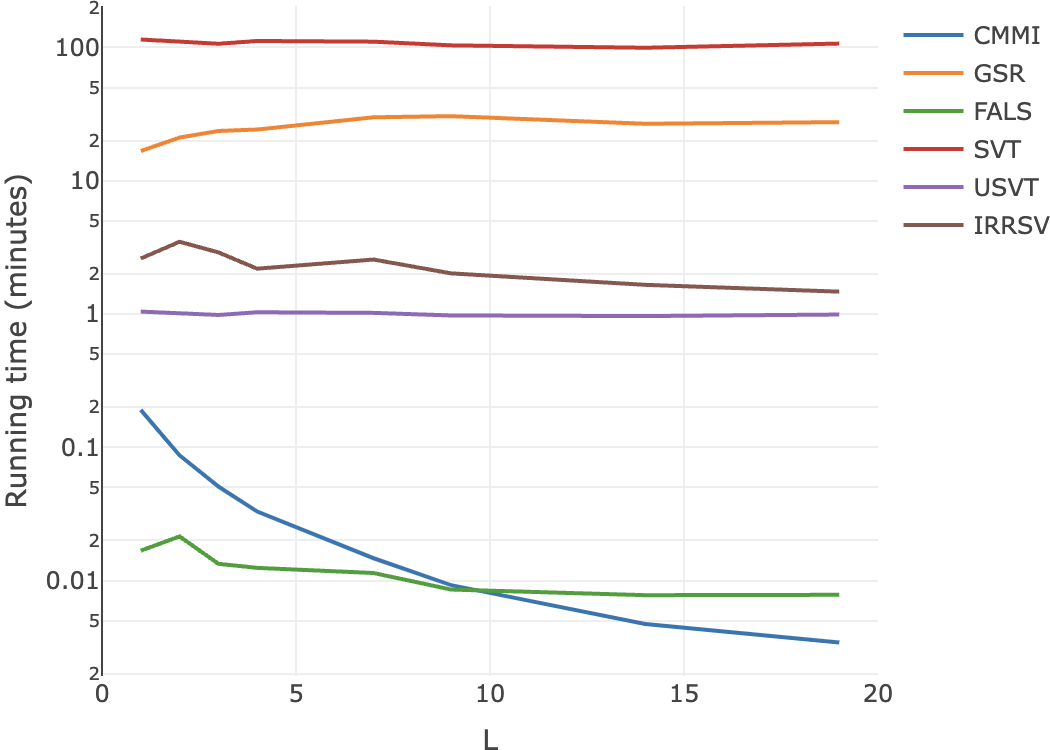}}
%\subfigure[]{\includegraphics[width=4cm]{image/simulation_disPCA/n_2.png}}
\caption{\footnotesize %Comparison with other algorithms when we tune $L$.
The left panel reports relative $F$-norm errors $\|\hat\mpp_{\mathcal{U}_{0},\mathcal{V}_{L}}-\mpp_{\mathcal{U}_{0},\mathcal{V}_{L}}\|_F/\|\mpp_{\mathcal{U}_{0},\mathcal{V}_{L}}\|_F$ for CMMI and other low-rank matrix completion algorithms as $L$ changes.
In particular, we vary $L \in \{1,2,3,4,7,9,14,19\}$ while fixing $N\approx 2200$, $M\approx 2800$, $\breve p_n=\breve p_m=0.1$, $q=0.8$, $\sigma=0.5$.
The results are averaged over $100$ independent Monte Carlo replicates.
Note that the averaged relative $F$-norm errors of IRRSV are $\{0.9,   9.3,  26.1,  25.0,  19.2,  34.1, 172.0, 216.8\}$ and some of these values are too large to be displayed in this panel.
The average running time (in log scale) over 100 replicates for algorithms, using 20-core parallel computing and 256 GB memory, is shown in the right panel. % The machine is equipped with 256 GB of memory and two 64-core, 2.25 GHz, 225-watt AMD EPYC 7742 processors.
}
\label{fig:simulation_comparison3}
\end{figure}

\setcounter{figure}{0}
\setcounter{algorithm}{0}
\setcounter{theorem}{0}
\newpage

\section{Proof of Main Results}

\subsection{Proof of Theorem~\ref{thm:R(i,j)}}\label{sec:proof thm1}

We first state two important technical lemmas, one for the error of $\hat{\mx}^{(i)}$ as an estimate of the true latent position matrix $\mx_{\mathcal{U}_i}$ for each $i \in [K]$, and another for the difference between $\mw^{(i,j)}$ and $\mw^{(i)} \mw^{(j)\top}$. The proofs of these lemmas are presented in Section~\ref{App: Lemma:hat X(i)W(i)-X} and Section~\ref{App: lemma:|| W(i)T W(i,j) W(j)-I ||}.
\begin{lemma}
\label{lemma:hat X(i)W(i)-X}
{\em 
Fix an $i\in[K]$ and consider $\ma^{(i)}=(\mpp^{(i)}+\mn^{(i)})\circ \mathbf{\Omega}^{(i)}/q_i\in\mathbb{R}^{n_i\times n_i}$ as defined in Eq.~\eqref{eq:A(i)=...}.
    % where $\mpp^{(i)}$ is of rank $d$, the entries of $\me^{(i)}$ are independent sub-Gaussian noise with mean zero and sub-Gaussian norm $\|\me^{(i)}_{s,t}\|_{\psi_2}$, $q_i=\|\mathbf{\Omega}^{(i)}\|_F^2/n_i^2$. %Let $\sigma_i=\max_{s,t\in[|\mathcal{U}_i|]}\|\me^{(i)}_{s,t}\|_{\psi_2}, \lambda_{i,\max}=\lambda_1(\mpp^{(i)}), \lambda_{i,\min}=\lambda_d(\mpp^{(i)})$.
    Write the eigen-decompositions of $\mpp^{(i)}$ and $\ma^{(i)}$ as
    \[ \mpp^{(i)}=\muu^{(i)}\mLambda^{(i)}\muu^{(i)\top}, \quad \ma^{(i)}=\hat\muu^{(i)}\hat\mLambda^{(i)}\hat\muu^{(i)\top} + \hat\muu^{(i)}_\perp\hat\mLambda^{(i)}_\perp\hat\muu^{(i)\top}_\perp.\]
    Let $\hat\mx^{(i)}=\hat\muu^{(i)}(\hat\mLambda^{(i)})^{1/2}$, and define $\mw^{(i)}=\underset{\mo\in \mathcal{O}_d}{\arg\min}\|\hat\mx^{(i)}\mo-\mx_{\mathcal{U}_i}\|_F$.
    Suppose that
    \begin{itemize}
      \item	$\muu^{(i)}$ is a $n_i\times d$ matrix with bounded coherence, i.e., % with high probability, i.e.,
      $$\|\muu^{(i)}\|_{2\to\infty}\lesssim d^{1/2}n_i^{-1/2}.$$
      %$$\mu_i=\frac{n_i}{d}\cdot\|\muu^{(i)}\|_{2\to\infty}^2 \lesssim 1.$$
      \item $\mpp^{(i)}$ has bounded condition number, i.e.,
      $$\frac{\lambda_{i,\max}}{\lambda_{i,\min}}\leq M$$
      for some finite constant $M > 0$; here $\lambda_{i,\max}$ and $\lambda_{i,\min}$ denote the largest and smallest non-zero eigenvalues of $\mpp^{(i)}$, respectively. %, then set $\tau_i=\lambda_{i,\max}/\lambda_{i,\min}$. 
      %$$ \tau_i \lesssim 1.$$
      %\item The entries of $\mn^{(i)}$ are independent sub-Gaussian noise with mean zero and sub-Gaussian norm $\|\mn^{(i)}_{s,t}\|_{\psi_2}$. Let $\sigma_i=\max_{s,t\in[|\mathcal{U}_i|]}\|\mn^{(i)}_{s,t}\|_{\psi_2}$. 
      \item %We set $n_i\asymp pN$. % and $q_i\asymp q$. 
      The following conditions are satisfied.
      \begin{equation}
          \label{eq:cond_lemmaA1}
          q_i n_i\gtrsim \log^2{n_i} ,\quad
          \frac{(\|\mpp^{(i)}\|_{\max}+\sigma_i) n_i^{1/2}}{q_i^{1/2}\lambda_{i,\min}}=\frac{\gamma_i}{(q_in_i)^{1/2}\mu_i \log^{1/2} n_i}\ll 1.
      \end{equation}
    \end{itemize}    
    We then have
    \begin{equation}\label{eq:hatXW-X}
    	\hat\mx^{(i)} \mw^{(i)}-\mx_{\mathcal{U}_i}
	= \me^{(i)}\mx_{\mathcal{U}_i}(\mx_{\mathcal{U}_i}^{\top}\mx_{\mathcal{U}_i})^{-1}
	+\mr^{(i)},
    \end{equation}
	where the remainder term $\mr^{(i)}$ satisfies
	$$
	\|\mr^{(i)}\|
	\lesssim \frac{(\|\mpp^{(i)}\|_{\max}+\sigma_i)^2 n_i}{q_i\lambda_{i,\min}^{3/2}}
		+\frac{(\|\mpp^{(i)}\|_{\max}+\sigma_i)\log^{1/2}n_i}{q_i^{1/2}\lambda_{i,\min}^{1/2}}
    $$
    with high probability.
	If we further assume 
 \begin{equation}
     \label{eq:cond2_lemmaA1}
     \frac{(\|\mpp^{(i)}\|_{\max}+\sigma_i) n_i^{1/2}\log^{1/2}n_i}{q_i^{1/2}\lambda_{i,\min}}=\frac{\gamma_i}{(q_i n_i)^{1/2}\mu_i}\ll 1,
 \end{equation}
 then we also have 
$$
\begin{aligned}
&\|\mr^{(i)}\|_{2\to\infty}
	\lesssim \frac{(\|\mpp^{(i)}\|_{\max}+\sigma_i)^2 n_i^{1/2}\log n_i}{q_i\lambda_{i,\min}^{3/2}}
		+\frac{ (\|\mpp^{(i)}\|_{\max}+\sigma_i)\log^{1/2} n_i}{q_i^{1/2}n_i^{1/2}\lambda_{i,\min}^{1/2}}
	=\frac{\gamma_i^2}{q_in_i \mu_{i}^{3/2}}
		+\frac{\gamma_i}{q_i^{1/2}n_i  \mu_i^{1/2}}, \\
        	&\|\hat\mx^{(i)} \mw^{(i)}-\mx_{\mathcal{U}_i}\|_{2\to\infty}
	    \lesssim \|\me^{(i)}\muu^{(i)}\|_{2\to\infty}
	    \cdot \|(\mLambda^{(i)})^{-1/2}\|
	    +\|\mr^{(i)}\|_{2\to\infty}
	    \lesssim \frac{(\|\mpp^{(i)}\|_{\max}+\sigma_i) \log^{1/2} n_i}{q_i^{1/2}\lambda_{i,\min}^{1/2}}=\frac{\gamma_i}{(q_in_i  \mu_i)^{1/2}} 
\end{aligned}
$$
    with high probability. }
\end{lemma}

\begin{rem}
    As $\mpp^{(i)} \in \mathbb{R}^{n_i \times n_i}$, we generally have $\|\mpp^{(i)}\|_{F} = \Theta(n_i)$, e.g.,  
    $\mpp^{(i)}$ has $\Theta(n_i^2)$ entries that are lower bounded by some constant $c_0$ not depending on $N$ or $n_i$. Thus, as $\mpp^{(i)}$ is low-rank with bounded condition number, we also have $\lambda_{i,\min} = \Theta(n_i)$ and the second condition in Eq.~\eqref{eq:cond_lemmaA1} simplifies to
    \[ \frac{\|\mpp^{(i)}\|_{\max}+\sigma_i}{(q_i n_i )^{1/2}} \ll 1.\]
    Similarly, the condition in Eq.~\eqref{eq:cond2_lemmaA1} simplifies to
    \[ \frac{(\|\mpp^{(i)}\|_{\max}+\sigma_i) \log^{1/2}{n_i}}{(q_i n_i )^{1/2}} \ll 1.\]
   Both conditions are then trivially satisfied whenever 
    $q_i  n_i \gg \log{n_i}$ and $\|\mpp^{(i)}\|_{\max} + \sigma_i = O(1)$.  
    Finally when  $\lambda_{i,\min} = \Theta(n_i)$, the bounds in Lemma~\ref{lemma:hat X(i)W(i)-X} simplify to
    \begin{gather*}
        \|\mr^{(i)}\|
	\lesssim \frac{(\|\mpp^{(i)}\|_{\max}+\sigma_i)^2}{q_in_i^{1/2} }
		+\frac{(\|\mpp^{(i)}\|_{\max}+\sigma_i)\log^{1/2}{n_i}}{(q_in_i )^{1/2}}, \\
  \|\mr^{(i)}\|_{2\to\infty}
	\lesssim \frac{(\|\mpp^{(i)}\|_{\max}+\sigma_i)^2 \log n_i}{q_in_i }
		+\frac{ (\|\mpp^{(i)}\|_{\max}+\sigma_i)\log^{1/2} n_i}{q_i^{1/2}n_i }, \\
  \|\hat\mx^{(i)} \mw^{(i)}-\mx_{\mathcal{U}_i}\|_{2\to\infty}
	     \lesssim \frac{(\|\mpp^{(i)}\|_{\max}+\sigma_i) \log^{1/2} n_i}{(q_in_i )^{1/2}} 
    \end{gather*}
    with high probability. 
\end{rem}

\begin{lemma}
\label{lemma:|| W(i)T W(i,j) W(j)-I ||}
{\em 
    Consider the setting of Theorem~\ref{thm:R(i,j)}. We then have
	$$
    \begin{aligned}
    	\|\mw^{(i)\top}\mw^{(i,j)}\mw^{(j)} -\mi\|
    	%\lesssim \frac{\|\mathbf{F}\|}{\lambda_{ij,\min}}
    % 	&\lesssim \frac{Nn_{i,j}\log N}{\theta_{i,j}\lambda_{\min}} \Big(\frac{(\|\mpp^{(i)}\|_{\max}+\sigma_i)^2}{q_in_i}
    % +\frac{(\|\mpp^{(j)}\|_{\max}+\sigma_j)^2}{q_jn_j}\Big)\\
    % &+\frac{n_{i,j}^{1/2}\log^{1/2}N}{\theta_{i,j}}\Big(\frac{\|\mpp^{(i)}\|_{\max}+\sigma_i}{q_i^{1/2}n_i^{1/2}}+\frac{\|\mpp^{(j)}\|_{\max}+\sigma_j}{q_j^{1/2}n_j^{1/2}}\Big)\\
    %     %&\lesssim \frac{(\|\mpp\|_{\max}+\sigma)^2N\log N}{q\lambda_{\min}^2}\cdot \Big(\frac{N}{n_i}+\frac{N}{n_j}\Big)+\frac{(\|\mpp\|_{\max}+\sigma)\log^{1/2}N}{q^{1/2}\lambda_{\min}}\Big(\frac{N}{n_i^{1/2}n_{i,j}^{1/2}}+\frac{N}{n_j^{1/2}n_{i,j}^{1/2}}
    % %\Big) \\
    % &\lesssim \frac{(\|\mpp\|_{\max}+\sigma)^2N  \breve p\log N}{q\theta_{i,j}\lambda_{\min}}
    % +\frac{(\|\mpp\|_{\max}+\sigma)\breve p^{1/2}\log^{1/2}N}{q^{1/2}\theta_{i,j}} \\
&\lesssim \frac{n_{i,j} \gamma_i  \gamma_j }{\theta_{i,j}(q_in_i  \mu_i)^{1/2} (q_jn_j  \mu_j)^{1/2} }
    +\frac{n_{i,j}^{1/2}\|\mx_{\mathcal{U}_i\cap \mathcal{U}_j}\|_{2\to\infty}}{\theta_{i,j}}\Big(\frac{\gamma_i}{(q_i n_i\mu_i)^{1/2}} +\frac{\gamma_j}{(q_j n_j\mu_j)^{1/2}}\Big)
\\
      & + \frac{(n_{i,j} \vartheta_{i,j})^{1/2}}{\theta_{i,j}}\Big(\frac{\gamma_i^2 }{q_in_i  \mu_i^{3/2}}
    +\frac{\gamma_j^2 }{q_jn_j  \mu_j^{3/2}}\Big)=:\alpha_{i,j}
    \end{aligned}
    $$
    with high probability, %, where $\gamma_i = \|\mpp^{(i)}\|_{\max} + \sigma_i$, $\vartheta_{ij} = \|\mx_{\mathcal{U}_{i} \cap \mathcal{U}_j}\|$ and $\theta_{ij} = \lambda_{\min}(\mx_{\mathcal{U}_i \cap \mathcal{U}_j}^{\top} \mx_{\mathcal{U}_i \cap \mathcal{U}_j})$.
    where $\mw^{(i)},\mw^{(j)}$ are defined in Lemma~\ref{lemma:hat X(i)W(i)-X} and $\mw^{(i,j)}=\underset{\mo\in \mathcal{O}_{d}}{\operatorname{argmin}} \|\hat\mx^{(i)}_{ \langle\mathcal{U}_i\cap \mathcal{U}_j\rangle}\mo-\hat\mx^{(j)}_{ \langle\mathcal{U}_i\cap \mathcal{U}_j\rangle}\|_F$.
    }
\end{lemma}
%The proofs of Lemma~\ref{lemma:hat X(i)W(i)-X} and Lemma~\ref{lemma:|| W(i)T W(i,j) W(j)-I ||} are presented in 

We now proceed with the proof of Theorem~\ref{thm:R(i,j)}. Recall Eq.~\eqref{eq:hatXW-X} and
let $\xi_i := \hat\mx^{(i)} \mw^{(i)}-\mx_{\mathcal{U}_i}$ for any $i$. Also denote $\tilde{\mw}^{(i,j)} = \mw^{(i)\top}\mw^{(i,j)} \mw^{(j)}$.
%From Eq.~\eqref{eq:hatXW-X} 
%We can then decompose $\hat\mx^{(i)} \mw^{(i,j)} \hat\mx^{(j)\top}
%	-\mx_{\mathcal{U}_i}\mx_{\mathcal{U}_j}^\top$ as
We then have
	$$
\begin{aligned}
	\hat\mx^{(i)} \mw^{(i,j)} \hat\mx^{(j)\top}
	-\mx_{\mathcal{U}_i}\mx_{\mathcal{U}_j}^\top
=&(\hat\mx^{(i)} \mw^{(i)})( \mw^{(i)\top}\mw^{(i,j)} \mw^{(j)})(\mw^{(j)\top}\hat\mx^{(j)\top})
	-\mx_{\mathcal{U}_i}\mx_{\mathcal{U}_j}^\top\\
 =&(\mx_{\mathcal{U}_i} + \xi_i) \tilde{\mw}^{(i,j)} (\mx_{\mathcal{U}_j} + \xi_j)^{\top}
	-\mx_{\mathcal{U}_i}\mx_{\mathcal{U}_j}^\top\\
  =&(\mx_{\mathcal{U}_i} + \xi_i) (\tilde{\mw}^{(i,j)} - \mi) (\mx_{\mathcal{U}_j} + \xi_j)^{\top} +  (\mx_{\mathcal{U}_i} + \xi_i) (\mx_{\mathcal{U}_j} + \xi_j)^{\top}
	-\mx_{\mathcal{U}_i}\mx_{\mathcal{U}_j}^\top\\
=& \xi_i \mx_{\mathcal{U}_j}^{\top} + \mx_{\mathcal{U}_i} \xi_j^{\top} + \xi_i \xi_j^{\top} + (\mx_{\mathcal{U}_i} + \xi_i) (\tilde{\mw}^{(i,j)} - \mi) (\mx_{\mathcal{U}_j} + \xi_j)^{\top} 
 %	=&\xi_i \tilde{\mw}^{(i,j)}  \xi_j^{\top} 
%	+ \xi_i \tilde{\mw}^{(i,j)} \mx_{\mathcal{U}_j}^\top
%	+ \mx_{\mathcal{U}_i} \tilde{\mw}^{(i,j)} \xi_j^{\top}
%	+ \mx_{\mathcal{U}_i} \big(\tilde{\mw}^{(i,j)} -\mi\big)\mx_{\mathcal{U}_j}^\top \\
% =&\xi_i (\tilde{\mw}^{(i,j)} - \mi)  \xi_j^{\top} + \xi_i \xi_j^{\top}
%	+ \xi_i \tilde{\mw}^{(i,j)} \mx_{\mathcal{U}_j}^\top
%	+ \mx_{\mathcal{U}_i} \tilde{\mw}^{(i,j)} \xi_j^{\top}
%	+ \mx_{\mathcal{U}_i} \big(\tilde{\mw}^{(i,j)} -\mi\big)\mx_{\mathcal{U}_j}^\top \\
%	=&\xi_i (\tilde{\mw}^{(i,j)} \mw^{(j)}-\mi)
%	\xi_j + \xi_i \xi_j^{\top} + 
%	+&\big(\hat\mx^{(i)} \mw^{(i)}-\mx_{\mathcal{U}_i}\big) 
%	\big(\hat\mx^{(j)}\mw^{(j)}-\mx_{\mathcal{U}_j}\big)^\top\\
%	+&\big(\hat\mx^{(i)} \mw^{(i)}-\mx_{\mathcal{U}_i}\big) 
%	(\mw^{(i)\top}\mw^{(i,j)} \mw^{(j)}-\mi)
%	\mx_{\mathcal{U}_j}^\top\\
%	+& \mx_{\mathcal{U}_i} 
%	(\mw^{(i)\top}\mw^{(i,j)} \mw^{(j)}-\mi)
%%	+&\big(\hat\mx^{(i)} \mw^{(i)}-\mx_{\mathcal{U}_i}\big) 
%	\mx_{\mathcal{U}_j}^\top
%	+ \mx_{\mathcal{U}_i}
%	\big(\hat\mx^{(j)}\mw^{(j)}-\mx_{\mathcal{U}_j}\big)^\top\\
%	+& \mx_{\mathcal{U}_i} \big(\mw^{(i)\top}\mw^{(i,j)}\mw^{(j)} -\mi\big)\mx_{\mathcal{U}_j}^\top\\
	\\ = & \me^{(i)}\mx_{\mathcal{U}_i}(\mx_{\mathcal{U}_i}^\top\mx_{\mathcal{U}_i})^{-1} \mx_{\mathcal{U}_j}^\top
	+\mx_{\mathcal{U}_i} (\mx_{\mathcal{U}_j}^\top\mx_{\mathcal{U}_j})^{-1} \mx_{\mathcal{U}_j}^\top \me^{(j)}
	+\mr^{(i,j)}+\ms^{(i,j)},
\end{aligned}
$$
where we set
\begin{equation}\label{eq:Rij}
	\begin{aligned}
	\mr^{(i,j)} :=
%	:=&\xi_i 	(\tilde{\mw}^{(i,j)}-\mi) \xi_j^{\top} + \xi_i \xi_j^{\top} + \xi_i (\tilde{\mw}^{(i,j)} - \mi) \mx_{\mathcal{U}_j}^{\top}
% + \mx_{\mathcal{U}_i} (\tilde{\mw}^{(i,j)} - \mi) \xi_j^{\top} \\
	%\big(\hat\mx^{(j)}\mw^{(j)}-\mx_{\mathcal{U}_j}\big)^\top\\
	%+&\big(\hat\mx^{(i)} \mw^{(i)}-\mx_{\mathcal{U}_i}\big) 
	%\big(\hat\mx^{(j)}\mw^{(j)}-\mx_{\mathcal{U}_j}\big)^\top\\
	%+&\big(\hat\mx^{(i)} \mw^{(i)}-\mx_{\mathcal{U}_i}\big) 
	%(\mw^{(i)\top}\mw^{(i,j)} \mw^{(j)}-\mi)
	%\mx_{\mathcal{U}_j}^\top\\
	%+& \mx_{\mathcal{U}_i} 
	%(\mw^{(i)\top}\mw^{(i,j)} \mw^{(j)}-\mi)
	%\big(\hat\mx^{(j)}\mw^{(j)}-\mx_{\mathcal{U}_j}\big)^\top\\
%	+&
\mr^{(i)} 
	\mx_{\mathcal{U}_j}^\top+\mx_{\mathcal{U}_i}  
	\mr^{(j)\top} + \xi_i \xi_j^{\top} 
	,\quad
	\ms^{(i,j)}:=(\mx_{\mathcal{U}_i} + \xi_i) (\tilde{\mw}^{(i,j)} - \mi) (\mx_{\mathcal{U}_j} + \xi_j)^{\top}.
	%+ \mx_{\mathcal{U}_i} (\tilde{\mw}^{(i,j)} -\mi\big)\mx_{\mathcal{U}_j}^\top
\end{aligned}
\end{equation}
We now bound $\mr^{(i,j)}$ and $\ms^{(i,j)}$. Note that,
for matrices $\mathbf{M}_1$ and $\mathbf{M}_2$ of conformal dimensions, we have
\[ \|\mathbf{M}_1 \mathbf{M}_2^{\top}\|_{\max} \leq \|\mathbf{M}_1\|_{2 \to \infty} \times \|\mathbf{M}_2\|_{2 \to \infty}, \quad \text{and} \quad
  \|\mathbf{M}_1 \mathbf{M}_2\|_{2 \to \infty} \leq \|\mathbf{M}_1\|_{2 \to \infty} \times \|\mathbf{M}_2\|.\]
We therefore have
\begin{equation}
  \label{eq:rij_technical0}
  \begin{split}
   &	\|\mr^{(i,j)}\|_{\max}%\|\hat\mx^{(i)} \mw^{(i)}-\mx_{\mathcal{U}_i}\|_{2\to\infty}
	\leq 
	\|\xi_i\|_{2 \to \infty} \times \|\xi_j\|_{2 \to \infty} 
	+ \|\mr^{(i)}\|_{2\to\infty}
	\times \|\mx_{\mathcal{U}_j}\|_{2\to\infty}
	+ \|\mx_{\mathcal{U}_i}\|_{2\to\infty}
	\times \|\mr^{(j)}\|_{2\to\infty},
% \|\tilde{\mw}^{(i,j)} - \mi\| \cdot \|\xi_j\|_{2 \to \infty} 
 %\|\mw^{(i)\top}\mw^{(i,j)} \mw^{(j)}-\mi\|
	%\cdot\|\hat\mx^{(j)}\mw^{(j)}-\mx_{\mathcal{U}_j}\|_{2\to\infty}\\
	\\ &\|\ms^{(i,j)}\|_{\max}\leq
	(\|\mx_{\mathcal{U}_i}\|_{2 \to \infty} + \|\xi_i\|_{2 \to \infty})
	\times (\|\mx_{\mathcal{U}_j}\|_{2 \to \infty} + \|\xi_j\|_{2 \to \infty})
	\times \|\tilde{\mw}^{(i,j)} - \mi \|. 
  \end{split}
\end{equation}
%Now by the fact that $\|\mathbf{M}\|_{2 \to \infty} = \|\mathbf{M} \mathbf{M}^{\top}\|_{\max}^{1/2}$ we have
%\begin{equation}\label{eq:Rij_X}
%	\begin{aligned}
%	\|\mx_{\mathcal{U}_i}\|_{2\to\infty}
%%	\leq \|\mpp^{(i)}\|_{\max}^{1/2} 
%	\leq \gamma_i^{1/2}/\log^{1/4}n_i.
%	\lesssim \frac{\lambda_{\max}^{1/2}}{N^{1/2}}.
%\end{aligned}
%\end{equation}
%where $\gamma_i = \|\mpp^{(i)}\|_{\max} + \sigma_i$. 
%Lemma~\ref{lemma:i->global} shows that the assumptions of Lemma~\ref{lemma:hat X(i)W(i)-X} are satisfied, then a
Next, by Lemma~\ref{lemma:hat X(i)W(i)-X}, for any $i$  %and Lemma~\ref{lemma:i->global} and
  we have
\begin{equation}\label{eq:Rij_R hatXW-X}
	\begin{aligned}
	\|\mr^{(i)}\|_{2\to\infty}
	\lesssim \frac{\gamma_i^2}{q_in_i  \mu_i^{3/2}}
	+\frac{\gamma_i }{q_i^{1/2}n_i \mu_i^{1/2}},\quad
	%&\lesssim \frac{(\|\mpp\|_{\max}+\sigma)^2 N^{1/2}\log N}{pq\lambda_{\min}^{3/2}}
	%+\frac{(\|\mpp\|_{\max}+\sigma) \log^{1/2}N}{pq^{1/2}N^{1/2}\lambda_{\min}^{1/2}},\\
%\end{aligned}
%$$
%with high probability,
%and
%$$
%\begin{aligned}
	%\|\hat\mx^{(i)} \mw^{(i)}-\mx_{\mathcal{U}_i}\|_{2\to\infty}
    \|\xi_i\|_{2 \to \infty} 
    \lesssim \frac{\gamma_i }{(q_i n_i \mu_i)^{1/2}}
    %\\&\lesssim \frac{(\|\mpp\|_{\max}+\sigma)\log^{1/2} N}{p^{1/2}q^{1/2}\lambda_{\min}^{1/2}}
\end{aligned}
\end{equation}
with high probability. Finally, by Lemma~\ref{lemma:|| W(i)T W(i,j) W(j)-I ||} we have
\begin{equation}\label{eq:tildeW-I}
\begin{aligned} 
	\|\tilde{\mw}^{(i,j)} -\mi\| 
&\lesssim \alpha_{i,j}%\frac{n_{i,j} \gamma_i  \gamma_j }{\theta_{i,j} (q_in_i  \mu_i)^{1/2} (q_jn_j  \mu_j)^{1/2}}
%    +\frac{n_{i,j}^{1/2}\|\mx_{\mathcal{U}_i\cap \mathcal{U}_j}\|_{2\to\infty}}{\theta_{i,j}}\Big(\frac{\gamma_i}{(q_i n_i\mu_i)^{1/2}} +\frac{\gamma_j}{(q_j n_j\mu_j)^{1/2}}\Big)
%\\
%      & + \frac{(n_{i,j} \vartheta_{i,j})^{1/2}}{\theta_{i,j}}\Big(\frac{\gamma_i^2 }{q_in_i  \mu_i^{3/2}}
%    +\frac{\gamma_j^2 }{q_jn_j  \mu_j^{3/2}}\Big)
\end{aligned}
\end{equation}
with high probability. Substituting the above bounds in %Eq.~\eqref{eq:Rij_X}, 
Eq.~\eqref{eq:Rij_R hatXW-X} and Eq.~\eqref{eq:tildeW-I} into Eq.~\eqref{eq:rij_technical0}, we obtain the desired bounds of $\|\mr^{(i,j)}\|_{\max}$ and $\|\ms^{(i,j)}\|_{\max}$.

\subsection{Proof of Theorem~\ref{thm:R(i0,...,iL)}}
Theorem~\ref{thm:R(i0,...,iL)} follows the same argument as that for Theorem~\ref{thm:R(i,j)}, with the only change being the use of 
Lemma~\ref{lemma: || W(i0)T ... W(iL)-I||} below (see Section~\ref{sec:technical_w_chain} for a proof) to bound the difference between $\mw^{(i_0,i_1)}
    	\cdots
    	\mw^{(i_{L-1},i_L)}$ and $\mw^{(i_0)}\mw^{(i_L)\top}$.
\begin{lemma}
\label{lemma: || W(i0)T ... W(iL)-I||} {\em 
    Consider the setting of Theorem~\ref{thm:R(i0,...,iL)}. Let $\mt^{(L)} := \mw^{(i_0)\top}\mw^{(i_0,i_1)}
    	\cdots
    	\mw^{(i_{L-1},i_L)}\mw^{(i_L)}$. Then
    	$$
    \begin{aligned}
    	\|\mt^{(L)} -\mi\|
&\lesssim \sum_{\ell=1}^{L}
\Big[\frac{n_{{i_{\ell-1}},{i_{\ell}}} \gamma_{i_{\ell-1}}  \gamma_{i_{\ell}} }{\theta_{{i_{\ell-1}},{i_{\ell}}}(q_{i_{\ell-1}}n_{i_{\ell-1}}  \mu_{i_{\ell-1}})^{1/2} (q_{i_{\ell}}n_{i_{\ell}}  \mu_{i_{\ell}})^{1/2} }
    \\&+\frac{n_{{i_{\ell-1}},{i_{\ell}}}^{1/2}\|\mx_{\mathcal{U}_{i_{\ell-1}}\cap \mathcal{U}_{i_{\ell}}}\|_{2\to\infty}}{\theta_{{i_{\ell-1}},{i_{\ell}}}}\Big(\frac{\gamma_{i_{\ell-1}}}{(q_{i_{\ell-1}} n_{i_{\ell-1}}\mu_{i_{\ell-1}})^{1/2}} +\frac{\gamma_{i_{\ell}}}{(q_{i_{\ell}} n_{i_{\ell}}\mu_j)^{1/2}}\Big)
\\
      & + \frac{(n_{{i_{\ell-1}},{i_{\ell}}} \vartheta_{{i_{\ell-1}},{i_{\ell}}})^{1/2}}{\theta_{{i_{\ell-1}},{i_{\ell}}}}\Big(\frac{\gamma_{i_{\ell-1}}^2 }{q_{i_{\ell-1}}n_{i_{\ell-1}}  \mu_{i_{\ell-1}}^{3/2}}
    +\frac{\gamma_{i_{\ell}}^2 }{q_{i_{\ell}}n_{i_{\ell}}  \mu_{i_{\ell}}^{3/2}}\Big)
    \Big]=\sum_{\ell=1}^{L}\alpha_{{i_{\ell-1}},{i_{\ell}}}
    \end{aligned}
    $$
    with high probability.}
\end{lemma}

%By Lemma~\ref{lemma: || W(i0)T ... W(iL)-I||} and the identical analysis of Theorem~\ref{thm:R(i,j)} the desired result in Theorem~\ref{thm:R(i0,...,iL)} is derived.
%%\hspace*{\fill} \qedsymbol

\subsection{Proof of Theorem~\ref{thm:normal}}

%In this proof, for ease of exposition we suppose for $i=i_0$ and $i_L$ we have
%\begin{equation}\label{eq:weaker condition}
%	\frac{n_i}{N}\lambda_{\min} \lesssim \lambda_{i,\min} \leq \lambda_{i,\max}\lesssim \frac{dn_i}{N}\lambda_{\max}
%  \text{ and }\|\muu^{(i)}\|_{2\to\infty}\lesssim\frac{d^{1/2}}{n_i^{1/2}}.
%\end{equation}
%The proof can be adapted for the weaker condition that for $i=i_0$ and $i_L$ Eq.~\eqref{eq:weaker condition} holds with high probability and
%$$
%\max\Big\{\frac{\|\muu^{(i)}\|_{2\to\infty}}{d^{1/2}/n_i^{1/2}},\frac{\lambda_{i,\max}}{dn_i\lambda_{\max}/N},\frac{n_i\lambda_{\min}/N}{\lambda_{i,\min}}\Big\}\lesssim N^k \text{ for some }k>0.
%$$

%Denote $\xi^{(i_0, i_{L})} := \hat\mx^{(i_0)} 
%		\tilde\mw^{(i_0,i_L)}%\cdots\mw^{(i_{L-1},i_L)}
%		\hat\mx^{(i_L)\top}
%	    -\mx_{\mathcal{U}_{i_0}}\mx_{\mathcal{U}_{i_L}}^\top$.
By Theorem~\ref{thm:R(i0,...,iL)}, for any fixed $s\in[n_{i_0}],t\in[n_{i_L}]$, we have\begin{equation}\label{eq:normal_1}
\begin{aligned}
	\bigl( \hat{\mpp}_{\mathcal{U}_{i_0},\mathcal{U}_{i_L}} - \mpp_{\mathcal{U}_{i_0},\mathcal{U}_{i_L}}\bigr)_{s,t}
  &=\sum_{k_1=1}^{n_{i_0}}\me^{({i_0})}_{s,k_1}\mb^{(i_0,i_L)}_{k_1,t} 
  + \sum_{k_2=1}^{n_{i_L}}\me^{({i_L})}_{t,k_2}\mb^{(i_L,i_0)}_{k_2,s}
	    % +\big[\operatorname{vec}[(\mb^{(i_L,i_0)})^\top (\me^{({i_L})})^\top]\big]_{\mathcal{M}}
	    +\mr^{(i_0,i_L)}_{s,t}
	    +\ms^{(i_0,\dots,i_L)}_{s,t}.
\end{aligned}
\end{equation}
%For the estimation error, the main term $\operatorname{vec}[\me^{({i_0})}\mb^{(i_0,i_L)} +(\mb^{(i_L,i_0)})^\top (\me^{({i_L})})^\top]$  is the sum of two independent terms
%Note that $\operatorname{vec}[\me^{({i_0})}\mb^{(i_0,i_L)}]$ and $\operatorname{vec}[(\mb^{(i_L,i_0)})^\top (\me^{({i_L})})^\top]$ are independent. 
%For $\operatorname{vec}[\me^{({i_0})}\mb^{(i_0,i_L)}]$,
%the same analysis also applies to $\operatorname{vec}[(\mb^{(i_L,i_0)})^\top (\me^{({i_L})})^\top]$.
%can be obtained with the identical analysis.
%because 
As $\me^{(i_0)} =\ma^{(i_0)}-\mpp^{(i_0)} =(\mpp^{(i_0)}+\mn^{(i_0)})\circ\mathbf\Omega^{(i_0)}/q_{i_0}
    	-\mpp^{(i_0)}$,
we have for any $s,k_1\in[n_{i_0}]$ that
\begin{equation}\label{eq:D}
	\begin{aligned}
	\mathrm{Var}\big[\me^{({i_0})}_{s,k_1}\big]
	=&\mathbb{E}\Big[\mathrm{Var}\big[\me^{({i_0})}_{s,k_1}\big|\mathbf{\Omega}^{(i_0)}_{s,k_1}\big]\Big]
	+\mathrm{Var}\Big[\mathbb{E}\big[\me^{({i_0})}_{s,k_1}\big|\mathbf{\Omega}^{(i_0)}_{s,k_1}\big]\Big]\\
	=&\mathbb{E}\big[\mathrm{Var}(\mn^{(i_0)}_{s,k_1})\mathbf{\Omega}^{(i_0)}_{s,k_1}/q_{i_0}^2\big]
	+\mathrm{Var}\big[\mpp^{(i_0)}_{s,k_1}\mathbf{\Omega}^{(i_0)}_{s,k_1}/q_{i_0}\big]% \\
	% =&q_{i_0}\mathrm{Var}(\mn^{(i_0)}_{k_1,k_2})/q_{i_0}^2
	% +q_{i_0}(1-q_{i_0})(\mpp^{(i_0)}_{k_1,k_2})^2/q_{i_0}^2
          \\
	=&\bigl[\mathrm{Var}(\mn^{(i_0)}_{s,k_1}) +(1 - q_{i_0})(\mpp^{(i_0)}_{s,k_1})^2\bigr]/q_{i_0}
	=\md^{(i_0)}_{s,k_1}.
\end{aligned}
\end{equation}
Similarly, we also have that for any $k_2,t\in[n_{i_L}]$, $\mathrm{Var}\big[\me^{({i_L})}_{t,k_2}\big]=\md^{(i_L)}_{t,k_2}$.
Note that $\{\me^{({i_0})}_{s,k_1},\me^{({i_L})}_{t,k_2}\}_{k_1\in [n_{i_1}],k_2\in[n_{i_L}]}$ are independent.
We thus have
\begin{equation*}
	\begin{aligned}
\mathrm{Var}\Bigl[\sum_{k_1=1}^{n_{i_0}}\me^{({i_0})}_{s,k_1}\mb^{(i_0,i_L)}_{k_1,t} 
  + \sum_{k_2=1}^{n_{i_L}}\me^{({i_L})}_{t,k_2}\mb^{(i_L,i_0)}_{k_2,s}\Bigr]
    =&\sum_{k_1=1}^{n_{i_0}}(\mb^{(i_0,i_L)}_{k_1,t})^2\md^{({i_0})}_{s,k_1}
    +\sum_{k_2=1}^{n_{i_L}}(\mb^{(i_L,i_0)}_{k_2,s})^2\md^{({i_L})}_{t,k_2}
    =\tilde\sigma^2_{s,t}.
\end{aligned}
\end{equation*}
Let 
\begin{gather*}
\my^{(i_0)}_{k_1}:=\me^{({i_0})}_{s,k_1}\mb^{(i_0,i_L)}_{k_1,t} \text{ for any }k_1\in[n_{i_0}], \quad
\my^{(i_L)}_{k_2}:=\me^{({i_L})}_{t,k_2}\mb^{(i_L,i_0)}_{k_2,s} \text{ for any }k_2\in[n_{i_L}],
\end{gather*}
and note that $\{\my^{(i_0)}_{k_1},\my^{(i_L)}_{k_2}\}_{k_1\in[n_{i_0}],k_2\in[n_{i_L}]}$ are mutually independent zero-mean random variables.
Let 
$$
\begin{aligned}
	\tilde\my^{(i_0)}_{k_1}:=\tilde\sigma^{-1}_{s,t}\my^{(i_0)}_{k_1} =
\tilde\sigma^{-1}_{s,t} \me^{({i_0})}_{s,k_1}\mb^{(i_0,i_L)}_{k_1,t} \text{ for any }k_1\in[n_{i_0}], \\
\tilde\my^{(i_L)}_{k_2}:=\tilde\sigma^{-1}_{s,t}\my^{(i_L)}_{k_2}=\tilde\sigma^{-1}_{s,t}\me^{({i_L})}_{t,k_2}\mb^{(i_L,i_0)}_{k_2,s} \text{ for any }k_2\in[n_{i_L}].
\end{aligned}
$$

We now analyze $\tilde\my^{(i_0)}_{k_1}$ for any fixed $k_1\in[n_{i_0}]$; the same analysis also applies to $\tilde\my^{(i_L)}_{k_2}$ for any fixed $k_2\in[n_{i_L}]$.
% Note $\tilde\my^{(i_0)}_{k_1,k_2}%:=(\mSigma^{(i_0,i_L)})^{-1/2}\my_{k_1,k_2}
% =(\mSigma^{(i_0,i_L)}_{\mathcal{M},\mathcal{M}})^{-1/2}\me^{(i_0)}_{k_1,k_2}\mathbf{z}^{(i_0)}_{k_1+(k_2-1)n_{i_0}}.$
Rewrite $\me^{(i_0)}_{s,k_1}$ as
    \begin{equation}\label{eq:E1E2_initial}
    	\begin{aligned}
    	\me^{(i_0)}_{s,k_1}%&
    	=\ma^{(i_0)}_{s,k_1}-\mpp^{(i_0)}_{s,k_1}%\\&
    	%=(\mpp^{(i_0)}+\mn^{(i_0)})\circ\mathbf\Omega^{(i_0)}/q_{i_0}
    	%-\mpp^{(i_0)}%\\&
    	=%\underbrace{
    	[\mpp^{(i_0)}_{s,k_1}\cdot \mathbf\Omega^{(i_0)}_{s,k_1}/q_{i_0}-\mpp^{(i_0)}_{s,k_1}]
    	%}_{\me^{(i_0,1)}}
    	+%\underbrace{
    	\mn^{(i_0)}_{s,k_1}\cdot \mathbf\Omega^{(i_0)}_{s,k_1}/q_{i_0}
    	%}_{\me^{(i_0,2)}}
    	.
    \end{aligned}
    \end{equation}
Then by Eq.~\eqref{eq:E1E2_initial} we have 
$$
\begin{aligned}
	\tilde\my^{(i_0)}_{k_1}
%	 &=(\mSigma^{(i_0,i_L,1)})^{-1/2}\me^{(i_0)}_{k_1,k_2}\mathbf{z}_{k_1+(k_2-1)n_{i_0}}\\
	 &=\underbrace{\tilde\sigma^{-1}_{s,t}[\mpp^{(i_0)}_{s,k_1} \mathbf{\Omega}^{(i_0)}_{s,k_1}/q_{i_0}-\mpp^{(i_0)}_{s,k_1}]\mb^{(i_0,i_L)}_{k_1,t}}_{\tilde\my_{k_1}^{(i_0,1)}}
	 +\underbrace{\tilde\sigma^{-1}_{s,t}[\mn^{(i_0)}_{s,k_1} \mathbf{\Omega}^{(i_0)}_{s,k_1}/q_{i_0}]\mb^{(i_0,i_L)}_{k_1,t}}_{\tilde\my_{k_1}^{(i_0,2)}}.%\\
	 %&=\underbrace{(\mSigma^{(i_0,i_L,1)})^{-1/2}\me^{(i_0,1)}_{k_1,k_2}\mathbf{z}_{k_1+(k_2-1)n_{i_0}}}_{\tilde\my_{k_1,k_2}^{(1)}}
   % +\underbrace{(\mSigma^{(i_0,i_L,1)})^{-1/2}\me^{(i_0,2)}_{k_1,k_2}\mathbf{z}_{k_1+(k_2-1)n_{i_0}}}_{\tilde\my_{k_1,k_2}^{(2)}},
\end{aligned}
$$
%where $\me^{(i_0,1)}$ and $\me^{(i_0,2)}$ are defined in Eq.~\eqref{eq:E1E2}.
{\color{black} The condition in Eq.~\eqref{eq:con_2} 
%$\lambda_{\min}(\mSigma^{(i_0,i_L)})\gtrsim \frac{\sigma^2+(1-q)\|\mpp\|_{\max}^2}{q}$
	} implies
	\begin{equation}\label{eq:mSigma}
		\tilde\sigma_{s,t}^{-1} \lesssim \zeta_{i_0,i_L}^{-1/2}
		=\min\Bigl\{\frac{(q_{i_0}n_{i_0}  \mu_{i_0})^{1/2} }{(\sigma_{i_0}^2 + (1 - q_{i_0})\|\mpp^{(i_0)}\|_{\max}^2)^{1/2}\|\bm{x}_{i_{L},t}\|},
  \frac{(q_{i_L}n_{i_L}  \mu_{i_L})^{1/2} }{(\sigma_{i_L}^2 + (1 - q_{i_0})\|\mpp^{(i_L)}\|_{\max}^2)^{1/2} \|\bm{x}_{i_0,s}\| }\Bigr\},
	\end{equation}
 where $\bm{x}_{i_0,s}$ and $\bm{x}_{i_L,t}$ denote the $s$th row and $t$th row of $\mx_{\mathcal{U}_{i_0}}$ and $\mx_{\mathcal{U}_{i_L}}$, respectively. 
Next we have
\begin{equation}\label{eq:z}
	\begin{aligned}
|\mb^{(i_0,i_L)}_{k_1,t}|
\leq \|\bm{x}_{i_0,k_1}\|
\cdot \|\mLambda^{(i_0)}\|^{-1} 
\cdot \|\bm{x}_{i_L,t}\|
&\lesssim (n_{i_0}\mu_{i_0})^{-1} \|\bm{x}_{i_0,k_1}\| \times \|\bm{x}_{i_L,t}\|.
\end{aligned}
\end{equation}
%To bound $|\me^{(i_0)}_{k_1,k_2}|$ we consider two cases: $1-q=o(1)$ or $1-q=\Omega(1)$. When $1-q=o(1)$, we have $|\me^{(i_0)}_{k_1,k_2}|\lesssim \frac{(1-q)\|\mpp\|_{\max}+\sigma\log^{1/2}N}{q}$ with probability $q$. When $1-q=\Omega(1)$, we have $|\me^{(i_0)}_{k_1,k_2}|\lesssim \frac{\|\mpp\|_{\max}+\sigma\log^{1/2}N}{q}$ with high probability. Therefore whatever the case is we always have
%$$
%\begin{aligned}
%	\frac{|\me^{(i_0)}_{k_1,k_2}|}{(\sigma^2+(1-q)\|\mpp\|_{\max}^2)^{1/2}}
%	\lesssim q^{-1}
%\end{aligned}
%$$
%with probability converging to $1$.
For any fixed but arbitrary $\epsilon>0$, we have
\begin{equation}\label{eq:E[Y]1}
	\begin{aligned}
	\mathbb{E}\big[|\tilde\my^{(i_0)}_{k_1}|^2\cdot\mathbb{I}\{|\tilde\my^{(i_0)}_{k_1}|>\epsilon\}\big]
	&\leq \mathbb{E}\big[|\tilde\my^{(i_0)}_{k_1}|^2\cdot\mathbb{I}\{|\tilde\my_{k_1}^{(i_0,1)}|>\epsilon/2\}\big]
	+\mathbb{E}\big[|\tilde\my^{(i_0)}_{k_1}|^2\cdot\mathbb{I}\{|\tilde\my_{k_1}^{(i_0,2)}|>\epsilon/2\}\big]\\
	%&\leq \mathbb{E}\big[\|\tilde\my_{k_1,k_2}\|^2\cdot\mathbb{I}\{\|\tilde\my_{k_1,k_2}^{(1)}\|>\epsilon/2\}|\mathbf{\Omega}^{(i_0)}_{k_1,k_2}=1\big]\cdot\mathbb{P}[\mathbf{\Omega}^{(i_0)}_{k_1,k_2}=1]\\
	%&+\mathbb{E}\big[\|\tilde\my_{k_1,k_2}\|^2\cdot\mathbb{I}\{\|\tilde\my_{k_1,k_2}^{(1)}\|>\epsilon/2\}|\mathbf{\Omega}^{(i_0)}_{k_1,k_2}=0\big]\cdot\mathbb{P}[\mathbf{\Omega}^{(i_0)}_{k_1,k_2}=0]
	%\\&+\mathbb{E}\big[\|\tilde\my_{k_1,k_2}\|^2\cdot\mathbb{I}\{\|\tilde\my_{k_1,k_2}^{(2)}\|>\epsilon/2\}|\mathbf{\Omega}^{(i_0)}_{k_1,k_2}=1\big]\cdot\mathbb{P}[\mathbf{\Omega}^{(i_0)}_{k_1,k_2}=1]\\
	%&+\mathbb{E}\big[\|\tilde\my_{k_1,k_2}\|^2\cdot\mathbb{I}\{\|\tilde\my_{k_1,k_2}^{(2)}\|>\epsilon/2\}|\mathbf{\Omega}^{(i_0)}_{k_1,k_2}=0\big]\cdot\mathbb{P}[\mathbf{\Omega}^{(i_0)}_{k_1,k_2}=0].
	&\leq \mathbb{E}\big[|\tilde\my^{(i_0)}_{k_1}|^2\cdot\mathbb{I}\{|\tilde\my_{k_1}^{(i_0,1)}|>\epsilon/2\}]%|\mathbf{\Omega}^{(i_0)}_{k_1,k_2}=1\big]
%	\\	&+ \mathbb{E}\big[\|\tilde\my^{(i_0)}_{k_1,k_2}\|^2\cdot\mathbb{I}
% \{\|\tilde\my_{k_1,k_2}^{(i_0,1)}\|>\epsilon/2\}|\mathbf{\Omega}^{(i_0)}_{k_1,k_2}=0\big]\cdot (1 - q) 
	\\ &+\mathbb{E}\big[|\tilde\my^{(i_0)}_{k_1}|^2\cdot\mathbb{I}\{|\tilde\my_{k_1}^{(i_0,2)}|>\epsilon/2\big\}|\mathbf{\Omega}^{(i_0)}_{s,k_1}=1]\cdot q_{i_0},
	%\mathbb{E}\big[\|\tilde\my_{k_1,k_2}\|^2\cdot\mathbb{I}\{\|\tilde\my_{k_1,k_2}^{(2)}\|>\epsilon/2\}|\mathbf{\Omega}^{(i_0)}_{k_1,k_2}=1\big]\cdot q\\
	%&+\mathbb{E}\big[\|\tilde\my_{k_1,k_2}\|^2\cdot\mathbb{I}\{\|\tilde\my_{k_1,k_2}^{(2)}\|>\epsilon/2\}|\mathbf{\Omega}^{(i_0)}_{k_1,k_2}=0\big]\cdot (1-q).
\end{aligned}
\end{equation}
where the last inequality follows from the fact that $\tilde{\my}^{(i_0,2)}_{k_1} =0$ whenever $\bm{\Omega}_{s,k_1}^{(i_0)} = 0$.

Now if $\mathbf{\Omega}^{(i_0)}_{s,k_1}=1$ then by Eq.~\eqref{eq:mSigma} and  Eq.~\eqref{eq:z} we have
\begin{equation}\label{eq:case1:tildeY1}
	\begin{aligned}
		|\tilde\my_{k_1}^{(i_0,1)}|
		&\leq \tilde\sigma^{-1}_{s,t}
		\cdot |\mpp^{(i_0)}_{s,k_1}/q_{i_0}-\mpp^{(i_0)}_{s,k_1}|
		\cdot |\mb^{(i_0,i_L)}_{k_1,t}|%\\
	%	&\lesssim 
  %\frac{(n_{i_0} q_{i_0} \mu_{i_0})^{1/2} \|\mx_{\mathcal{U}_{i_{L}}}\|_{2 \to \infty}^{-1}}{(\sigma_{i_0}^2 + (1 - q_{i_0})\|\mpp^{(i_0)}\|_{\max}^2)^{1/2}}\cdot 
  %\frac{1-q_{i_0}}{q_{i_0}}\|\mpp^{(i_0)}\|_{\max}
	%	\cdot (n_{i_0} \mu_{i_0})^{-1} \|\mx_{\mathcal{U}_{i_0}}\|_{2\to\infty} \cdot \|\mx_{\mathcal{U}_{i_L}}\|_{2\to\infty} \\
		%&
		\lesssim \frac{(1-q_{i_0})\|\mpp^{(i_0)}\|_{\max} \cdot \|\bm{x}_{i_0,k_1}\|}{(n_{i_0} q_{i_0} \mu_{i_0})^{1/2} (\sigma_{i_0}^2+(1-q_{i_0})\|\mpp^{i_0)}\|_{\max}^2)^{1/2}}.
	\end{aligned}
\end{equation}
Similarly, if $\mathbf{\Omega}^{(i_0)}_{s,k_1}=0$ we have
\begin{equation}\label{eq:case2:tildeY1}
	\begin{aligned}
		|\tilde\my_{k_1}^{(i_0,1)}|
		&\leq \tilde\sigma^{-1}_{s,t}
		\cdot |-\mpp^{(i_0)}_{s,k_1}|
		\cdot |\mb^{(i_0,i_L)}_{k_1,t}|%\\
	%	&\lesssim \frac{(qpN)^{1/2}}{(\sigma^2+(1-q)\|\mpp\|_{\max}^2)^{1/2}}
	%	\cdot \|\mpp^{(i)}\|_{\max}
	%	\cdot (pN)^{-1}\\
 %& 
 \lesssim  \frac{q_{i_0}^{1/2} \|\mpp^{(i_0)}\|_{\max} \cdot \|\bm{x}_{i_0,k_1}\|}{(n_{i_0} \mu_{i_0})^{1/2} (\sigma_{i_0}^2+(1-q_{i_0})\|\mpp^{i_0)}\|_{\max}^2)^{1/2}}.
	%	&\lesssim \frac{q^{1/2}\|\mpp\|_{\max}}{(pN)^{1/2}(\sigma^2+(1-q)\|\mpp\|_{\max}^2)^{1/2}}.
	\end{aligned}
\end{equation}
Eq.~\eqref{eq:case1:tildeY1} and Eq.~\eqref{eq:case2:tildeY1} together imply 
\begin{equation}\label{eq:E[Y]1_12}
	\sup_{k_1 \in [n_{i_0}]} |\tilde\my_{k_1}^{(i_0,1)}|\leq \epsilon/2
\end{equation}
asymptotically almost surely, provided by the condition in Eq.~\eqref{eq:con_3}.
%that 
%\[ 
%\frac{\|\mx_{\mathcal{U}_{i_0}}\|_{2 \to \infty}}{(n_{i_0} q_{i_0} \mu_{i_0})^{1/2}} \rightarrow 0 %, 
%\quad \text{and} \quad \frac{\|\mpp^{(i_0)}\|_{\max} \cdot \|\mx_{\mathcal{U}_{i_0}}\|_{2 \to \infty}}{(n_{i_0} \mu_{i_0})^{1/2} \sigma_{i_0}} \rightarrow 0\]
%as $n_{i_0} \rightarrow \infty$. 
%{\color{black}$qpN=\omega(1)$ and Eq.~\eqref{eq:con_3}}.
%and hence, by Eq.~\eqref{eq:E[Y]1}, we have
Returning to Eq.~\eqref{eq:E[Y]1}, note that $\tilde\my_{k_1}^{(i_0,1)}$ is a deterministic function of $\mathbf{\Omega}^{(i_0)}_{s,k_1}$ and hence
\begin{equation}\label{eq:E[Y]2}
	\begin{aligned}
	%	&\mathbb{E}\big[\|\tilde\my_{k_1,k_2}\|^2\cdot\mathbb{I}\{\|\tilde\my_{k_1,k_2}\|>\epsilon\}\big]\\
		%&\leq \mathbb{E}\big[\|\tilde\my_{k_1,k_2}\|^2\cdot\mathbb{I}\{\|\tilde\my_{k_1,k_2}^{(2)}\|>\epsilon/2\big]\\
		%&\leq \mathbb{E}\big[\|\tilde\my_{k_1,k_2}\|^2\cdot\mathbb{I}\{\|\tilde\my_{k_1,k_2}^{(2)}\|>\epsilon/2|\mathbf{\Omega}^{(i_0)}_{k_1,k_2}=1\big]\cdot q
		%+ \mathbb{E}\big[\|\tilde\my_{k_1,k_2}\|^2\cdot\mathbb{I}\{\|\tilde\my_{k_1,k_2}^{(2)}\|>\epsilon/2|\mathbf{\Omega}^{(i_0)}_{k_1,k_2}=0\big]\cdot (1-q)\\
	\mathbb{E}\big[|\tilde\my^{(i_0)}_{k_1}|^2\cdot\mathbb{I}\{|\tilde\my_{k_1}^{(i_0,2)}|>\epsilon/2\}|\mathbf{\Omega}^{(i_0)}_{s,k_1}=1\big]
		%&=  \mathbb{E}\big[|\tilde\my_{k_1}^{(i_0,1)}|^2\cdot\mathbb{I}\{|\tilde\my_{k_1}^{(i_0,2)}|>\epsilon/2\}|\mathbf{\Omega}^{(i_0)}_{s,k_1}=1\big]
		%2\mathbb{E}\big[\|\tilde\my_{k_1,k_2}^{(1)}\|^2\big]\cdot\mathbb{P}\big[\|\tilde\my_{k_1,k_2}^{(2)}\|>\epsilon/2\big]
	%+ \mathbb{E}\big[|\tilde\my_{k_1}^{(i_0,2)}|^2\cdot\mathbb{I}\{|\tilde\my_{k_1}^{(i_0,2)}|>\epsilon/2\}|\mathbf{\Omega}^{(i_0)}_{s,k_1}=1\big]\\
	&\leq 
	2 \big[|\tilde\my_{k_1}^{(i_0,1)}|^2|\mathbf{\Omega}^{(i_0)}_{s,k_1}=1\big]
	\cdot\mathbb{P}\big[|\tilde\my_{k_1}^{(i_0,2)}|>\epsilon/2|\mathbf{\Omega}^{(i_0)}_{s,k_1}=1\big]
		%2\mathbb{E}\big[\|\tilde\my_{k_1,k_2}^{(1)}\|^2\big]\cdot\mathbb{P}\big[\|\tilde\my_{k_1,k_2}^{(2)}\|>\epsilon/2\big]
\\ &+ 2 \mathbb{E}\big[|\tilde\my_{k_1}^{(i_0,2)}|^2\cdot\mathbb{I}\{|\tilde\my_{k_1}^{(i_0,2)}|>\epsilon/2\}|\mathbf{\Omega}^{(i_0)}_{s,k_1}=1\big].
	\end{aligned}
\end{equation}
%where the third inequality follows that $\tilde\my_{k_1,k_2}^{(2)}=\mathbf{0}$ given that $\mathbf{\Omega}^{(i_0)}_{k_1,k_2}=0$,
%where the inequality follows from the fact that $\tilde\my_{k_1}^{(i_0,1)}$ is a deterministic function of $\mathbf{\Omega}^{(i_0)}_{s,k_1}$.
%from the independence of $\tilde\my_{k_1,k_2}^{(1)}$ and $\tilde\my_{k_1,k_2}^{(2)}$ (as $\me_{k_1,k_2}^{(i_0,1)}$ is independent of $\me_{k_1,k_2}^{(i_0,2)}$).

We now bound the terms appearing %in the right hand side %of %Eq.~\eqref{eq:E[Y]2} 
%of 
in the above display. 
%For $\mathbb{E}\big[\|\tilde\my_{k_1,k_2}^{(1)}\|^2|\mathbf{\Omega}^{(i_0)}_{k_1,k_2}=1\big]$, 
First, by Eq.~\eqref{eq:case1:tildeY1}, we have
\begin{equation}\label{eq:E[Y]3}
	\begin{aligned}
	%\mathbb{E}
	\big[|\tilde\my_{k_1}^{(i_0,1)}|^2|\mathbf{\Omega}^{(i_0)}_{s,k_1}=1\big]
	& \lesssim \frac{(1-q_{i_0})^2\|\mpp^{(i_0)}\|_{\max}^2 \cdot \|\bm{x}_{i_0,k_1}\|^2}{n_{i_0} q_{i_0} \mu_{i_0} (\sigma_{i_0}^2+(1-q_{i_0})\|\mpp^{i_0)}\|_{\max}^2)}.
% &\lesssim \Big(\frac{(1-q)\|\mpp\|_{\max}}{(qpN)^{1/2}(\sigma^2+(1-q)\|\mpp\|_{\max}^2)^{1/2}}\Big)^2
	%\\ &\lesssim \frac{(1-q)^2\|\mpp\|_{\max}^2}{qpN(\sigma^2+(1-q)\|\mpp\|_{\max}^2)}.
	%\lesssim \frac{1-q}{qpN}.
\end{aligned}
\end{equation}
Next,  %suppose $\mathbf{\Omega}^{(i_0)}_{k_1,k_2}=1$. 
as $\mn^{(i_0)}_{s,k_1}$ is sub-Gaussian with $\|\mn^{(i_0)}_{s,k_1}\|_{\psi_2}\leq \sigma_{i_0}$, we have by a similar analysis to Eq.~\eqref{eq:case1:tildeY1} that 
$\tilde\my_{k_1}^{(i_0,2)}$ is also sub-Gaussian with
\[
\|\tilde\my_{k_1}^{(i_0,2)}\|_{\psi_2} \lesssim 
\frac{\sigma_{i_0}\|\mx_{\mathcal{U}_{i_0}}\|_{2 \to \infty}}{(n_{i_0} q_{i_0} \mu_{i_0})^{1/2} (\sigma_{i_0}^2+(1-q_{i_0})\|\mpp^{i_0)}\|_{\max}^2)^{1/2}} =: \nu_{i_0}.
\]
There thus exists a constant $C>0$ such that 
\begin{equation*}%\label{eq:E[Y]4}
	\begin{aligned}
		\mathbb{P}\big[|\tilde\my_{k_1}^{(i_0,2)}|>t|\mathbf{\Omega}^{(i_0)}_{s,k_1}=1\big]
		&\leq 2\exp\Bigl(\frac{-C t^2}{\nu_{i_0}^2} \Bigr)\text{ for any }t>0;
	\end{aligned}
\end{equation*}
see Eq.~(2.14) in \cite{vershynin2018high} for more details on tail bounds for sub-Gaussian random variables.  
We therefore have
%\begin{equation}\label{eq:E[Y]4}
	%\begin{aligned}
	%	\mathbb{P}\big[|\tilde\my_{k_1}^{(i_0,2)}|>\epsilon/2 |\mathbf{\Omega}^{(i_0)}_{s,k_1}=1\big]
	%	&\leq 2\exp\Big(\frac{-C\epsilon^2 qpN(\sigma^2+(1-q)\|\mpp\|_{\max}^2)}{4\sigma^2}\Big),
	%\end{aligned}
%\end{equation} 
\begin{equation}\label{eq:E[Y]5}
	\begin{aligned}
		&\mathbb{E}\big[|\tilde\my_{k_1}^{(i_0,2)}|^2\cdot\mathbb{I}\{|\tilde\my_{k_1}^{(i_0,2)}|>\epsilon/2\}|\mathbf{\Omega}^{(i_0)}_{s,k_1}=1\big]\\	&=\int_{t=0}^{\infty} \mathbb{P}\big[|\tilde\my_{k_1}^{(i_0,2)}|^2 \cdot \mathbb{I}\big\{\|\tilde\my_{k_1}^{(i_0,2)}\|^2>\epsilon^2/4\big\} \geq t \big|\mathbf{\Omega}^{(i_0)}_{s,k_1}=1\big] d t \\
       & =\epsilon^2/4 \times \mathbb{P}\big[|\tilde\my_{k_1}^{(i_0,2)}|^2 \geq \epsilon^2/4\big|\mathbf{\Omega}^{(i_0)}_{s,k_1}=1\big]
+
        \int_{t=\epsilon^2/4}^{\infty} \mathbb{P}\big[|\tilde\my_{k_1}^{(i_0,2)}|^2\geq t\big|\mathbf{\Omega}^{(i_0)}_{s,k_1}=1 \big] d t \\
        & \leq 
        \epsilon^2/4\times 2\exp\Big(\frac{-C\epsilon^2}{4 \nu_{i_0}^2}\Big)
        +
        2 \int_{\epsilon^2/4}^{\infty} 
        \exp\Big(\frac{-C t}{\nu_{i_0}^2}\Bigr) \, \mathrm{dt} %qpN(\sigma^2+(1-q)\|\mpp\|_{\max}^2)}{4\sigma^2}\Big)
       % +\int_{t=\epsilon^2/4}^{\infty} 2\exp\Big(\frac{-Ct qpN(\sigma^2+(1-q)\|\mpp\|_{\max}^2)}{\sigma^2}\Big) d t \\
       = \Big[\epsilon^2/2+\frac{2 \nu_{i_0}^2}{C} \Big]\exp\Bigl(-\frac{C \epsilon^2}{4 \nu_{i_0}^2}\Bigr).
       %\\ & =\Big(\epsilon^2/2+\frac{2 \sigma^2}{C qpN(\sigma^2+(1-q)\|\mpp\|_{\max}^2)}\Big) 
        %\times \exp\Big(\frac{-C\epsilon^2 qpN(\sigma^2+(1-q)\|\mpp\|_{\max}^2)}{4\sigma^2}\Big) .
	\end{aligned}
\end{equation}
Combining Eq.~\eqref{eq:E[Y]1} and Eq.~\eqref{eq:E[Y]1_12} through Eq.~\eqref{eq:E[Y]5}, we have
\begin{equation}\label{eq:lin1}
	\begin{aligned}
	%\lim_{N\rightarrow\infty}
	\sum_{k_1=1}^{n_{i_0}}\mathbb{E}\big[|\tilde\my^{(i_0)}_{k_1}|^2\cdot\mathbb{I}\{|\tilde\my^{(i_0)}_{k_1}|>\epsilon\}\big]
	&\lesssim n_{i_0} \cdot q_{i_0}
 \Bigl(\frac{(1-q_{i_0})^2 \|\mpp^{(i_0)}\|_{\max}^2 \cdot \|\mx_{\mathcal{U}_{i_0}}\|_{2 \to \infty}^2}{q_{i_0} n_{i_0}\mu_{i_0} (\sigma_{i_0}^2+(1-q_{i_0})\|\mpp^{i_0)}\|_{\max}^2)} + \epsilon^2+\nu_{i_0}^2\Bigr) \cdot \exp\Bigl(-\frac{C \epsilon^2}{4 \nu_{i_0}^2}\Bigr)\\
 &\lesssim 
 \Bigl(q_{i_0}n_{i_0}\epsilon^2+\frac{\|\mx_{\mathcal{U}_{i_0}}\|_{2 \to \infty}^2}{\mu_{i_0}}\Bigr) \cdot \exp\Bigl(-\frac{C \epsilon^2}{4 \nu_{i_0}^2}\Bigr)
% \Big(\frac{(1-q)^2\|\mpp\|_{\max}^2}{qpN(\sigma^2+(1-q)\|\mpp\|_{\max}^2)}
%	+\epsilon^2
%	+\frac{\sigma^2}{qpN(\sigma^2+(1-q)\|\mpp\|_{\max}^2)}\Big) \\
%\\
%        &\times \exp\Big(\frac{-C\epsilon^2 qpN(\sigma^2+(1-q)\|\mpp\|_{\max}^2)}{4\sigma^2}\Big) \\
%        &\lesssim \Big(\epsilon^2qpN+1
        %\frac{pN(\sigma^2+(1-q)^2\|\mpp\|_{\max}^2)}{(\sigma^2+(1-q)\|\mpp\|_{\max}^2)}
%        \Big)
%        \times \exp\Big(\frac{-C\epsilon^2 qpN(\sigma^2+(1-q)\|\mpp\|_{\max}^2)}{4\sigma^2}\Big) \\
\\ & \rightarrow 0
\end{aligned}
\end{equation}
as $n_{i_0} \rightarrow \infty$, under the assumption that {\color{black} $\frac{n_{i_0} q_{i_0}\mu_{i_0}}{\|\mx_{\mathcal{U}_{i_0}}\|_{2\to\infty}}=\omega(\log (q_{i_0}n_{i_0}))$ provided by Eq.~\eqref{eq:con_3}}. 
%, because $x^k/e^{x}\rightarrow 0$ as $x\rightarrow\infty$ for any $k$.
Using the same argument we also have 
\begin{equation}\label{eq:lin2}
	\begin{aligned}
	\lim_{n_{i_L} \rightarrow\infty}\sum_{k_2=1}^{n_{i_L}}\mathbb{E}\big[|\tilde\my^{(i_L)}_{k_2}|^2\cdot \mathbb{I}\{|\tilde\my^{(i_L)}_{k_2}|>\epsilon\}\big]
	= 0.
\end{aligned}
\end{equation}
By Eq.~\eqref{eq:lin1}, Eq.~\eqref{eq:lin2}, and applying the Lindeberg-Feller central limit theorem (see e.g.,
Proposition~2.27 in \cite{van2000asymptotic})  we have
\begin{equation}\label{eq:normal_2}
\tilde\sigma_{s,t}^{-1}\Big[\sum_{k_1=1}^{n_{i_0}}\me^{({i_0})}_{s,k_1}\mb^{(i_0,i_L)}_{k_1,t} 
  + \sum_{k_2=1}^{n_{i_L}}\me^{({i_L})}_{t,k_2}\mb^{(i_L,i_0)}_{k_2,s}\Big]
\rightsquigarrow \mathcal{N}(0,1)
\end{equation}
as $\min\{n_{i_0},n_{i_{L}}\} \rightarrow \infty$. Finally, invoking Theorem~\ref{thm:R(i0,...,iL)} and the assumption in Eq.~\eqref{eq:con_4}, % $\zeta_{i_0,i_L}^{-1/2} (r_{\infty} + s_{\infty}) \rightarrow 0$, 
we have $\tilde{\sigma}_{s,t}^{-1}(\|\mr^{(i_0,\dots,i_L)}\|_{\max} + \|\ms^{(i_0,\dots,i_L)}\|_{\max}) \rightarrow 0$ in probability. Then applying Slutsky's theorem we obtain $\tilde{\sigma}_{s,t}^{-1}(\hat{\mpp}_{\mathcal{U}_{i_0},\mathcal{U}_{i,L}} - \mpp_{\mathcal{U}_{i_0}, \mathcal{U_{i_L}}})_{s,t} \rightsquigarrow \mathcal{N}(0,1)$ as claimed.

\setcounter{theorem}{0}

\section{Technical Lemmas}

\subsection{Technical lemmas for Lemma~\ref{lemma:hat X(i)W(i)-X}}

\begin{lemma}
\label{lemma:|E|,|EU|,|UT EU|}
{\em Consider the setting of Lemma~\ref{lemma:hat X(i)W(i)-X}.
	Then for any $i$, for
    $
    	\me^{(i)}
    	=\ma^{(i)}-\mpp^{(i)}
    $
    we have 
	$$
	\begin{aligned}
		&\|\me^{(i)}\|
    	\lesssim q_i^{-1/2}(\|\mpp^{(i)}\|_{\max}+\sigma_i) n_i^{1/2},\\
    	&\|\muu^{(i)\top}\me^{(i)}\muu^{(i)}\|
    	\lesssim d^{1/2}q_i^{-1/2}(\|\mpp^{(i)}\|_{\max}+\sigma_i)\log^{1/2} n_i ,\\
    	&\|\me^{(i)}\muu^{(i)}\|_{2\to\infty}
    	\lesssim d^{1/2}q_i^{-1/2}(\|\mpp^{(i)}\|_{\max}+\sigma_i)\log^{1/2} n_i 
	\end{aligned}
	$$
	with high probability. }
\end{lemma}

\begin{proof}
    First write $\me^{(i)}$ as the sum of two matrices, namely
    \begin{equation}\label{eq:E1E2}
    	\begin{aligned}
    	\me^{(i)}%&
    	=(\mpp^{(i)}+\mn^{(i)})\circ\mathbf\Omega^{(i)}/q_i
    	-\mpp^{(i)}%\\&
    	=\underbrace{(\mpp^{(i)}\circ\mathbf\Omega^{(i)}/q_i-\mpp^{(i)})}_{\me^{(i,1)}}
    	+\underbrace{\mn^{(i)}\circ\mathbf\Omega^{(i)}/q_i}_{\me^{(i,2)}}.
    \end{aligned}
    \end{equation}
     If {\color{black}$ n_i q_i \gg  \log n_i$} then,
     following the same arguments as that for Lemma~13 in \cite{abbe2020entrywise} we obtain
    \begin{equation}\label{eq:||Ei1||}
    \begin{aligned}
    	\|\me^{(i,1)}\|
    	&\lesssim (n_i/q_i)^{1/2}\|\mpp^{(i)}\|_{\max}
    \end{aligned}	
    \end{equation}
    with high probability. Next for any arbitrary 
    $s,t\in[n_i]$ with $s\leq t$, let 
    $$
    \begin{aligned}
    \mz^{(i;s,t)}:=\left\{
    \begin{aligned}
    	&\me^{(i,1)}_{s,t}\big[\mathbf{u}^{(i)}_s(\mathbf{u}^{(i)}_t)^\top+\mathbf{u}^{(i)}_t(\mathbf{u}^{(i)}_s)^\top\big]\text{ for }s<t,\\
    	&\me^{(i,1)}_{s,s}\mathbf{u}^{(i)}_s(\mathbf{u}^{(i)}_s)^\top\qquad\qquad\qquad\quad \text{for }s=t,
    \end{aligned}
    \right.
    \end{aligned}
    $$
    where $\mathbf{u}^{(i)}_s$ denotes the $s$th row of $\muu^{(i)}$ for any $s\in[n_i]$.
    Then $\muu^{(i)\top}\me^{(i,1)}\muu^{(i)} = \sum_{s \leq t} \mz^{(i;s,t)}$; note that 
    $\{\mz^{(i;s,t)}\}_{s\leq t}$ are independent, random, self-adjoint matrices of $d$ dimension with $\mathbb{E}[\mz^{(i;s,t)}]=\mathbf{0}$ and
    $$
    \begin{aligned}
    	%\max_{s\leq t}
    	\|\mz^{(i;s,t)}\|
    	&\leq %\max_{s\leq t}
    	|\me^{(i,1)}_{s,t}|
    	\cdot 2%\max_{s\leq t} 
    	\|\mathbf{u}^{(i)}_s\|\cdot\|\mathbf{u}^{(i)}_t\|\\
    	&\leq \|\mpp^{(i)}\|_{\max}/q_i
    	\cdot 2\|\muu^{(i)}\|^2_{2\to\infty}
    	\lesssim d n_i^{-1}q_i^{-1}\|\mpp^{(i)}\|_{\max}.
    \end{aligned}
    $$
    Now for any square matrix $\mm$, we have $(\mm+\mm^\top)^2\preceq 2\mm\mm^\top+2\mm^\top\mm$, where $\preceq$ denotes the Loewner ordering for positive semidefinite matrices. Therefore for any $s<t$ we have
    $$
    \begin{aligned}
    	\big[\mathbf{u}^{(i)}_s(\mathbf{u}^{(i)}_t)^\top+\mathbf{u}^{(i)}_t(\mathbf{u}^{(i)}_s)^\top\big]^2
    	&\preceq 2\mathbf{u}^{(i)}_s(\mathbf{u}^{(i)}_t)^\top\mathbf{u}^{(i)}_t(\mathbf{u}^{(i)}_s)^\top
    	+2\mathbf{u}^{(i)}_t(\mathbf{u}^{(i)}_s)^\top\mathbf{u}^{(i)}_s(\mathbf{u}^{(i)}_t)^\top\\
    	&\preceq 2\|\mathbf{u}^{(i)}_t\|^2\mathbf{u}^{(i)}_s(\mathbf{u}^{(i)}_s)^\top
    	+2\|\mathbf{u}^{(i)}_s\|^2\mathbf{u}^{(i)}_t(\mathbf{u}^{(i)}_t)^\top.\\
    	%&\preceq 2\|\muu^{(i)}\|^2_{2\to\infty} [\muu^{(i)}_s(\muu^{(i)}_s)^\top+\muu^{(i)}_t(\muu^{(i)}_t)^\top].
    \end{aligned}
    $$
    Furthermore, as
    $\mathbb{E}[(\me^{(i,1)}_{s,t})^2]
    	=\frac{1-q_i}{q_i}(\mpp^{(i)}_{s,t})^2
    	\leq q_i^{-1} \|\mpp^{(i)}\|_{\max}^2
    $ for all $s,t \in [n_i]$, 
    we have
    $$
    \begin{aligned}
    	\Big\|\sum_{s\leq t} \mathbb{E}[(\mz^{(i;s,t)})^2]\Big\|
    	&\leq \max_{s\leq t} \mathbb{E}[(\me^{(i,1)}_{s,t})^2] 
    	\cdot 2\Big\|\sum_{s\in[n_i]}\sum_{t\in[n_i]}\|\mathbf{u}^{(i)}_s\|^2\mathbf{u}^{(i)}_t(\mathbf{u}^{(i)}_t)^\top\Big\|\\
       &\leq 2q_i^{-1} \|\mpp^{(i)}\|_{\max}^2
    	\cdot \sum_{s\in[n_i]}\|\mathbf{u}^{(i)}_s\|^2
    	\cdot \Big\|\sum_{t\in[n_i]}\mathbf{u}^{(i)}_t(\mathbf{u}^{(i)}_t)^\top\Big\|\\
&\leq 2q_i^{-1} \|\mpp^{(i)}\|_{\max}^2
    \cdot d \cdot \|\muu^{(i)\top}\muu^{(i)}\|%\\
    	%&%\leq 2q_i^{-1} \|\mpp^{(i)}\|_{\max}^2
    	%\cdot d 
    	%\cdot 1
    	%\leq 
     \leq 2d q_i^{-1} \|\mpp^{(i)}\|_{\max}^2.
    \end{aligned}
    $$
    Therefore according to Theorem~1.4 in \cite{tropp2012user}, for all $t>0$, we have
    $$
    \begin{aligned}
    	\mathbb{P}\Big\{\|\muu^{(i)\top}\me^{(i,1)}\muu^{(i)}\|\geq t\Big\}
    	&\leq d\cdot \exp \Big(\frac{-t^2/2}{2d q_i^{-1} \|\mpp^{(i)}\|_{\max}^2+d n_i^{-1}q_i^{-1}\|\mpp^{(i)}\|_{\max}t/3}\Big),
    \end{aligned}
    $$
    and hence
    \begin{equation}\label{eq:||UtEiU1||}
    	\begin{aligned}
    	\|\muu^{(i)\top}\me^{(i,1)}\muu^{(i)}\|
    	&\lesssim d^{1/2}q_i^{-1/2} \|\mpp^{(i)}\|_{\max}\log^{1/2}{n_i}
    	+dn_i^{-1}q_i^{-1} \|\mpp^{(i)}\|_{\max}\log n_i \\
    	&\lesssim d^{1/2}q_i^{-1/2} \|\mpp^{(i)}\|_{\max}\log^{1/2}{n_i} 
    \end{aligned}
    \end{equation}
    with high probability.% under the assumption {\color{black}$q_in_i^2\gtrsim \log N$}.
    
    Next note that
    $
    \|\me^{(i,1)}\muu^{(i)}\|_{2\to\infty}
    =\max_{s\in[n_i]}\|(\me^{(i,1)}\muu^{(i)})_s\|
    $, where $(\me^{(i,1)}\muu^{(i)})_s$ is the $s$th row of $\me^{(i,1)}\muu^{(i)}$.
    Thus, for any fixed $s\in[n_i]$, $(\me^{(i,1)}\muu^{(i)})_s=\sum_{t\in[n_i]}[\me^{(i,1)}_{s,t} \mathbf{u}^{(i)}_t]$; note that $\{\me^{(i,1)}_{s,t} \mathbf{u}^{(i)}_t\}_{t\in[n_i]}$ are independent, random matrices of dimension $1\times d$ with $\mathbb{E}[\me^{(i,1)}_{s,t} \mathbf{u}^{(i)}_t]=\mathbf{0}$ and
    $$
    \begin{aligned}
    	\|\me^{(i,1)}_{s,t} \mathbf{u}^{(i)}_t\|
    	\leq |\me^{(i,1)}_{s,t}| 
    	\cdot\|\mathbf{u}^{(i)}_t\|
    	\leq \|\mpp^{(i)}\|_{\max}/q_i\cdot \|\muu^{(i)}\|_{2\to\infty}%\\&
    	\leq d^{1/2}n_i^{-1/2}q_i^{-1}\|\mpp^{(i)}\|_{\max}.
    \end{aligned}
    $$
    Let $\sigma_*^2 = \max_{t \in [n_i]} \mathbb{E}[
    	(\me^{(i,1)}_{s,t})^2]$. 
    We then have
    $$
    \begin{aligned}
    \max\Big\{
    	\Big\|\sum_{t\in[n_i]}
    	\mathbb{E}
    	[
    	(\me^{(i,1)}_{s,t})^2\mathbf{u}^{(i)}_t (\mathbf{u}^{(i)}_t)^\top
    	]\Big\|,
    	\Big\|\sum_{t\in[n_i]}
    	\mathbb{E}
    	[
    	(\me^{(i,1)}_{s,t})^2 (\mathbf{u}^{(i)}_t)^\top\mathbf{u}^{(i)}_t 
    	]\Big\|
    \Big\}
    %\leq &
    %\max_{t\in[n_i]}\mathbb{E}[(\me^{(i)}_1)_{s,t}^2]
    %\cdot \max\Big\{
    %\Big\|\sum_{t\in[n_i]}\muu^{(i)}_t (\muu^{(i)}_t)^\top\Big\|,
    %\Big\|\sum_{t\in[n_i]} (\muu^{(i)}_t)^\top\muu^{(i)}_t\Big\|
    %\Big\}\\
    &\leq \sigma_{*}^2 \max\Bigl\{
    \bigl\|\muu^{(i)}\bigr\|^2,
    \sum_{t\in[n_i]}\|\bm{u}^{(i)}_t\|^2
    \Bigr\}\\
  %  &\lesssim &q_i^{-1}\|\mpp^{(i)}\|_{\max}^2
   % \cdot \max\{1,d\}
     & \lesssim d q_i^{-1} \|\mpp^{(i)}\|_{\max}^2.
    \end{aligned}
    $$
    Therefore, by Theorem~1.6 in \cite{tropp2012user}, for any $t > 0$ we have
    $$
    \begin{aligned}
    	\mathbb{P}\Big\{\|(\me^{(i,1)}\muu^{(i)})_s\|\geq t\Big\}
    	\leq (1+d)\cdot \exp \Big(\frac{-t^2/2}{d q_i^{-1} \|\mpp^{(i)}\|_{\max}^2+d^{1/2}n_i^{-1/2}q_i^{-1}\|\mpp^{(i)}\|_{\max}t/3}\Big),
    \end{aligned}
    $$
    and hence
    $$
    \begin{aligned}
    	\|(\me^{(i,1)}\muu^{(i)})_s\|
    	&\lesssim d^{1/2} q_i^{-1/2}\|\mpp^{(i)}\|_{\max}\log^{1/2}{n_i}
    	+d^{1/2}n_i^{-1/2}q_i^{-1}\|\mpp^{(i)}\|_{\max}\log n_i\\
    	%&\lesssim q^{-1/2}\|\mpp^{(i)}\|_{\max}\log^{1/2}N(1+n_i^{-1/2}q^{-1/2}\log^{1/2}N)\\
    	&\lesssim d^{1/2}q_i^{-1/2}\|\mpp^{(i)}\|_{\max}\log^{1/2} n_i
    \end{aligned}
    $$
    with high probability, where the final inequality follows from the assumption 
    $n_i q_i\gtrsim \log n_i$. Taking a union over all $s \in [n_i]$ we obtain
    \begin{equation}\label{eq:||EiU1||}
    	\|\me^{(i,1)}\muu^{(i)}\|_{2\to\infty}
    \lesssim d^{1/2}q_i^{-1/2}\|\mpp^{(i)}\|_{\max}\log^{1/2} n_i
    \end{equation}
    with high probability.

    For $\me^{(i,2)}$, its upper triangular entries are independent random variables. Because for any $s,t\in[n_i]$ $\{\mn_{s,t}^{(i)}\}$ is a sub-gaussian random variable with $\|\mn_{s,t}^{(i)}\|_{\psi_2}\leq\sigma_i$, we have
    $
    	\mathbb{E}[(\mn_{s,t}^{(i)})^2]
    	\leq 2\sigma_i^2,
    $
    and it follows that
    $$
    \begin{aligned}
        \mathbb{E}[(\me^{(i,2)}_{s,t})^2]
    	=\mathbb{E}\Big[\frac{(\mn_{s,t}^{(i)})^2}{q_i^2}\circ\mathbf\Omega^{(i)}\Big]
    	\leq \frac{2\sigma_i^2q_i}{q_i^2}
    	=\frac{2\sigma_i^2}{q_i},
    \text{ and }
    	%\tilde\sigma:=
    	\max_{s\in[n_i]}\sqrt{\sum_{t=1}^{n_i} 
    	\mathbb{E}[(\me^{(i,2)}_{s,t})^2]}
    	\leq \sqrt{\frac{2\sigma_i^2 n_i}{q_i}}.
    \end{aligned}
    $$
    For sub-gaussian random variables, we also have
    %$
    %	\mathbb{E}[\mn_{s,t}^{(i)}]=0
    %$
    $
    |\mn_{s,t}^{(i)}|
    \leq c\sigma_i \log^{1/2} n_i 
    $
    with high probability; here $c$ is some finite constant not depending on $n_i$ or $\sigma_i$. 
    Then with high probability $$%\tilde\sigma_*:=
    \max_{s,t\in[n_i]}|\me^{(i,2)}_{s,t}|\leq c q_i^{-1}\sigma_i \log^{1/2} n_i.$$
    Then by combining Corollary 3.12 and Remark~3.13
    in \cite{bandeira2016sharp} with Proposition~A.7 in \cite{truncated_bernstein}, 
    there exists some constant $c'>0$ such that for any $t>0$
    $$
    \begin{aligned}
    	\mathbb{P}
    	\Big\{\|\me^{(i,2)}\|
    	\geq 3\sqrt{\frac{2\sigma_i^2 n_i}{q_i}} +t\Big\}
    	\leq n \exp\Big(-\frac{t^2}{c' (c q_i^{-1}\sigma_i\log^{1/2} n_i)^2 }\Big).
    \end{aligned}
    $$
    Let $t=C(n_i/q_i)^{1/2}\sigma_i$ for some $c>0$, then {\color{black} from $n_i q_i\gtrsim \log^2 n_i$} we have
    \begin{equation}\label{eq:||Ei2||}
    \begin{aligned}
    	\|\me^{(i,2)}\|
    	\lesssim (n_i/q_i)^{1/2}\sigma_i
    \end{aligned}	
    \end{equation}
    with high probability.
    
    For $\muu^{(i)\top}\me^{(i,2)}\muu^{(i)}$ and $\me^{(i,2)}\muu^{(i)}$, with the similar analysis with $\muu^{(i)\top}\me^{(i,1)}\muu^{(i)}$ and $\me^{(i,1)}\muu^{(i)}$ we have
    \begin{equation}\label{eq:||UtEiU2||}
    	\begin{aligned}
    	\|\muu^{(i)\top}\me^{(i,2)}\muu^{(i)}\|
    	&\lesssim d^{1/2} q_i^{-1/2} \sigma_i \log^{1/2} n_i
    	+dn_i^{-1}q_i^{-1} \sigma_i\log^{3/2} n_i\\
    	&\lesssim d^{1/2} q_i^{-1/2} \sigma_i \log^{1/2} n_i
    \end{aligned}
    \end{equation}
    with high probability
    and
    \begin{equation}\label{eq:||EiU2||}
    	\begin{aligned}
    	\|\me^{(i,2)}\muu^{(i)}\|_{2\to\infty}
    	&\lesssim d^{1/2} q_i^{-1/2}\sigma_i\log^{1/2} n_i
    	+d^{1/2} n_i^{-1/2}q_i^{-1}\sigma_i\log^{3/2} n_i\\
    	&\lesssim d^{1/2} q_i^{-1/2}\sigma_i\log^{1/2} n_i
    	%(1+n_i^{-1/2}q^{-1/2}\log N)
    	%{\color{red}\text{ (if assume }qpN\gtrsim \log^2 N)??? \me_2 \text{ is sub-gaussian with }\sigma_i q_i^{-1/2}?}
    \end{aligned}
    \end{equation}
    with high probability. %{\color{black} by $q_i n_i \gtrsim \log^2 N$}.
    
    Finally we combine Eq.~\eqref{eq:||Ei1||} and Eq.~\eqref{eq:||Ei2||}, Eq.~\eqref{eq:||UtEiU1||} and Eq.~\eqref{eq:||UtEiU2||}, Eq.~\eqref{eq:||EiU1||} and Eq.~\eqref{eq:||EiU2||} and obtain the desired results for $\me^{(i)}$, $\muu^{(i)\top}\me^{(i)}\muu^{(i)}$, and $\me^{(i)}\muu$.
\end{proof}

\begin{lemma}
\label{lemma:||sinTheta(Uhat,U)||,Lambdahat}
   {\em Consider the setting of Lemma~\ref{lemma:hat X(i)W(i)-X}.
	Then for any $i\in[K]$ we have 
    $$
    \begin{aligned}
    	    &\lambda_k(\ma^{(i)})\asymp \lambda_k(\mpp^{(i)})\text{ for }k= 1,\dots,d,\\
    	    %\text{and }
    	    &\lambda_{k}(\ma^{(i)})\lesssim q_i^{-1/2}(\|\mpp^{(i)}\|_{\max}+\sigma_i) n_i^{1/2}\text{ for }k=d+1,\dots, n_i
    \end{aligned}
    $$
    with high probability, and for $\muu^{(i)}$ and $\hat\muu^{(i)}$ we have
    $$
    \begin{aligned}
    	&\|\sin\Theta(\hat\muu^{(i)},\muu^{(i)})\|
    	\leq \frac{(\|\mpp^{(i)}\|_{\max}+\sigma_i) n_i^{1/2}}{q_i^{1/2}\lambda_{i,\min}},\\
    	&\|\muu^{(i)\top}\hat\muu^{(i)}-\mw^{(i)\top}_\muu\|
    	\lesssim \frac{(\|\mpp^{(i)}\|_{\max}+\sigma_i)^2 n_i}{q_i\lambda_{i,\min}^2},\\
    	&\|\hat\muu^{(i)}\mw^{(i)}_\muu-\muu^{(i)}\|
    	\lesssim \frac{(\|\mpp^{(i)}\|_{\max}+\sigma_i) n_i^{1/2}}{q_i^{1/2}\lambda_{i,\min}}
    \end{aligned}
    $$
    with high probability. }
\end{lemma}

\begin{proof}
    %For ease of exposition we will fix a value of $i$ and thereby drop the index from our matrices.

	By perturbation theorem for singular values (see Problem~III.6.13 in \cite{horn2012matrix}) and Lemma~\ref{lemma:|E|,|EU|,|UT EU|}, for any $k\in[n_i]$ we have
    $$
    \begin{aligned}
    	|\lambda_k(\ma^{(i)})-\lambda_k(\mpp^{(i)})|
    	\leq \|\me^{(i)}\|
    	\lesssim q_i^{-1/2}(\|\mpp^{(i)}\|_{\max}+\sigma_i) n_i^{1/2}
    \end{aligned}
    $$
    with high probability. Then the condition in Eq.~\eqref{eq:cond_lemmaA1} %Eq.~\eqref{eq:cond2_lemmaA1} implies $\frac{(\|\mpp^{(i)}\|_{\max}+\sigma_i)n_i^{1/2}}{q_i^{1/2}\lambda_{i,\min}}=o(1)$, which 
    yields the stated claim for $\lambda_k(\ma^{(i)})$. 
   Next, by Wedin's $\sin\Theta$ Theorem (see e.g., Theorem~4.4 in Chapter~4 of  \cite{stewart_sun}) and Lemma~\ref{lemma:|E|,|EU|,|UT EU|}, we have
	$$
	\begin{aligned}
		\|\sin \Theta(\hat\muu^{(i)},\muu^{(i)})\|
		\leq \frac{\|\me^{(i)}\|}{\lambda_{d}(\ma^{(i)})-\lambda_{d+1}(\mpp^{(i)})}
		\lesssim \frac{\|\me^{(i)}\|}{\lambda_{i,\min}}
		\lesssim \frac{(\|\mpp^{(i)}\|_{\max}+\sigma_i) n_i^{1/2}}{q_i^{1/2}\lambda_{i,\min}}
	\end{aligned}
	$$
	with high probability, and hence
    $$
    \begin{aligned}
    	&\|\muu^{(i)\top}\hat\muu^{(i)}-\mw^{(i)\top}_\muu\|
    	\leq \|\sin\Theta(\hat\muu^{(i)},\muu^{(i)})\|^2
    	\lesssim \frac{(\|\mpp^{(i)}\|_{\max}+\sigma_i)^2 n_i}{q_i\lambda_{i,\min}^2},\\
    	&\|\hat\muu^{(i)}\mw^{(i)}_\muu-\muu^{(i)}\|
    	\leq \|\sin\Theta(\hat\muu^{(i)},\muu^{(i)})\|
    	+\|\muu^{(i)\top}\hat\muu^{(i)}-\mw^{(i)\top}_\muu\|
    	\lesssim \frac{(\|\mpp^{(i)}\|_{\max}+\sigma_i) n_i^{1/2}}{q_i^{1/2}\lambda_{i,\min}}
    \end{aligned}
    $$
    with high probability.
\end{proof}

\begin{lemma}
\label{lemma:||Lambda UT Uhat-UT Uhat Lambdahat||, ||Lambda WT-WT Lambdahat||, ||Lambdahat1/2 W-W Lambda1/2||}
{\em     Consider the setting of Lemma~\ref{lemma:hat X(i)W(i)-X}.
	Then for each $i $, we have
	$$
	\begin{aligned}
		&\|\mLambda^{(i)}\muu^{(i)\top}\hat\muu^{(i)}-\muu^{(i)\top}\hat\muu^{(i)} \hat\mLambda^{(i)}\|
		\lesssim \frac{ (\|\mpp^{(i)}\|_{\max}+\sigma_i)^2 n_i}{q_i\lambda_{i,\min}}
		+d^{1/2} q_i^{-1/2}(\|\mpp^{(i)}\|_{\max}+\sigma_i)\log^{1/2} n_i ,\\
		&\|\mLambda^{(i)}\mw_\muu^{(i)\top}
		-\mw_\muu^{(i)\top}\hat\mLambda^{(i)}\|
	    \lesssim \frac{(\|\mpp^{(i)}\|_{\max}+\sigma_i)^2 n_i}{q_i\lambda_{i,\min}}
		+d^{1/2}q_i^{-1/2}(\|\mpp^{(i)}\|_{\max}+\sigma_i)\log^{1/2} n_i,\\
		&\|(\hat\mLambda^{(i)})^{1/2}\mw^{(i)}_\muu-\mw^{(i)}_\muu(\mLambda^{(i)})^{1/2}\|
		\lesssim \frac{d^{1/2}(\|\mpp^{(i)}\|_{\max}+\sigma_i)^2 n_i}{q_i\lambda_{i,\min}^{3/2}}
		+\frac{d(\|\mpp^{(i)}\|_{\max}+\sigma_i)\log^{1/2} n_i}{q_i^{1/2}\lambda_{i,\min}^{1/2}},\\
		&\|\breve\mw^{(i)}-\mw_\muu^{(i)}\|
		\lesssim \frac{d^{1/2}(\|\mpp^{(i)}\|_{\max}+\sigma_i)^2 n_i}{q_i\lambda_{i,\min}^2}
		+\frac{d(\|\mpp^{(i)}\|_{\max}+\sigma_i)\log^{1/2} n_i}{q_i^{1/2}\lambda_{i,\min}}
	\end{aligned}
	$$
	with high probability.}
\end{lemma}

\begin{proof}
    For ease of exposition we will fix a value of $i$ and thereby drop the index from our matrices.
    
    For $\mLambda\muu^\top\hat\muu-\muu^\top\hat\muu \hat\mLambda$, because
	$$
	\begin{aligned}
		\mLambda\muu^\top\hat\muu-\muu^\top\hat\muu \hat\mLambda
		&=\muu^\top\mpp\hat\muu-\muu^\top\ma\hat\muu
		=-\muu^\top\me\hat\muu
		=-\muu^\top\me(\hat\muu\mw_\muu-\muu)\mw_\muu^\top-\muu^\top\me\muu\mw_\muu^\top,
	\end{aligned}
	$$
	by Lemma~\ref{lemma:|E|,|EU|,|UT EU|} and Lemma~\ref{lemma:||sinTheta(Uhat,U)||,Lambdahat} we have
	\begin{equation}\label{eq:lambdaUtU-UtUlambda}
		\begin{aligned}
		\|\mLambda\muu^\top\hat\muu-\muu^\top\hat\muu \hat\mLambda\|
		&\leq \|\me\|\cdot\|\hat\muu\mw_\muu-\muu\|
		+\|\muu^\top\me\muu\|\\
		%&\lesssim \frac{\|\me\|^2}{\lambda_d(\mpp)}
		%+\|\muu^\top\me\muu\|\\
		&\lesssim \frac{ (\|\mpp^{(i)}\|_{\max}+\sigma_i)^2 n_i}{q_i\lambda_{i,\min}}
		+d^{1/2} q_i^{-1/2}(\|\mpp^{(i)}\|_{\max}+\sigma_i)\log^{1/2}{n_i}
	\end{aligned}
	\end{equation}
	with high probability. 
	
	For $\mLambda\mw_\muu^{\top}
		-\mw_\muu^{\top}\hat\mLambda$, we notice
	$$
	\begin{aligned}
		\mLambda\mw_\muu^{\top}
		-\mw_\muu^{\top}\hat\mLambda
		=\mLambda(\mw_\muu^{\top}-\muu^\top\hat\muu)
		+(\mLambda\muu^\top\hat\muu-\muu^\top\hat\muu\hat\mLambda)
		+(\muu^\top\hat\muu-\mw_\muu^{\top})\hat\mLambda.
	\end{aligned}
	$$
	For the right hand side of the above display, we have bounded the second term in Eq.~\eqref{eq:lambdaUtU-UtUlambda}. 
	For the first term and the third term, by Lemma~\ref{lemma:||sinTheta(Uhat,U)||,Lambdahat} %and Lemma~\ref{lemma:i->global}
	we have
	\begin{equation}\label{eq:lambda...+...lambda}
		\begin{aligned}
		\|\mLambda(\mw_\muu^{\top}-\muu^\top\hat\muu)
		+(\muu^\top\hat\muu-\mw_\muu^{\top})\hat\mLambda\|
		&\leq (\|\mLambda\|+\|\hat\mLambda\|)\cdot\|\muu^\top\hat\muu-\mw_\muu^{\top}\|\\
		&\lesssim \lambda_{i,\max}\cdot\frac{(\|\mpp^{(i)}\|_{\max}+\sigma_i)^2 n_i}{q_i\lambda_{i,\min}^2}
		\lesssim \frac{(\|\mpp^{(i)}\|_{\max}+\sigma_i)^2 n_i}{q_i\lambda_{i,\min}}
	\end{aligned}
	\end{equation}
	with high probability.
	Combining Eq.~\eqref{eq:lambdaUtU-UtUlambda} and Eq.~\eqref{eq:lambda...+...lambda}, we have 
	\begin{equation}\label{eq:lambdaWt+Wtlambda}
		\begin{aligned}
			\|\mLambda\mw_\muu^{\top}
		-\mw_\muu^{\top}\hat\mLambda\|
	    \lesssim \frac{(\|\mpp^{(i)}\|_{\max}+\sigma_i)^2 n_i}{q_i\lambda_{i,\min}}
		+d^{1/2}q_i^{-1/2}(\|\mpp^{(i)}\|_{\max}+\sigma_i)\log^{1/2} n_i 
		\end{aligned}
	\end{equation}
    with high probability.
	
	For $\hat\mLambda^{1/2}\mw_\muu-\mw_\muu\mLambda^{1/2}$, for any $k,l\in[d]$ the $(k,l)$ entry can be written as
	\begin{equation}\label{eq:lambdaW-Wlambda}
		\begin{aligned}
		(\hat\mLambda^{1/2}\mw_\muu-\mw_\muu\mLambda^{1/2})_{k,l}
	    &=(\mw_\muu)_{k,l}\cdot\big(\sqrt{\lambda_{l}(\ma)}-\sqrt{\lambda_{k}(\mpp)}\big)\\
	    &=(\mw_\muu)_{k,l}\cdot\big({\lambda_{l}(\ma)}-{\lambda_{k}(\mpp)}\big)\cdot \big(\sqrt{\lambda_{l}(\ma)}+\sqrt{\lambda_{k}(\mpp)}\big)^{-1}\\
	    &=(\hat\mLambda\mw_\muu-\mw_\muu\mLambda)_{k,l}\cdot \big(\sqrt{\lambda_{l}(\ma)}+\sqrt{\lambda_{k}(\mpp)}\big)^{-1}.
	\end{aligned}
	\end{equation}
	We define $\mh$ as a $d\times d$ matrix whose entries are $\mh_{k,\ell}=\big(\sqrt{\lambda_{\ell}(\ma)}+\sqrt{\lambda_{k}(\mpp)}\big)^{-1}$.  Eq.~\eqref{eq:lambdaW-Wlambda} means
	$$
	\hat\mLambda^{1/2}\mw_\muu-\mw_\muu\mLambda^{1/2}
	=(\hat\mLambda\mw_\muu-\mw_\muu\mLambda)\circ \mh,
	$$
	where $\circ$ denotes the Hadamard matrix product, and by Eq.~\eqref{eq:lambdaWt+Wtlambda} and Lemma~\ref{lemma:||sinTheta(Uhat,U)||,Lambdahat} it follows that
	\begin{equation}\label{eq:lambda1/2W-Wlambda1/2}
		\begin{aligned}
		\|\hat\mLambda^{1/2}\mw_\muu-\mw_\muu\mLambda^{1/2}\|
	    &\leq d^{1/2}\|\mh\|_{\max}\cdot \|\hat\mLambda\mw_\muu-\mw_\muu\mLambda\|\\
	    &\leq d^{1/2}(\sqrt{\lambda_d(\ma)}+\sqrt{\lambda_d(\mpp)})^{-1}
	    \cdot \|\mLambda\mw_\muu^{\top}
		-\mw_\muu^{\top}\hat\mLambda\|\\
	    &\lesssim \frac{d^{1/2}}{\lambda_{i,\min}^{1/2}}\Big(\frac{(\|\mpp^{(i)}\|_{\max}+\sigma_i)^2 n_i}{q_i\lambda_{i,\min}}
		+d^{1/2}q_i^{-1/2}(\|\mpp^{(i)}\|_{\max}+\sigma_i)\log^{1/2} n_i\Big)\\
		&\lesssim \frac{d^{1/2}(\|\mpp^{(i)}\|_{\max}+\sigma_i)^2 n_i}{q_i\lambda_{i,\min}^{3/2}}
		+\frac{d(\|\mpp^{(i)}\|_{\max}+\sigma_i)\log^{1/2} n_i}{q_i^{1/2}\lambda_{i,\min}^{1/2}}
	\end{aligned}
	\end{equation}
	with high probability.
	
	For $\breve\mw-\mw_\muu$, notice
	$$
    \begin{aligned}
    	&\breve\mw=\underset{\mo\in \mathcal{O}_d}{\arg\min}\|\hat\mx\mo-\mx\|_F%=\hat H(\hat\mx^\top\mx)
    	\text{ and }
    	\mw_\muu
    	%=\underset{\mo\in \mathbb{O}_d}{\arg\min}\|\hat\muu\mo-\muu\|_F
    	=\hat H(\mw_\muu \mLambda),
    \end{aligned}
    $$
    where $\hat H(\cdot)$ is a matrix-valued function, and for any $d\times d$ invertible matrix $\mc$, $\hat H(\mc)=\mc(\mc^\top\mc)^{-1/2}$. Then if $\|\hat\mx^\top\mx-\mw_\muu \mLambda\|\leq \lambda_{d}(\mw_\muu\mLambda)$, according to Theorem~1 in \cite{rencang} we have
	$$
	\begin{aligned}
		\|\breve\mw-\mw_\muu\|
		&\leq \frac{2\|\hat\mx^\top\mx-\mw_\muu \mLambda\|}{\sigma_{d}(\hat{\mx}^{\top} \mx) + \sigma_{d}(\mw_\muu\mLambda)} \leq \frac{2 \|\hat\mx^\top\mx-\mw_\muu \mLambda\|}{\lambda_{i,\min}}
	\end{aligned}
    $$
    % We know $\lambda_{d}(\mw_\muu\mLambda)=\lambda_{i,\min}$ and now
    We now bound $\|\hat\mx^\top\mx-\mw_\muu \mLambda\|$.
    By Lemma~\ref{lemma:||sinTheta(Uhat,U)||,Lambdahat} and Eq.~\eqref{eq:lambda1/2W-Wlambda1/2} we have
    $$
    \begin{aligned}
    	\|\hat\mx^\top\mx-\mw_\muu \mLambda\|
    	&=\|\hat\mLambda^{1/2}\hat\muu^\top\muu\mLambda^{1/2}-\mw_\muu \mLambda\|\\
    	&=\|\hat\mLambda^{1/2}(\hat\muu^\top\muu-\mw_\muu)\mLambda^{1/2}
    	+(\hat\mLambda^{1/2}\mw_\muu-\mw_\muu\mLambda^{1/2})\mLambda^{1/2}\|\\
    	&\leq \|\hat\mLambda\|^{1/2}\cdot \|\hat\muu^\top\muu-\mw_\muu\| \cdot \|\mLambda\|^{1/2}
    	+\|\hat\mLambda^{1/2}\mw_\muu-\mw_\muu\mLambda^{1/2}\|\cdot \|\mLambda\|^{1/2}\\
    	&\lesssim \lambda_{i,\max}\cdot \frac{(\|\mpp^{(i)}\|_{\max}+\sigma_i)^2 n_i}{q_i\lambda_{i,\min}^2}
    	+\Big(\frac{d^{1/2}(\|\mpp^{(i)}\|_{\max}+\sigma_i)^2 n_i}{q_i\lambda_{i,\min}^{3/2}}
		+\frac{d(\|\mpp^{(i)}\|_{\max}+\sigma_i)\log^{1/2} n_i}{q_i^{1/2}\lambda_{i,\min}^{1/2}}\Big)\cdot \lambda_{i,\max}^{1/2}\\
		%&\lesssim \frac{\tau_i \sigma_i^2 n_i}{\lambda_{i,\min}}
		%+\frac{\tau_i^{1/2}\sigma_i^2 n_i}{\lambda_{i,\min}}
		%+d\tau_i^{1/2}\sigma_i\log^{1/2} n_i
		&\lesssim \frac{d^{1/2}(\|\mpp^{(i)}\|_{\max}+\sigma_i)^2 n_i}{q_i\lambda_{i,\min}}
		+dq_i^{-1/2}(\|\mpp^{(i)}\|_{\max}+\sigma_i)\log^{1/2} n_i 
    \end{aligned}
    $$
    with high probability. Now by Eq.~\eqref{eq:cond_lemmaA1}
%implies  $\frac{(\|\mpp^{(i)}\|_{\max}+\sigma_i)n_i^{1/2}}{q_i^{1/2}\lambda_{i,\min}} \ll 1$ and hence
we have
$\|\hat\mx^\top\mx-\mw_\muu \mLambda\|\leq \lambda_{i,\min}$ with high probability. 
%    {\color{black} under the conditions that  $\frac{(\|\mpp^{(i)}\|_{\max}+\sigma_i)n_i^{1/2}}{q_i^{1/2}\lambda_{i,\min}} \ll 1$ is implied by Eq.~\eqref{eq:con_1}.}
In summary we have  %  Then the condition of Lemma~36 in \cite{ma2018implicit} is satisfied and we have
	$$
	\begin{aligned}
		\|\breve\mw-\mw_\muu\|
		&\leq \frac{2\|\hat\mx^\top\mx-\mw_\muu \mLambda\|}{\lambda_{i,\min}}
		% &\lesssim \Big(\frac{d^{1/2}(\|\mpp^{(i)}\|_{\max}+\sigma_i)^2 n_i}{q_i\lambda_{i,\min}}
		% +dq_i^{-1/2}(\|\mpp^{(i)}\|_{\max}+\sigma_i)\log^{1/2} n_i\Big)\cdot\frac{1}{\lambda_{i,\min}}\\
		\lesssim \frac{d^{1/2}(\|\mpp^{(i)}\|_{\max}+\sigma_i)^2 n_i}{q_i\lambda_{i,\min}^2}
		+\frac{d(\|\mpp^{(i)}\|_{\max}+\sigma_i)\log^{1/2} n_i}{q_i^{1/2}\lambda_{i,\min}}
	\end{aligned}
    $$
    with high probability.
\end{proof}

\subsection{Proof of Lemma~\ref{lemma:hat X(i)W(i)-X}}
\label{App: Lemma:hat X(i)W(i)-X}

%Assume the SVD of $\mpp^{(i)}$ and $\ma^{(i)}$ as 
%    $$
%    \begin{aligned}
%    	\mpp^{(i)}=\muu^{(i)}\mLambda^{(i)}\muu^{(i)\top},
%    	\text{ and }
%    	\ma^{(i)}=\hat\muu^{(i)}\hat\mLambda^{(i)}\hat\muu^{(i)\top} + \hat\muu^{(i)}_\perp\hat\mLambda^{(i)}_\perp\hat\muu^{(i)\top}_\perp.\\
%    \end{aligned}
%    $$
We first state a lemma for bounding the 
%The result for the error of $\hat\mx^{(i)}$ as an estimate of $\mx_{\mathcal{U}_i}$ is based on the following result for 
error of $\hat\muu^{(i)}$ as an estimate of $\muu^{(i)}$; see Section~\ref{App: lemma:hat U(i)W(i)-U(i)} for a proof.
\begin{lemma} 
\label{lemma:hat U(i)W(i)-U(i)}
    {\em Consider the setting of Lemma~\ref{lemma:hat X(i)W(i)-X}.
    Define $\mw^{(i)}_\muu=\underset{\mo\in \mathcal{O}_d}{\arg\min}\|\hat\muu^{(i)}\mo-\muu^{(i)}\|_F$.
    We then have
	\begin{equation}\label{eq:hatUW-U=...}
		\hat\muu^{(i)}\mw^{(i)}_\muu-\muu^{(i)}=\me^{(i)}\muu^{(i)}(\mLambda^{(i)})^{-1}+\mr_\muu^{(i)},
	\end{equation}
	where $\me^{(i)}=\ma^{(i)}-\mpp^{(i)}$ and $\mr_\muu^{(i)}$ is a $n_i\times d$ random matrix satisfying
	$$
	\|\mr_\muu^{(i)}\|
	\lesssim \frac{ (\|\mpp^{(i)}\|_{\max}+\sigma_i)^2 n_i}{q_i\lambda_{i,\min}^2}
		+\frac{(\|\mpp^{(i)}\|_{\max}+\sigma_i)\log^{1/2}n_i}{q_i^{1/2}\lambda_{i,\min}}
	$$
    with high probability. Furthermore, suppose
	\[\frac{(\|\mpp^{(i)}\|_{\max}+\sigma_i) n_i^{1/2}\log^{1/2}n_i}{q_i^{1/2}\lambda_{i,\min}}\ll1.\]
 We then have
	$$\|\mr_\muu^{(i)}\|_{2\to\infty}
	\lesssim \frac{(\|\mpp^{(i)}\|_{\max}+\sigma_i)^2 n_i^{1/2}\log n_i}{q_i\lambda_{i,\min}^2}\
		+\frac{(\|\mpp^{(i)}\|_{\max}+\sigma_i)\log^{1/2} n_i}{q_i^{1/2}n_i^{1/2}\lambda_{i,\min}}
	$$
	with high probability. }
\end{lemma}
%The proof of Lemma~\ref{lemma:hat U(i)W(i)-U(i)} is in Section~\ref{App: lemma:hat U(i)W(i)-U(i)}.
%Now we bound the error of $\hat\mx^{(i)}$ as an estimate of $\mx_{\mathcal{U}_i}$.
Now recall that
    $$
        \mx_{\mathcal{U}_i}=\muu_{\mathcal{U}_i}\mLambda^{1/2}, \quad \mx^{(i)}=\muu^{(i)}(\mLambda^{(i)})^{1/2}
        \text{ and }
        \hat\mx^{(i)}=\hat\muu^{(i)}(\hat\mLambda^{(i)})^{1/2}.
    $$
    As $\mx_{\mathcal{U}_i}\mx_{\mathcal{U}_i}^\top=\mpp^{(i)}=\mx^{(i)}\mx^{(i)\top}$, there exists an orthogonal $\tilde\mw^{(i)}$ such that $\mx_{\mathcal{U}_i}=\mx^{(i)}\tilde\mw^{(i)}$.
    Define 
    \begin{equation}
    	\begin{aligned}
    	\breve\mw^{(i)}=\underset{\mo\in \mathcal{O}_d}{\arg\min}\|\hat\mx^{(i)}\mo-\mx^{(i)}\|_F,
    \end{aligned}
    \end{equation}
    and recall
    $$
    \begin{aligned}
    	\mw^{(i)}=\underset{\mo\in \mathcal{O}_d}{\arg\min}\|\hat\mx^{(i)}\mo-\mx_{\mathcal{U}_i}\|_F \text{ and }\mw^{(i)}_\muu=\underset{\mo\in \mathcal{O}_d}{\arg\min}\|\hat\muu^{(i)}\mo-\muu^{(i)}\|_F.
    \end{aligned}
    $$
    Note that $\mw^{(i)}=\breve\mw^{(i)}\tilde\mw^{(i)}$.
    %For $\hat\mx^{(i)}\mw^{(i)}-\mx_{\mathcal{U}_i}$, 
    Next, by Eq.~\eqref{eq:hatUW-U=...} we have
    $$
    \begin{aligned}
    	\hat\mx^{(i)}\mw^{(i)}-\mx_{\mathcal{U}_i}
    	%=&\big(\hat\mx^{(i)}\breve\mw^{(i)}\tilde\mw^{(i)}-\mx_{\mathcal{U}_i}\big)\tilde\mw^{(i)\top}\\
    	=&[\hat\mx^{(i)}\breve\mw^{(i)}-\mx^{(i)}]\tilde\mw^{(i)}\\
    	%=&\Big[\hat\muu^{(i)}(\hat\mLambda^{(i)})^{1/2}(\breve\mw^{(i)}-\mw^{(i)}_\muu)
    	%+\hat\muu^{(i)}[(\hat\mLambda^{(i)})^{1/2}\mw^{(i)}_\muu-\mw^{(i)}_\muu(\mLambda^{(i)})^{1/2}]\\
    	%&+(\hat\muu^{(i)}\mw^{(i)}_\muu-\muu^{(i)})(\mLambda^{(i)})^{1/2}\Big]\tilde\mw^{(i)}\\
    	=&\Big[[(\hat\muu^{(i)}\mw_\muu^{(i)}-\muu^{(i)})+\muu^{(i)}](\mw_\muu^{(i)})^\top(\hat\mLambda^{(i)})^{1/2}(\breve\mw^{(i)}-\mw^{(i)}_\muu)
    	\\&+[(\hat\muu^{(i)}\mw_\muu^{(i)}-\muu^{(i)})+\muu^{(i)}](\mw_\muu^{(i)})^\top[(\hat\mLambda^{(i)})^{1/2}\mw^{(i)}_\muu-\mw^{(i)}_\muu(\mLambda^{(i)})^{1/2}]\\
    	&+(\hat\muu^{(i)}\mw^{(i)}_\muu-\muu^{(i)})(\mLambda^{(i)})^{1/2}\Big]\tilde\mw^{(i)}\\
    	=&\me^{(i)}\muu^{(i)}(\mLambda^{(i)})^{-1/2}\tilde\mw^{(i)}
    	+\mr^{(i)} \\
    	=&\me^{(i)}\mx^{(i)}(\mx^{(i)\top}\mx^{(i)})^{-1}\tilde\mw^{(i)}
    	+\mr^{(i)} \\
    	=&\me^{(i)}\mx_{\mathcal{U}_i}(\mx_{\mathcal{U}_i}^\top\mx_{\mathcal{U}_i})^{-1}
    	+\mr^{(i)},
    \end{aligned}
    $$
    where the remainder term $\mr^{(i)}$ is
    $$
    \begin{aligned}
    	\mr^{(i)}
    	=&\Big[\mr^{(i)}_\muu(\mLambda^{(i)})^{1/2}
    	%+\hat\muu^{(i)}(\hat\mLambda^{(i)})^{1/2}(\breve\mw^{(i)}-\mw^{(i)}_\muu)
    	%+\hat\muu^{(i)}[(\hat\mLambda^{(i)})^{1/2}\mw^{(i)}_\muu-\mw^{(i)}_\muu(\mLambda^{(i)})^{1/2}]
    	+[(\hat\muu^{(i)}\mw_\muu^{(i)}-\muu^{(i)})+\muu^{(i)}](\mw_\muu^{(i)})^\top(\hat\mLambda^{(i)})^{1/2}(\breve\mw^{(i)}-\mw^{(i)}_\muu)
    	\\&+[(\hat\muu^{(i)}\mw_\muu^{(i)}-\muu^{(i)})+\muu^{(i)}](\mw_\muu^{(i)})^\top[(\hat\mLambda^{(i)})^{1/2}\mw^{(i)}_\muu-\mw^{(i)}_\muu(\mLambda^{(i)})^{1/2}]
    	\Big]\tilde\mw^{(i)}.
    \end{aligned}
    $$
    By Lemma~\ref{lemma:hat U(i)W(i)-U(i)}, Lemma~\ref{lemma:||sinTheta(Uhat,U)||,Lambdahat} and Lemma~\ref{lemma:||Lambda UT Uhat-UT Uhat Lambdahat||, ||Lambda WT-WT Lambdahat||, ||Lambdahat1/2 W-W Lambda1/2||} we obtain
    $$
    \begin{aligned}
    	\|\mr^{(i)}\|
    	&\lesssim \|\mr^{(i)}_\muu\|\cdot \|\mLambda^{(i)}\|^{1/2}
    	+[\|\hat\muu^{(i)}\mw_\muu^{(i)}-\muu^{(i)}\|+\|\muu^{(i)}\|]\cdot\|\hat\mLambda^{(i)}\|^{1/2}\cdot\|\breve\mw^{(i)}-\mw^{(i)}_\muu\|
    	\\&+[\|\hat\muu^{(i)}\mw_\muu^{(i)}-\muu^{(i)}\|+\|\muu^{(i)}\|]\cdot\|(\hat\mLambda^{(i)})^{1/2}\mw^{(i)}_\muu-\mw^{(i)}_\muu(\mLambda^{(i)})^{1/2}\|\\
    	&\lesssim \Big(\frac{ (\|\mpp^{(i)}\|_{\max}+\sigma_i)^2 n_i}{q_i\lambda_{i,\min}^2}
		+\frac{(\|\mpp^{(i)}\|_{\max}+\sigma_i)\log^{1/2}n_i}{q_i^{1/2}\lambda_{i,\min}}\Big)
		\cdot\lambda_{i,\max}^{1/2}\\
		&+\lambda_{i,\max}^{1/2}\cdot 
		\Big(\frac{(\|\mpp^{(i)}\|_{\max}+\sigma_i)^2 n_i}{q_i\lambda_{i,\min}^2}
		+\frac{(\|\mpp^{(i)}\|_{\max}+\sigma_i)\log^{1/2} n_i}{q_i^{1/2}\lambda_{i,\min}}\Big)	\\
		&+\frac{(\|\mpp^{(i)}\|_{\max}+\sigma_i)^2 n_i}{q_i\lambda_{i,\min}^{3/2}}
		+\frac{(\|\mpp^{(i)}\|_{\max}+\sigma_i)\log^{1/2} n_i}{q_i^{1/2}\lambda_{i,\min}^{1/2}}\\
		&\lesssim \frac{(\|\mpp^{(i)}\|_{\max}+\sigma_i)^2 n_i}{q_i\lambda_{i,\min}^{3/2}}
		+\frac{(\|\mpp^{(i)}\|_{\max}+\sigma_i)\log^{1/2} n_i}{q_i^{1/2}\lambda_{i,\min}^{1/2}}
    \end{aligned}
    $$
    with high probability and 
    $$
    \begin{aligned}
    	\|\mr^{(i)}\|_{2\to\infty}
    	&\lesssim \|\mr^{(i)}_\muu\|_{2\to\infty}\cdot \|\mLambda^{(i)}\|^{1/2}
    	+[\|\hat\muu^{(i)}\mw_\muu^{(i)}-\muu^{(i)}\|_{2\to\infty}+\|\muu^{(i)}\|_{2\to\infty}]\cdot \|\hat\mLambda^{(i)}\|^{1/2}\cdot\|\breve\mw^{(i)}-\mw^{(i)}_\muu\|\\
    	&+[\|\hat\muu^{(i)}\mw_\muu^{(i)}-\muu^{(i)}\|_{2\to\infty}+\|\muu^{(i)}\|_{2\to\infty}]\cdot\|(\hat\mLambda^{(i)})^{1/2}\mw^{(i)}_\muu-\mw^{(i)}_\muu(\mLambda^{(i)})^{1/2}\|\\
    	&\lesssim \Big(\frac{(\|\mpp^{(i)}\|_{\max}+\sigma_i)^2 n_i^{1/2}\log N}{q_i\lambda_{i,\min}^2}\
		+\frac{(\|\mpp^{(i)}\|_{\max}+\sigma_i)\log^{1/2} n_i}{q_i^{1/2}n_i^{1/2}\lambda_{i,\min}}\Big)
		\cdot \lambda_{i,\max}^{1/2}\\
		&+ \frac{1}{n_i^{1/2}}\cdot
		\lambda_{i,\max}^{1/2}\cdot 
		\Big(\frac{(\|\mpp^{(i)}\|_{\max}+\sigma_i)^2 n_i}{q_i\lambda_{i,\min}^2}
		+\frac{(\|\mpp^{(i)}\|_{\max}+\sigma_i)\log^{1/2} n_i}{q_i^{1/2}\lambda_{i,\min}}\Big)\\
		&+\frac{1}{n_i^{1/2}}\cdot
		\Big(\frac{(\|\mpp^{(i)}\|_{\max}+\sigma_i)^2 n_i}{q_i\lambda_{i,\min}^{3/2}}
		+\frac{(\|\mpp^{(i)}\|_{\max}+\sigma_i)\log^{1/2} n_i}{q_i^{1/2}\lambda_{i,\min}^{1/2}}\Big)\\
		&\lesssim \frac{(\|\mpp^{(i)}\|_{\max}+\sigma_i)^2 n_i^{1/2}\log n_i}{q_i\lambda_{i,\min}^{3/2}}
		+\frac{ (\|\mpp^{(i)}\|_{\max}+\sigma_i)\log^{1/2} n_i}{q_i^{1/2}n_i^{1/2}\lambda_{i,\min}^{1/2}}
    \end{aligned}
    $$
    with high probability.
%\hspace*{\fill} \qedsymbol

\subsection{Proof of Lemma~\ref{lemma:hat U(i)W(i)-U(i)}}
\label{App: lemma:hat U(i)W(i)-U(i)}

   % This lemma is about one of the blocks.
    For ease of exposition we fix a value of $i$ and drop this index from our matrices.
	First note that
	$$
	\begin{aligned}
		\hat\muu
		&=\ma\hat\muu\hat\mLambda^{-1}
		%=+\me) \hat\muu\hat\mLambda^{-1}
		%+\me\hat\muu\hat\mLambda^{-1}
		=(\muu\mLambda\muu^\top + \me) \hat\muu\hat\mLambda^{-1}
		%+\me\hat\muu\hat\mLambda^{-1}\\
		%&=\muu\muu^\top\hat\muu
		%+\muu\mLambda(\muu^\top\hat\muu\hat\mLambda^{-1}-\mLambda^{-1}\muu^\top\hat\muu)
		%+\me\hat\muu\hat\mLambda^{-1}\\
		=\muu\muu^\top\hat\muu
		+\muu(\mLambda\muu^\top\hat\muu-\muu^\top\hat\muu\hat\mLambda)\hat\mLambda^{-1}
		+\me\hat\muu\hat\mLambda^{-1}.
	\end{aligned}
	$$
	Hence for any $d\times d$ orthogonal matrix $\mw$, we have
	$$
	\begin{aligned}
		\hat\muu\mw-\muu
		&=\me\muu\mLambda^{-1}
		+\underbrace{\muu(\muu^\top\hat\muu-\mw^\top)\mw}_{\mr_{\muu,1}}
		+\underbrace{\muu(\mLambda\muu^\top\hat\muu-\muu^\top\hat\muu\hat\mLambda)\hat\mLambda^{-1}\mw}_{\mr_{\muu,2}}\\
		&+\underbrace{\me\muu\mLambda^{-1}(\mLambda\mw^\top-\mw^\top\hat\mLambda)\hat\mLambda^{-1}\mw}_{\mr_{\muu,3}}
		+\underbrace{\me(\hat\muu\mw-\muu)\mw^\top\hat\mLambda^{-1}\mw}_{\mr_{\muu,4}}.
	\end{aligned}
	$$
    Let $\mw_\muu$ be the minimizer of $\|\hat\muu\mo-\muu\|_F$ over all $d\times d$ orthogonal matrices $\mo$. We now bound the spectral norms of $\mr_{\muu,1},\dots,\mr_{\muu,4}$ when $\mw = \mw_{\muu}$.
	
	For $\mr_{\muu,1}$, by Lemma~\ref{lemma:||sinTheta(Uhat,U)||,Lambdahat} we have
	\begin{equation*}
		\begin{aligned}
		&\|\mr_{\muu,1}\|
		\leq \|\muu^{\top}\hat\muu-\mw^{\top}_\muu\|
		%\leq \frac{\|\me\|^2}{\lambda_d^2(\mpp)}
		\lesssim \frac{ (\|\mpp^{(i)}\|_{\max}+\sigma_i)^2 n_i}{q_i\lambda_{i,\min}^2},\\
		&\|\mr_{\muu,1}\|_{2\to\infty}
		\leq \|\muu\|_{2\to\infty}
		\cdot\|\muu^{\top}\hat\muu-\mw^{\top}_\muu\|
		%\leq \|\muu\|_{2\to\infty} 
		%\cdot \frac{\|\me\|^2}{\lambda_d^2(\mpp)}
		\lesssim \frac{(\|\mpp^{(i)}\|_{\max}+\sigma_i)^2n_i^{1/2} }{q_i\lambda_{i,\min}^2}
	\end{aligned}
	\end{equation*}
	with high probability.
	
	For $\mr_{\muu,2}$, by Lemma~\ref{lemma:||sinTheta(Uhat,U)||,Lambdahat} and Lemma~\ref{lemma:||Lambda UT Uhat-UT Uhat Lambdahat||, ||Lambda WT-WT Lambdahat||, ||Lambdahat1/2 W-W Lambda1/2||} we have
	$$
    \begin{aligned}
    	\|\mr_{\muu,2}\|
    	\leq& \|\mLambda\muu^{\top}\hat\muu-\muu^{\top}\hat\muu\hat\mLambda^{(i)}\|
    	\cdot \|\hat\mLambda^{-1}\|\\
    	\lesssim& \Big(\frac{ (\|\mpp^{(i)}\|_{\max}+\sigma_i)^2 n_i}{q_i\lambda_{i,\min}}
		+q_i^{-1/2}(\|\mpp^{(i)}\|_{\max}+\sigma_i)\log^{1/2} n_i\Big)\lambda_{i,\min}^{-1}\\
		\lesssim &\frac{ (\|\mpp^{(i)}\|_{\max}+\sigma_i)^2 n_i}{q_i\lambda_{i,\min}^2}
		+\frac{(\|\mpp^{(i)}\|_{\max}+\sigma_i)\log^{1/2} n_i}{q_i^{1/2}\lambda_{i,\min}},\\
		\|\mr_{\muu,2}\|_{2\to\infty}
		&\leq \|\muu\|_{2\to\infty}
		\cdot\|\mLambda\muu^{\top}\hat\muu-\muu^{\top}\hat\muu\hat\mLambda^{(i)}\|
    	\cdot \|\hat\mLambda^{-1}\|\\
		&\lesssim \frac{ (\|\mpp^{(i)}\|_{\max}+\sigma_i)^2 n_i^{1/2}}{q_i\lambda_{i,\min}^2}
		+\frac{(\|\mpp^{(i)}\|_{\max}+\sigma_i)\log^{1/2} n_i}{q_i^{1/2}n_i^{1/2}\lambda_{i,\min}}
	\end{aligned}
	$$
	with high probability.

	For $\mr_{\muu,3}$, by Lemma~\ref{lemma:|E|,|EU|,|UT EU|}, Lemma~\ref{lemma:||sinTheta(Uhat,U)||,Lambdahat} and Lemma~\ref{lemma:||Lambda UT Uhat-UT Uhat Lambdahat||, ||Lambda WT-WT Lambdahat||, ||Lambdahat1/2 W-W Lambda1/2||} we have
	$$
	\begin{aligned}
		\|\mr_{\muu,3}\|
		&\leq \|\me\|\cdot\|\mLambda^{-1}\|\cdot\|\mLambda\mw_{\muu}^\top-\mw^\top_{\muu}\hat\mLambda\|\cdot\|\hat\mLambda^{-1}\|\\
		&\lesssim q_i^{-1/2}(\|\mpp^{(i)}\|_{\max}+\sigma_i)n_i^{1/2}
		\cdot \lambda_{i,\min}^{-2}
		\cdot \Big(\frac{ (\|\mpp^{(i)}\|_{\max}+\sigma_i)^2 n_i}{q_i\lambda_{i,\min}}
		+\frac{(\|\mpp^{(i)}\|_{\max}+\sigma_i)\log^{1/2} n_i}{q_i^{1/2}}\Big)
		\\
		&\lesssim \frac{(\|\mpp^{(i)}\|_{\max}+\sigma_i)^3 n_i^{3/2}}{q_i^{3/2}\lambda_{i,\min}^3}
		+\frac{(\|\mpp^{(i)}\|_{\max}+\sigma_i)^2n_i^{1/2}\log^{1/2} n_i}{q_i\lambda_{i,\min}^2},\\
		\|\mr_{\muu,3}\|_{2\to\infty}
		&\leq \|\me\muu\|_{2\to\infty}\cdot\|\mLambda^{-1}\|\cdot\|\mLambda\mw^\top_{\muu}-\mw^\top_{\muu}\hat\mLambda\|\cdot\|\hat\mLambda^{-1}\|\\
		&\lesssim q_i^{-1/2}(\|\mpp^{(i)}\|_{\max}+\sigma_i)\log^{1/2} n_i
		\cdot \lambda_{i,\min}^{-2}
		\cdot \Big(\frac{ (\|\mpp^{(i)}\|_{\max}+\sigma_i)^2 n_i}{q_i\lambda_{i,\min}}
		+\frac{(\|\mpp^{(i)}\|_{\max}+\sigma_i)\log^{1/2} n_i}{q_i^{1/2}}\Big)
		\\
		&\lesssim \frac{(\|\mpp^{(i)}\|_{\max}+\sigma_i)^3 n_i\log^{1/2} n_i}{q_i^{3/2}\lambda_{i,\min}^3}
		+\frac{(\|\mpp^{(i)}\|_{\max}+\sigma_i)^2\log n_i}{q_i\lambda_{i,\min}^2}
	\end{aligned}
	$$
	with high probability.
	
	For $\mr_{\muu,4}$, by Lemma~\ref{lemma:|E|,|EU|,|UT EU|}, Lemma~\ref{lemma:||sinTheta(Uhat,U)||,Lambdahat} and Lemma~\ref{lemma:|| E(Uhat W -U) ||2toinf}, we have
	$$
	\begin{aligned}
		\|\mr_{\muu,4}\|
		&\leq \|\me\|
		\cdot \|\hat\muu\mw_\muu-\muu\|
		\cdot \|\hat\mLambda^{-1}\|\\
		&\lesssim q_i^{-1/2}(\|\mpp^{(i)}\|_{\max}+\sigma_i) n_i^{1/2}\cdot \frac{(\|\mpp^{(i)}\|_{\max}+\sigma_i) n_i^{1/2}}{q_i^{1/2}\lambda_{i,\min}}\cdot \frac{1}{\lambda_{i,\min}}\\
		&\lesssim \frac{(\|\mpp^{(i)}\|_{\max}+\sigma_i)^2 n_i}{q_i\lambda_{i,\min}^2},\\
	    \|\mr_{\muu,4}\|_{2\to\infty}
	    &\leq \|\me(\hat\muu\mw_\muu-\muu)\|_{2\to\infty}
	    \cdot \|\hat\mLambda^{-1}\|\\
		&\lesssim \frac{(\|\mpp^{(i)}\|_{\max}+\sigma_i)^2 n_i^{1/2}\log n_i}{q_i\lambda_{i,\min}^2}
	\end{aligned}
	$$
	with high probability. 
	Combining the above bounds we obtain
	$$
	\begin{aligned}
		\Big\|\sum_{r=1}^4\mr_{\muu,r}\Big\|
		%&\leq \|\mr_{\muu_1}\| +\|\mr_{\muu,2}\| +\|\mr_{\muu,3}\| +\|\mr_{\muu,4}\|\\
		&\lesssim \frac{ (\|\mpp^{(i)}\|_{\max}+\sigma_i)^2 n_i}{q_i\lambda_{i,\min}^2}\\
		&+\frac{(\|\mpp^{(i)}\|_{\max}+\sigma_i)^2 n_i}{q_i\lambda_{i,\min}^2}
		+\frac{(\|\mpp^{(i)}\|_{\max}+\sigma_i)\log^{1/2} n_i}{q_i^{1/2}\lambda_{i,\min}}\\
		&+\frac{(\|\mpp^{(i)}\|_{\max}+\sigma_i)^3 n_i^{3/2}}{q_i^{3/2}\lambda_{i,\min}^3}
		+\frac{(\|\mpp^{(i)}\|_{\max}+\sigma_i)^2n_i^{1/2}\log^{1/2} n_i}{q_i\lambda_{i,\min}^2}\\
		&+\frac{(\|\mpp^{(i)}\|_{\max}+\sigma_i)^2 n_i}{q_i\lambda_{i,\min}^2}\\
		&\lesssim \frac{ (\|\mpp^{(i)}\|_{\max}+\sigma_i)^2 n_i}{q_i\lambda_{i,\min}^2}
		+\frac{(\|\mpp^{(i)}\|_{\max}+\sigma_i)\log^{1/2}n_i}{q_i^{1/2}\lambda_{i,\min}},\\
		\Big\|\sum_{r=1}^4\mr_{\muu,r}\Big\|_{2\to\infty}
		%&\leq \|\mr_{\muu,1}\|_{2\to\infty} +\|\mr_{\muu,2}\|_{2\to\infty} +\|\mr_{\muu,3}\|_{2\to\infty} +\|\mr_{\muu,4}\|_{2\to\infty}\\
		&\lesssim \frac{(\|\mpp^{(i)}\|_{\max}+\sigma_i)^2n_i^{1/2} }{q_i\lambda_{i,\min}^2}\\
		&+\frac{ (\|\mpp^{(i)}\|_{\max}+\sigma_i)^2 n_i^{1/2}}{q_i\lambda_{i,\min}^2}
		+\frac{(\|\mpp^{(i)}\|_{\max}+\sigma_i)\log^{1/2} n_i}{q_i^{1/2}n_i^{1/2}\lambda_{i,\min}}\\
		&+\frac{(\|\mpp^{(i)}\|_{\max}+\sigma_i)^3 n_i \log^{1/2}n_i}{q_i^{3/2}\lambda_{i,\min}^3}
		+\frac{(\|\mpp^{(i)}\|_{\max}+\sigma_i)^2\log n_i}{q_i\lambda_{i,\min}^2}\\
		&+\frac{(\|\mpp^{(i)}\|_{\max}+\sigma_i)^2 n_i^{1/2}\log n_i}{q_i\lambda_{i,\min}^2}\\
		&\lesssim \frac{(\|\mpp^{(i)}\|_{\max}+\sigma_i)^2 n_i^{1/2}\log n_i}{q_i\lambda_{i,\min}^2}\
		+\frac{(\|\mpp^{(i)}\|_{\max}+\sigma_i)\log^{1/2} n_i}{q_i^{1/2}n_i^{1/2}\lambda_{i,\min}}
	\end{aligned}
	$$
	with high probability
	as {\color{black} $\frac{(\|\mpp^{(i)}\|_{\max}+\sigma_i)n_i^{1/2}\log^{1/2} n_i}{q_i^{1/2}\lambda_{i,\min}}\lesssim 1$ as implied by Eq.~\eqref{eq:cond2_lemmaA1}.}
	%, provided that 
	%The proof is completed by defining $\mr_{\muu}=\mr_{\muu,1}+\mr_{\muu,2}+\mr_{\muu,3}+\mr_{\muu,4}$.
%\hspace*{\fill} \qedsymbol

\subsection{Technical lemmas for $\mr_{\muu,4}$ in Lemma~\ref{lemma:hat U(i)W(i)-U(i)}}
Our bound for $\mr_{\muu,4}$ in the above proof of Lemma~\ref{lemma:hat U(i)W(i)-U(i)} is based on a series of technical lemmas which culminate in a high-probability bound for $\|\me^{(i)}(\hat\muu^{(i)}\mw^{(i)}_\muu-\muu^{(i)})\|_{2\to\infty}$. These lemmas are derived using an adaptation of the leave-one-out analysis presented in Theorem~3.2 of \cite{xie2021entrywise}; the noise model for $\mn^{(i)}$ in the current paper is, however, somewhat different from that of \cite{xie2021entrywise} and thus we chose to provide self-contained proofs of these lemmas here. 
Once again, for ease of exposition we will fix a value of $i$ and thereby drop this index from our matrices in this section. 

We introduce some notations.
For $\ma$ whose entries are independent Bernoulli random variables sampled according to $\mpp$, 
we define the following collection of auxiliary matrices $\ma^{[1]},\dots,\ma^{[n_i]}$ generated from $\ma$.
For each row index $h\in[n_i]$, the matrix $\ma^{[h]}=(\ma^{[h]}_{s,t})_{n_i\times n_i}$ is obtained by replacing the entries in the $h$th row of $\ma$ with their expected values, i.e.,
$$
\mathbf{A}_{s,t}^{[h]}= 
\begin{cases}
\mathbf{A}_{s,t}, & \text { if } s \neq h \text{ and } t \neq h, \\ 
\mathbf{P}_{s,t}, & \text { if } s=h \text{ or } s=h.
\end{cases}
$$
Denote the singular decompositions of $\ma$ and $\ma^{[h]}$ as
$$
\begin{aligned}
	&\ma=\hat\muu\hat\mLambda\hat\muu^{\top}+\hat\muu_\perp\hat\mLambda_\perp\hat\muu^{\top}_\perp,\\
	\ma^{[h]}&=\hat\muu^{[h]}\hat\mLambda^{[h]}\hat\muu^{[h]\top}+\hat\muu^{[h]}_\perp\hat\mLambda^{[h]}_\perp\hat\muu^{[h]\top}_\perp.
\end{aligned}
$$

\begin{lemma}
\label{lemma:||Uhat|| and ||Uhat[h]||}
	{\em Consider the setting in Lemma~\ref{lemma:hat X(i)W(i)-X}, we then have
	$$
	    \|\hat\muu^{(i)}\|_{2 \to \infty}
	    \lesssim \frac{ d^{1/2}}{n_i^{1/2}},\quad
	    \|(\hat\muu^{(i)})^{[h]}\|_{2\to\infty}
	    \lesssim \frac{d^{1/2}}{n_i^{1/2}}
	$$
	with high probability. Furthermore, let $\mw^{[h]}$ be the solution of orthogonal Procrustes problem between $\hat\muu^{[h]}$ and $\muu$. We then have
    $$
    \begin{aligned}
    	\|(\hat\muu^{(i)})^{[h]}(\mw^{(i)})^{[h]}-\muu^{(i)}\|_{2\to\infty}
    	\lesssim \frac{d^{1/2}(\|\mpp^{(i)}\|_{\max}+\sigma_i)\log^{1/2} n_i}{q_i^{1/2}\lambda_{i,\min}}
    \end{aligned}
    $$
    with high probability.}
\end{lemma}

\begin{proof}

The proof is based on verifying the conditions in Theorem~2.1 of \cite{abbe2020entrywise}, and this can be done by following the exact same derivations as that in Lemma~12 of \cite{abbe2020entrywise}.
More specifically, $\Delta^*$ in \cite{abbe2020entrywise} corresponds to $\lambda_{i,\min}$ in our paper while $\kappa$ in \cite{abbe2020entrywise} corresponds to $M$ in our setting, where $M$ appears in the assumptions of Lemma~\ref{lemma:hat X(i)W(i)-X} and is bounded,
and $\gamma$ in \cite{abbe2020entrywise} can be set to be $\frac{c(\|\mpp^{(i)}\|_{\max}+\sigma_i)n_i^{1/2}}{q_i^{1/2}\lambda_{i,\min}}$ for some sufficiently large constant $c>0$, based on the bound of $\|\me\|$ from Lemma~\ref{lemma:|E|,|EU|,|UT EU|}. Then all desired results of Lemma~\ref{lemma:||Uhat|| and ||Uhat[h]||} can by obtained {\color{black}under the conditions that $q_in_i\gg \log n_i $ and condition in Eq.~\eqref{eq:cond2_lemmaA1}.
%$ \frac{%\kappa
%    (\|\mpp^{(i)}\|_{\max}+\sigma_i)n_i^{1/2}\log^{1/2}n_i}{q_i^{1/2}\lambda_{i,\min}}\ll 1$, the latter of which is implied by Eq.~\eqref{eq:con_1}.
    }
\end{proof}

\begin{lemma}
\label{lemma:ehC}
{\em 
Consider the setting in Lemma~\ref{lemma:hat X(i)W(i)-X}.
Recall in Lemma~\ref{lemma:|E|,|EU|,|UT EU|}, we decompose $\me^{(i)}$ as $\me^{(i,1)}+\me^{(i,2)}$.
Let $\mathbf{e}^{(i)}_{h}$ denote the $h$th row of $\me^{(i)}$.
Then for any deterministic $n_i\times d$ matrix $\mc$, we have
$$
\begin{aligned}
	\|\mathbf{e}^{(i)}_{h} \mc\|
&\lesssim q_i^{-1/2}(\|\mpp^{(i)}\|_{\max}+\sigma_i)n_i^{1/2}\|\mc\|_{2\to\infty}\log^{1/2}n_i
	+q_i^{-1}(\|\mpp^{(i)}\|_{\max}+\sigma_i\log^{1/2} n_i)\|\mc\|_{2\to\infty}\log n_i
\end{aligned}
$$
with high probability. }
\end{lemma}
\begin{proof}
For any $h\in[n_i]$, let $\mathbf{e}^{(i,1)}_{h}$ and $\mathbf{e}^{(i,2)}_{h}$ denote the $h$th row of $\me^{(i,1)}$ and $\me^{(i,2)}$. Following the same arguments as that for $\|\me^{(i,1)}\muu^{(i)}\|_{2\to\infty}$ and $\|\me^{(i,2)}\muu^{(i)}\|_{2\to\infty}$ in the proof of Lemma~\ref{lemma:|E|,|EU|,|UT EU|}, we have
	$$
\begin{aligned}
	&\mathbb{P}\Big\{\|(\mathbf{e}^{(i,1)}_{h})^\top \mc\|
	\geq t\Big\}
	\leq  (d+1)\exp\Big(\frac{-t^2/2}{q_i^{-1}\|\mpp^{(i)}\|_{\max}^2\max\{\|\mc\|^2,n_i\|\mc\|_{2\to\infty}^2\}+q_i^{-1}\|\mpp^{(i)}\|_{\max}\|\mc\|_{2\to\infty}t/3}\Big),\\
	&\mathbb{P}\Big\{\|(\mathbf{e}^{(i,2)}_{h})^\top \mc \|
	\geq t\Big\}
	\leq  (d+1)\exp\Big(\frac{-t^2/2}{2q_i^{-1}\sigma_i^2\max\{\|\mc\|^2,n_i\|\mc\|_{2\to\infty}^2\}+q_i^{-1}\sigma_i\log^{1/2} N\|\mc\|_{2\to\infty}t/3}\Big).
	\end{aligned}
$$	
We therefore have
$$
\begin{aligned}
	&\|(\mathbf{e}^{(i,1)}_{h})^\top \mc\|
	\lesssim q_i^{-1/2}\|\mpp^{(i)}\|_{\max}n_i^{1/2}\|\mc\|_{2\to\infty}\log^{1/2}n_i
	+q_i^{-1}\|\mpp^{(i)}\|_{\max}\|\mc\|_{2\to\infty}\log n_i,\\
	&\|(\mathbf{e}^{(i,2)}_{h})^\top \mc\|
	\lesssim q_i^{-1/2}\sigma_in_i^{1/2}\|\mc\|_{2\to\infty}\log^{1/2}n_i+q_i^{-1}\sigma_i\|\mc\|_{2\to\infty}\log^{3/2} n_i
	\end{aligned}
$$	
with high probability. 
Combining the above bounds yields the desired claim. 
%Then the desired results of $\|\me^{(i)} \mc\|_{2\to\infty}$ are followed.
\end{proof}

\begin{lemma}
\label{lemma:||eh E Uhat||, ||sin Theta(Uhat[h],Uhat)||}
{\em 
	Consider the setting in Lemma~\ref{lemma:hat U(i)W(i)-U(i)},	 we then have
	$$
	\begin{aligned}
		&\|\sin \Theta((\hat\muu^{(i)})^{[h]},\hat\muu^{(i)})\|
		\lesssim \frac{d^{1/2}(\|\mpp^{(i)}\|_{\max}+\sigma_i) \log^{1/2}n_i}{q_i^{1/2}\lambda_{i,\min}}
	\end{aligned}
	$$
	with high probability.}
\end{lemma}

\begin{proof}
    %From Lemma~\ref{lemma:|E|,|EU|,|UT EU|}, we have $\|\me\|\lesssim q_i^{-1/2}(\|\mpp^{(i)}\|_{\max}+\sigma_i) n_i^{1/2}$ with high probability.
    By the construction of $\ma^{[h]}$ and Lemma~\ref{lemma:|E|,|EU|,|UT EU|} we have
	$$
	\|\ma^{[h]}-\ma\|
	\leq 2\|\mathbf{e}_h\|
	%\leq \Big(\sum_{s\in[n_i]}\me_{s,h}^2\Big)^{1/2}+\Big(\sum_{t\in[n_i]}\me_{h,t}^2\Big)^{1/2}
	\leq 2\|\me\|_{2\to\infty}
	\leq 2\|\me\|
	\lesssim q_i^{-1/2}(\|\mpp^{(i)}\|_{\max}+\sigma_i) n_i^{1/2}
	$$
	with high probability, and hence
	\begin{equation}\label{eq:Ah-P_1}
		\begin{aligned}
			\|\ma^{[h]}-\mpp\|
	\leq \|\ma^{[h]}-\ma\|+\|\me\|
	\lesssim q_i^{-1/2}(\|\mpp^{(i)}\|_{\max}+\sigma_i) n_i^{1/2}
		\end{aligned}
	\end{equation}
	with high probability.
    %&With the similar proof with Lemma~\ref{lemma:|E|,|EU|,|UT EU|}, we also have $\|\ma^{[h]}-\mpp\|\lesssim \sigma_i n_i^{1/2}$ with high probability.
    Then by Weyl's inequality we have
    \begin{equation}\label{eq:Ah-P_2}
    	\begin{aligned}
    	%&|\lambda_{d+1}(\ma)|
    	%=|\lambda_{d+1}(\ma)-\lambda_{d+1}(\mpp)|
    	%\leq \|\me\|
    	%\lesssim \sigma_i n_i^{1/2},\\
    	&|\lambda_{d}(\ma^{[h]})-\lambda_{i,\min}|
    	\leq \|\ma^{[h]}-\mpp\|
    	\lesssim q_i^{-1/2}(\|\mpp^{(i)}\|_{\max}+\sigma_i) n_i^{1/2}
    \end{aligned}
    \end{equation}
    with high probability. The condition in Eq.~\eqref{eq:cond2_lemmaA1} implies
    $\frac{(\|\mpp^{(i)}\|_{\max}+\sigma_i) n_i^{1/2}}{q_i^{1/2}\lambda_{i,\min}}=o(1)$ and hence
     %{\color{black}as $\frac{(\|\mpp^{(i)}\|_{\max}+\sigma_i) n_i^{1/2}}{q_i^{1/2}\lambda_{i,\min}}\lesssim 1$ as implied by Eq.~\eqref{eq:con_1}}, 
    %we have 
    $\lambda_{d}(\ma^{[h]})\asymp \lambda_{i,\min}$.
    Furthermore, by Lemma~\ref{lemma:||sinTheta(Uhat,U)||,Lambdahat} we have $\lambda_{d+1}(\ma)\lesssim q_i^{-1/2}(\|\mpp^{(i)}\|_{\max}+\sigma_i) n_i^{1/2}$ with high probability. Applying
	Wedin's $\sin\Theta$ Theorem (see e.g., Theorem~4.4 of \cite{stewart_sun}) we have	
	\begin{equation}
 \label{eq:wedin_lemD7}
	\begin{aligned}
		\|\sin \Theta(\hat\muu^{[h]},\hat\muu)\|
		\leq \frac{\|(\ma^{[h]}-\ma)\hat\muu^{[h]}\|}{\lambda_{d}(\ma^{[h]})-\lambda_{d+1}(\ma)}
		\lesssim \frac{\|(\ma^{[h]}-\ma)\hat\muu^{[h]}\|_F}{\lambda_{i,\min}}
	\end{aligned}
    \end{equation}
	with high probability.
		We now bound $\|(\ma^{[h]}-\ma)\hat\muu^{[h]}\|_F$. Note that
	\begin{equation}\label{eq:Ah-A hatU}
		\begin{aligned}
		\|(\ma^{[h]}-\ma)\hat\muu^{[h]}\|_F
		&= \Big[\sum_{s\in[n_i],s\neq h} \sum_{r\in[d]} (\me_{s,h}\hat\muu^{[h]}_{h,r})^2
		+\Big \|\mathbf{e}_{h}^\top\hat\muu^{[h]}\Big\|^2\Big]^{1/2},%\\
		%&\leq \Big[\|\me\|_{2\to\infty}^2\cdot\|\hat\muu^{[h]}\|_{2\to\infty}^2
		%+ \|\mathbf{e}_{h}^\top\hat\muu^{[h]}\|^2\Big]^{1/2}.
	\end{aligned}
	\end{equation}
	where $\mathbf{e}_h$ is the $h$th row of $\me$.
	%where $\hat\muu^{[h]}_k$ represents the $k$th row of $\hat\muu^{[h]}$.
	For the first term on the right side of Eq.~\eqref{eq:Ah-A hatU}, by Cauchy-Schwarz inequality, Lemma~\ref{lemma:||Uhat|| and ||Uhat[h]||} and Lemma~\ref{lemma:|E|,|EU|,|UT EU|} we have
	\begin{equation}\label{eq:Ah-A hatU_1}
		\begin{aligned}
	\big[\sum_{s\in[n_i],s\neq h} \sum_{r\in[d]} (\me_{s,h}\hat\muu^{[h]}_{h,r})^2\big]^{1/2}
	    &\leq \|\me\|_{2\to\infty}\cdot\|\hat\muu^{[h]}\|_{2\to\infty}\\
		&\leq \|\me\|\cdot\|\hat\muu^{[h]}\|_{2\to\infty}%\\
		%&%\lesssim q_i^{-1/2}(\|\mpp^{(i)}\|_{\max}+\sigma_i) n_i^{1/2} \cdot d^{1/2} n_i^{-1/2}
		\lesssim d^{1/2}q_i^{-1/2}(\|\mpp^{(i)}\|_{\max}+\sigma_i)
	\end{aligned}
	\end{equation}
	with high probability.
	For the second term on the right side of Eq.~\eqref{eq:Ah-A hatU}, 
	    as $\mathbf{e}_{h}$ and $\hat\muu^{[h]}$ are independent, we have by Lemma~\ref{lemma:ehC}, Lemma~\ref{lemma:||Uhat|| and ||Uhat[h]||}, and the assumption $n_i q_i \gtrsim \log^2n_i$  that
	\begin{equation}\label{eq:Ah-A hatU_2}
		\begin{aligned}
		\|\mathbf{e}_{h}^\top\hat\muu^{[h]}\|
		&\lesssim q_i^{-1/2}(\|\mpp^{(i)}\|_{\max}+\sigma_i)n_i^{1/2}\|\hat\muu^{[h]}\|_{2\to\infty}\log^{1/2}n_i%\\ &
 +q_i^{-1}(\|\mpp^{(i)}\|_{\max}+\sigma_i\log^{1/2} n_i)\|\hat\muu^{[h]}\|_{2\to\infty}\log n_i\\
	&\lesssim q_i^{-1/2}(\|\mpp^{(i)}\|_{\max}+\sigma_i)n_i^{1/2}\cdot d^{1/2}n_i^{-1/2}\cdot\log^{1/2}N%\\ &
	+q_i^{-1}(\|\mpp^{(i)}\|_{\max}+\sigma_i\log^{1/2} N)\cdot d^{1/2}n_i^{-1/2}\cdot \log n_i\\
	&\lesssim d^{1/2}q_i^{-1/2}(\|\mpp^{(i)}\|_{\max}+\sigma_i) \log^{1/2}n_i
	\end{aligned}
	\end{equation}
	with high probability. % {\color{black}by $q_i n_i\gtrsim \log^2 N$}.
	Combining Eq.~\eqref{eq:Ah-A hatU}, Eq.~\eqref{eq:Ah-A hatU_1} and Eq.~\eqref{eq:Ah-A hatU_2}, we have
	\begin{equation}
 \label{eq:bound_Ah-A_lemD7}
		\begin{aligned}
			\|(\ma^{[h]}-\ma)\hat\muu^{[h]}\|_F
			\lesssim d^{1/2}q_i^{-1/2}(\|\mpp^{(i)}\|_{\max}+\sigma_i) \log^{1/2}n_i
		\end{aligned}
	\end{equation}
	with high probability. Substituting Eq.~\eqref{eq:bound_Ah-A_lemD7} into Eq.~\eqref{eq:wedin_lemD7}
	yields the desired claim. 
\end{proof}

\begin{lemma}
\label{lemma:||eh E (Uhat[h] Uhat[h]T U-U)||}{\em 
	Consider the setting in Lemma~\ref{lemma:hat U(i)W(i)-U(i)}, we then have
	$$
	\begin{aligned}
		\|(\mathbf{e}_h^{(i)})^\top [(\hat\muu^{(i)})^{[h]}(\hat\muu^{(i)})^{[h]\top}\muu^{(i)}-\muu^{(i)}]\|
		\lesssim \frac{d^{1/2}(\|\mpp^{(i)}\|_{\max}+\sigma_i)^2n_i^{1/2}\log n_i}{q_i\lambda_{i,\min}}
	\end{aligned}
	$$
	with high probability. }
\end{lemma}

\begin{proof}
  Eq.~\eqref{eq:Ah-P_1} and Eq.~\eqref{eq:Ah-P_2} implies
  $\|\ma^{[h]}-\mpp\|\lesssim q_i^{-1/2}(\|\mpp^{(i)}\|_{\max}+\sigma_i) n_i^{1/2}$ and $\lambda_d(\ma^{[h]})\asymp \lambda_{i,\min}$ with high probability.
  Then by Wedin's $\sin\Theta$ Theorem (see e.g., Theorem~4.4 in \cite{stewart_sun})
we have
    $$
    \begin{aligned}
    	\|\sin\Theta(\hat\muu^{[h]},\muu)\|
		\leq \frac{\|\ma^{[h]}-\mpp\|}{\lambda_d(\ma^{[h]})-\lambda_{d+1}(\mpp)}
		\lesssim \frac{(\|\mpp^{(i)}\|_{\max}+\sigma_i) n_i^{1/2}}{q_i^{1/2}\lambda_{i,\min}}
    \end{aligned}
    $$
    with high probability.
    Let $\mw^{[h]}$ be the orthogonal Procrustes alignment between $\hat\muu^{[h]}$ and $\muu$. Then 
    \begin{equation}\label{eq:hatUtU-Wh}
    	\begin{aligned}
    	\|\hat\muu^{[h]\top}\muu-\mw^{[h]}\|
    	\leq \|\sin\Theta(\hat\muu^{[h]},\muu)\|^2
    	\lesssim \frac{(\|\mpp^{(i)}\|_{\max}+\sigma_i)^2 n_i}{q_i\lambda_{i,\min}^2},\\
    	%&\|\hat\muu^{[h]}\mw^{[h]}-\muu\|
    	%\leq \|\sin\Theta(\hat\muu^{[h]},\muu)\|+\|\hat\muu^{[h]\top}\muu-\mw^{[h]}\|
    	%    	\lesssim \frac{(\|\mpp^{(i)}\|_{\max}+\sigma_i) n_i^{1/2}}{q_i^{1/2}\lambda_{i,\min}}
    \end{aligned}
    \end{equation}
    with high probability.
    
    Now let $\mz^{[h]}=\hat\muu^{[h]}\hat\muu^{[h]\top}\muu-\muu$. By Eq.~\eqref{eq:hatUtU-Wh} and Lemma~\ref{lemma:||Uhat|| and ||Uhat[h]||}
    we obtain
    $$
    \begin{aligned}
    	\|\mz^{[h]}\|_{2\to\infty}
    	&\leq \|\hat\muu^{[h]}\hat\muu^{[h]\top}\muu-\hat\muu^{[h]}\mw^{[h]}\|_{2\to\infty}
    	+\|\hat\muu^{[h]}\mw^{[h]}-\muu\|_{2\to\infty}\\
    	&\leq \|\hat\muu^{[h]}\|_{2\to\infty}\cdot\|\hat\muu^{[h]\top}\muu-\mw^{[h]}\|
    	+\|\hat\muu^{[h]}\mw^{[h]}-\muu\|_{2\to\infty}\\
    	&\lesssim \frac{d^{1/2}}{n_i^{1/2}}\cdot\frac{(\|\mpp^{(i)}\|_{\max}+\sigma_i)^2 n_i}{q_i\lambda_{i,\min}^2}
    	+\frac{d^{1/2}(\|\mpp^{(i)}\|_{\max}+\sigma_i)\log^{1/2} n_i}{q_i^{1/2}\lambda_{i,\min}}\\
    	&\lesssim \frac{d^{1/2}(\|\mpp^{(i)}\|_{\max}+\sigma_i)^2 n_i^{1/2}}{q_i\lambda_{i,\min}^2}
    	+\frac{d^{1/2}(\|\mpp^{(i)}\|_{\max}+\sigma_i)\log^{1/2} n_i}{q_i^{1/2}\lambda_{i,\min}}\\
    	&\lesssim \frac{d^{1/2}(\|\mpp^{(i)}\|_{\max}+\sigma_i)\log^{1/2} n_i}{q_i^{1/2}\lambda_{i,\min}},\\
    	%\|\mz^{[h]}\|
    	%&\leq \|\hat\muu^{[h]\top}\muu-\mw^{[h]}\|
    	%+\|\hat\muu^{[h]}\mw^{[h]}-\muu\|\\
    	%&\lesssim \frac{(\|\mpp^{(i)}\|_{\max}+\sigma_i) n_i^{1/2}}{q_i^{1/2}\lambda_{i,\min}}\\
    \end{aligned}
    $$
    with high probability, where the final inequality follows from the fact that, under the conditions in Eq.~\eqref{eq:cond2_lemmaA1} we have
    $\frac{(\|\mpp\|_{\max}+\sigma_i)n_i^{1/2}}{q^{1/2}\lambda_{i,\min}\log^{1/2} n_i}\lesssim 1$. 
    
    Finally, as $\mathbf{e}_{h}$ and $\mz^{[h]}$ are independent, by Lemma~\ref{lemma:ehC} and the assumption $ q_i n_i \gtrsim \log^2n_i$ we have
    $$
\begin{aligned}
	\|\mathbf{e}^{(i)}_{h} \mz^{[h]}\|
&\lesssim q_i^{-1/2}(\|\mpp^{(i)}\|_{\max}+\sigma_i)n_i^{1/2}\|\mz^{[h]}\|_{2\to\infty}\log^{1/2}n_i\\
	&+q_i^{-1}(\|\mpp^{(i)}\|_{\max}+\sigma_i\log^{1/2} n_i)\|\mz^{[h]}|_{2\to\infty}\log n_i\\
	&\lesssim q_i^{-1/2}(\|\mpp^{(i)}\|_{\max}+\sigma_i) n_i^{1/2}\cdot \frac{d^{1/2}(\|\mpp^{(i)}\|_{\max}+\sigma_i)\log^{1/2} n_i}{q_i^{1/2}\lambda_{i,\min}}\cdot \log^{1/2}n_i\\
	&+q_i^{-1}(\|\mpp^{(i)}\|_{\max}+\sigma_i\log^{1/2} n_i)\cdot \frac{d^{1/2}(\|\mpp^{(i)}\|_{\max}+\sigma_i)\log^{1/2} n_i}{q_i^{1/2}\lambda_{i,\min}}\cdot\log n_i\\
	&\lesssim  \frac{d^{1/2}(\|\mpp^{(i)}\|_{\max}+\sigma_i)^2n_i^{1/2}\log n_i}
	{q_i\lambda_{i,\min}}
	+\frac{d^{1/2}(\|\mpp^{(i)}\|_{\max}+\sigma_i)^2\log^{2} n_i}{q_i^{3/2}\lambda_{i,\min}}\\
	&\lesssim  \frac{d^{1/2}(\|\mpp^{(i)}\|_{\max}+\sigma_i)^2n_i^{1/2}\log n_i}
	{q_i\lambda_{i,\min}}
\end{aligned}
$$
with high probability. %{\color{black} by $q_i n_i \gtrsim \log^2N$}.
\end{proof}

\begin{lemma}
\label{lemma:|| E(Uhat W -U) ||2toinf}
{\em Consider the setting in Lemma~\ref{lemma:hat U(i)W(i)-U(i)},	 we then have
	$$
	\begin{aligned}
		\|\me^{(i)}(\hat\muu^{(i)}\mw_\muu^{(i)}-\muu^{(i)})\|_{2\to\infty}
		\lesssim \frac{d^{1/2}(\|\mpp^{(i)}\|_{\max}+\sigma_i)^2 n_i^{1/2}\log n_i}{q_i\lambda_{i,\min}}
	\end{aligned}
	$$
	with high probability. }
\end{lemma}

\begin{proof}
	For each $h\in[n_i]$, let $\mathbf{e}_h$ denote the $h$th row of $\me$. Notice
	\begin{equation}\label{eq:eh(hat UW-U)}
		\begin{aligned}
		\mathbf{e}_h^\top(\hat\muu\mw_\muu-\muu)
		&=\mathbf{e}_h^\top (\hat\muu\mw_\muu-\muu)\mw_\muu^\top(\mw_\muu-\hat\muu^\top\muu)
		+\mathbf{e}_h^\top \muu\mw_\muu^\top(\mw_\muu-\hat\muu^\top\muu)\\
		&+\mathbf{e}_h^\top (\hat\muu\hat\muu^\top-\hat\muu^{[h]}\hat\muu^{[h]\top})\muu
		+\mathbf{e}_h^\top (\hat\muu^{[h]}\hat\muu^{[h]\top}\muu-\muu).
	\end{aligned}
	\end{equation}
	We now bound all terms on the right hand side of Eq.~\eqref{eq:eh(hat UW-U)}.
	For $\mathbf{e}_h^\top (\hat\muu\mw_\muu-\muu)\mw_\muu^\top(\mw_\muu-\hat\muu^\top\muu)$, by Lemma~\ref{lemma:||sinTheta(Uhat,U)||,Lambdahat} we have
	\begin{equation}\label{eq:eh(hat UW-U)_1}
		\begin{aligned}
		\|\mathbf{e}_h^\top (\hat\muu\mw_\muu-\muu)\mw_\muu^\top(\mw_\muu-\hat\muu^\top\muu)\|
		&\lesssim \|\mathbf{e}_h^\top(\hat\muu\mw_\muu-\muu)\|
		\cdot \|\mw_\muu-\hat\muu^\top\muu\|\\
		&\lesssim \|\mathbf{e}_h^\top (\hat\muu\mw_\muu-\muu)\|
		\cdot \frac{(\|\mpp^{(i)}\|_{\max}+\sigma_i)^2 n_i}{q_i\lambda_{i,\min}^2}\\
		&=O(\|\mathbf{e}_h^\top (\hat\muu\mw_\muu-\muu)\|)
	\end{aligned}
	\end{equation}
	with high probability {\color{black}as $\frac{(\|\mpp^{(i)}\|_{\max}+\sigma_i) n_i^{1/2}}{q_i^{1/2}\lambda_{i,\min}}\lesssim 1$ as implied by Eq.~\eqref{eq:cond2_lemmaA1}}.
	
	For $\mathbf{e}_h^\top \muu\mw_\muu^\top(\mw_\muu-\hat\muu^\top\muu)$, by Lemma~\ref{lemma:|E|,|EU|,|UT EU|} and Lemma~\ref{lemma:||sinTheta(Uhat,U)||,Lambdahat} we have
	\begin{equation}\label{eq:eh(hat UW-U)_2}
		\begin{aligned}
		\|\mathbf{e}_h^\top \muu\mw_\muu^\top(\mw_\muu-\hat\muu^\top\muu)\|
		&\leq \|\me\muu\|_{2\to\infty}
		\cdot \|\mw_\muu-\hat\muu^\top\muu\|\\
		&\lesssim d^{1/2}q_i^{-1/2}(\|\mpp^{(i)}\|_{\max}+\sigma_i)\log^{1/2}n_i
		\cdot \frac{(\|\mpp^{(i)}\|_{\max}+\sigma_i)^2 n_i}{q_i\lambda_{i,\min}^2}\\
		&\lesssim \frac{d^{1/2}(\|\mpp^{(i)}\|_{\max}+\sigma_i)^3 n_i\log^{1/2}n_i}{q_i^{3/2}\lambda_{i,\min}^2}
	\end{aligned}
	\end{equation}
	with high probability.
	
	For $\mathbf{e}_h^\top (\hat\muu\hat\muu^\top-\hat\muu^{[h]}\hat\muu^{[h]\top})\muu$, by Lemma~\ref{lemma:|E|,|EU|,|UT EU|} and Lemma~\ref{lemma:||eh E Uhat||, ||sin Theta(Uhat[h],Uhat)||} we have
	\begin{equation}\label{eq:eh(hat UW-U)_3}
		\begin{aligned}
		\|\mathbf{e}_h^\top (\hat\muu\hat\muu^\top-\hat\muu^{[h]}\hat\muu^{[h]\top})\muu\|
		&\leq \|\me\|
		\cdot 2\|\sin \Theta(\hat\muu^{[h]},\hat\muu)\|\\
		&\lesssim q_i^{-1/2}(\|\mpp^{(i)}\|_{\max}+\sigma_i) n_i^{1/2}\cdot \frac{d^{1/2}(\|\mpp^{(i)}\|_{\max}+\sigma_i) \log^{1/2}n_i}{q_i^{1/2}\lambda_{i,\min}}\\
		&\lesssim \frac{d^{1/2}(\|\mpp^{(i)}\|_{\max}+\sigma_i)^2 n_i^{1/2} \log^{1/2}n_i}{q_i\lambda_{i,\min}}
	\end{aligned}
	\end{equation}
	with high probability.
	
	For $\mathbf{e}_h^\top (\hat\muu^{[h]}\hat\muu^{[h]\top}\muu-\muu)$, by Lemma~\ref{lemma:||eh E (Uhat[h] Uhat[h]T U-U)||} we have
	\begin{equation}\label{eq:eh(hat UW-U)_4}
		\begin{aligned}
		\|\mathbf{e}_h^\top (\hat\muu^{[h]}\hat\muu^{[h]\top}\muu-\muu)\|
		\lesssim \frac{d^{1/2}(\|\mpp^{(i)}\|_{\max}+\sigma_i)^2n_i^{1/2}\log n_i}
	{q_i\lambda_{i,\min}}
	\end{aligned}
	\end{equation}
	with high probability.
	
	Combining Eq.~\eqref{eq:eh(hat UW-U)}, Eq.~\eqref{eq:eh(hat UW-U)_1}, \dots, Eq.~\eqref{eq:eh(hat UW-U)_4} we finally obtain
	$$
	\begin{aligned}
		\|\mathbf{e}_h^\top (\hat\muu\mw_\muu-\muu)\|
		%&\lesssim \frac{d^{1/2}\sigma_i^3 n_i \log^{1/2}n_i}{\lambda_d^2(\mpp)}
		%+\frac{d^{1/2}\sigma_i^2 n_i^{1/2} \log^{1/2}n_i}{\lambda_{d}(\mpp)}
		%+\frac{d^{1/2} \sigma_i^2 n_i^{1/2}\log^{1/2}n_i}{\lambda_d(\mpp)}\\
		&\lesssim \frac{d^{1/2}(\|\mpp^{(i)}\|_{\max}+\sigma_i)^2 n_i^{1/2}\log n_i}{q_i\lambda_{i,\min}}
	\end{aligned}
	$$
    with high probability {\color{black}as $\frac{(\|\mpp\|_{\max}+\sigma_i) n_i^{1/2}}{q_i^{1/2}\lambda_{i,\min}\log^{1/2} n_i} \lesssim 1$ as implied by Eq.~\eqref{eq:cond2_lemmaA1}}.
\end{proof}

\subsection{Proof of Lemma~\ref{lemma:|| W(i)T W(i,j) W(j)-I ||}}
\label{App: lemma:|| W(i)T W(i,j) W(j)-I ||}
%Define $\hat{H}(\mathbf{C}) = \mathbf{C} (\mathbf{C}^{\top} \mathbf{C})^{-1/2}$ for any invertible matrix $\mathbf{C}$.
%Note that
Recall that
\begin{gather*}
	\mw^{(i)\top}\mw^{(i,j)}\mw^{(j)}
	= \underset{\mo\in \mathcal{O}_d}{\arg\min}
	\|\hat \mx^{(i)}_{\langle\mathcal{U}_i\cap \mathcal{U}_j\rangle}\mw^{(i)}\mo-\hat\mx^{(j)}_{\langle\mathcal{U}_i\cap \mathcal{U}_j\rangle}\mw^{(j)}\|_F.
	%=\hat H\big(\mw^{(i)\top}(\hat\mx^{(i)}_{\langle\mathcal{U}_i\cap \mathcal{U}_j\rangle})^\top\hat\mx^{(j)}_{\langle\mathcal{U}_i\cap \mathcal{U}_j\rangle}\mw^{(j)}\big)\\ 
	%\mi
	%=\underset{\mo\in \mathbb{O}_d}{\arg\min}
	%\|\mx_{\mathcal{U}_i\cap \mathcal{U}_j}\mo-\mx_{\mathcal{U}_i\cap \mathcal{U}_j}\|_F
	%=\hat H\big(\mx_{\mathcal{U}_i\cap \mathcal{U}_j}^\top\mx_{\mathcal{U}_i\cap \mathcal{U}_j}\big),
	%=G\big(\mpp^{(i,j)}\big).
\end{gather*}
Denote \[\mathbf{F}:=\mw^{(i)\top}(\hat\mx^{(i)}_{\langle\mathcal{U}_i\cap \mathcal{U}_j\rangle})^\top\hat\mx^{(j)}_{\langle\mathcal{U}_i\cap \mathcal{U}_j\rangle}\mw^{(j)}
		-\mx_{\mathcal{U}_i\cap \mathcal{U}_j}^\top\mx_{\mathcal{U}_i\cap \mathcal{U}_j}.\]
We therefore have, by perturbation bounds for polar decompositions, that
\begin{equation}
 \label{eq:rencang} 
  \|\mw^{(i)\top}\mw^{(i,j)}\mw^{(j)} -\mi\| \leq \frac{2\|\mathbf{F}\|}{\sigma_{\min}(\mx_{\mathcal{U}_i\cap \mathcal{U}_j}^\top\mx_{\mathcal{U}_i\cap \mathcal{U}_j})}. 
\end{equation}
Indeed, we suppose $\mx_{\mathcal{U}_i\cap \mathcal{U}_j}^\top\mx_{\mathcal{U}_i\cap \mathcal{U}_j}$ is invertible in Theorem~\ref{thm:R(i,j)}.
Now suppose $\|\mathbf{F}\| <  \sigma_{\min}(\mx_{\mathcal{U}_i\cap \mathcal{U}_j}^\top\mx_{\mathcal{U}_i\cap \mathcal{U}_j})$. Then 
$(\hat\mx^{(i)}_{\langle\mathcal{U}_i\cap \mathcal{U}_j\rangle})^{\top}\hat \mx^{(j)}_{\langle\mathcal{U}_i\cap \mathcal{U}_j\rangle}$ is also invertible and hence, by Theorem~1 in \cite{rencang} we have
\[ \begin{split} \|\mw^{(i)\top}\mw^{(i,j)}\mw^{(j)} -\mi\| &\leq \frac{2\|\mathbf{F}\|}{\sigma_{\min}((\hat\mx^{(i)}_{\langle\mathcal{U}_i\cap \mathcal{U}_j\rangle})^{\top}\hat \mx^{(j)}_{\langle\mathcal{U}_i\cap \mathcal{U}_j\rangle}) + \sigma_{\min}(\mx_{\mathcal{U}_i\cap \mathcal{U}_j}^\top\mx_{\mathcal{U}_i\cap \mathcal{U}_j})} \leq \frac{2\|\mathbf{F}\|}{\sigma_{\min}(\mx_{\mathcal{U}_i\cap \mathcal{U}_j}^\top\mx_{\mathcal{U}_i\cap \mathcal{U}_j})}. \end{split} \]
Otherwise if $\|\mathbf{F}\| \geq  \sigma_{\min}(\mx_{\mathcal{U}_i\cap \mathcal{U}_j}^\top\mx_{\mathcal{U}_i\cap \mathcal{U}_j})$
then, as $\|\mw^{(i)\top}\mw^{(i,j)}\mw^{(j)} -\mi\| \leq 2$, Eq.~\eqref{eq:rencang} holds trivially. 
%\[ \|\mw^{(i)\top}\mw^{(i,j)}\mw^{(j)} -\mi\| \leq 2 \leq \frac{2\|\mathbf{F}\|}{\sigma_{\min}(\mx_{\mathcal{U}_i\cap \mathcal{U}_j}^\top\mx_{\mathcal{U}_i\cap \mathcal{U}_j})}. \]
\iffalse
    By Lemma~36 in \cite{ma2018implicit}, we have
    $$
	\begin{aligned}
		\|\mw^{(i)\top}\mw^{(i,j)}\mw^{(j)} -\mi\|
		&\leq \frac{\|\mw^{(i)\top}(\hat\mx^{(i)}_{\langle\mathcal{U}_i\cap \mathcal{U}_j\rangle})^\top\hat\mx^{(j)}_{\langle\mathcal{U}_i\cap \mathcal{U}_j\rangle}\mw^{(j)}
		-\mx_{\mathcal{U}_i\cap \mathcal{U}_j}^\top\mx_{\mathcal{U}_i\cap \mathcal{U}_j}\|}
		{\sigma_{d-1}(\mx_{\mathcal{U}_i\cap \mathcal{U}_j}^\top\mx_{\mathcal{U}_i\cap \mathcal{U}_j})
		+\sigma_{d}(\mx_{\mathcal{U}_i\cap \mathcal{U}_j}^\top\mx_{\mathcal{U}_i\cap \mathcal{U}_j})},
	\end{aligned}
    $$
provided that $\|\mw^{(i)\top}(\hat\mx^{(i)}_{\mathcal{U}_i\cap \mathcal{U}_j})^\top\hat\mx^{(j)}_{\mathcal{U}_i\cap \mathcal{U}_j}\mw^{(j)}
		-\mx_{\mathcal{U}_i\cap \mathcal{U}_j}^\top\mx_{\mathcal{U}_i\cap \mathcal{U}_j}\|\leq \lambda_{d}(\mx_{\mathcal{U}_i\cap \mathcal{U}_j}^\top\mx_{\mathcal{U}_i\cap \mathcal{U}_j})$.
              \fi
    
    We now bound $\|\mathbf{F}\|$. First note that
      \begin{equation*}
    	\begin{aligned}
    	%\mw^{(i)^\top}(\hat\mx^{(i)}_{\mathcal{U}_i\cap \mathcal{U}_j})^\top\hat\mx^{(j)}_{\mathcal{U}_i\cap \mathcal{U}_j}\mw^{(j)}
		%-\mx_{\mathcal{U}_i\cap \mathcal{U}_j}^\top\mx_{\mathcal{U}_i\cap \mathcal{U}_j}
		\mathbf{F}
		&= (\hat\mx^{(i)}_{\langle\mathcal{U}_i\cap \mathcal{U}_j\rangle}\mw^{(i)}-\mx_{\mathcal{U}_i\cap \mathcal{U}_j})^\top
		(\hat\mx^{(j)}_{\langle\mathcal{U}_i\cap \mathcal{U}_j\rangle}\mw^{(j)}-\mx_{\mathcal{U}_i\cap \mathcal{U}_j})
		\\ &+(\hat\mx^{(i)}_{\langle\mathcal{U}_i\cap \mathcal{U}_j\rangle}\mw^{(i)}-\mx_{\mathcal{U}_i\cap \mathcal{U}_j})^\top\mx_{\mathcal{U}_i\cap \mathcal{U}_j}%\\&
		+\mx_{\mathcal{U}_i\cap \mathcal{U}_j}^\top(\hat\mx^{(j)}_{\langle\mathcal{U}_i\cap \mathcal{U}_j\rangle}\mw^{(j)}-\mx_{\mathcal{U}_i\cap \mathcal{U}_j}).
    \end{aligned}
    \end{equation*}
    Next, by Lemma~\ref{lemma:hat X(i)W(i)-X}, we have
    \begin{equation*}
    \begin{split}
    	\hat\mx^{(i)}_{\langle\mathcal{U}_i\cap \mathcal{U}_j\rangle}\mw^{(i)}-\mx_{\mathcal{U}_i\cap \mathcal{U}_j}
    	&=\me^{(i)}_{\langle\mathcal{U}_i\cap \mathcal{U}_j\rangle}
    	\mx_{\mathcal{U}_i}
    	(\mx_{\mathcal{U}_i}^\top\mx_{\mathcal{U}_i})^{-1}
    	+\mr^{(i)}_{\langle\mathcal{U}_i\cap \mathcal{U}_j\rangle}%,
    	\\ &=\me^{(i)}_{\langle\mathcal{U}_i\cap \mathcal{U}_j\rangle}
    	\muu^{(i)}
    	(\mLambda^{(i)})^{-1/2}
    	\tilde\mw^{(i)}
    	+\mr^{(i)}_{\langle\mathcal{U}_i\cap \mathcal{U}_j\rangle},
    	%\\ &=\me^{(i)}_{\langle\mathcal{U}_i\cap \mathcal{U}_j\rangle}
    	%\muu^{(i)}
    	%(\mLambda^{(i)})^{-1/2}
    	%\tilde\mw^{(i)}
    	%+\mr^{(i)}_{\langle\mathcal{U}_i\cap \mathcal{U}_j\rangle},
    \end{split}
    \end{equation*}
    where $\me^{(i)}_{\langle\mathcal{U}_i\cap \mathcal{U}_j\rangle}$ and $\mr^{(i)}_{\langle\mathcal{U}_i\cap \mathcal{U}_j\rangle}$ contain the rows in $\me^{(i)}$ and $\mr^{(i)}$ corresponding to entities $\mathcal{U}_i\cap \mathcal{U}_j$, respectively. A similar expansion holds for
    	$\hat\mx^{(j)}_{\langle\mathcal{U}_i\cap \mathcal{U}_j\rangle}\mw^{(j)}-\mx_{\mathcal{U}_i\cap \mathcal{U}_j}$. 
    We therefore have
    $$
    \begin{aligned}
    \mathbf{F}
    	%&=\mw^{(i)^\top}(\hat\mx^{(i)}_{\mathcal{U}_i\cap \mathcal{U}_j})^\top\hat\mx^{(j)}_{\mathcal{U}_i\cap \mathcal{U}_j}\mw^{(j)}
		%-\mx_{\mathcal{U}_i\cap \mathcal{U}_j}^\top\mx_{\mathcal{U}_i\cap \mathcal{U}_j}\\
		%&=(\hat\mx^{(i)}_{\langle\mathcal{U}_i\cap \mathcal{U}_j\rangle}\mw^{(i)}-\mx_{\mathcal{U}_i\cap \mathcal{U}_j})^\top
		%(\hat\mx^{(j)}_{\langle\mathcal{U}_i\cap \mathcal{U}_j\rangle}\mw^{(j)}-\mx_{\mathcal{U}_i\cap \mathcal{U}_j})		\\
    	%&=(\tilde\mw^{(i)\top}
    	%(\mLambda^{(i)})^{-1/2}
    	%\muu^{(i)\top}
    	%\me^{(i)\top}_{\mathcal{U}_i\cap \mathcal{U}_j}
    	%+\mr^{(i)\top}_{\mathcal{U}_i\cap \mathcal{U}_j})
    	%(\me^{(j)}_{\mathcal{U}_i\cap \mathcal{U}_j}
    	%\muu^{(j)}
    	%(\mLambda^{(j)})^{-1/2}
    	%\tilde\mw^{(j)}
    	%+\mr^{(j)}_{\mathcal{U}_i\cap \mathcal{U}_j})\\
    	%&+(\tilde\mw^{(i)\top}
    	%(\mLambda^{(i)})^{-1/2}
    	%\muu^{(i)\top}
    	%\me^{(i)\top}_{\langle\mathcal{U}_i\cap \mathcal{U}_j\rangle}
    	%+\mr^{(i)\top}_{\langle\mathcal{U}_i\cap \mathcal{U}_j\rangle})
    	%\muu_{\mathcal{U}_i\cap \mathcal{U}_j}\mLambda^{1/2}\\
    	%&+\mLambda^{1/2}\muu_{\mathcal{U}_i\cap \mathcal{U}_j}^\top 	(\me^{(j)}_{\langle\mathcal{U}_i\cap \mathcal{U}_j\rangle}
    	%\muu^{(j)}
    	%(\mLambda^{(j)})^{-1/2}
    	%\tilde\mw^{(j)}
    	%+\mr^{(j)}_{\langle\mathcal{U}_i\cap \mathcal{U}_j\rangle})\\
    	&=\underbrace{(\hat\mx^{(i)}_{\langle\mathcal{U}_i\cap \mathcal{U}_j\rangle}\mw^{(i)}-\mx_{\mathcal{U}_i\cap \mathcal{U}_j})^\top
		(\hat\mx^{(j)}_{\langle\mathcal{U}_i\cap \mathcal{U}_j\rangle}\mw^{(j)}-\mx_{\mathcal{U}_i\cap \mathcal{U}_j})}_{\mathbf{F}_1}\\
    	%&\underbrace{
    	%\tilde\mw^{(i)\top}
    	%(\mLambda^{(i)})^{-1/2}
    	%\muu^{(i)\top}
    	%\me^{(i)\top}_{\mathcal{U}_i\cap \mathcal{U}_j}
    	%\me^{(j)}_{\mathcal{U}_i\cap \mathcal{U}_j}
    	%\muu^{(j)}
    	%(\mLambda^{(j)})^{-1/2}
    	%\tilde\mw^{(j)}
    	%}_{\mathbf{F}_1}
    	%+\underbrace{
    	%\mr^{(i)\top}_{\mathcal{U}_i\cap \mathcal{U}_j}
    	%\mr^{(j)}_{\mathcal{U}_i\cap \mathcal{U}_j}
    	%}_{\mathbf{F}_2}\\
    	%&+\underbrace{
    	%\tilde\mw^{(i)\top}
    	%(\mLambda^{(i)})^{-1/2}
    	%\muu^{(i)\top}
    	%\me^{(i)\top}_{\mathcal{U}_i\cap \mathcal{U}_j}
    	%\mr^{(j)}_{\mathcal{U}_i\cap \mathcal{U}_j}
    	%}_{\mathbf{F}_3}
    	%+\underbrace{
    	%\mr^{(i)\top}_{\mathcal{U}_i\cap \mathcal{U}_j}
    	%\me^{(j)}_{\mathcal{U}_i\cap \mathcal{U}_j}
    	%\muu^{(j)}
    	%(\mLambda^{(j)})^{-1/2}
    	%\tilde\mw^{(j)}
    	%}_{\mathbf{F}_4}\\
    	&+\underbrace{
    	\tilde\mw^{(i)\top}
    	(\mLambda^{(i)})^{-1/2}
    	\muu^{(i)\top}
   	\me^{(i)\top}_{\langle\mathcal{U}_i\cap \mathcal{U}_j\rangle}
    	\mx_{\mathcal{U}_i\cap \mathcal{U}_j}
    	}_{\mathbf{F}_2}
    	+\underbrace{
    	\mx_{\mathcal{U}_i\cap \mathcal{U}_j}^\top 	
    	\me^{(j)}_{\langle\mathcal{U}_i\cap \mathcal{U}_j\rangle}
    	\muu^{(j)}
    	(\mLambda^{(j)})^{-1/2}
    	\tilde\mw^{(j)}
    	}_{\mathbf{F}_3}\\
    	&+\underbrace{
    	\mr^{(i)\top}_{\langle\mathcal{U}_i\cap \mathcal{U}_j\rangle}
    	\mx_{\mathcal{U}_i\cap \mathcal{U}_j}
    	}_{\mathbf{F}_4}
    	+\underbrace{
    	\mx_{\mathcal{U}_i\cap \mathcal{U}_j}^\top 	
    	\mr^{(j)}_{\langle\mathcal{U}_i\cap \mathcal{U}_j\rangle}
    	}_{\mathbf{F}_5}.
    \end{aligned}
    $$
  %  We bound $\mathbf{F}_1,\dots,\mathbf{F}_5$, respectively.
    For $\mathbf{F}_1$, by Lemma~\ref{lemma:hat X(i)W(i)-X} %and Lemma~\ref{lemma:i->global} 
    we have
    %$$
    %\begin{aligned}
    %	\|\mathbf{F}_1\|
    %	&\leq \|(\mLambda^{(i)})^{-1/2}\|
    %	\cdot \|\muu^{(i)\top}
    %	\me^{(i)\top}_{\mathcal{U}_i\cap \mathcal{U}_j}
    %	\me^{(j)}_{\mathcal{U}_i\cap \mathcal{U}_j}
    %	\muu^{(j)}\|_F
    %	\cdot \|(\mLambda^{(j)})^{-1/2}\|\\
    %	&\lesssim \lambda_{i,\min}^{-1/2}
    %	\cdot d n_{ij}^{1/2}\sigma_i\sigma_j\log^{3/2} N
    %	\cdot \lambda_{j,\min}^{-1/2}
    %\end{aligned}
    %$$
    $$
    \begin{aligned}
    	\|\mathbf{F}_1\|
    	&\leq n_{i,j}^{1/2}\|\hat\mx^{(i)}\mw^{(i)}-\mx_{\mathcal{U}_i}\|_{2\to\infty}
    	\cdot n_{i,j}^{1/2}\|\hat\mx^{(j)}\mw^{(j)}-\mx_{\mathcal{U}_j}\|_{2\to\infty}\\
    	&\lesssim n_{i,j}
    	\cdot \frac{(\|\mpp^{(i)}\|_{\max}+\sigma_i)\log^{1/2}{n_i}}{q_i^{1/2}\lambda_{i,\min}^{1/2}}
    	\cdot \frac{(\|\mpp^{(j)}\|_{\max}+\sigma_j)\log^{1/2}{n_j}}{q_j^{1/2}\lambda_{j,\min}^{1/2}}\\
      & \lesssim \frac{n_{i,j} \gamma_i  \gamma_j }{(q_in_i  \mu_i)^{1/2} (q_jn_j  \mu_j)^{1/2}}
    	% &\lesssim \frac{n_{ij}\log(n_i + n_j)}{(\lambda_{i,\min} \lambda_{j,\min})^{1/2}}
    	% \cdot \frac{(\|\mpp^{(i)}\|_{\max}+\sigma_i)(\|\mpp^{(j)}\|_{\max}+\sigma_j)}{(n_i q_i)^{1/2}(n_j q_j)^{1/2}}
    \end{aligned}
    $$
    with high probability. %, where $\mu_i = \lambda_{i,\min}/n_i$ and $\mu_j = \lambda_{j,\min}/n_j$. 
    For $\mathbf{F}_2$, by Lemma~\ref{lemma:||E(i)_ij||,||UiT E(i)_ij Uj||F} %and Lemma~\ref{lemma:i->global} 
    we have
    $$
    \begin{aligned}
    	\|\mathbf{F}_2\|
    	&\leq \|(\bm{\Lambda}^{(i)})^{-1/2}\|
    	\cdot\|\muu^{(i){\top}}
    	\me^{(i)\top}_{\langle\mathcal{U}_i\cap \mathcal{U}_j\rangle}
    \mx_{\mathcal{U}_i\cap \mathcal{U}_j}\|
   \\
    	&\lesssim \lambda_{i,\min}^{-1/2} 
    	\cdot q_i^{-1/2} n_{i,j}^{1/2}\|\mx_{\mathcal{U}_i\cap \mathcal{U}_j}\|_{2\to\infty} (\|\mpp^{(i)}\|_{\max}+\sigma_i) \log^{1/2}{n_i} 
    	\lesssim \frac{n_{i,j}^{1/2}\|\mx_{\mathcal{U}_i\cap \mathcal{U}_j}\|_{2\to\infty}\gamma_i}{(q_i n_i)^{1/2}\mu_i^{1/2}}  
    \end{aligned}
    $$
    with high probability. The same argument also yields
    $$
    \begin{aligned}
    	\|\mathbf{F}_3\|
\lesssim \frac{n_{i,j}^{1/2}\|\mx_{\mathcal{U}_i\cap \mathcal{U}_j}\|_{2\to\infty}\gamma_j}{(q_j n_j)^{1/2}\mu_j^{1/2}} 
    \end{aligned}
    $$
    with high probability. For $\mathbf{F}_4$, once again by Lemma~\ref{lemma:||E(i)_ij||,||UiT E(i)_ij Uj||F} %and Lemma~\ref{lemma:i->global} 
    we have
    $$
    \begin{aligned}
    	\|\mathbf{F}_4\|
    	&\leq \|\mr^{(i)}_{\langle\mathcal{U}_i\cap \mathcal{U}_j\rangle}\|
    	\cdot \|\mx_{\mathcal{U}_i \cap \mathcal{U}_j}\|
    	\lesssim n_{i,j}^{1/2}\Big(\frac{(\|\mpp^{(i)}\|_{\max}+\sigma_i)^2n_i^{1/2} \log n_i}{q_i\lambda_{i,\min}^{3/2}}
    +\frac{(\|\mpp^{(i)}\|_{\max}+\sigma_i)\log^{1/2}{n_i}}{q_i^{1/2}n_i^{1/2}\lambda_{i,\min}^{1/2}}\Big)
    \cdot \vartheta_{i,j}^{1/2} \\
    	& \lesssim (n_{i,j} \vartheta_{i,j})^{1/2}\Big(\frac{\gamma_i^2 }{q_in_i  \mu_i^{3/2}}
    +\frac{\gamma_i}{q_i^{1/2} n_i \mu_i^{1/2}}\Big)
	%&\lesssim \Big(\frac{(\|\mpp\|_{\max}+\sigma)^2N^{3/2}n_i^{1/2}n_{i,j}^{1/2}\log N}{q\lambda_{\min}^{3/2}n_i^{3/2}}
   % +\frac{(\|\mpp^{(i)}\|_{\max}+\sigma_i)N^{1/2}n_{i,j}^{1/2}\log^{1/2}N}{q^{1/2}n_i^{1/2}\lambda_{\min}^{1/2}n_i^{1/2}}\Big)
	%	\cdot \frac{n_{i,j}^{1/2}\lambda_{\max}^{1/2}}{N^{1/2}}
	%\\
    % &\lesssim \frac{Nn_{i,j}\log N}{\lambda_{\min}}\cdot \frac{(\|\mpp^{(i)}\|_{\max}+\sigma_i)^2}{q_in_i}
    % +n_{i,j}\log^{1/2}N\cdot\frac{(\|\mpp^{(i)}\|_{\max}+\sigma_i)}{q_i^{1/2}n_i}
    \end{aligned}
    $$
    with high probability. The same argument also yields
    $$
    \begin{aligned}
    	\|\mathbf{F}_5\|
    	& \lesssim (n_{i,j} \vartheta_{i,j})^{1/2}\Big(\frac{\gamma_j^2 }{q_jn_j  \mu_j^{3/2}}
    +\frac{\gamma_j}{q_j^{1/2} n_j \mu_j^{1/2}}\Big)
    %     &\lesssim n_{ij}^{1/2} \Big(\frac{(\|\mpp^{(j)}\|_{\max}+\sigma_j)^2n_j^{1/2} \log{n_j}}{q_j\lambda_{j,\min}^{3/2}}
    % +\frac{(\|\mpp^{(j)}\|_{\max}+\sigma_j)\log^{1/2}{n_j}}{q_j^{1/2}n_j^{1/2}\lambda_{j,\min}^{1/2}}\Big)
    % \cdot \vartheta_{ij}^{1/2} \\
    % 	\lesssim \frac{Nn_{i,j}\log N}{\lambda_{\min}}\cdot \frac{(\|\mpp^{(j)}\|_{\max}+\sigma_j)^2}{q_jn_j}
    % +n_{i,j}\log^{1/2}N\cdot\frac{(\|\mpp^{(j)}\|_{\max}+\sigma_j)}{q_j^{1/2}n_j}
    \end{aligned}
    $$
    with high probability.
    Combining the above bounds for $\mathbf{F}_1,\dots,\mathbf{F}_5$ and simplifying, we obtain
    $$
    \begin{aligned}
    	\|\mathbf{F}\|
    	&\leq \|\mathbf{F}_1\|+\cdots+\|\mathbf{F}_5\| \\ 
   %  	&\lesssim \frac{Nn_{i,j}\log N}{\lambda_{\min}}
   %  	\cdot \frac{(\|\mpp^{(i)}\|_{\max}+\sigma_i)(\|\mpp^{(j)}\|_{\max}+\sigma_j)}{q_i^{1/2}n_i^{1/2}q_j^{1/2}n_j^{1/2}}
   %  	+n_{i,j}^{1/2}\log^{1/2}N \Big(\frac{\|\mpp^{(i)}\|_{\max}+\sigma_i}{q_i^{1/2}n_i^{1/2}}+\frac{\|\mpp^{(j)}\|_{\max}+\sigma_j}{q_j^{1/2}n_j^{1/2}}\Big)\\
   %  	&+\frac{Nn_{i,j}\log N}{\lambda_{\min}}\cdot \Big(\frac{(\|\mpp^{(i)}\|_{\max}+\sigma_i)^2}{q_in_i}
   %  	+\frac{(\|\mpp^{(j)}\|_{\max}+\sigma_j)^2}{q_jn_j}\Big)
   % +n_{i,j}\log^{1/2}N\cdot
   %  \Big(\frac{\|\mpp^{(i)}\|_{\max}+\sigma_i}{q_i^{1/2}n_i}
   %  +\frac{\|\mpp^{(j)}\|_{\max}+\sigma_j}{q_j^{1/2}n_j}\Big)\\
%     	&\lesssim d^{1/2} (n_i q_i \mu_i)^{-1/2} (\vartheta_{ij} \log n_i)^{1/2} \gamma_i  + d^{1/2} (n_j q_j \mu_j)^{-1/2} (\vartheta_{ij} \log n_j)^{1/2} \gamma_j  \\
% &+ (n_{ij} \vartheta_{ij})^{1/2}\Big(\frac{\gamma_i^2 \log n_i}{n_i q_i \mu_i^{3/2}}
%     +\frac{\gamma_i\log^{1/2}{n_i}}{n_i q_i^{1/2} \mu_i^{1/2}}\Bigr) + 
%            (n_{ij} \vartheta_{ij})^{1/2}\Bigl(\frac{\gamma_j^2 \log n_j}{n_j q_j \mu_j^{3/2}}
%   +\frac{\gamma_j\log^{1/2}{n_j}}{n_j q_j^{1/2} \mu_j^{1/2}}\Bigr) 
&\lesssim \frac{n_{i,j} \gamma_i  \gamma_j }{(q_in_i  \mu_i)^{1/2} (q_jn_j  \mu_j)^{1/2}}
+n_{i,j}^{1/2}\|\mx_{\mathcal{U}_i\cap \mathcal{U}_j}\|_{2\to\infty}\Big(\frac{\gamma_i}{(q_i n_i \mu_i)^{1/2}} +\frac{\gamma_j}{(q_j n_j\mu_j)^{1/2}}\Big)\\
&+(n_{i,j} \vartheta_{i,j})^{1/2}\Big(\frac{\gamma_i^2 }{q_in_i  \mu_i^{3/2}}
    +\frac{\gamma_i}{q_i^{1/2} n_i \mu_i^{1/2}}
    +\frac{\gamma_j^2 }{q_jn_j  \mu_j^{3/2}}
    +\frac{\gamma_j}{q_j^{1/2} n_j \mu_j^{1/2}}\Big)\\
    &\lesssim \frac{n_{i,j} \gamma_i  \gamma_j }{(q_in_i  \mu_i)^{1/2} (q_jn_j  \mu_j)^{1/2}}
    +n_{i,j}^{1/2}\|\mx_{\mathcal{U}_i\cap \mathcal{U}_j}\|_{2\to\infty}\Big(\frac{\gamma_i}{(q_i n_i\mu_i)^{1/2}} +\frac{\gamma_j}{(q_j n_j\mu_j)^{1/2}}\Big)
\\
      & +(n_{i,j} \vartheta_{i,j})^{1/2}\Big(\frac{\gamma_i^2 }{q_in_i  \mu_i^{3/2}}
    +\frac{\gamma_j^2 }{q_jn_j  \mu_j^{3/2}}\Big)
      %\lesssim (n_{ij} \vartheta_{ij})^{1/2}\Bigl(\frac{\gamma_i^2 \log n_i}{n_i q_i \mu_i^{3/2}} + 
%\frac{\gamma_j^2 \log n_j}{n_j q_j \mu_j^{3/2}}\Bigr) + (d \vartheta_{ij})^{1/2} \Bigl(\frac{\gamma_i \log^{1/2}{n_i}}{(n_i q_i \mu_i)^{1/2}} + \frac{\gamma_j \log^{1/2}{n_j}}{(n_j q_j \mu_j)^{1/2}}\Bigr)  % &\lesssim 
    % +\frac{(\|\mpp^{(j)}\|_{\max}+\sigma_j)^2}{q_jn_j}\Big)
    % +n_{i,j}^{1/2}\log^{1/2}N\Big(\frac{\|\mpp^{(i)}\|_{\max}+\sigma_i}{q_i^{1/2}n_i^{1/2}}+\frac{\|\mpp^{(j)}\|_{\max}+\sigma_j}{q_j^{1/2}n_j^{1/2}}\Big)
    \end{aligned}
    $$
    with high probability. %Note $\lambda_{d}(\mx_{\mathcal{U}_i\cap \mathcal{U}_j}^\top\mx_{\mathcal{U}_i\cap \mathcal{U}_j})=\lambda_{\min}(\mpp_{\mathcal{U}_i\cap\mathcal{U}_j,\mathcal{U}_i\cap\mathcal{U}_j})\gtrsim \frac{n_{i,j}}{N}\lambda_{\min}$. 
    Substituting the above bound for $\|\mathbf{F}\|$ into Eq.~\eqref{eq:rencang} yields the stated claim. % under the condition $\alpha_{i,j}\ll \theta_{i,j}$.

\subsection{Technical lemmas for Lemma~\ref{lemma:|| W(i)T W(i,j) W(j)-I ||}}

\begin{lemma}
\label{lemma:||E(i)_ij||,||UiT E(i)_ij Uj||F}
{\em     Consider the setting of Theorem~\ref{thm:R(i,j)}. We then have
	$$
	\begin{aligned}
	    &\|\mr^{(i)}_{\langle\mathcal{U}_i\cap \mathcal{U}_j\rangle}\|
	    \lesssim n_{i,j}^{1/2}\Bigl(\frac{(\|\mpp^{(i)}\|_{\max}+\sigma_i)^2n_i^{1/2} \log n_i }{q_i\lambda_{i,\min}^{3/2}}
    +\frac{(\|\mpp^{(i)}\|_{\max}+\sigma_i)\log^{1/2} n_i}{q_i^{1/2}n_i^{1/2}\lambda_{i,\min}^{1/2}}\Bigr),\\
        % 	&\|\muu_{\mathcal{U}_i\cap \mathcal{U}_j} \bm{\Lambda}^{1/2}\|
    	% \lesssim n_{ij}^{1/2} \|\mpp^{(i)}_{\mathcal{U}_{i} \cap \mathcal{U}_j}\|_{\max}^{1/2} \\
		&\|\muu^{(i)\top}
    	\me^{(i)\top}_{\langle\mathcal{U}_i\cap \mathcal{U}_j\rangle}\mx_{\mathcal{U}_i\cap \mathcal{U}_j}\|
    	\lesssim q_i^{-1/2} n_{i,j}^{1/2}\|\mx_{\mathcal{U}_i\cap \mathcal{U}_j}\|_{2\to\infty} (\|\mpp^{(i)}\|+\sigma_i) \log^{1/2}{n_i}
        %&\|
    	%\me^{(i)}_{\mathcal{U}_i\cap \mathcal{U}_j}\muu^{(i)}\|_{2\to\infty}
    	%\lesssim d^{1/2}\sigma_i\log^{1/2}n_i%\\
    	%&\|\muu^{(i)\top}
    	%\me^{(i)\top}_{\mathcal{U}_i\cap \mathcal{U}_j}
    	%\me^{(j)}_{\mathcal{U}_i\cap \mathcal{U}_j}
    	%\muu^{(j)}\|_F
    	%\lesssim d n_{ij}^{1/2}\sigma_i\sigma_j\log^{3/2} N
	\end{aligned}
	$$
	with high probability. %, where $\vartheta_{ij} = \|\mx_{\mathcal{U}_i \cap \mathcal{U}_j}\|^2$.
	}
\end{lemma}

\begin{proof}
    From Lemma~\ref{lemma:hat X(i)W(i)-X} we have
    $$
    \begin{aligned}
    \|\mr^{(i)}_{\langle\mathcal{U}_i\cap \mathcal{U}_j\rangle}\|
    \lesssim n_{i,j}^{1/2}\|\mr^{(i)}\|_{2\to\infty}
    \lesssim n_{i,j}^{1/2}\Bigl(\frac{(\|\mpp^{(i)}\|_{\max}+\sigma_i)^2n_i^{1/2} \log n_i}{q_i\lambda_{i,\min}^{3/2}}
    +\frac{(\|\mpp^{(i)}\|_{\max}+\sigma_i)\log^{1/2} n_i}{q_i^{1/2}n_i^{1/2}\lambda_{i,\min}^{1/2}}\Bigr)
    \end{aligned}
    $$
    with high probability. % Next, as $\|\mathbf{M}\|_{2 \to \infty} \leq \|\mathbf{M} \mathbf{M}^{\top}\|_{\max}^{1/2}$ for all matrices $\mathbf{M}$, we have 
%     $$
%     \begin{aligned}
%     	\|\muu_{\mathcal{U}_i\cap \mathcal{U}_j} \bm{\Lambda}^{1/2}\|
%     	\leq n_{ij}^{1/2} \|\muu_{\mathcal{U}_i\cap \mathcal{U}_j} \bm{\Lambda}^{1/2}\|_{2 \to \infty}
%       \leq n_{ij}^{1/2} \|\mpp^{(i)}_{\mathcal{U}_{i} \cap \mathcal{U}_j}\|_{\max}^{1/2}
%     	% =\|\sum_{s\in {\mathcal{U}_i\cap \mathcal{U}_j}}\mathbf{u}_s\mathbf{u}_s^\top\big\|^{1/2}
%     	% \leq n_{i,j}\|\muu\|_{2\to\infty}^2
% %    	\lesssim dn_{i,j}N^{-1}.
%     \end{aligned}
%     $$
    Furthermore, following the same derivations as that for $\|\muu^{(i)\top}
    	\me^{(i)}\muu^{(i)}\|$ in the proof of Lemma~\ref{lemma:|E|,|EU|,|UT EU|}, we have
%    Based on $\|\muu\|_{2\to\infty}\lesssim d^{1/2} N^{-1/2}$ and $\|\muu_{\mathcal{U}_i\cap \mathcal{U}_j}^\top\muu_{\mathcal{U}_i\cap \mathcal{U}_j}\|\lesssim dn_{i,j}N^{-1}$, by the identical proof of $\|\muu^{(i)\top}
%    	\me^{(i)\top}\muu^{(i)}\|$ in Lemma~\ref{lemma:|E|,|EU|,|UT EU|}, we have
	$$
\|\muu^{(i)\top}
    	\me^{(i)\top}_{\langle\mathcal{U}_i\cap \mathcal{U}_j\rangle}\mx_{\mathcal{U}_i\cap \mathcal{U}_j}\| 
    	%\lesssim q_i^{-1/2} (d^{1/2}\vartheta_{ij}^{1/2}+n_{i,j}^{1/2}\|\mx_{\mathcal{U}_i\cap \mathcal{U}_j}\|_{2\to\infty}) (\|\mpp^{(i)}\|+\sigma_i) \log^{1/2}{n_i}
    	\lesssim q_i^{-1/2} n_{i,j}^{1/2}\|\mx_{\mathcal{U}_i\cap \mathcal{U}_j}\|_{2\to\infty} (\|\mpp^{(i)}\|+\sigma_i) \log^{1/2}{n_i}
    $$
    with high probability, provided that %$q_i n_i n_{i,j}\gtrsim \log n_i$ 
    $q_i n_i \gtrsim \log^2 n_i$.
\end{proof}

\begin{lemma}
	\label{lemma:i->global}
        {\em Suppose the entities in $\mathcal{U}_i$ are selected uniformly at random from all $N$ entities.
	Write the eigen-decomposition of $\mpp_{\mathcal{U}_i,\mathcal{U}_i}$ as $\mpp_{\mathcal{U}_i,\mathcal{U}_i}=\muu^{(i)}\mLambda^{(i)}\muu^{(i)\top}$ where 
	$\muu^{(i)}$ are the eigenvectors corresponding to the non-zero eigenvalues. Let $\lambda_{i,\max}$ and $\lambda_{i,\min}$ denote the largest and smallest non-zero eigenvalue of
        $\mpp_{\mathcal{U}_i,\mathcal{U}_i}$, respectively. Then for $n_i \gg \log N$ we have 
	$$
	\begin{aligned}
		&%\gtrsim ,% \gtrsim p_i\lambda_{\min},
		%\quad 
		\frac{n_i}{N}\lambda_{\min} \lesssim \lambda_{i,\min} \leq \lambda_{i,\max}\lesssim \frac{dn_i}{N}\lambda_{\max},
	% 	\quad \frac{\lambda_{i,\max}}{\lambda_{i,\min}}
	% \lesssim d, 
                  \quad \text{and } \|\muu^{(i)}\|_{2\to\infty}
	\lesssim\frac{d^{1/2}}{n_i^{1/2}}
	\end{aligned}
	$$
	with high probability.
        }
\end{lemma}
\begin{proof}
           %For $\lambda_{i,\min}$,  f
  If the entities in $\mathcal{U}_i$ are chosen uniformly at random from all $N$ entities then, by Proposition~S.3. in \cite{zhou2021multi} together with the assumption 
  $n_i\gg \log N$, we have
  \[\lambda_{\min}(\muu_{\mathcal{U}_i}^{\top} \muu_{\mathcal{U}_i}) \gtrsim \frac{n_i}{N} \]
	with high probability. As $\mpp^{(i)}=\muu_{\mathcal{U}_i}\mLambda \muu_{\mathcal{U}_i}^\top$, given the above bound we have
	\begin{equation}\label{eq:lambda_i min}
	\lambda_{i,\min}
		\geq \lambda_{\min}(\muu_{\mathcal{U}_i}^{\top}\muu_{\mathcal{U}_i})
		\cdot\lambda_d(\mLambda)
		%\cdot\lambda_{\min}(\muu_{\mathcal{U}_i}^\top)
		\gtrsim \frac{n_i}{N}\lambda_{\min},
	\end{equation}
        and hence
    %according to Lemma~S.4. in \cite{zhou2021multi}, we have
    $$
    \begin{aligned}
    	\|\muu^{(i)}\|_{2\to\infty}^2
    	\leq %\frac{d}{n_i}
    	%\cdot \frac{n_i}{N}
    	%\cdot \frac{N}{d}
    	\frac{\lambda_{\max}\|\muu\|_{2\to\infty}^2}{\lambda_{i,\min}}
    	\lesssim \lambda_{\max}\cdot \frac{d}{N}\cdot \frac{N}{n_i\lambda_{\min}}
    	%\lesssim \frac{d}{N}
    	%\cdot \frac{\lambda_{\max}}{\lambda_{\min}}\cdot \frac{N}{n_i}
    	\lesssim \frac{d}{n_i}. 
    \end{aligned}
    $$
    %, then the desired result of $\|\muu^{(i)}\|_{2\to\infty}$ is obtained.
    Finally, from
    $$
    \begin{aligned}
    	\|\muu_{\mathcal{U}_i}\|
    	\leq n_i^{1/2} \|\muu_{\mathcal{U}_i}\|_{2\to\infty}
    	\leq n_i^{1/2} \|\muu\|_{2\to\infty}
    	\lesssim \frac{d^{1/2}n_i^{1/2}}{N^{1/2}},
    \end{aligned}
    $$  %by Proposition~S.5. in \cite{zhou2021multi} 
    we obtain
    \begin{equation}\label{eq:lambda_i max}
    	\lambda_{i,\max}
    	= \|\mpp^{(i)}\|
    	%=\|\muu_{\mathcal{U}_i}\mLambda \muu_{\mathcal{U}_i}^\top\|
    	\leq \|\muu_{\mathcal{U}_i}\|^2\cdot \|\mLambda\|
    	\lesssim \frac{d n_i}{N}\lambda_{\max}
    \end{equation}
    as claimed. 
    % Then combining Eq.~\eqref{eq:lambda_i max} and Eq.~\eqref{eq:lambda_i min} we have the relationship between $\lambda_{i,\max}$ and $\lambda_{i,\min}$
    % $$
    % \begin{aligned}
    % 	\frac{\lambda_{i,\max}}{\lambda_{i,\min}}
    % 	\lesssim \frac{d n_i \lambda_{\max}}{N}
    % 	\cdot \frac{N}{n_i\lambda_{\min}}
    % 	=d\frac{\lambda_{\max}}{\lambda_{\min}}
    % 	\lesssim d
    % \end{aligned}
    % $$
    % with high probability.
%    as claimed. 
%\end{proof}
\end{proof}

\subsection{Proof of Lemma~\ref{lemma: || W(i0)T ... W(iL)-I||}}
\label{sec:technical_w_chain}
%\begin{lemma}
%\label{lemma: || W(i0)T ... W(iL)-I||}
%    Consider the setting of Theorem~\ref{thm:R(i0,...,iL)}, we have
%	$$
%    \begin{aligned}
%    	&\|\mw^{(i_0)\top}\mw^{(i_0,i_1)}\mw^{(i_1,i_2)}
%    	\cdots
%    	\mw^{(i_{L-1},i_L)}\mw^{(i_L)} -\mi\|
%    	\lesssim L\Big(\frac{(\|\mpp\|_{\max}+\sigma)^2N\log N}{pq\lambda_{\min}^2}+\frac{(\|\mpp\|_{\max}+\sigma)\log^{1/2}N}{p\breve p^{1/2}\tilde\sigma_{s,t}^{-1} \lesssim \max\Bigl\{\frac{(n_{i_0} q_{i_0} \mu_{i_0})^{1/2} \q^{1/2}\lambda_{\min}}\Big)
%    \end{aligned}
%    $$
%    with high probability.
%\end{lemma}

%\begin{proof}
    We proceed by induction on $L$. The case $L = 1$ follows from Lemma~\ref{lemma:|| W(i)T W(i,j) W(j)-I ||}.
    Now for any $L\geq 2$, define
    \[ \mt^{(L-1)} = \mw^{(i_0)\top}\mw^{(i_0,i_1)} \cdots \mw^{(i_{L-2},i_{L-1})}\mw^{(i_{L-1})}, \quad \mt^{(L)} = \mw^{(i_0)\top}\mw^{(i_0,i_1)} \cdots \mw^{(i_{L-1},i_{L})}\mw^{(i_{L})}. \]
     and suppose the stated bound holds for $L-1$, i.e.
     \begin{equation}\label{eq:WWW-I_L-1}
     \begin{aligned}
    	\|\mt^{(L-1)} -\mi\|
&\lesssim \sum_{\ell=1}^{L-1}\alpha_{{i_{\ell-1}},{i_{\ell}}}
    \end{aligned}
    \end{equation}
    with high probability. % Hence, recalling the conditions in Eq.~\eqref{eq:con_1}, we have
    % \begin{equation}\label{eq:WWW-I_|L-1|}
    % 	\begin{aligned}
    % 		\|\mt^{(L-1)} -\mi\|
    % 	=o(1)
    % 	\end{aligned}
    % \end{equation}
    % with high probability. 
    Next note that
    \begin{equation*}
    \begin{split}
    	\mt^{(L)}-\mi
    	&= \mt^{(L-1)} \mw^{(i_{L-1})\top}\mw^{(i_{L-1},i_L)}\mw^{(i_L)} -\mi\\
    	&= (\mt^{(L-1)}-\mi) \mw^{(i_{L-1})\top}\mw^{(i_{L-1},i_L)}\mw^{(i_L)} + (\mw^{(i_{L-1})\top}\mw^{(i_{L-1},i_L)}\mw^{(i_L)}-\mi),
    \end{split}
    \end{equation*}
    and hence, by combining Eq.~\eqref{eq:WWW-I_L-1} and Lemma~\ref{lemma:|| W(i)T W(i,j) W(j)-I ||} (for $\mw^{(i_{L-1})\top}\mw^{(i_{L-1},i_L)}\mw^{(i_L)}-\mi$),
 we obtain
    $$
    \begin{aligned}
    	\|\mt^{(L)} -\mi
    	\| 
    	\leq \|\mt^{(L-1)}-\mi\| + \|\mw^{(i_{L-1})\top}\mw^{(i_{L-1},i_L)}\mw^{(i_L)}-\mi\|
 \lesssim \sum_{\ell=1}^{L}\alpha_{{i_{\ell-1}},{i_{\ell}}}
    \end{aligned}
    $$
    with high probability.
%\end{proof}

\subsection{Extension to symmetric indefinite matrices}
We first state an analogue of Lemma~\ref{lemma:hat X(i)W(i)-X} for symmetric but possibly indefinite matrices. 
\begin{lemma}
\label{lemma:hat X(i)W(i)-X:indefinite}
{\em 
Fix an $i\in[K]$ and consider $\ma^{(i)}=(\mpp^{(i)}+\mn^{(i)})\circ \mathbf{\Omega}^{(i)}/q_i\in\mathbb{R}^{n_i\times n_i}$ as defined in Eq.~\eqref{eq:A(i)=...}.
    % where $\mpp^{(i)}$ is of rank $d$, the entries of $\me^{(i)}$ are independent sub-Gaussian noise with mean zero and sub-Gaussian norm $\|\me^{(i)}_{s,t}\|_{\psi_2}$, $q_i=\|\mathbf{\Omega}^{(i)}\|_F^2/n_i^2$. %Let $\sigma_i=\max_{s,t\in[|\mathcal{U}_i|]}\|\me^{(i)}_{s,t}\|_{\psi_2}, \lambda_{i,\max}=\lambda_1(\mpp^{(i)}), \lambda_{i,\min}=\lambda_d(\mpp^{(i)})$.
    Write the eigen-decompositions of $\mpp^{(i)}$ and $\ma^{(i)}$ as
    \[ \mpp^{(i)}=\muu^{(i)}\mLambda^{(i)}\muu^{(i)\top}, \quad \ma^{(i)}=\hat\muu^{(i)}\hat\mLambda^{(i)}\hat\muu^{(i)\top} + \hat\muu^{(i)}_\perp\hat\mLambda^{(i)}_\perp\hat\muu^{(i)\top}_\perp.\]
    Let $d_{+}$ and $d_{-}$ denote the number of positive and negative eigenvalues of $\mpp^{(i)}$, respectively, and denote $d = d_{+} + d_{-}$. 
    Let $\hat\mx^{(i)}=\hat\muu^{(i)}|\hat\mLambda^{(i)}|^{1/2}$ and $\mx^{(i)} = \muu^{(i)} |\bm{\Lambda}^{(i)}|^{1/2}$. 
%    Next define $\breve{\mw}^{(i)}$ as a minimizer of $\|\hat\mx^{(i)}\mo-\mx^{(i)}\|_F$ over all $d\times d$ orthogonal matrices $\mo$.
    Suppose that
    \begin{itemize}
      \item	$\muu^{(i)}$ is a $n_i\times d$ matrix with bounded coherence, i.e., % with high probability, i.e.,
      $$\|\muu^{(i)}\|_{2\to\infty}\lesssim d^{1/2}n_i^{-1/2}.$$
      %$$\mu_i=\frac{n_i}{d}\cdot\|\muu^{(i)}\|_{2\to\infty}^2 \lesssim 1.$$
      \item $\mpp^{(i)}$ has bounded condition number, i.e.,
      $$\frac{\sigma_{i,\max}}{\sigma_{i,\min}}\leq M$$
      for some finite constant $M > 0$; here $\sigma_{i,\max}$ and $\sigma_{i,\min}$ denote the largest and smallest %{\bf in magnitude} 
      non-zero %eigenvalues 
      singular values of $\mpp^{(i)}$, respectively. %, then set $\tau_i=\lambda_{i,\max}/\lambda_{i,\min}$. 
      %$$ \tau_i \lesssim 1.$$
      %\item The entries of $\mn^{(i)}$ are independent sub-Gaussian noise with mean zero and sub-Gaussian norm $\|\mn^{(i)}_{s,t}\|_{\psi_2}$. Let $\sigma_i=\max_{s,t\in[|\mathcal{U}_i|]}\|\mn^{(i)}_{s,t}\|_{\psi_2}$. 
      \item %We set $n_i\asymp pN$. % and $q_i\asymp q$. 
      The following conditions are satisfied.
      \begin{equation*}
         % \label{eq:cond_lemmaA1}
          q_i n_i\gtrsim \log^2{n_i} ,\quad
          \frac{(\|\mpp^{(i)}\|_{\max}+\sigma_i) n_i^{1/2}}{q_i^{1/2}\sigma_{i,\min}}=\frac{\gamma_i}{(q_in_i)^{1/2}\mu_i \log^{1/2} n_i}\ll 1.
      \end{equation*}
    \end{itemize}    
    Then there exists an matrix $\mathring{\mw}^{(i)} \in \mathcal{O}_{d} \cap \mathcal{O}_{d_{+}, d_{-}}$ such that
    \begin{equation}\label{eq:hatXW-X_indefinite}
    	\hat\mx^{(i)} \mathring{\mw}^{(i)}-\mx^{(i)}
	= \me^{(i)}\mx^{(i)}(\mx^{(i)\top}\mx^{(i)})^{-1} \mi_{d_+,d_-} +\mr^{(i)},
    \end{equation}
	where the remainder term $\mr^{(i)}$ satisfies
	$$
	\|\mr^{(i)}\|
	\lesssim \frac{(\|\mpp^{(i)}\|_{\max}+\sigma_i)^2 n_i}{q_i\sigma_{i,\min}^{3/2}}
		+\frac{(\|\mpp^{(i)}\|_{\max}+\sigma_i)\log^{1/2}n_i}{q_i^{1/2}\sigma_{i,\min}^{1/2}}
    $$
    with high probability. Recall that $\mathcal{O}_{d}$ and $\mathcal{O}_{d_+,d_{-}}$ denote the set of $d \times d$ orthogonal and {\em indefinite} orthogonal matrices, respectively. 
	If we further assume 
 \begin{equation*}
     %\label{eq:cond2_lemmaA1}
     \frac{(\|\mpp^{(i)}\|_{\max}+\sigma_i) n_i^{1/2}\log^{1/2}n_i}{q_i^{1/2}\sigma_{i,\min}}=\frac{\gamma_i}{(q_i n_i)^{1/2}\mu_i}\ll 1,
 \end{equation*}
 then we also have 
$$
\begin{aligned}
&\|\mr^{(i)}\|_{2\to\infty}
	\lesssim \frac{(\|\mpp^{(i)}\|_{\max}+\sigma_i)^2 n_i^{1/2}\log n_i}{q_i\sigma_{i,\min}^{3/2}}
		+\frac{ (\|\mpp^{(i)}\|_{\max}+\sigma_i)\log^{1/2} n_i}{q_i^{1/2}n_i^{1/2}\sigma_{i,\min}^{1/2}}
	=\frac{\gamma_i^2}{q_in_i \mu_{i}^{3/2}}
		+\frac{\gamma_i}{q_i^{1/2}n_i  \mu_i^{1/2}}, \\
        	&\|\hat\mx^{(i)} \mathring{\mw}^{(i)}-\mx^{(i)}\|_{2\to\infty}
	    \lesssim \|\me^{(i)}\muu^{(i)}\|_{2\to\infty}
	    \cdot \|(\mLambda^{(i)})^{-1/2}\|
	    +\|\mr^{(i)}\|_{2\to\infty}
	    \lesssim \frac{(\|\mpp^{(i)}\|_{\max}+\sigma_i) \log^{1/2} n_i}{q_i^{1/2}\sigma_{i,\min}^{1/2}}=\frac{\gamma_i}{(q_in_i  \mu_i)^{1/2}} 
\end{aligned}
$$
    with high probability. %Here $|\mu_i| = |\lambda_{i,min}|/n_i$.
     }
\end{lemma} The proof of Lemma~\ref{lemma:hat X(i)W(i)-X:indefinite}
follows the same argument as that for Lemma~\ref{lemma:hat X(i)W(i)-X}
and is thus omitted.  The main difference between the statements of
these two results is that Lemma~\ref{lemma:hat X(i)W(i)-X} bounds
$\hat{\mx}^{(i)} \mw^{(i)} - \mx_{\mathcal{U}_i}$ while
Lemma~\ref{lemma:hat X(i)W(i)-X:indefinite} bounds $\hat{\mx}^{(i)}
\mathring{\mw}^{(i)} - \mx^{(i)}$. If $\mpp$ is positive semidefinite then $\mx_{\mathcal{U}_i} = \mx^{(i)} \tilde{\mw}^{(i)}$ for some orthogonal matrix $\tilde{\mw}^{(i)}$
and thus we can combine both $\mathring{\mw}^{(i)}$ and $\tilde{\mw}^{(i)}$ into a single orthogonal transformation $\mw^{(i)}$; see the argument in Section~\ref{App: Lemma:hat X(i)W(i)-X}. In contrast, if $\mpp$ is indefinite then $\mx_{\mathcal{U}_i} = \mx^{(i)} \mq^{(i)}$ for some indefinite orthogonal $\mq^{(i)}$. As indefinite
orthogonal matrices behave somewhat differently from orthogonal matrices, it is simpler to consider $\mathring{\mw}^{(i)}$ and $\mq^{(i)}$ separately. 

\begin{lemma}\label{lemma:WTW-QQ:indefinite}
  {\em 
  Consider the setting of Theorem~\ref{thm:R(il)_ind}
   for just two overlapping submatrices $\ma^{(i)}$ and $\ma^{(j)}$.
    Let $\mathring{\mw}^{(i,j)} = 
    (\hat{\mx}_{\langle\mathcal{U}_{i} \cap \mathcal{U}_{j}\rangle}^{(i)\top} \hat{\mx}_{\langle\mathcal{U}_{i} \cap \mathcal{U}_{j}\rangle}^{(i)})^{-1} \hat{\mx}_{\langle\mathcal{U}_{i} \cap \mathcal{U}_{j}\rangle}^{(i)\top} \hat{\mx}_{\langle\mathcal{U}_{i} \cap \mathcal{U}_{j}\rangle}^{(j)} = (\hat{\mx}_{\langle\mathcal{U}_{i} \cap \mathcal{U}_{j}\rangle}^{(i)})^{\dagger} \hat{\mx}_{\langle\mathcal{U}_{i} \cap \mathcal{U}_{j}\rangle}^{(j)}$ be the least square alignment between $\hat{\mx}_{\langle\mathcal{U}_{i} \cap \mathcal{U}_{j}\rangle}^{(i)}$ and $\hat{\mx}_{\langle\mathcal{U}_{i} \cap \mathcal{U}_{j}\rangle}^{(j)}$. 
    Here $(\cdot)^{\dagger}$ denotes the Moore-Penrose pseudoinverse of a matrix.
    Also notice $\mq^{(i)} (\mq^{(j)})^{-1} = (\mx_{\langle\mathcal{U}_{i} \cap \mathcal{U}_{j}\rangle}^{(i)})^{\dagger} \mx_{\langle\mathcal{U}_{i} \cap \mathcal{U}_{j}\rangle}^{(j)}$ is the corresponding alignment between
    $\mx_{\langle\mathcal{U}_{i} \cap \mathcal{U}_{j}\rangle}^{(i)}$ and $\mx_{\langle\mathcal{U}_{i} \cap \mathcal{U}_{j}\rangle}^{(j)}$.  We then have
    \[ \begin{split} \|\mathring{\mw}^{(i)\top} \mathring{\mw}^{(i,j)} \mathring{\mw}^{(j)} - \mq^{(i)} (\mq^{(j)})^{-1}\| &\lesssim \frac{n_{i,j} \gamma_i  \gamma_j }{\theta_{i,j}(q_in_i  \mu_i)^{1/2} (q_jn_j  \mu_j)^{1/2} }
    +\frac{n_{i,j}\vartheta_{i,j}^{1/2}\gamma_i^2}{\theta_{i,j}^{3/2}q_in_i  \mu_i} 
    \\&+\frac{n_{i,j}^{1/2}\|\mx^{(i)}_{\langle\mathcal{U}_i\cap\mathcal{U}_j\rangle}\|_{2\to\infty}}{\theta_{i,j}}\Big(\frac{\gamma_i}{(q_in_i\mu_i)^{1/2}}+\frac{\vartheta_{i,j}^{1/2}\gamma_j}{\theta_{i,j}^{1/2}(q_jn_j\mu_j)^{1/2}}\Big)\\
    &+\frac{n_{i,j}^{1/2}}{\theta_{i,j}^{1/2}}\Bigl(\frac{\gamma_i^2}{q_i n_i \mu_i^{3/2}}+\frac{\vartheta_{i,j}^{1/2}\gamma_j^2}{\theta_{i,j}^{1/2}q_jn_j  \mu_j^{3/2}}\Bigr)=:\alpha_{i,j}
    \end{split}
    \]
    with high probability. 
  }
\end{lemma}

\begin{proof}
    For ease of notation, we will let $\hat{\my} := \hat{\mx}_{\langle\mathcal{U}_{i} \cap \mathcal{U}_{j}\rangle}^{(i)} \mathring{\mw}^{(i)}$ and $\hat{\mz} := \hat{\mx}_{\langle\mathcal{U}_{i} \cap \mathcal{U}_{j}\rangle}^{(j)} \mathring{\mw}^{(j)}$, and define $\my := \mx_{\langle\mathcal{U}_{i} \cap \mathcal{U}_{j}\rangle}^{(i)}$ and $\mz := \mx_{\langle\mathcal{U}_{i} \cap \mathcal{U}_{j}\rangle}^{(j)}$. We then have
    \begin{equation}\label{eq:WTW-QQ1}
    	\mathring{\mw}^{(i)\top} \mathring{\mw}^{(i,j)} \mathring{\mw}^{(j)}- \mq^{(i)} (\mq^{(j)})^{-1} = \hat{\my}^{\dagger} \hat{\mz} - \my^{\dagger} \mz 
    = (\hat{\my}^{\dagger} - \my^{\dagger})(\hat{\mz} - \mz) + \my^{\dagger}(\hat{\mz} - \mz) + (\hat{\my}^{\dagger} - \my^{\dagger})\mz,
    \end{equation}
    where the first equality follows from the fact that if $\mathbf{M}$ is a $p  \times d$ matrix and $\mathbf{W}$ is a $d \times d$ orthogonal matrix then $(\mathbf{M} \mathbf{W})^{\dagger} = \mathbf{W}^{\top} \mathbf{M}^{\dagger}
    $. 

    Now under the assumption $\frac{n_{i,j}^{1/2}\gamma_i}{(q_in_i  \mu_i)^{1/2}} \ll \theta_{i,j}^{1/2}$, we have
    %under our assumption, 
    $\|\hat{\my} - \my\| \ll
\sigma_{d}(\my) %= (\lambda_{d}(\mx_{\mathcal{U}_{i} \cap\mathcal{U}_{j}}^{\top}\mx_{\mathcal{U}_{i} \cap\mathcal{U}_{j}}))^{1/2}
$ with high probability, and hence both
$\my^{\top} \my$ and $\hat{\my}^{\top} \hat{\my}$ are invertible with
high probability.  We thus have $\hat{\my}^{\dagger} \hat{\my} = \my^{\dagger} \my = 
\mi$ and
 \[ \hat{\my}^{\dagger} - \my^{\dagger} = \hat{\my}^{\dagger}(\mathbf{I} - \my (\my^{\top} \my)^{-1} \my^{\top})
 - \hat{\my}^{\dagger} (\hat{\my} - \my) \my^{\dagger}.
    \]
    Furthermore, as $\mz$ and $\my$ share the same column space, we have $(\mi - \my (\my^{\top} \my)^{-1} \my^{\top}) \mz = \mathbf{0}$ so that
    \begin{equation}\label{eq:WTW-QQ2}
    	(\hat{\my}^{\dagger} - \my^{\dagger}) \mz = - \hat{\my}^{\dagger} (\hat{\my} - \my) \my^{\dagger} \mz = (\my^{\dagger} - \hat{\my}^{\dagger}) (\hat{\my} - \my) \my^{\dagger} \mz - \my^{\dagger} (\hat{\my} - \my) \my^{\dagger} \mz.
    \end{equation}
    Combining Eq.~\eqref{eq:WTW-QQ1} and Eq.~\eqref{eq:WTW-QQ2} we obtain
    \[
\mathring{\mw}^{(i)\top} \mt^{(i,j)} \mathring{\mw}^{(j)}
      - \mq^{(i)} (\mq^{(j)})^{-1} =(\hat{\my}^{\dagger} - \my^{\dagger})(\hat{\mz} - \mz) + \my^{\dagger}(\hat{\mz} - \mz) + (\my^{\dagger} - \hat{\my}^{\dagger}) (\hat{\my} - \my) \my^{\dagger} \mz - \my^{\dagger} (\hat{\my} - \my) \my^{\dagger} \mz, \]
    and hence 
    \begin{equation}\label{eq:WTW-QQb}
    	\begin{aligned}
\|
\mathring{\mw}^{(i)\top} \mathring{\mw}^{(i,j)} \mathring{\mw}^{(j)}
- \mq^{(i)} (\mq^{(j)})^{-1}\| &\leq \|\hat{\my}^{\dagger} - \my^{\dagger}\| \cdot \|\hat{\mz} - \mz\| + \|\my^{\dagger} (\hat{\mz} - \mz)\| \\ 
&+ \|\hat{\my}^{\dagger} - \my^{\dagger}\| \cdot \|\hat{\my} - \my\| \cdot \|\my^{\dagger}\|\cdot\| \mz\| + \|\my^{\dagger}(\hat{\my} - \my)\| \cdot \|\my^{\dagger} \|\cdot\|\mz\|.
    	\end{aligned}
    \end{equation}
      
    We now bound the terms on the right side of Eq.~\eqref{eq:WTW-QQb}.
    By Lemma~\ref{lemma:hat X(i)W(i)-X:indefinite} we have
    \begin{equation}\label{eq:WTW-QQb1}
    	\begin{aligned}
    		&\|\hat{\my} - \my\|\lesssim n_{i,j}^{1/2}\cdot \|\hat\mx^{(i)} \mathring{\mw}^{(i)}-\mx^{(i)}\|_{2\to\infty}
    		\lesssim \frac{n_{i,j}^{1/2}\gamma_i}{(q_in_i  \mu_i)^{1/2}},\\&
    		\|\hat{\mz} - \mz\|\lesssim n_{i,j}^{1/2}\cdot \|\hat\mx^{(j)} \mathring{\mw}^{(j)}-\mx^{(j)}\|_{2\to\infty}
    		\lesssim \frac{n_{i,j}^{1/2}\gamma_j}{(q_jn_j  \mu_j)^{1/2}} 
    	\end{aligned}
    \end{equation}
    with high probability.
    Recall that
    $$
    \begin{aligned}
    	\vartheta_{i,j} := \|\mx^{(j)}_{\langle\mathcal{U}_{i} \cap \mathcal{U}_{j}\rangle}\|^2=\|\mz\|^2, % = \|\mx_{\mathcal{U}_{i} \cap \mathcal{U}_{j}}^{(j)}\|^2, 
    \quad \theta_{i,j} := \lambda_d(\mx^{(i)\top}_{\langle\mathcal{U}_{i} \cap \mathcal{U}_{j}\rangle} \mx^{(i)}_{\langle\mathcal{U}_{i} \cap \mathcal{U}_{j}\rangle}) = \lambda_d(\my^{\top} \my)%,\\
    %\quad \vartheta_{j,i} := \sigma_1^2(\mx^{(j)}_{\mathcal{U}_{i} \cap \mathcal{U}_{j}})=\sigma_{1}^2(\mz), % = \|\mx_{\mathcal{U}_{i} \cap \mathcal{U}_{j}}^{(j)}\|^2, 
    %\quad \theta_{j,i} := \sigma_{d}^2(\mx^{(j)}_{\mathcal{U}_{i} \cap \mathcal{U}_{j}})
    %=\sigma_{d}^2(\mz).
    \end{aligned}
    $$
    Thus $\|{\my}^{\dagger}\|= 
     \theta_{i,j}^{-1/2}.$
     Because $\|\hat{\my} - \my\| \ll \lambda_d^{1/2}(\my^{\top} \my)$ with high probability, we also have
     \[
     \|\hat{\my}^{\dagger}\| = \lambda_{d}^{-1/2}(\hat\my^{\top} \hat{\my})\leq 2\lambda_{d}^{-1/2}(\my^{\top} \my)\lesssim 
     \theta_{i,j}^{-1/2}
     \]
     with high probability.
     Then by Theorem~4.1 in \cite{wedin-pseudoinverse} we have
     \begin{equation}\label{eq:WTW-QQb2}
     	\|\hat{\my}^{\dagger} - \my^{\dagger}\| \leq \sqrt{2} \|\hat{\my}^{\dagger}\| \cdot \|\my^{\dagger}\| \cdot \|\hat{\my} - \my\|
    \lesssim \frac{n_{i,j}^{1/2}\gamma_i}{\theta_{i,j}(q_in_i  \mu_i)^{1/2}}
     \end{equation}
    with high probability.
    Next, by Eq.~\eqref{eq:hatXW-X_indefinite}, we have
\begin{gather*}
         \|\my^{\dagger}(\hat{\mz} - \mz)\| \leq \|(\my^{\top} \my)^{-1}\| \cdot \|\my^{\top} \me^{(i)}_{\langle\mathcal{U}_i \cap \mathcal{U}_j\rangle} \mx^{(i)}(\mx^{(i)\top} \mx^{(i)})^{-1}\| + \|\my^{\dagger}\| \cdot \|\mr^{(i)}_{\langle\mathcal{U}_i \cap \mathcal{U}_j\rangle}\|, \\
         \|\my^{\dagger}(\hat{\my} - \my)\| \leq \|(\my^{\top} \my)^{-1}\| \cdot \|\my^{\top} \me^{(j)}_{\langle\mathcal{U}_i \cap \mathcal{U}_j\rangle} \mx^{(j)}(\mx^{(j)\top} \mx^{(j)})^{-1}\| + \|\my^{\dagger}\| \cdot \|\mr^{(j)}_{\langle\mathcal{U}_i \cap \mathcal{U}_j\rangle}\|. 
     \end{gather*}
     By similar results to Lemma~\ref{lemma:||E(i)_ij||,||UiT E(i)_ij Uj||F} we have
     \begin{gather*}
     \|\my^{\top} \me^{(i)}_{\langle\mathcal{U}_i \cap \mathcal{U}_j\rangle} \mx^{(i)}(\mx^{(i)\top} \mx^{(i)})^{-1}\|
	    \lesssim \frac{n_{i,j}^{1/2}\gamma_i\|\mx^{(i)}_{\langle\mathcal{U}_i\cap\mathcal{U}_j\rangle}\|_{2\to\infty}}{(q_in_i\mu_i)^{1/2}}, \\
         \|\my^{\top} \me^{(j)}_{\langle\mathcal{U}_i \cap \mathcal{U}_j\rangle} \mx^{(j)}(\mx^{(j)\top} \mx^{(j)})^{-1}\|
	    \lesssim \frac{n_{i,j}^{1/2}\gamma_j\|\mx^{(i)}_{\langle\mathcal{U}_i\cap\mathcal{U}_j\rangle}\|_{2\to\infty}}{(q_jn_j\mu_j)^{1/2}},\\
         \|\mr^{(i)}_{\langle\mathcal{U}_i\cap \mathcal{U}_j\rangle}\|
	    \lesssim n_{i,j}^{1/2}\Bigl(\frac{\gamma_i^2}{q_in_i  \mu_i^{3/2}}
    +\frac{\gamma_i}{q_i^{1/2}n_i  \mu_i^{1/2}}\Bigr), \\
         \|\mr^{(j)}_{\langle\mathcal{U}_i\cap \mathcal{U}_j\rangle}\|
	    \lesssim n_{i,j}^{1/2}\Bigl(\frac{\gamma_j^2}{q_jn_j  \mu_j^{3/2}}
    +\frac{\gamma_j}{q_j^{1/2}n_j  \mu_j^{1/2}}\Bigr)
     \end{gather*} 
     with high probability. Therefore we have
     \begin{equation}\label{eq:WTW-QQb3}
     	\begin{aligned}
     		\|\my^{\dagger}(\hat{\mz} - \mz)\|
     		\lesssim \frac{n_{i,j}^{1/2}\gamma_i\|\mx^{(i)}_{\langle\mathcal{U}_i\cap\mathcal{U}_j\rangle}\|_{2\to\infty}}{\theta_{i,j}(q_in_i\mu_i)^{1/2}}
     		+\frac{n_{i,j}^{1/2}}{\theta_{i,j}^{1/2}}\Bigl(\frac{\gamma_i^2}{q_i n_i \mu_i^{3/2}}
    +\frac{\gamma_i}{q_i^{1/2} n_i  \mu_i^{1/2}}\Bigr),\\
    \|\my^{\dagger}(\hat{\my} - \my)\|
     		\lesssim \frac{n_{i,j}^{1/2}\gamma_j\|\mx^{(i)}_{\langle\mathcal{U}_i\cap\mathcal{U}_j\rangle}\|_{2\to\infty}}{\theta_{i,j}(q_jn_j\mu_j)^{1/2}}
     		+\frac{n_{i,j}^{1/2}}{\theta_{i,j}^{1/2}}\Bigl(\frac{\gamma_j^2}{q_jn_j  \mu_j^{3/2}}
    +\frac{\gamma_j}{q_j^{1/2}n_j  \mu_j^{1/2}}\Bigr)
     	\end{aligned}
     \end{equation}
     with high probability.
     Combining Eq.~\eqref{eq:WTW-QQb}, Eq.~\eqref{eq:WTW-QQb1}, Eq.~\eqref{eq:WTW-QQb2}, and Eq.~\eqref{eq:WTW-QQb3} we have the desired error rate of $\|\mathring{\mw}^{(i)\top} \mathring{\mw}^{(i,j)} \mathring{\mw}^{(j)}
- \mq^{(i)} (\mq^{(j)})^{-1}\|$.
      \end{proof}
      
      We now prove Theorem~\ref{thm:R(il)_ind}. 
      In the case of a chain $(i_0, i_1, \dots, i_L)$ we have
    $$\mathring{\mw}^{(i_0)\top} \Bigl(\prod_{\ell=1}^{L} \mathring{\mw}^{(i_{\ell-1},i_{\ell})}\Bigr) \mathring{\mw}^{(i_L) }
    = \prod_{\ell=1}^{L} \mathring{\mw}^{(i_{\ell-1})\top} \mathring{\mw}^{(i_{\ell-1},i_{\ell})} \mathring{\mw}^{(i_{\ell})} = \prod_{\ell=1}^{L} \tilde{\mt}^{(i_{\ell-1}, i_{\ell})},$$
    where $\prod$ is matrix product and $\tilde{\mt}^{(i_{\ell-1}, i_{\ell})} := \mathring{\mw}^{(i_{\ell-1})\top} \mathring{\mw}^{(i_{\ell-1},i_{\ell})} 
    \mathring{\mw}^{(i_{\ell})}$.  %Let $\mq^{(i_0,\dots,i_L)} = \prod_{\ell=1}^{L} \mq^{(i_{\ell-1},i_{\ell})}$. Recall that $\mq^{(i,j)}$ is the mapping
    % from $\mx_{\mathcal{U}_{i} \cap \mathcal{U}_{j}}^{(i)}$ to $\mx_{\mathcal{U}_{i} \cap \mathcal{U}_{j}}^{(j)}$.
    % Now recall that $\mq^{(i)}$ and $\mq^{(j)}$ also satisfies
    % $\mx_{\mathcal{U}_i} = \mx^{(i)} \mq^{(i)}$ and $\mx_{\mathcal{U}_j} = \mx^{(j)} \mq^{(j)}$. 
    % As both $\mx_{\mathcal{U}_i \cap \mu_{j}}^{(i)}$
    % and $\mx_{\mathcal{U}_i \cap \mu_{j}}^{(j)}$ are of rank $d$, the transformation between $\mx_{\mathcal{U}_i \cap \mu_{j}}^{(i)}$ and $\mx_{\mathcal{U}_i \cap \mu_{j}}^{(j)}$ is unique and hence
%for indefinite orthogonal matrices $\mq^{(i)}$
%    and $\mq^{(j)}$. % Hence $\mx_{\mathcal{U}_i \cap \mu_{j}}^{(j)} = \mx_{\mathcal{U}_i \cap \mathcal{U}_j}^{(i)} \mq^{(i)}
%    (\mq^{(j)})^{-1}$.
%    We therefore have, 
    % \[ (\mx^{(i)\top} \mx^{(i)})^{-1} \mx^{(i)\top} \mx_{\mathcal{U}_i} (\mx_{\mathcal{U}_{j}}^{\top} \mx_{\mathcal{U}_{j}})^{-1} \mx_{\mathcal{U}_{j}}^{\top} \mx^{(j)} =
    %   \mq^{(i)} (\mq^{(j)})^{-1} = (\mx_{\mathcal{U}_i \cap \mu_{j}}^{(i)\top}\mx_{\mathcal{U}_i \cap \mu_{j}}^{(i)})^{-1} \mx_{\mathcal{U}_i \cap \mu_{j}}^{(i)\top}\mx_{\mathcal{U}_i \cap \mu_{j}}^{(j)}. % = \mq^{(i,j)}
    % \]
Furthermore, for any $2 \leq \ell \leq L$ we also have
\[ \mq^{(i_0)} (\mq^{(i_\ell)})^{-1}  = \prod_{k=1}^{\ell} \mq^{(i_{k-1})} \bigl(\mq^{(i_k)}\bigr)^{-1}=
\prod_{k=1}^{\ell} \bigl(\mx_{\mathcal{U}_{i_{k-1}} \cap \mathcal{U}_{i_{k}}}^{(i_{k-1})}\bigr)^{\dagger} 
\mx_{\mathcal{U}_{i_{k-1}} \cap \mathcal{U}_{i_{k}}}^{(i_{k})}.
%\mq^{(i_0)} \mq^{(i_{\ell})} = (\mx^{(i_0)\top} \mx^{(i_0)})^{-1} \mx^{(i_0)\top} \mx_{\mathcal{U}_{i_0}}
%(\mx_{\mathcal{U}_{i_{\ell}}}^{\top} \mx_{\mathcal{U}_{i_{\ell}}})^{-1} \mx_{\mathcal{U}_{i_{\ell}}}^{\top} \mx^{(i_{\ell})}.
\]
%\[ \mq^{(i_0,\dots,i_{\ell})} = \prod_{k=1}^{\ell} \mq^{(i_{k-1},i_{k})} = \mq^{(i_0)} \mq^{(i_{\ell})} = (\mx^{(i_0)\top} \mx^{(i_0)})^{-1} \mx^{(i_0)\top} \mx_{\mathcal{U}_{i_0}}
%(\mx_{\mathcal{U}_{i_{\ell}}}^{\top} \mx_{\mathcal{U}_{i_{\ell}}})^{-1} \mx_{\mathcal{U}_{i_{\ell}}}^{\top} \mx^{(i_{\ell})}.
%\]
Therefore, for $\ell \geq 2$, 
\begin{equation}\label{eq:l>=2}
\begin{aligned}
	 \Bigl(\prod_{k=1}^{\ell} 
    \tilde{\mt}^{(i_{k-1},i_{k})} \Bigr)
  - \mq^{(i_0)} (\mq^{(i_{\ell})})^{-1} &= \Bigl(\prod_{k=1}^{\ell-1} 
    \tilde{\mt}^{(i_{k-1},i_{k})} \Bigr)\tilde{\mt}^{(i_{\ell-1},i_{\ell})} -
     \mq^{(i_{0})} (\mq^{(i_{\ell-1})})^{-1} \mq^{(i_{\ell}-1)} (\mq^{(i_{\ell})})^{-1}  \\
     &= \Bigl(\Bigl(\prod_{k=1}^{\ell-1} 
    \tilde{\mt}^{(i_{k-1},i_{k})} \Bigr) - \mq^{(i_0)} (\mq^{(i_{\ell-1})})^{-1} \Bigr) (\tilde{\mt}^{(i_{\ell-1},i_{\ell})}  - \mq^{(i_{\ell-1})} (\mq^{(i_{\ell})})^{-1}) 
    \\ &+ 
\Bigl(\Bigl(\prod_{k=1}^{\ell-1} 
    \tilde{\mt}^{(i_{k-1},i_{k})} \Bigr) -  \mq^{(i_0)} (\mq^{(i_{\ell-1})})^{-1}
    \Bigr) \mq^{(i_{\ell-1})} (\mq^{(i_{\ell})})^{-1}
        \\ &+ 
    \mq^{(i_0)} (\mq^{(i_{\ell-1})})^{-1} \big(\tilde{\mt}^{(i_{\ell-1},i_{\ell})}  - \mq^{(i_{\ell-1})} (\mq^{(i_{\ell})})^{-1}\big).
\end{aligned}
\end{equation}    
    Define $a_{\ell} := \|
\prod_{\ell=1}^{L} \tilde{\mt}^{(i_{\ell-1}, i_{\ell})} - \mq^{(i_0)} (\mq^{(i_{\ell})})^{-1}\|$ for $1 \leq \ell \leq L$. We then have $a_{1} \leq \alpha_{i_0,i_1}$ with high probability by Lemma~\ref{lemma:WTW-QQ:indefinite}, and have
    \[ a_{\ell} \leq a_{\ell-1} \cdot \Bigl({\alpha_{i_{\ell-1},i_{\ell}}}  + \Bigl[\frac{\vartheta_{i_{\ell-1},i_{\ell}}}{\theta_{i_{\ell-1},i_{\ell}}}\Bigr]^{1/2} \Bigr) +  \varrho_{\ell-1} \cdot \alpha_{i_{\ell-1},i_{\ell}} 
     \]
     for $2 \leq \ell \leq L$ with high probability by Eq.~\eqref{eq:l>=2}, where 
     \[\varrho_{\ell-1} 
     := \bigl\|\mq^{(i_0)} \bigl(\mq^{(i_{\ell-1})}\bigr)^{-1} \bigr\|
     = \Big\|\prod_{k=1}^{\ell-1}\mq^{(i_{k-1})} \bigl(\mq^{(i_{k})}\bigr)^{-1} \Big\| 
     = \Big\|\prod_{k=1}^{\ell-1} \bigl(\mx_{\langle\mathcal{U}_{i_{k-1}}\cap \mathcal{U}_{i_k}\rangle}^{(i_{k-1})}\bigr)^{\dagger} \mx_{\langle\mathcal{U}_{i_{k-1}}\cap \mathcal{U}_{i_k}\rangle}^{(i_{k})} \Bigr\|.\] 

Finally, we have
\[ \begin{split} \hat{\mpp}_{\mathcal{U}_{i_0}, \mathcal{U}_{i_{L}}} - \mpp_{\mathcal{U}_{i_0}, \mathcal{U}_{i_{L}}}
&= \hat{\mx}^{(i_0)} \bigl( \prod_{\ell=1}^{L} 
  \mathring{\mw}^{(i_{\ell-1},i_{\ell})} \bigr) \mi_{d_+,d_-} \hat{\mx}^{(i_{L})\top} - \mx_{\mathcal{U}_{i_0}} \mi_{d_+,d_-} \mx_{\mathcal{U}_{i_L}}^{\top} \\
   &= \hat{\mx}^{(i_0)} \mathring{\mw}^{(i_0)}
\bigl(\prod_{\ell=1}^{L} 
    \tilde{\mt}^{(i_{\ell-1},i_{\ell})} \bigr) \mathring{\mw}^{(i_L)\top} \mi_{d_+,d_-} \hat{\mx}^{(i_{L})\top} - \mx^{(i_0)} \mq^{(i_0)} \mi_{d_+,d_-} \mq^{(i_L)\top} \mx^{(i_L)\top}
%\\ &=\hat{\mx}^{(i_0)} \mathring{\mw}^{(i_0)}
%\bigl(\prod_{k=1}^{L} 
%    \tilde{\mt}^{(i_{k-1},i_{k})} \bigr) \mi_{d_+,d_-} \mathring{\mw}^{(i_L)\top} \hat{\mx}^{(i_{L})\top} - \mx^{(i_0)} \mq^{(i_0)} \mi_{d_+,d_-} \mq^{(i_L)\top} \mx^{(i_L)\top}
\\ &= \hat{\mx}^{(i_0)} \mathring{\mw}^{(i_0)}
\bigl(\prod_{\ell=1}^{L} 
    \tilde{\mt}^{(i_{\ell-1},i_{\ell})} \bigr) \mi_{d_+,d_-} \mathring{\mw}^{(i_L)\top} \hat{\mx}^{(i_{L})\top} - \mx^{(i_0)} \mq^{(i_0)} (\mq^{(i_L)})^{-1} \mi_{d_+,d_-} \mx^{(i_L)\top},
% \\ &= \hat{\mx}^{(i_0)} \mt^{(i_0,i_1,\dots,i_L)} \mi_{d_+,d_-} \hat{\mx}^{(i_{L})\top} - 
% \mx^{(i_0)} \mq^{(i_0)} \bigl(\mq^{(i_L)}\bigr)^{-1} \mi_{d_+,d_-} \mx^{(i_L)} \\
% &= \hat{\mx}^{(i_0)} \mw^{(i_0)} \tilde{\mt}^{(i_0,i_1,\dots,i_L)} \mi_{d_{+},d_{-}} \mw^{(i_L)\top}  \hat{\mx}^{(i_L)} - 
% \mx^{(i_0)} \mq^{(i_0)} \bigl(\mq^{(i_L)}\bigr)^{-1} \mi_{d_+,d_-} \mx^{(i_L)}
\end{split}
\]
where the last equality follows from the facts that $\mathring{\mw}^{(i_L)} \in  \mathcal{O}_d\cap\mathcal{O}_{d_+,d_-}$ and $\mq^{(i_L)}\in  \mathcal{O}_{d_+,d_-}$. 
Let $\xi_i = \hat{\mx}^{(i)} \mathring\mw^{(i)} - \mx^{(i)}$ for $i \in \{i_0,i_L\}$. Following the same derivations as that for Eq.~\eqref{eq:Rij}, with Lemma~\ref{lemma:hat X(i)W(i)-X:indefinite} replacing Lemma~\ref{lemma:hat X(i)W(i)-X}, 
we obtain
\[\begin{split}
  \hat{\mpp}_{\mathcal{U}_{i_0}, \mathcal{U}_{i_{L}}} - \mpp_{\mathcal{U}_{i_0}, \mathcal{U}_{i_{L}}} &= \me^{(i_0)} \mx^{(i_0)} (\mx^{(i_0)\top} \mx^{(i_0)})^{-1} \mi_{d_+,d_-} \mq^{(i_0)} (\mq^{(i_L)})^{-1}  \mi_{d_+,d_{-}} \mx^{(i_L)\top} \\ & + \mx^{(i_0)}
     \mq^{(i_0)} (\mq^{(i_L)})^{-1} (\mx^{(i_L)\top}\mx^{(i_L)})^{-1} \mx^{(i_L)\top} \me^{(i_L)} \\ &+ \mr^{(i_0,i_L)} + \ms^{(i_0,i_1, \dots, i_L)}, \end{split}
\]
where we set
\begin{gather*}
    \mr^{(i_0,i_L)} = \mr^{(i_0)} \mq^{(i_0)} (\mq^{(i_L)})^{-1} \mi_{d_+,d_-} \mx^{(i_L)\top} 
    + \mx^{(i_0)} \mq^{(i_0)} (\mq^{(i_L)})^{-1} \mi_{d_+,d_{-}} \mr^{(i_L)\top}
    + \xi_{i_0} \mq^{(i_0)} (\mq^{(i_L)})^{-1} \mi_{d_{+},d_{-}} \xi_{i_L}^{\top}, \\
\ms^{(i_0,i_1,\dots,i_L)} = (\mx^{(i_0)} + \xi_{i_0}) (\prod_{\ell=1}^{L} 
    \tilde{\mt}^{(i_{\ell-1},i_{\ell})} - \mq^{(i_0)} (\mq^{(i_L)})^{-1}) \mi_{d_+,d_-}(\mx^{(i_L)} + \xi_{i_L})^{\top}.
\end{gather*}
Substituting the bounds for $\mr^{(i)}$ and $\xi_i$ in Lemma~\ref{lemma:hat X(i)W(i)-X:indefinite}, we obtain
\[ \begin{split} \|\mr^{(i_0,i_L)}\|_{\max} & \lesssim   \Big(\frac{\gamma_{i_0}^2}{q_{i_0} n_{i_0}  \mu_{i_0}^{3/2}}
	+ \frac{\gamma_{i_0} }{q_{i_0}^{1/2}n_{i_0} \mu_{i_0}^{1/2}}\Big)
	\|\mx^{(i_L)}\|_{2\to\infty} +
	\Big(\frac{\gamma_{i_L}^2}{q_{i_L} n_{i_L}  \mu_{i_L}^{3/2}}
	+\frac{\gamma_{i_L} }{q_{i_L}^{1/2}n_{i_L} \mu_{i_L}^{1/2}}\Big) 
	\|\mx^{(i_0)}\|_{2\to\infty} \\ &+ \frac{\gamma_{i_0} \gamma_{i_L}}{(q_{i_0} n_{i_0} \mu_{i_0})^{1/2}(q_{i_L} n_{i_L} \mu_{i_L})^{1/2}},
 \end{split}
 \]
 \[ \|\ms^{(i_0,i_1,\dots,i_L)}\|_{\max} \lesssim a_{L} \|\mx^{(i_0)}\|_{2 \to \infty} \cdot \|\mx^{(i_L)}\|_{2 \to \infty} \]
 with high probability under the assumption $\frac{\gamma_i}{(q_in_i\mu_i)^{1/2}}\lesssim \|\mx^{(i)}\|_{2\to\infty}$ 
for $i \in \{i_0,i_L\}$. Finally, as $\mx_{\mathcal{U}_i} = \mx^{(i)} \mq^{(i)}$ for $i \in \{i_0, i_L\}$, with $\mq^{(i)} \in \mathcal{O}_{d_+,d_-}$, we have after some straightforward algebra that
 \begin{gather*}
   \me^{(i_0)} \mx^{(i_0)} (\mx^{(i_0)\top} \mx^{(i_0)})^{-1} \mi_{d_{+},d_{-}}\mq^{(i_0)}  (\mq^{(i_L)})^{-1}  \mi_{d_+,d_{-}} \mx^{(i_L)\top}  = \me^{(i_0)} \mx_{\mathcal{U}_{i_0}} (\mx_{\mathcal{U}_{i_0}}^{\top} \mx_{\mathcal{U}_{i_0}})^{-1}  \mx_{\mathcal{U}_{i_L}}^{\top}, \\
   \mx^{(i_0)} \mq^{(i_0)} (\mq^{(i_L)})^{-1} (\mx^{(i_L)\top}\mx^{(i_L)})^{-1} \mx^{(i_L)\top} \me^{(i_L)} = \mx_{\mathcal{U}_{i_0}} (\mx_{\mathcal{U}_{i_L}}^{\top} \mx_{\mathcal{U}_{i_L}})^{-1} \mx_{\mathcal{U}_{i_L}}^{\top} \me^{(i_L)}
 \end{gather*}
 as desired. 

\subsection{Additional discussions on theoretical results}
\label{sec:add discussion}

We now provide further discussion on Assumption~\ref{ass:main} and the conditions in Theorems~\ref{thm:R(i,j)}, \ref{thm:R(i0,...,iL)}, and \ref{thm:normal}.
For the necessity of Eq.~\eqref{eq:assm1_con} in Assumption~\ref{ass:main}, consider the setting where $n_i \asymp n$, $q_i \asymp q$, $n_{ij} \asymp m$, $\mu_i \asymp 1$, $\sigma_i = O(1)$ and $\|\mathbf{P}\|_{\max} \asymp 1$ as discussed in Remark~\ref{rem:example1_setting}. This is a standard setting for many noisy matrix completion problems where the entries of $\mathbf{P}$ are bounded, the noise (while sub-Gaussian) has bounded Orlicz-2 norms, and each block of $\mpp$ has bounded condition number. Then both conditions in Eq.~\eqref{eq:assm1_con} simplify to $n q \gtrsim \log n$ (as $\|\mathbf{X}_{\mathcal{U}_i}\|_{2 \to \infty} = \|\mathbf{P}_{\mathcal{U}_i,\mathcal{U}_{i}}\|_{\max}^{1/2}$). The condition $n q \gtrsim \log n$ is very mild and is furthermore also necessary for matrix completion to work (see for example the discussion after Theorem~3.22 in \cite{chen2021spectral}). In summary our conditions in Eq.~\eqref{eq:assm1_con} match those in the existing literature for standard matrix completion.
Regarding Eq.~\eqref{eq:thm1_add_condition} in Theorem~\ref{thm:R(i,j)}, it ensures that the remainder term $\ms^{(i,j)}$ of the estimation error  is bounded by the two dominant terms, which is a mild and natural assumption in many settings.
  For example, continuing with the above setting, 
    the expression for $\alpha_{i,j}$ simplifies to 
    \[ \alpha_{i,j} \lesssim \frac{\log n}{n q} + \frac{\log^{1/2}{n}}{\sqrt{n q m}},\]
    in which case Eq.~\eqref{eq:thm1_add_condition} %(we number this inequality as Eq.~\eqref{eq:thm1_add_condition} in the revision)
     simplifies to
    \[ \frac{\sqrt{\log n}}{\sqrt{nq}} + \frac{1}{\sqrt{m}} \lesssim 1\]
    which is then satisfied for all $m \geq 1$ (assuming $n q \gtrsim \log n$).  
This discussion also extends to Theorem~\ref{thm:R(i0,...,iL)}, where Eq.~\eqref{eq:thm2_add_condition} %(which is the condition {\bf after} Eq.(3.8) in the original submission) 
becomes
 \[ L\Bigl(\frac{\sqrt{\log n}}{\sqrt{nq}} + \frac{1}{\sqrt{m}}\Bigr) \precsim 1.\]
 This condition is then satisfied whenever $n q = \Omega(L^2 \log n)$ and $m = \Omega(L^2)$. In other words, the number of matrices in the chain between $\ma^{(i_0)}$ and $\ma^{(i_L)}$ is not too large compared to $m$ (the overlap size) and $nq$ (the average number of non-zero entries in each row of the $\ma^{(i)}$).
Finally, Eq.~\eqref{eq:con_2} to Eq.~\eqref{eq:con_4} provide the technical conditions required for the central limit theorem stated in Theorem~\ref{thm:normal}; see Remark~\ref{rm:qleq1-c} for details.

 We next provide further discussion on $\mu_i$ and $\|\mx_{\mathcal{U}i}\|_{2\to\infty}$.
In our analysis, both $\mu_i$ and $\|\mx_{\mathcal{U}i}\|_{2\to\infty}$ appeared due to their roles in controlling related but {\em distinct} 
        quantities in our model. For example, $\mu_i$ is used to bound the subspace estimation error between $\hat{\muu}_i$ and $\muu_{i}$ when applying the Davis-Kahan theorem as $n_i \mu_i$ is the gap between the leading eigenvalues of $\mpp^{(i)}$ and the remaining eigenvalues. See, for example, the statement and proof of Lemma~\ref{lemma:||sinTheta(Uhat,U)||,Lambdahat}. The Davis-Kahan theorem, however, also depends on an upper bound for $\|\me^{(i)}\|$. As $\me^{(i)}$ accounts for both noise and missingness (see Eq.~\eqref{eq:E1E2}), the magnitude of the entries of $\me^{(i)}$ depend on those of $\mpp^{(i)}$ and hence standard matrix completion bounds typically depend on 
        $\|\mpp^{(i)}\|_{\max}$ (see for example Theorem~3.4 in \cite{abbe2020entrywise}). Now $\|\mpp^{(i)}\|_{\max} = \|\mx_{\mathcal{U}_i}\|_{2 \to \infty}^{2}$ whenever $\mpp^{(i)}$ is positive semidefinite, and this explains the need to also include $\|\mx_{\mathcal{U}_i}\|_{2 \to \infty}$ in our bounds for Lemma~\ref{lemma:||sinTheta(Uhat,U)||,Lambdahat}. The same observation also extends to other bounds in the paper, including those in the main theorems. 
        In addition, notice that we can potentially simplify our bounds to depend only on $\mu_i$ under Assumption~\ref{ass:main}. More specifically if
        $\|\muu_i\|_{2 \to \infty} \lesssim (d/n_i)^{1/2}$ (as in Assumption~\ref{ass:main}) then
            \[ \|\mx_{\mathcal{U}_i}\|_{2 \to \infty}^2 = \|\mpp^{(i)}\|_{\max} \lesssim d \frac{\lambda_{i,\max}}{n_i} \asymp \mu_i, \]
            provided that $d$ is fixed (not depending on $N$) and $\mpp^{(i)}$ has bounded condition number. Hence, under Assumption~\ref{ass:main}, all of our bounds can be simplified to depend only on $\mu_i$. However, if we make no assumptions on $\|\muu_i\|_{2 \to \infty}$ then we only have
        $\|\mx_{\mathcal{U}_i}\|_{2 \to \infty} \leq \|\mx_{\mathcal{U}_i}\|$ and $n_i \mu_i \leq \|\mx_{\mathcal{U}_i}\|$, but this does not yield any explicit 
        relationship between $\mu_i$ and $\|\mx_{\mathcal{U}_i}\|_{2 \to \infty}$. For conciseness we have chosen to keep the stated bounds as they are more general.

%\bibliographystyle{chicago}
%\bibliography{ref}
\end{document}
%%% Local Variables:
%%% mode: latex
%%% TeX-master: t
%%% End: